\documentclass[abstract, fontsize = 12pt]{scrartcl}

\newcommand{\Routputpath}{data/derived}
\usepackage{xr-hyper}
\usepackage{hyperref}
\usepackage{amsmath, amssymb, amsthm, mathtools}	
\usepackage{tabularx, array, booktabs, multirow, longtable}
\usepackage{etoolbox}
\usepackage{setspace}
\usepackage{textcomp}

\usepackage{microtype}
\makeatletter
\AtBeginEnvironment{table}{%
  \def\@floatboxreset{\reset@font\footnotesize\@setminipage}%
}
\patchcmd{\@xfloat}{\normalsize}{\selectfont}{}{}
\makeatother

\usepackage[style = authoryear, uniquename = false, hyperref = false]{biblatex}
\usepackage{enumitem}
\usepackage{tikz}
\usepackage{placeins}

\newtheorem{assumption}{Assumption}
\newtheorem{theorem}{Theorem}

\newtheorem{remark}{Remark}

\makeatletter
\newcommand{\leqnomode}{\tagsleft@true}
\newcommand{\reqnomode}{\tagsleft@false}
\makeatother

\newcommand{\abs}[1]{\left\lvert #1 \right\rvert}
\DeclarePairedDelimiter\norm\lVert\rVert

\newcommand{\E}{\mathbb{E}}

\newcommand{\reals}{\mathbb{R}}

\DeclareMathOperator{\cov}{cov}

\DeclareMathOperator{\argmin}{argmin}
\DeclareMathOperator{\diag}{diag}

\title{Confidence Set for Group Membership\thanks{The authors would like to thank St\`ephane Bonhomme (Editor), a co-editor, two anonymous referees, Otilia Boldea, Christoph Breunig, Le-Yu Chen, Elena Erosheva, Eric Gautier, Hidehiko Ichimura, Hiroaki Kaido, Hiroyuki Kasahara, Kengo Kato, Toru Kitagawa, Arthur Lewbel,  Artem Prokhorov, Adam Rosen, Myung Hwan Seo, Katsumi Shimotsu, Liangjun Su, Michael Vogt, Wendun Wang, Wuyi Wang, Martin Weidner, Yoon-Jae Whang and seminar participants at the Centre for Panel Data Analysis Symposium at the University of York, St Gallen, HKUST, SUFE, Sydney Econometric Reading Group, Xiamen, CUHK Workshop on Econometrics, Chinese Academy of Sciences, Asian Meeting of Econometric Society 2017, STJU, Workshop on Advances in Econometrics 2017 at Hakodate, SNU, Academia Sinica, ESEM 2017, Berlin Humboldt, CFE-CMStatistics 2017, International Panel Data Conference, Tsinghua, Fudan, Bonn, Hanyang, Lund, Exeter, Montreal, Cambridge, Panel data workshop Amsterdam, Barcelona GSE summer forum and 2021 Nanyang Econometric Workshop for valuable comments. 
We are particularly grateful to one anonymous referee whose comment led to correcting an error in a previous version of the manuscript.
Sophie Li and Heejun Lee provided excellent research assistance. A part of this research was done while Okui was at Kyoto University, Vrije Universiteit Amsterdam, NYU Shanghai, and Seoul National University. This work is supported by JSPS KAKENHI Grant number 15H03329, 16K03598, 22H00833, 22K20154 and 23H00804, and Jan Wallanders och Tom Hedelius stiftelse samt Tore Browaldhs stiftelse grant number P19-0079.}}

\author{Andreas Dzemski\thanks{Department of Economics, University of Gothenburg,
P.O. Box 640, SE-405 30 Gothenburg, Sweden.
Email: andreas.dzemski@economics.gu.se
} \ 
and Ryo Okui\thanks{Corresponding Author. Graduate School of Economics and Faculty of Economics, University of Tokyo,
7-3-1 Hongo, Bunkyou-ku, Tokyo, 113-0033, Japan.
		Email: okuiryo@e.u-tokyo.ac.jp;
}}

\date{\today}

\allowdisplaybreaks

\begin{document}

\maketitle

\onehalfspacing


\begin{abstract}
Our confidence set quantifies the statistical uncertainty from data-driven group assignments in grouped panel models. 
It covers the true group memberships jointly for all units with pre-specified probability and is constructed by inverting many simultaneous unit-specific one-sided tests for group membership. We justify our approach under $N, T \to \infty$ asymptotics using tools from high-dimensional statistics, some of which we extend in this paper. 
We provide  Monte Carlo evidence that the confidence set has adequate coverage in finite samples. 
An empirical application illustrates the use of our confidence set.

\vspace{4mm}
\par 
\vspace{1mm}
\end{abstract}

\section{Introduction}

Clustering units into discrete groups is one of the oldest problems in statistics \parencite{pearson1896mathematical}.  
It has received interest in the recent econometric literature on grouped panel models \parencite{lin2012estimation, bonhomme2015grouped, sarafidis2015partially, ando2016panel,vogt2015classification, su2016identifying, wang2016homogeneity, vogt2017clustering, lu2017determining, GuQuantileClustering, liu2020identification, wang2021identifying, mammen2022estimation,mehrabani2022estimation,mugnier2022simple,chetverikov2022spectral,yu2022group, mugnier2023nonlinear}.

In grouped panel models, a data-driven clustering algorithm is used to estimate a latent group structure.
As statistical procedures, clustering algorithms suffer from sampling errors and produce a noisy version of the true group structure.
The existing literature gives little guidance on how to assess clustering uncertainty in a given application.
Inferential theory for grouped panel models has focused on group characteristics, but is underdeveloped for assessing the uncertainty about individual group memberships \parencite[see e.g.][]{mclachlan2004finite}.

In this paper, we quantify the statistical uncertainty about the true group memberships in a grouped panel model. As far as we know, we are the first to propose and justify a rigorous frequentist method to evaluate clustering uncertainty. 

In grouped panel models, individual regression curves are heterogeneous and exhibit a grouped pattern. All units that belong to the same group face the same regression curve. 
Group memberships are unobserved and estimated by a clustering algorithm. Clustering uncertainty means that the algorithm may misclassify some units and assign them an incorrect regression curve. 

We propose a confidence set for group membership that quantifies clustering uncertainty jointly for all units in the panel. For a panel of $N$ units, an element of a joint confidence set is an $N$-dimensional vector that specifies a group assignment for every unit. Our confidence set gathers all $N$-dimensional vectors of group assignments that are ``not ruled out by the data'' and is guaranteed to contain the vector of the true group memberships with a pre-specified probability, say 95\%. 

\label{QE:referee1_2}
The latent groups in a grouped panel model have no natural labels and can only be identified up to a permutation. For notational convenience, we write our confidence set using an arbitrary ordering of the groups. We interpret this as a shorthand for linking units to regression curves. 
For example, suppose that there are a ``group 1'' with a slope coefficient of $0.2$ and a ``group 2'' with a slope coefficient of $0.8$. If our confidence set rules out that unit $i$ belongs to ``group 1'' then we take this to mean that unit $i$ does not face a slope coefficient of $0.2$. This interpretation does not depend on the ordering of the groups. Similarly, if our confidence set determines that units $i$ and $j$ belong to different groups, then we take this to mean that they face different slope coefficients. 

\label{QE2:referee1_1_intro}
This interpretation of a confidence set presumes that the data are rich enough to recover the group-specific coefficients. If the data do not provide any clue about the group-specific coefficients, then we are not able to statistically examine group memberships either. 
On the other hand, if the data are known to be ``very rich,'' and group memberships are guaranteed to be estimated correctly, then our confidence set is not needed.

Settings between these two extreme scenarios are relevant in practice. 
In Monte Carlo experiments calibrated to their empirical application, \textcite{bonhomme2015grouped} find that units are frequently misclassified, whereas group-specific coefficients are estimated precisely (see Table~S.III in their supplemental appendix). 
\textcite{dzemskiokui2021convergence} provide a theoretical framework to explain this observation.\footnote{%
For mathematical convenience, the theoretical analysis of clustered panel models often proceeds under assumptions that rule out any misclassification in the asymptotic limit \parencite[see e.g.][]{bonhomme2015grouped,vogt2015classification}.}
They assume that unit $i$ faces an error term with unit-specific variance $\sigma_i^2$. Units with small $\sigma_i$ are classified reliably. Units with large $\sigma_i$ are potentially misclassified. \textcite{dzemskiokui2021convergence} show that the group-specific coefficients can be estimated consistently if the proportion of potentially misclassified units is sufficiently small. This is the main setting that we have in mind for applications of our confidence set. 
It is less restrictive than assuming, as is typically done in the literature, that both group-specific coefficients and all group assignments can be reliably estimated. 

\label{QE:intro theta gi inference}Our empirical application illustrates how our confidence set provides new economic insights. 
We follow \textcite{wang2019heterogeneous} who estimate a panel model that allows the effect of a minimum wage on unemployment to vary between US states, depending on the assignment of each state to one of four latent groups. 
This group assignment is potentially estimated with error. We use our confidence set to identify, up to a small pre-specified error probability, states without clustering uncertainty. In terms of the framework discussed above, these are states with low values of $\sigma_i$. 
For these states, we can identify the state-specific effects of the minimum wage. 

Our confidence set can also be used to enhance a plot of the estimated groups by adding information about clustering uncertainty.  We illustrate this in our empirical application.
Providing a plot of the estimated groups is standard practice.\footnote{For example, see Figure~2 in \textcite{wang2019heterogeneous}, Figure~2 in \textcite{bonhomme2015grouped} and Figure~6 in \textcite{wang2016homogeneity}.} This is true even in applications where the group structure is considered merely a nuisance parameter.\footnote{\label{QE:interactive fixed-effects}For example, time-varying unobserved heterogeneity can be controlled by imposing a latent group structure with group-specific time fixed-effects. In this context, the heterogeneity in the fixed-effect is a nuisance parameter similar to the interacted fixed effects in, for example, \textcite{moonweidner2019}.}

An alternative to our frequentist approach is Bayesian inference. Fully parametric grouped-panel models are finite mixtures models that can be estimated by the EM algorithm \parencite{dempster1977maximum}. The E-step of the EM algorithm computes unit-wise posterior probabilities for group membership. These are valid if the units in the panel are independently drawn from the assumed parametric distribution. Our frequentist approach is more general. We do not assume a parametric distribution of the error term and show that our approach is valid for error distributions in a broad nonparametric class. We also allow for cross-sectional dependence, non-random patterns of heteroscedasticity, and non-random group assignments. 
Another advantage of our procedure is that it allows for joint inference on the $N$-dimensional vector of all group memberships, whereas unit-wise posterior probabilities only address uncertainty about the group membership of a single unit.

Quantifying the uncertainty about the true group structure is a high-dimensional inference problem. The $N$-dimensional vector of true group memberships is high-dimensional since its size grows as $N \to \infty$. 
To construct a confidence set for this high-dimensional parameter, we invert a test of the many moment inequalities that characterize group memberships. Testing many moment inequalities is the problem considered in \textcite{chernozhukov2013testing}. Their test is based on a single test statistic, whereas ours combines many simultaneous group membership tests for individual units. The advantage of our approach is that it can be inverted without running a computationally infeasible exhaustive search over the space of all partitions.

Our confidence set is valid in the presence of weak-dependence and serial correlation, whereas \textcite{chernozhukov2013testing} assume independence over time. We account for serial correlation by using a heteroskedasticity--and--autocorrelation--robust (HAC) variance estimator \parencite{Andrews91,NeweyWest87} when constructing the unit-wise test statistics.

The asymptotic analysis of our procedure accounts for the high-dimensional nature of our setting and allows for weak time-dependence. We built on \textcite{chang2022central} who provide results for HAC estimators and a high-dimensional central limit theorem for dependent data. Our setting requires extending their approach in different ways. Specifically, we use a Nasarov-type inequality \parencite{chernozhukov2016central} to control for the effect of parameter estimation. Moreover, we develop a regularization scheme for the HAC estimates. This regularization scheme allows us to control the estimation error in our data-driven critical values by using a comparison bound based on \textcite{li2002normal}.

We provide several extensions of our method. First, we suggest alternative critical values. These are slightly conservative but much easier to compute than our benchmark critical values. Second, we show that HAC estimation is not needed if there is no serial correlation. In this case, test statistics based on the usual variance estimators provide valid confidence sets and are easier to implement than those based on the HAC estimator. Third, we propose a two-step method, called unit selection, to shrink the cardinality of the confidence set when there are many units for which clustering uncertainty is low. The two-step procedure identifies and discards such units. We then construct a confidence set for the remaining units. The method accounts for errors in the unit selection and provides a valid confidence set. This additional error control can inflate the confidence set if unit selection does not eliminate sufficiently many units.

The remainder of this paper is organized as follows. Section~\ref{sec:model} introduces the grouped panel model. Section~\ref{sec:cs} defines our confidence set for group membership. Section~\ref{sec:theory} proves the asymptotic validity of our confidence sets. Section~\ref{sec:extensions} discusses the extensions of our method. Section~\ref{sec:application} provides an empirical application. Section~\ref{sec:simulations} presents Monte Carlo simulations that investigate the validity and power of our confidence set based on simulation designs inspired by our empirical application. Section~\ref{sec:conclusion} concludes.

An R package implementing the methods proposed in this paper is included in the replication package that is published together with this article.

\section{\label{sec:model}Model}

We observe panel data $(y_{it}, w_{it}', x_{it}')'$ for units $i = 1, \dotsc, N$ and time periods $t = 1, \dotsc, T$, where $y_{it}$ is a scalar dependent variable and $w_{it}$ and $x_{it}$ are covariate vectors. 
Unit $i$ belongs to group $g_i^0 \in \mathbb{G} = \{1, \dotsc, G\}$. Group memberships are unobserved. 
The data are generated from the model
\begin{align}
	\label{eq:linear-model}
	y_{it} = w_{it}'\theta^w + x_{it}'\theta_{g_i^0} + \sigma_i v_{it},  
\end{align}
where $v_{it}$ is a noise term with variance one and is potentially serially correlated, and $\sigma_i$ is a latent heteroscedasticity parameter. The slope coefficient $\theta^w$ on $w_{it}$ is common to all units. The slope coefficient on $x_{it}$ is group-specific and given by $\theta_g$ for units $i$ belonging to group $g \in \mathbb{G}$. 

\label{QE:uncorrelatedness}
We assume that the regressors $(w_{it}', x_{it}')'$ are uncorrelated with the contemporaneous error term $\sigma_i v_{it}$, 
\begin{align}
	\label{eq:regressors are orthogonal}
	\E \left[ (x_{it}', w_{it}')' \sigma_i v_{it} \right] = 0. 
\end{align}
This assumption does not rule out predetermined regressors such as lagged dependent variables. 

Different estimation strategies for estimating the common and group-specific coefficients ($\theta^w$ and $\theta_g$) have been proposed in the literature. For example, \textcite{bonhomme2015grouped} estimate slope coefficients $\hat{\theta}^w, \hat{\theta}_1, \dotsc, \hat{\theta}_G$ and group memberships $\hat{g}_1, \dotsc, \hat{g}_N$ simultaneously by solving the least-squares problem 
\label{QE:referee2 typos1}
\begin{align*}
	(\hat{\theta}^w, \hat{\theta}_1, \dotsc, \hat{\theta}_G, \hat{g}_1, \dotsc, \hat{g}_N) = \argmin_{\substack{\theta^w, \theta_1, \dotsc, \theta_G \\ g_1, \dotsc, g_N \in \mathbb{G}}} \frac{1}{N T} \sum_{i = 1}^N \sum_{t = 1}^T \left(y_{it} - w_{it}'\theta^w - x_{it}'\theta_{g_i} \right)^2
\end{align*}
via the \emph{kmeans} algorithm. The choice of squared loss is justified under the orthogonality condition \eqref{eq:regressors are orthogonal}. The estimators in \textcite{su2016identifying, wang2016homogeneity} augment a squared loss function by a penalization scheme that imposes the grouped structure.

\label{QE:choose G strategies}
We assume that the number of groups $G$ is either pre-specified or consistently estimated. 
Consistent estimates of $G$ can be obtained, for example, by using the information criteria proposed by \textcite{bonhomme2015grouped} or \textcite{su2016identifying} or by employing the testing procedure in \textcite{lu2017determining}.

\begin{remark}
For ease of exposition, we describe our procedures for balanced panels. The extension to unbalanced panels is straightforward. 
\end{remark}

\section{\label{sec:cs}Confidence set for group membership}

This section describes our confidence set for group membership. First, we define the formal requirements for an asymptotically valid confidence set for group membership. Second, we show that each group allocation corresponds to a set of moment inequalities and that a confidence set can be obtained by inverting a test of these inequalities. Finally, we introduce the test statistic and critical values. 

\subsection{Definition}
\label{QE:definition CS}
A joint confidence set of group membership at confidence level $1 - \alpha$ is a random set $\widehat{C}_{\alpha}$ of vectors in $\mathbb{G}^{N}$ that satisfies
\begin{align}
	\label{eq:definition:CS:coverage}
	\liminf_{N, T \to \infty} \inf_{P \in \mathbb{P}_N} P\left( \left\{g_i^0\right\}_{1 \leq i \leq N} \in \widehat{C}_{\alpha}\right) \geq 1 - \alpha, 
\end{align}
where $\mathbb{P}_N$ is a class of data-generating processes. If we observe $ \{g_i\}_{1 \leq i \leq N} \in \widehat{C}_{\alpha}$, then the group structure that assigns unit $i = 1, \dotsc, N$ to group $g_i$ is not ruled out by the data at confidence level $1 - \alpha$. Inequality \eqref{eq:definition:CS:coverage} ensures that the confidence set is asymptotically valid in the sense that it rules out the population partition at most with probability $\alpha$. 

We impose uniform validity over sequences $P_N$ on $\mathbb{P}_N$. Changing the data-generating process along the asymptotic sequence allows for $\sigma_i$ that diverge as $N \to \infty$ for some units $i$, rendering these units potentially misclassified in the limit. Data-generating processes that are constant in $N$ cannot model asymptotic clustering uncertainty \parencite{dzemskiokui2021convergence}.

We construct our joint confidence set by combining unit-wise marginal confidence sets. This approach is computationally simple and can be tabulated and visualized easily.
The marginal confidence set for unit $i$ is computed by inverting a test for group membership 
\begin{align*}
	\widehat{C}_{\alpha, N, i}  =& \left\{
		g \in \mathbb{G} : \widehat{T}_i(g) \leq \hat{c}_{\alpha, N, i} (g)
	\right\} \cup \left\{\hat{g}_i\right\},
\end{align*}
where $\hat{g}_i$ denotes an estimator of the group membership of unit $i$, $\widehat{T}_i(g)$ is a test statistic and $\hat{c}_{\alpha, N, i}$ is a unit-specific and data-dependent critical value. Test statistic and critical value are defined in Sections~\ref{sec:test stat} and \ref{sec:crit val}, respectively.

By explicitly adding $\hat{g}_i$ to the marginal confidence set, we guarantee that the joint confidence set is never empty and can always be interpreted as containing the estimated group structure padded by a margin of error. 
For the typical unit, inverting the test already includes the unit's estimated group membership in its marginal confidence set.\footnote{In our simulations, we find that the probability of the test rejecting the estimated group membership to be very close to, but not equal to, zero (see Supplemental Material~\ref*{sec:testing ghat}).}

Our joint confidence set is given by the Cartesian product of the unit-wise confidence sets: 
\begin{align*}
	\widehat{C}_{\alpha} = \bigtimes_{1 \leq i \leq N} \widehat{C}_{\alpha, N, i}.
\end{align*}
We use Bonferroni correction to control dependence between units and compute each unit-wise marginal confidence set at a nominal level of $1 - \alpha/N$. 

In principle, it is possible to construct a joint confidence set directly without first computing unit-wise confidence sets. This can be accomplished by inverting a \emph{joint} test for group membership. Testing group memberships for all groups simultaneously avoids the possible power loss from Bonferroni correction. However, inverting the test to obtain the confidence set requires testing all $G^N$ possible groupings. This task is computationally infeasible unless $N$ is very small. In contrast, our approach carries out only $G \times N$ tests and is feasible even if $N$ is large.

An additional advantage of Bonferroni correction is that it produces confidence sets that are easy to report and to interpret. The joint confidence set can be fully described by reporting the marginal unit-wise confidence sets without enumerating all $N$-dimensional vectors contained in $\widehat{C}_{\alpha}$.
Interpreting a potentially large collection of such high-dimensional vectors would be challenging.

\label{QE:new Theorem 1}
The Bonferroni correction renders our confidence set robust to any kind of cross-sectional dependence. This correction is only minimally conservative if the unit-wise confidence sets are approximately independent, i.e., if 
\begin{align}
	\label{eq:units approximately independent}
	P \left(
		\{g_i^0\}_{1 \leq i \leq N} \in \widehat{C}_\alpha
	\right) = \prod_{1 \leq i \leq N} P \left( 
		g_i^0 \in \widehat{C}_{\alpha, N, i} \right) + p_\Delta,
\end{align}
where $p_\Delta$ is a small number. Approximate independence holds if units are cross-sectionally independent and $N$ and $T$ are large enough to estimate the group-specific coefficients precisely. 

\begin{theorem}
\label{thm:Bonferroni:bounds}
Let $0 < \alpha < 1$ and suppose that the unit-wise confidence sets satisfy
\begin{align}
	\label{eq:unit-wise CS have correct coverage}
	P \left( g_i^0 \in \widehat{C}_{\alpha, N, i} \right) = 1 - \alpha/N
\end{align} 
and that condition \eqref{eq:units approximately independent} holds.
Then 
\begin{align*}
	\alpha - \frac{\alpha^2}{2}  - p_\Delta\leq 
	P \left( \{g_i^0\}_{1 \leq i \leq N} \notin \widehat{C}_{\alpha}\right) \leq 
	\alpha - \frac{\alpha^2}{2} \left(1 - \frac{\alpha}{3} + \frac{1}{N} \left(1 - \frac{\alpha}{N} \right)^{-2} \right) + p_\Delta.
\end{align*}
\end{theorem}
\noindent For example, suppose that $N \geq 8$ and $1 - \alpha = 0.9$ and that conditions \eqref{eq:units approximately independent} and \eqref{eq:unit-wise CS have correct coverage} hold with $p_\Delta \approx 0$. Theorem~\ref{thm:Bonferroni:bounds} implies that the Bonferroni correction inflates the joint coverage probability by only about $0.5$-$0.55\%$.

\label{QE:K units}
\begin{remark}
	Our approach can be adapted to produce joint confidence sets for subsets of units. For $1 \leq K < N$, suppose that the researcher is only interested in the $K$ first units. A joint confidence set for these units is given by 
	\begin{align*}
	\bigtimes_{1 \leq i \leq K}\left\{
		g \in \mathbb{G} : \widehat{T}_i(g) \leq \hat{c}_{\alpha, K, i} (g)
	\right\} \cup \left\{\hat{g}_i\right\}.	
	\end{align*}
	For inference on the first $K$ units, this confidence set is more powerful than the joint confidence set for all units. This is because less Bonferroni correction is needed when testing fewer units and therefore $\hat{c}_{\alpha, K, i} (g) < \hat{c}_{\alpha, N, i} (g)$ (see the definition of the critical values in Section~\ref{sec:crit val} below).
\end{remark}

\begin{remark}
	For applications where we are interested in inference on a single pre-specified unit i, a confidence set for the group membership of $i$ is given by $\widehat{C}_{\alpha, 1, i}$.
\end{remark}

\subsection{Motivation of our test of group membership}

Our approach to testing the group membership hypothesis $H_0: g_i^0  = g$ is based on 
\begin{align*}
	d_{it}(g, h) = \frac{1}{2} \left( \left(y_{it} - w_{it}'\theta^w - x_{it}'\theta_{g}\right)^2 - \left(y_{it} - w_{it}' \theta^w - x_{it}'\theta_{h}\right)^2
	+ \left(x_{it}' \left(\theta_{g} - \theta_{h}\right)\right)^2 \right).
\end{align*}
The first two terms on the right-hand side are squared residuals representing the fit of assigning unit $i$ to group $g$ and the fit of assigning unit $i$ to group $h$, respectively.
The third term applies moment re-centering and ensures that $d_{it}(g, h)$ has mean zero under the null hypothesis.
This can be seen by re-writing $d_{it}(g, h)$ as 
\begin{align*}
	d_{it}(g, h) = - \sigma_i v_{it} x_{it}'\left(\theta_{g} - \theta_{h} \right)
	+ \left(\theta_{g} - \theta_{g_{i}^0} \right)' x_{it} x_{it}' \left(\theta_{g} - \theta_{h} \right).
\end{align*}
The first term on the right-hand side has mean zero under the orthogonality assumption~\eqref{eq:regressors are orthogonal}. If the null hypothesis is true, i.e., if $g_i^0 = g$, then the second term vanishes and $\sum_{t = 1}^T \E [d_{it} (g, h)] = 0$ for all $h \neq g$.

\label{QE:referee1 comment7}
If $\sum_{t = 1}^T \E[x_{it} x_{it}']$ has full rank and the null hypothesis is false, then $\sum_{t=1}^T \E \left[ d_{it}(g, h) \right] > 0$ for some $h \neq g$.
The strict inequality holds because for $h = g_i^0 \in \mathbb{G}\setminus \{g\}$ the second term in $d_{it}(g, h)$ is a quadratic form with strictly positive mean. 

In summary, testing $H_0: g_i^0 = g$ is equivalent to testing
\begin{align*}
	H'_0: \sum_{t = 1}^T  \E [d_{it}(g, h)] = 0 \quad \text{for all $h \in \mathbb{G} \setminus \{g\}$}
\end{align*}
against 
\begin{align*}
H_1': \text{there exists $h \in \mathbb{G} \setminus \{g\}$ such that} \quad  \sum_{t = 1}^T  \E [d_{it}(g, h)] > 0.	
\end{align*}
This is a one-sided significance test for a vector of moments.

\subsection{Test statistic} \label{sec:test stat}

Our test statistic $\widehat{T}_i (g)$ for unit $i=1, \dotsc, N$ is the maximum of $(G-1)$ statistics $\widehat{D}_i(g, h)$ that test a hypothesized group membership $g$ against alternative group assignments $h \neq g$:
\begin{align*}
	\widehat{T}_i(g) =& \max_{h \in \mathbb{G} \setminus \{g\}} \widehat{D}_{i}(g, h).
\end{align*}
$\widehat{D}_i (g, h)$ tests group $g$ against group $h$ based on the restriction $\sum_{t=1}^T \E[d_{it}(g, h)]=0$ from the previous section and is equal to the $t$-statistic
\begin{align*}
	\widehat{D}_{i}(g, h) = \frac{\sum_{t = 1}^T \hat{d}_{it}(g, h) /\sqrt{T}}{ \sqrt{ \widehat{\Xi}_i (g,h,h)} },
\end{align*}
where $\hat{d}_{it}(g, h)$ is a sample counterpart of $d_{it} (g,h)$ that replaces the true slope coefficients $\theta^w$ and $\theta_g$ by their estimated values $\hat{\theta}^w$ and $\hat{\theta}_g$,  
\begin{align*}
	\hat{d}_{it}(g, h) = \frac{1}{2} \left( \left(y_{it}  -w_{it}'\hat{\theta}^w -  x_{it}'\hat{\theta}_{g}\right)^2 - \left(y_{it} - w_{it}'\hat{\theta}^w - x_{it}'\hat{\theta}_{h}\right)^2
	+ \left(x_{it}' \left(\hat{\theta}_{g} - \hat{\theta}_{h}\right)\right)^2 \right),
\end{align*}
and $\widehat{\Xi}_i (g,h,h)$ is an estimator of the long-run variance of $d_{it}(g, h)$. 

The long-run variance estimator is kernel-based as in \textcite{Andrews91} and \textcite{NeweyWest87}. In order to define the estimator, let $\widehat{H}_{ij} (g, h, h')$ denote the sample covariance of order $j$ between $\hat{d}_{it}(g, h)$ and $\hat{d}_{it} (g, h')$, 
	\begin{align}
		\label{eq:H_hat}
		\widehat{H}_{ij} (g,h,h')= \frac{1}{T}\sum_{t=|j|+1}^T  \left( \hat{d}_{i, t+ \min(0,j)} (g, h) - \bar{\hat{d}}_{i} (g, h) \right)\left( \hat{d}_{i,t- \max(0,j)} (g, h') - \bar{\hat{d}}_{i} (g, h') \right), 
	\end{align}
where $\bar{\hat{d}}_{i} (g, h) = \frac{1}{T}\sum_{t = 1}^T \hat{d}_{it} (g, h)$. 
The long-run variance-covariance estimator is given by 
\begin{align*}
\widehat{\Xi}_{i} (g,h,h')= \sum_{j=-T+1}^{T-1} K \left( \frac{j}{\kappa_N} \right) \widehat{H}_{ij} (g,h,h'),
\end{align*}
where $K(\cdot)$ is a kernel function and $\kappa_N$ is a bandwidth parameter.
In our simulation studies and empirical application, we use the quadratic spectral (QS) kernel \parencite[see, for example,][]{Andrews91} given by 
\begin{align*}
	K_{QS} (x) = \frac{25}{12\pi^2 x^2} \left( \frac{\sin (6 \pi x /5)}{6\pi x /5} - \cos (6 \pi x /5) \right).
\end{align*}
We select a data-driven bandwidth $\hat{\kappa}_N$ by using the following algorithm adapted from \textcite{chang2022central}:
\begin{enumerate}[label=Step~\Alph*:, align = left]
    \item For $g \in \mathbb{G}$ and $h \in \mathbb{G} \setminus \{g\}$, take each unit $i$ such that $\hat{g}_i = g$ and fit an $AR(1)$-model on $(\hat{d}_{it} (g, h))_{t = 1}^T$. Let 
    $\hat{\rho}_{igh}$ denote the estimated autoregressive coefficients and $\hat{\sigma}^2_{igh}$ the estimated variance of the innovation.
    \item 
    Select the bandwidth 
    \begin{align*}
        \hat{\kappa}_N = 1.3221\left( T \times \frac{\sum_{i = 1}^N \sum_{g \in \mathbb{G}} \sum_{h \in \mathbb{G} \setminus \{g\}} \hat{\rho}_{igh}^2  \hat{\sigma}_{igh}^4  / (1 - \hat{\rho}_{igh}^2 )^8 }{\sum_{i = 1}^N \sum_{g \in \mathbb{G}} \sum_{h \in \mathbb{G} \setminus \{g\}} \hat{\sigma}_{igh}^4  / (1 - \hat{\rho}_{igh}^2 )^4}
        \right)^{1/5}.
    \end{align*}
\end{enumerate}

\subsection{\label{sec:crit val}Critical values}

The critical value $\hat{c}_{\alpha, N, i}(g)$ is computed from the multivariate $t$-distribution (MVT) in $(G-1)$ dimensions.\footnote{%
Using the multivariate $t$-distribution instead of a Gaussian distribution improves the finite sample performance of our confidence set if $T$ is small.
} In the case of two groups ($G=2$), this distribution is equal to Student's $t$-distribution and our critical value is given by 
\begin{align*}
	c_{\alpha, N, i} (g)
	= \sqrt{\frac{T}{T - 1}} t_{T- 1}^{-1}\left(1 - \frac{\alpha}{N}\right),
\end{align*}
where $t_{T - 1}^{-1}(p)$ denotes the $p$-quantile of Student's $t$-distribution with $(T - 1)$ degrees of freedom. 
This critical value is straightforward to evaluate in most statistical software packages.

For $G \geq 3$, the computation of the critical value is more involved and requires the estimation of unit-specific correlation matrixes. In Section~\ref{sec:extensions}, we discuss a conservative approximation of the critical value that is easy to implement and independent of the data. 

The critical value is given by
\begin{align*}
	\hat{c}_{\alpha, N, i} (g) =& c_{\alpha, N} \left( \rho(\widehat{\Omega}_i(g), \epsilon_N) \right) = \sqrt{\frac{T}{T - 1}} \left(t_{\max, \rho(\widehat{\Omega}_i(g), \epsilon_N), T - 1}\right)^{-1} \left(1 - \frac{\alpha}{N} \right),
\end{align*}
where $t_{\max, \Omega, T-1}$ denotes the distribution function of the maximal entry of a centered random vector with multivariate $t$-distribution with scale matrix $\Omega$ and $(T-1)$ degrees of freedom, $\widehat{\Omega}_i (g)$ is an estimator of the correlation matrix of the moment inequalities, $\rho$ is a regularization function and $\epsilon_N$ is a regularization parameter.\footnote{%
\label{footnote:normalprob}
The distribution function of the multivariate $t$-distribution can be efficiently approximated by modern algorithms \parencite{genz1992numerical}. Implementations exist for Stata \parencite{grayling2016mvtnorm} and R \parencite{GenzRPackage}.}

To define the estimated correlation matrix $\widehat{\Omega}_i (g)$ for given $g$ and $i$, map $j, j' = 1, \dotsc, G-1$ to $h, h' \in \mathbb{G}$ such that $h$ and $h'$ give the $j$th and $j'$th element of the vector $\mathbb{G} / \{g\}$, respectively.\footnote{\label{footnote:j to h}More formally, $h = j - \mathbf{1}_{\{j > g\}}$ and $h' = j' - \mathbf{1}_{\{j' > g\}}$, where $\mathbf{1}_{\{\cdot \}}$ is the indicator function.} $\widehat{\Omega}_i (g)$ is given by the  $(G - 1) \times (G - 1)$ matrix with entry $(j, j')$ equal to
\begin{align*}
	\left(\widehat{\Omega}_i(g)\right)_{j, j'} =
	\frac{ \widehat{\Xi}_i (g,h,h') }{\sqrt{ \widehat{\Xi}_i (g,h,h)   \widehat{\Xi}_i (g,h',h')}}.
\end{align*}
For a $(G-1) \times (G-1)$ correlation matrix $\Omega$, the regularization function $\rho$ is defined as 
\begin{align*}
    \rho (\Omega, \epsilon) = \diag^{-1/2 }\left(\Omega + \epsilon^*(\Omega, \epsilon) I_{G-1}\right) \left(\Omega + \epsilon^*(\Omega, \epsilon) I_{G-1}\right) \diag^{-1/2} \left(\Omega + \epsilon^*(\Omega, \epsilon) I_{G-1}\right),
\end{align*}
where $I_{G-1}$ is the identity matrix in $\reals^{G-1}$, $\diag(A)$, for a matrix $A$, returns a diagonal matrix of the same dimension as $A$ with the diagonal entries equal to the diagonal entries of $A$, and 
\begin{align*}
    \epsilon^*(\Omega, \epsilon) = \max\{0, \epsilon - (1 - \max_{i < j} \Omega_{ij} )\}. 
\end{align*}
We set the regularization parameter equal to $\epsilon_N = 0.01$. For robustness, we also conduct simulations with other values of $\epsilon_N$ and find that the exact choice of $\epsilon_N$ is not crucial to the validity of our method. Regularization is a technical tool needed to prove the asymptotic validity of our confidence set. 

Our regularization scheme bounds pairwise correlations away from one but may output a singular matrix. A related regularization scheme in \textcite{andrews2012inference} bounds the resulting matrix away from singularity.

\section{\label{sec:theory}Asymptotic validity of our confidence set}

Our asymptotic framework is of the long-panel variety and takes both the number of units $N$ and the number of time periods $T$ to infinity. In particular, we consider asymptotic sequences in which $T= T(N)$, where $T(\cdot)$ is increasing, but its exact form is unspecified except for conditions given in the statement of the theorems. In many panel data sets, the number of units far exceeds the number of time periods. We replicate this feature along the asymptotic sequence by allowing $N$ to diverge at a much faster rate than $T$.

We consider a sequence $\mathbb{P}_N$ of classes of probability measures. All our theoretical results hold uniformly over the sequence $\mathbb{P}_N$. For a probability measure $P$, let $\E_P$ denote the expectation operator that integrates with respect to measure $P$. 
The parameters $\theta^w$, $\{\theta_g\}_{g \in \mathbb{G}}$ and $\sigma_i$ depend potentially on $N$. For notational convenience, we keep this dependence implicit.

The number of latent groups $G$ is fixed and does not depend on $N$.

To state our assumptions, we define the matrix $\Omega_i (g_i^0)$ which provides the population counterpart to  $\widehat{\Omega}_i (g_i^0)$. 
For unit $i$, define the population long-run covariance of the averages of $d_{it} (g_i^0, h)$ and $d_{it}(g_i^0, h')$ as
\begin{align*}
    \Xi_{i} (h,h')= \sum_{j=-T+1}^{T-1} H_{ij} (h,h'),
\end{align*}
where 
\begin{align*}
    H_{it} (h, h') = \frac{1}{T}\sum_{t=|j|+1}^T  \E_P [ d_{i,t+\min (0,j)} (g_i^0, h) d_{i,t-\max (0,j)} (g_i^0, h') ].
\end{align*}
For $i = 1, \dotsc, N$, $\Omega_i (g_i^0)$ is the $(G-1) \times (G-1)$ correlation matrix with entries 
\begin{align*}
	\left(\Omega_i(g_i^0)\right)_{j, j'} =
	\frac{ \Xi_i (h,h') }{\sqrt{  \Xi_i (h,h)   \Xi_i (h',h')}}, 
\end{align*} 
where the convention relating $(j, j')$ to $(h, h')$ is defined in footnote~\ref{footnote:j to h}.

We now state our assumptions for the asymptotic validity of our joint confidence set. 
\begin{assumption}
	\label{assumption:basic}
	\begin{enumerate}
		 \item 
		 \label{assum:exogeneity}
		 (Regressors are uncorrelated with the error term) 
		 The linear panel model satisfies the orthogonality assumption~\eqref{eq:regressors are orthogonal}.
		\item 
		\label{assumption:max:G_consistent}
		(Number of groups)
		The number of latent groups $G$ is fixed along the asymptotic sequence.
		Its estimator $\widehat{G}_N$ satisfies 
		\begin{align*}
			\lim_{N \to \infty} \sup_{P \in \mathbb{P}_N} P \left(
				\widehat{G}_N \neq G
			\right) = 0.
		\end{align*}
		\item 
		\label{assum:est}
		(Estimation of auxiliary parameters)
		There are vanishing sequences $r_{\theta, N}$ and $a_{\theta, N}$ such that 
		\begin{align*}
			\sup_{P \in \mathbb{P}_N} P \bigg (
				\norm{\hat{\theta}^w - \theta^w} \vee
				\max_{g \in \mathbb{G}} \norm{\hat{\theta}_g - \theta_g} > r_{\theta, N}
			\bigg) \leq a_{\theta, N}. 
		\end{align*}	
		\item 
		\label{assumption:max:min_eigenvalue}
		(Full rank) Let $\lambda_{\min}(\cdot)$ denote the smallest eigenvalue of its argument. There is a finite constant $C_{\lambda} > 0$ such that 
		\begin{align*}
			\inf_{P \in \mathbb{P}_N} \min_{1 \leq i \leq N} \lambda_{\min} \left( \frac{1}{T}\sum_{t=1}^{T} \sum_{s=1}^T \E_P [v_{it} v_{is} x_{it} x_{is}' ] \right) \geq C_{\lambda}^{-1}.	
		 \end{align*}
		\item 
		\label{assum:gseparate}
		(Group separation) 
		$
			\iota_N \equiv \min_{g \in \mathbb{G}} \min_{h \in \mathbb{G} \setminus \{g\} } \Vert \theta_g - \theta_h \Vert > 0.
		$
		 \item 
		 \label{assum:tail}
		 (Exponential tail bound)
		 There exist constants $a$ and $d_1 >1$ such that $P(|Z| > z)< \exp(- (z/a)^{d_1} )$ for sufficiently large $z$, where $Z$ is any component of the random vector $(x_{it}', w_{it}', v_{it})$. 
		 \item
		 \label{assum:mixing}
		 (Mixing sequence)
		 For every $i = 1, 2, \dotsc$, the sequence  $(x_{it}', w_{it}', v_{it})_{t =1}^{T}$ is an $\alpha$-mixing sequence with mixing coefficient $\alpha_i$ satisfying $\sup_i \alpha_i[k] \leq \exp (1- bk^{d_2})$ for some $b>0$ and $d_2 >0$.
		 \item \label{assum:stationarity}
		 (Stationarity) For every $i = 1, 2, \dotsc$,  
		 $\{x_{it}', w_{it}', v_{it}\}_{t = 1}^T$ is a strictly stationary time series.
		 \item \label{assum:Omega negative correlations} 
		 (Correlation of moment inequalities) $\min_{1 \leq i \leq N} \min_{1 \leq j, j' \leq G-1} \left( \Omega_i (g_i^0) \right)_{j, j'} > -1 + 4 \epsilon_N/3$, where $\epsilon_N$ is the regularization parameter for $\rho(\cdot, \cdot)$. 
	\end{enumerate}
	\end{assumption}
Assumption~\ref{assumption:basic}.\ref{assum:exogeneity} requires the regressors to be uncorrelated with the contemporaneous error term. It is a much weaker exogeneity assumption than, for example, strict exogeneity and allows for a rich set of regressors, including lagged dependent variables. 
Assumption \ref{assumption:basic}.\ref{assumption:max:G_consistent} requires the number of groups to be consistently estimated. This condition is weak and can be guaranteed by using an appropriate procedure for choosing the number of groups \parencite[e.g.][]{lu2017determining,vogt2017clustering}. 

\label{QE:referee1_1_technical}
Assumption \ref{assumption:basic}.\ref{assum:est} requires the estimators $\hat{\theta}^w$ and $\hat{\theta}_g$ to be consistent for $\theta^w$ and $\theta_g$, respectively, at a rate that vanishes as fast as or faster than $r_{\theta, N}$. Theorem~\ref{thm:MAX-dep-lv} below require $r_{\theta, N}$ to vanish faster than $(T \log N)^{-1/2}$. Therefore, $\hat{\theta}^w$ and $\hat{\theta}_g$ have to converge faster than $T^{-1/2}$. Since $T^{-1/2}$ is the rate obtained by estimators based on time-series regression within units, it is important to choose estimators that also exploit cross-sectional variation such as the \emph{kmeans} estimator \parencite{bonhomme2015grouped} or the estimators in \textcite{su2016identifying,wang2016homogeneity}.
These estimators are known to be $\sqrt{NT}$-consistent under assumptions that rule out misclassification in the limit. \textcite{dzemskiokui2021convergence} study consistency of the estimated coefficients under weaker assumptions. They distinguish between units with $\sigma_i \prec \sqrt{T/\log N}$ that can be classified reliably and noisy units with $\sigma_i \succeq \sqrt{T/\log N}$ that are potentially misclassified in the limit. They show that the \emph{kmeans} estimator is $\sqrt{NT}$-consistent if the proportion of noisy units is sufficiently small. In Section~\ref{sec:simulations}, we complement this theoretical result by simulating designs where \emph{kmeans} misclassifies units but still recovers group-specific coefficients at a rate that is faster than $\sqrt{T}$.

The full-rank condition in Assumption \ref{assumption:basic}.\ref{assumption:max:min_eigenvalue} ensures that the denominator of $\widehat{D}_{i} (g,h)$ is not too close to zero. The term in the minimum eigenvalue function is the population long-run covariance matrix of $v_{it}x_{it}$. The assumption restricts $v_{it}$ (which has variance one) but does not limit the magnitude of the error term $\sigma_i v_{it}$ in our panel model \eqref{eq:linear-model}. It allows conditional heteroskedasticity of $v_{it}$ given $x_{it}$.

Assumption~\ref{assumption:basic}.\ref{assum:gseparate} maintains that groups are unique in the sense that there are not two groups that share the same coefficient values. The minimal distance between any two groups is measured by $\iota_N$ and is allowed to vanish asymptotically, provided that it satisfies additional rate conditions stated below. Vanishing group separation is an asymptotic modeling device to study settings where groups are distinct but difficult to distinguish. Most existing results that establish asymptotic properties of estimators of the group-specific coefficient assume strict group separation, i.e., that $\iota_N$ is bounded away from zero. This makes it difficult to verify Assumption~\ref{assumption:basic}.\ref{assum:est} if $\iota_N \to 0$. This difficulty can be overcome by using our result for \emph{kmeans} estimation in Supplemental Material~\ref*{appendix:weak separation}, which gives a consistency rate under vanishing group separation. 

Assumptions~\ref{assumption:basic}.\ref{assum:tail}-\ref{assumption:basic}.\ref{assum:stationarity} restrict the distribution of the time series $\{x_{it}', w_{it}', v_{it}\}_{t = 1}^T$. Assumption~\ref{assumption:basic}.\ref{assum:tail} imposes exponential decay of the tails of the marginal distributions. Assumption~\ref{assumption:basic}.\ref{assum:mixing} restricts the time-series dependence of the data by imposing exponential decay of the mixing coefficients. This assumption rules out processes with long memory. Assumption~\ref{assumption:basic}.\ref{assum:stationarity} requires the time series to be stationary. We use the stationarity assumption primarily to show that certain long-run variances are bounded. It can be replaced by other conditions that bound the long-run variances.

\label{QE:discuss regularization technical}
Assumption~\ref{assumption:basic}.\ref{assum:Omega negative correlations} rules out that the correlation matrix $\Omega_i (g_i^0)$ contains entries that are too close to negative one. This assumption does not rule out singularity of $\Omega_i (g_i^0)$. Singularity occurs mechanically in our setting whenever $G > p + 1$, where $p$ is the dimension of $x_{it}$.\footnote{Our testing approach is designed to be able to handle singular correlation matrices. In contrast, e.g., the quasi-likelihood ratio statistic used in \textcite{kudo1963multivariate} is not defined for singular correlation matrices.}
No restrictions are placed on positive correlations, which is important in settings with vanishing group separation, where groups $h$ and $h'$ have similar coefficients and hence highly positively correlated moment inequalities.\footnote{In Supplemental Material~\ref*{appendix:weak separation}, we develop a theoretical framework based on local alternatives to study settings with very similar groups.}
Our regularization approach controls positive correlations that are close to one.  

We now introduce the last assumption that restricts the choice of kernel function. The validity of this assumption is under complete control of the researcher and does not depend on the underlying data. It is satisfied by the QS kernel \parencite[see][]{Andrews91}.

\begin{assumption}
	\label{assum:kernel}
	The kernel function $K(\cdot): \mathbb{R} \to [-1,1]$ is continuously differentiable with bounded derivatives on $\mathbb{R}$ and satisfies (i) $K(0)=1$, (ii) $K(x) = K(-x)$ for any $x \in \mathbb{R}$, (iii) $\int_{-\infty}^{\infty} |K(x)| dx < \infty$, and (iv) $K(x) \precsim |x|^{-\vartheta} $ as $|x| \to \infty$ for some constant $\vartheta >1$.
\end{assumption}
%

The following theorem establishes the validity of our joint confidence set. In the limit, it covers the true group membership at least with pre-specified probability $1- \alpha$.
\begin{theorem}
	\label{thm:MAX-dep-lv}
	Let $\mathbb{P}_N$ be a class of probability measures that satisfy Assumption~\ref{assumption:basic} with identical choices of $a$, $b$, $d_1$, and $d_2$. 
    Assume that there are finite constants $0 < \delta_1 < \delta_2$ and $1 < k_1 \leq k_2$ such that $T^{\delta_1} \leq N \leq o(1)T^{\delta_2}$ and $(\log N)^{-k_2} \leq \epsilon_N \leq (\log N)^{-k_1}$. 
%
    Let Assumption~\ref{assum:kernel} hold with $\kappa_N \asymp T^{\rho}$, where $0 <\rho < (\vartheta -1)/(3\vartheta -2)$.\footnote{For sequences $a_N$ and $b_N$ we write $a_N \asymp b_N$ if and only if $a_N = O(b_N)$ and $b_N = O (a_N)$.}
    In addition, assume 
    \begin{align}
		\label{eq:rate condition r_theta}
        r_{\theta, N}  \sqrt{ T \log N} = o \left(1 \wedge \iota_N \min_{1 \leq i \leq N} \sigma_i\right).
    \end{align}
	Then, 
	\begin{align*}
		\liminf_{N, T \to \infty} \inf_{P \in \mathbb{P}_N} P \left( \left\{g_i^0\right\}_{1 \leq i \leq N} \in  \widehat{C}_{\alpha} \right) \geq 1 - \alpha. 
	\end{align*}
\end{theorem}

In addition to the abovementioned assumptions, this theorem introduces some rate conditions. The first condition restricts the relative magnitudes of $T$ and $N$. It still accommodates both ``short panels'' where $T$ is small relative to $N$ and ``long panels'' where $T$ is large relative to $N$. The second condition requires the regularization parameter $\epsilon_N$ to vanish at a sufficiently slow rate. The third rate condition controls the rate at which the bandwidth sequence $\kappa_N$ diverges. This rate condition is automatically satisfied for the QS kernel if the bandwidth is chosen by the procedure described in Section~\ref{sec:test stat}. Finally, condition \eqref{eq:rate condition r_theta} restricts the rate of convergence of the estimators $\hat{\theta}^w$ and $\hat{\theta}_g$. 

We now discuss the latter condition in the context of two examples. For both examples, assume that $\sigma_i$ is bounded away from zero uniformly over units $i$.  
For the first example, groups are well-separated, i.e., $\iota_N$ is a positive constant and does not depend on $N$. 
Existing results for clustering with well-separated groups suggest $r_{\theta, N} = (NT)^{-1/2} \zeta_N$ with $\zeta_N \to \infty$ \parencite{bonhomme2015grouped,su2016identifying,wang2016homogeneity}.\footnote{The sequence $\zeta_N$ can go to infinity at any slow rate.} 
Under this convergence rate, condition~\eqref{eq:rate condition r_theta} is equivalent to $\log N / N = o(1)$ and hence trivially satisfied.
\label{QE:theory validty CS weak separation}
For the second example, group separation is vanishing with $\iota_N = T^{-1/2 + e}$ for $0 < e < 1/2$. In Theorem~\ref*{thm:weak-sep_gamma-beta-w} in Supplemental Material~\ref*{sec:weak-sep_theory}, we show $r_{\theta, N} = (NT)^{-1/2} \zeta_N$, under some technical conditions. Then, condition~\eqref{eq:rate condition r_theta} becomes $T^{1 - 2e} \log N / N = o(1)$. This condition restricts the relative magnitudes of $N$ and $T$. In particular, the weaker group separation is, i.e., the closer $e$ is to zero, the larger $N$ has to be relative to $T$. It is sufficient that $N \geq T^{\delta_1}$ with $\delta_1 > 1 - 2e$. 

\label{QE:caveat kmeans misclassification}
A caveat to the calculations in the previous paragraph is that the results for the rate of consistency of the \emph{kmeans} estimator rely on assumptions that rule out misclassification in the limit. In the discussion following Theorem~\ref*{thm:weak-sep_gamma-beta-w} in Supplemental Material~\ref*{sec:weak-sep_theory}, we indicate how this limitation can be overcome based on the approach in \textcite{dzemskiokui2021convergence}, but leave a formal proof to future research.  

\label{QE:fixed effects}
\begin{remark}
	\label{remark:FE transformation}
	Our method can be extended to panel models with unit fixed effects by interpreting model \eqref{eq:linear-model} as representing the fixed-effect transformed model. This application of our procedure can be theoretically justified but is not covered by the asymptotic results in this section. Extending the results to the fixed-effect model requires some (possibly lengthy) modifications to our arguments.	

	To derive the asymptotic behavior of the test statistic under the transformed model, we have to examine $-\sigma_i \dot v_{it} \dot x_{it}' (\theta_g - \theta_h)$, where $\dot v_{it} = v_{it} - \sum_{s=1}^T v_{is}/ T$ and $\dot x_{it} = x_{it} - \sum_{s=1}^T x_{is}/ T$. When $x_{it}$ is strictly exogenous, i.e., $\E[ v_{it} \mid x_{i1}, \dots, x_{iT}]=0$, the orthogonality condition~\eqref{eq:regressors are orthogonal} holds also in the transformed model. However, the mixing condition in Assumption~\ref{assumption:basic}.\ref{assum:mixing} may not be satisfied. 
	For example, if $v_{it}$ is i.i.d. over time, the correlation between $\dot v_{it}$ and $\dot v_{is}$ for $t\neq s$ is $1/T$ regardless of the distance between $s$ and $t$. The mixing condition requires serial correlation between distant time periods to vanish and is therefore not satisfied for $\dot v_{it}$ under fixed $T$ as required by Assumption~\ref{assumption:basic}.\ref{assum:mixing}.  Mixing still holds asymptotically as $T\to \infty$, and our proof has to be modified to show that this is sufficient.  When $x_{it}$ is merely predetermined, i.e. $\E[ v_{it} \mid x_{i1}, \dots, x_{it}]=0$, the transformed model is not guaranteed to satisfy \eqref{eq:regressors are orthogonal} because 
	\begin{align*}
		\sum_{t=1}^T \E \left[ \dot v_{it} \dot x_{it}'\right] =& \E \left[ \sum_{t=1}^T v_{it} x_{it} - T \left(\sum_{s=1}^T v_{is}/ T\right) \left(\sum_{s=1}^T x_{is}'/ T\right) \right] 
		\\
		=& T^{-1} \E\left[ \left(\sum_{s=1}^T v_{is}\right) \left(\sum_{s=1}^T x_{is}'\right) \right] 
	\end{align*}
	may not be zero.
	In many cases, including the panel AR(1) model with $x_{it} = y_{i,t-1}$, $  (\sum_{s=1}^T v_{is}/ T) (\sum_{s=1}^T x_{is}'/ T) = O_P (1/T)$ for each $i$ \parencite[see, e.g.,][]{Nickell1981,hahn2002asymptotically} and the expected uniform order of this term is $O_P (T^{-1} \log N)$. Inspection of the proof of Theorem \ref{thm:MAX-dep-lv} reveals that a term of this order is asymptotically negligible.
\end{remark}

\section{\label{sec:extensions}Extensions}

This section discusses several extensions of our procedure. We first demonstrate the possibility of simplifying the procedure at the cost of power and/or losing robustness against serial correlation. We also propose a two-step method to increase the power of our confidence set.

\subsection{\label{sec:sns crit val}A simpler procedure with conservative critical values} 

The implementation of our confidence set can be greatly simplified by using different critical values that are slightly conservative but can be computed without estimating and regularizing a covariance matrix. We call these the \emph{SNS critical values}, borrowing a term from \textcite{chernozhukov2013testing} who propose similar critical values for a high-dimensional testing problem and justify them using the theory of self-normalized sums (SNS).  

For $G=2$, the SNS critical values are identical to our critical values.
For $G \geq 3$, the SNS critical values are an upper bound to our critical values and will always yield a weakly larger confidence set.  The SNS critical values are given by 
\begin{align*}
	\hat{c}_{\alpha, N, i}^{\mathrm{SNS}} (g)
	= c_{\alpha, N}^{\mathrm{SNS}}
	= \sqrt{\frac{T}{T - 1}} t_{T- 1}^{-1}\left(1 - \frac{\alpha}{(G- 1) N}\right)
\end{align*}
with $t_{T - 1}^{-1}(p)$ defined in Section~\ref{sec:crit val}.
The factor $(G - 1)$ carries out a Bonferroni correction to account for the $(G-1)$ moment inequalities that are simultaneously tested for unit $i$. The critical values do not depend on $i$ nor $g$; identical values can be used for all $i$ and $g$. Let $\widehat{C}_\alpha^{\text{SNS}}$ denote the confidence set computed by applying our procedure with SNS critical values. The following result establishes the asymptotic validity of this confidence set.

\begin{theorem}
	\label{thm:SNS critical values}
	Let $\mathbb{P}_N$ be a class of probability measures that satisfy Assumptions~\ref{assumption:basic}.\ref{assum:exogeneity}--\ref{assumption:basic}.\ref{assum:stationarity} with identical choices of $a$, $b$, $d_1$, and $d_2$. 
	Assume that there are finite constants $0 < \delta_1 < \delta_2$ such that $T^{\delta_1} \leq N \leq o(1)T^{\delta_2}$. 
	Let Assumption~\ref{assum:kernel} hold with $\kappa_N \asymp T^{\rho}$, where $0 <\rho < (\vartheta -1)/(3\vartheta -2)$, and assume that condition~\eqref{eq:rate condition r_theta} holds. Then, 
	\begin{align*}
		\liminf_{N, T \to \infty} \inf_{P \in \mathbb{P}_N} P \left( \left\{g_i^0\right\}_{1 \leq i \leq N} \in  \widehat{C}_{\alpha}^{\text{SNS}} \right) \geq 1 - \alpha. 
	\end{align*}
\end{theorem}






\subsection{\label{sec:no serial correlation}A simpler procedure under no serial correlation}

Another simplification of our procedure is possible in the absence of serial correlation. 
Suppose that, for each unit $i$, the time series $\{x_{it} v_{it}\}_{t=1}^T$ is serially uncorrelated. 
\begin{assumption}
	\label{assumption:no serial correlation}
	For every $i = 1, 2, \dotsc$ and all $s, t = 1,2, \dotsc$ such that $s \neq t$, $\E[v_{is} v_{it} x_{is} x_{it}] = 0$. 
\end{assumption}
Under this assumption
\begin{align*}
	\Xi_i (g, h, h') = \cov \left(d_{it}(g, h), d_{it}(g, h')\right)
\end{align*}
can be consistently estimated by $\widehat{\Xi}_i (g, h, h')$ by setting 
\begin{align}
	\label{eq:K if no serial correlation}
	K \left(\frac{j}{\kappa_N}\right) = \begin{cases}
		0 & \text{if $j \neq 0$} \\ 
		1 & \text{if $j = 0$}.
	\end{cases}
\end{align}
Then $\widehat{\Xi}_i (g, h, h')$ takes the simple form
\begin{align*}
	\widehat{\Xi}_i (g, h, h') =
	\frac{1}{T} \sum_{t = 1}^T \left(
		\hat{d}_{it} (g, h) - \bar{\hat{d}}_{i} (g, h) 
	\right)\left(
		\hat{d}_{it} (g, h') - \bar{\hat{d}}_{i} (g, h') 
	\right)
\end{align*}
and the test statistic is given by
\begin{align}
	\label{eq:Dhat no serial correlation}
	\widehat{D}_{i}(g, h) = \frac{\sum_{t = 1}^T \hat{d}_{it}(g, h) / \sqrt{T}}{\sqrt{\frac{1}{T} \sum_{t = 1}^T \left( \hat{d}_{it} (g, h) - \bar{\hat{d}}_{i} (g, h) \right)^2}}.
\end{align}
Critical values are computed by $\hat{c}_{\alpha, N,i} (g)$ with the simplified version of $\widehat{\Omega}_i (g)$.
The following result states the conditions for the validity of the simplified procedure. In contrast to Theorem~\ref{thm:MAX-dep-lv}, stationarity is not required.  
\begin{theorem}
	\label{thm:no serial correlation}
	Let $\mathbb{P}_N$ be a class of probability measures that satisfy Assumptions~\ref{assumption:basic}.\ref{assum:exogeneity}--\ref{assumption:basic}.\ref{assum:mixing} and Assumption~\ref{assumption:basic}.\ref{assum:Omega negative correlations} with identical choices of $a$, $b$, $d_1$, and $d_2$. 
	Assume that there are finite constants $0 < \delta_1 < \delta_2$ and $1 < k_1 \leq k_2$ such that $T^{\delta_1} \leq N \leq o(1)T^{\delta_2}$ and $(\log N)^{-k_2} \leq \epsilon_N \leq (\log N)^{-k_1}$. Let Assumption~\ref{assumption:no serial correlation} and condition \eqref{eq:K if no serial correlation} hold. In addition, suppose that
	\begin{align}
		\label{eq:rate condition r_theta no  serial correlation}
		r_{\theta, N} \sqrt{T \log N} = o \left(1 \wedge \iota_N \wedge  \min_{1 \leq i \leq N} \sigma_i \right).
	\end{align}
    Then, $\widehat{C}_{\alpha}$ based on \eqref{eq:K if no serial correlation} satisfies
	\begin{align*}
		\liminf_{N, T \to \infty} \inf_{P \in \mathbb{P}_N} P \left( \left\{g_i^0\right\}_{1 \leq i \leq N} \in  \widehat{C}_{\alpha} \right) \geq 1 - \alpha. 
	\end{align*}
\end{theorem}
Combining the test statistic \eqref{eq:Dhat no serial correlation} under no serial correlation with SNS critical values yields a particularly simple procedure that can be implemented with minimal programming effort.

\subsection{\label{sec:2step}Increasing power by using a two-step procedure}

For units that are very easy to classify (i.e., have very small $\sigma_i$), we estimate the true group memberships with probability strictly larger than $1 - \alpha / N$. This renders our confidence set conservative. Using the information provided by the units that are easy to classify, we may be able to shrink the marginal confidence sets for the units that are difficult to classify (i.e., have large $\sigma_i$). This idea inspires our two-step procedure that we call \emph{unit selection}.

The key part of our two-step procedure is the detection of units that are easy to classify. For these units, we report singleton marginal confidence sets. We then compute our joint confidence on the sub-sample of remaining units. We slightly adjust the nominal level of the confidence set to control for classification error in unit selection. 
Unit selection can increase the power of the confidence set because it carries out fewer simultaneous tests than our one-step procedure. For example, if unit selection eliminates $N/3$ units, the resulting confidence set is based on Bonferroni correction to adjust for $2N/3$ rather than $N$ simultaneous tests. 

For our two-step procedure we assume that $\hat{g}_i$ minimizes squared loss, i.e., we assume
\begin{align}
	\label{eq:ghat:squared_loss}
	\sum_{t = 1}^T \hat{d}_{it}^U (\hat{g}_i, h) \leq 0 \quad \text{for all $h \in \mathbb{G}$},
\end{align}
where 
\begin{align*}
	 \hat{d}_{it}^U (g, h) = (y_{it} - w_{it}'\hat{\theta}^w - x_{it}' \hat{\theta}_{g})^2 -  (y_{it} - w_{it}'\hat{\theta}^w - x_{it}' \hat{\theta}_{h})^2.
\end{align*}
This requirement is automatically satisfied if the grouped panel model is estimated by \emph{kmeans} clustering \parencite[e.g.][]{bonhomme2015grouped}.

Our algorithm for unit selection identifies a unit as easy to classify if it satisfies two conditions that we call \emph{moment selection} and \emph{hypothesis selection}.\footnote{%
The term ``moment selection'' is borrowed from the literature on testing moment inequalities \parencite[see, e.g.,][]{chernozhukov2013testing,andrews2010inference,andrews2012inference,romano2014practical,canay2016practical}. In this literature, moment selection reduces the power loss from possibly slack moment inequalities by identifying inequalities that are ``obviously'' satisfied. Our use is different. We use moment selection to reduce the power loss from running many simultaneous tests by identifying units that are ``obviously'' correctly classified.
} A unit $i$ satisfies the moment selection criterion if we detect substantial slackness in inequality \eqref{eq:ghat:squared_loss} for all $h \neq \hat{g}_i$. 
A unit $i$ satisfies the hypothesis selection criterion if all group memberships $h \neq \hat{g}_i$ are rejected. 

The unit selection procedure is parameterized by $\beta \in [0,\alpha/3)$.
The larger $\beta$, the more unit selection is carried out. 
Let
\begin{align*}
 \widehat{D}_i^U (g,h) = \frac{\sum_{t=1}^T \hat{d}_{it}^U (g, h) / \sqrt{T} }{\sqrt{\widehat{\Xi}^U_{i} (g, h, h)}},
\end{align*}
where $\widehat{\Xi}^U_{i}(g, h, h)$ is defined as $\widehat{\Xi}_{i}(g, h, h)$ with $\hat{d}_{it}^U$ replacing $\hat{d}_{it}$. 
$\widehat{D}_i^U (g,h)$ is a counterpart to $\widehat{D}_{i}(g, h)$ that does not adjust for the mean under the null hypothesis.

The first step of our two-step procedure carries out moment selection by computing the set 
\begin{align*}
	\widehat{M}_i (g) = \left\{ h \in \mathbb{G} \setminus \{ g\} \mid
	\widehat{D}_i^U (g, h) >  -2 c_{\beta, N}^{\mathrm{SNS}}
	\right\}
\end{align*}
for $g \in \mathbb{G}$ and $i = 1, \dotsc, N$. 
This set gives the selected moment inequalities for the hypothesis $H_0: g_i^0 = g$. For units $i$ for which $\widehat{M}_i (g)$ is empty, we have strong evidence that $g_i^0 = g$.
These units satisfy the moment selection criterion for elimination in the first step.
Condition \eqref{eq:ghat:squared_loss} ensures that $\widehat{M}_i(g)$ is never empty for $g \neq \hat{g}_i$. This property ensures that moment selection does not eliminate misclassified units. 

The second step of our two-step procedure is given by the following algorithm that carries out hypothesis selection: 
\begin{enumerate}[label=Step 2.\Alph*:, align = left]
	\item
	Set $s=0$ and $H_i (0) = \mathbb{G}$.
	\item
	Set $\widehat{N} (s) = \sum_{i=1}^N \max_{g \in H_i(s)} \mathbf{1} \{ \widehat{M}_i (g) \neq \emptyset\} $.
	\item Set
	\begin{align*}
	H_i (s+1) = \left\{ g \in \mathbb{G} \mid  \widehat{T}_i (g) \leq  \hat{c}_{\alpha - 2 \beta, \hat{N}(s), i} (g)
	\right\} \cup \left\{ \hat{g}_i\right\}.
	\end{align*}
	\item 
	If $H_i (s+1) = H_i (s)$ for all $i$, then exit the algorithm. Otherwise, set $s = s + 1$ and go to Step 2.B. 
\end{enumerate}
Step 2.A initializes the algorithm by designating all possible group assignments as hypotheses that have to be tested. 
Step 2.B counts the number $\widehat{N}(s)$ of units that are not easy to classify. 
Unit $i$ is easy to classify if and only if $\widehat{M}_i(\hat g_i)$ is empty (moment selection) and $H_i (s) = \{\hat g_i\}$ (hypothesis selection).
Step 2.C carries out hypothesis selection with critical values adjusted for $\widehat{N}(s)$ simultaneous tests of group membership. $H_i(s+1)$ gives a preliminary marginal confidence set for unit $i$ after $s$ iterations of hypothesis selection.  
Step 2.D checks the convergence of the algorithm.\footnote{The algorithm always converges since $\widehat{N} (s)$ and the cardinality of $H_i (s)$ are decreasing in $s$.}
If the algorithm has converged after $s = s^*$ iterations, the final joint confidence set is given by
\begin{align*}
	\widehat{C}_{ \mathrm{sel}, \alpha, \beta} = \bigtimes_{1 \leq i \leq N} H_i (s^*).
\end{align*}
The second-step confidence set is calculated at nominal confidence level $1 - \alpha + 2 \beta$ (see the definition of $H_i(s+1)$ in Step 2.C). 
The adjustment by $2 \beta$ represents the cost of unit selection and controls for two possible errors at the first step. The first error is estimating an incorrect group membership for a unit that is easy to classify ``in population''. The second error is erroneously declaring a unit as easy to classify.

Unit selection increases the power of the confidence set if its benefits (decreasing the number of units at the second step) outweigh the cost (adjustment of nominal level at the second step). If sufficiently many units can be eliminated, then $\widehat{C}_{\mathrm{sel},\alpha, \beta}$ is more powerful (``smaller'') than the corresponding one-step confidence set $\widehat{C}_{\alpha}$. If too few units are eliminated, then a two-step confidence set can be slightly more conservative (``larger'') than the corresponding one-step confidence set.



We establish the validity of the two-step procedure under the assumption of no serial correlation (Assumption~\ref{assumption:no serial correlation} above) and the following additional assumption. 
\begin{assumption}
\label{assumption:two step}
\begin{enumerate}
	\item The vector $({\theta^w}', \{\theta_g'\}'_{g \in \mathbb{G}})$ is contained in a compact parameter space $\Theta$.
	\item Let $p$ denote the dimension of $x_{it}$. For any $k_1, k_2, k_3 = 1, \dotsc, p$,
	\begin{align*}
		\E [\sigma_i v_{it} x_{it, k_1} x_{it, k_2} x_{it, k_3}] = 0.
	\end{align*} 
\end{enumerate}
\end{assumption}
The first part of this assumption is standard. The second part imposes an orthogonality condition on triples of time periods. It is stronger than the assumption of no serial correlation between pairs of time periods but weaker than independence across time.
We now state the formal result for the asymptotic validity of the two-step procedure.
\begin{theorem}
	\label{thm:two step}
	Let $\mathbb{P}_N$ be a set of probability measures that satisfy Assumptions~\ref{assumption:basic}, \ref{assumption:no serial correlation} and \ref{assumption:two step} with identical choices of $a$, $b$, $d_1$, and $d_2$. 
	Assume that there are finite constants $0 < \delta_1 < \delta_2$ and $1 < k_1 \leq k_2$ such that $T^{\delta_1} \leq N \leq o(1)T^{\delta_2}$ and $(\log N)^{-k_2} \leq \epsilon_N \leq (\log N)^{-k_1}$. 
	Let Assumption~\ref{assumption:no serial correlation} and condition \eqref{eq:K if no serial correlation} hold and $0 < 3 \beta < \alpha < 1$. If condition \eqref{eq:rate condition r_theta no  serial correlation} is satisfied, 
    then
	\begin{align*}
		\liminf_{N, T \to \infty} \inf_{P \in \mathbb{P}_N} P \left( \left\{g_i^0\right\}_{1 \leq i \leq N} \in  \widehat{C}_{\mathrm{sel},\alpha, \beta} \right) \geq 1 - \alpha. 
	\end{align*}
\end{theorem}

\section{\label{sec:application}Empirical application: Heterogeneous effects of a minimum wage}

In this section, we revisit the work by \textcite{wang2019heterogeneous} to illustrate our procedures.
\textcite{wang2019heterogeneous} estimate a grouped panel model to study heterogeneous effects of a minimum wage in the restaurant sector. Their analysis builds on \textcite{dube2010minimum} who employ a similar panel model but do not allow for effect heterogeneity. We assess the clustering uncertainty of the group memberships estimated in \textcite{wang2019heterogeneous} by computing confidence sets for group memberships.

We use the panel data described in \textcite{dube2010minimum}. It contains quarterly data for 1380 US counties, ranging from the first quarter of 1990 to the second quarter of 2006. The grouped panel model estimated in \textcite{wang2019heterogeneous} is given by 
\begin{align*}
	\log (\texttt{emp}_{ict}) = \theta_{g^0_i, 1} \log (\texttt{mw}_{ict}) + \theta_{g^0_i, 2} \log (\texttt{pop}_{ict}) + \theta_{g^0_i, 3} \log (\texttt{emp}_{ict}^{\texttt{TOT}}) + \phi_c + \tau_t + \sigma_i v_{ict}
\end{align*}
for state $i = 1, \dotsc, 51$, county $c = 1, \dotsc, n_i$ and time period $t = 1, \dotsc, T = 66$, where $n_i$ is the number of counties in state $i$. The variable $\texttt{emp}_{ict}$ gives employment in the restaurant sector, $\texttt{mw}_{ict}$ gives the minimum wage and $\texttt{emp}_{ict}^{\text{TOT}}$ gives total employment in all sectors. Finally, $\phi_c$ is a county fixed effect, $\tau_t$ is a time fixed effect, and $\sigma_i v_{ict}$ is an idiosyncratic error term. 

\label{QE:IC}
\textcite{su2016identifying} propose an information criterion to select the number of groups $G$ in a grouped-panel regression and prove its consistency. \textcite{wang2019heterogeneous} apply this criterion to the data and select $G = 4$. Table~\ref{tab:app:slope_coef} gives their CLasso \parencite{su2016identifying} estimates for the slope coefficients for the four groups. 

\begin{table}
\centering

\begin{tabular}{lrrrrrrrr}
\toprule
\multicolumn{1}{c}{ } & \multicolumn{2}{c}{$g = 1$} & \multicolumn{2}{c}{$g = 2$} & \multicolumn{2}{c}{$g = 3$} & \multicolumn{2}{c}{$g = 4$} \\
\cmidrule(l{3pt}r{3pt}){2-3} \cmidrule(l{3pt}r{3pt}){4-5} \cmidrule(l{3pt}r{3pt}){6-7} \cmidrule(l{3pt}r{3pt}){8-9}
 & coef & se & coef & se & coef & se & coef & se\\
\midrule
$\theta_{g, 1}$ & 0.55 & 0.05 & -0.03 & 0.04 & 0.06 & 0.04 & -0.25 & 0.04\\
$\theta_{g, 2}$ & 0.63 & 0.07 & 0.60 & 0.06 & 0.34 & 0.07 & 0.47 & 0.06\\
$\theta_{g, 3}$ & 0.51 & 0.04 & 0.61 & 0.03 & 0.41 & 0.03 & 0.53 & 0.03\\
\bottomrule
\end{tabular}
\caption{\label{tab:app:slope_coef}Estimates for the group-specific slope coefficients (coef) and corresponding standard errors (se). Standard errors clustered at the county level.}
\end{table}
There are two groups with estimated positive effects of the minimum wage on employment and two groups with negative effects (see $\theta_{g,1} $ in the table).

Based on these estimates for the slope coefficients, we estimate group memberships by running one update step of the \emph{kmeans} algorithm. This updating step ensures that the estimated group memberships satisfy inequality \eqref{eq:ghat:squared_loss} and produces estimated group memberships that are identical to the CLasso estimates in all but six states. 
\begin{figure}
\centering
\input{\Routputpath/figures/figure1.tex}
\caption{\label{fig:ghat_plot}Estimated clusters.}
\end{figure}
Estimated group memberships are displayed in Figure~\ref{fig:ghat_plot} and reported in Table~\ref*{tab:supp:cs_1step} in the Supplemental Material. Note that Figure~\ref{fig:ghat_plot} is slightly different from Figure~2 in \textcite{wang2019heterogeneous} due to the additional \emph{kmeans} step.

To translate the panel model to our framework, we identify states (subscript $i$) as the cross-sectional dimension and counties and quarters (subscripts $c$ and $t$) jointly as the second (``time'') dimension. 
We compute our confidence sets based on the fixed-effect transformation 
\begin{align}
	\label{eq:application FE model}
	\widetilde{\texttt{lemp}}_{ict} = \theta_{g^0_s, 1} \widetilde{\texttt{lmw}}_{ict} + \theta_{g^0_s, 2} \widetilde{\texttt{lpop}}_{ict} + \theta_{g^0_s, 3} \widetilde{\texttt{lemp}}_{ict}^{\texttt{TOT}} + \tilde{\tau}_t + \sigma_i \tilde{v}_{ict}.
\end{align}
Here, $\widetilde{\texttt{lemp}}_{ict} = \log (\texttt{lemp}_{ict}) - \sum_{t'= 1}^T \log (\texttt{lemp}_{ict'}) /T$ 
and $\widetilde{\texttt{lmw}}_{ict}$, $\widetilde{\texttt{lpop}}_{ict}$, and $\widetilde{\texttt{lemp}}_{ict}^{\texttt{TOT}}$ are defined similarly. Moreover, $\tilde{\tau}_t = \tau_t -  \sum_{t' = 1}^T \tau_{t'} /T $ and $\tilde{v}_{ict} = v_{ict} - \sum_{t' = 1}^T v_{ict'} /T$. Remark~\ref{remark:FE transformation} above heuristically justifies applying our procedure to fixed-effect-transformed data. 

Our second panel dimension comprises both cross-sectional variation between counties and time-series variation between different quarters. To adapt our definition of the estimated long-run variance to this setting, we redefine $\widehat{H}_{ij}$ in \eqref{eq:H_hat} as 
\begin{align*}
    & \widehat{H}_{ij} (g, h, h') 
\\
    = & \frac{1}{T n_i} \sum_{c = 1}^{n_i} \sum_{t = \abs{j} + 1}^T \left(\hat{d}_{ic, t + \min(0, j)} (g, h) - \bar{\hat{d}}_{ic} (g, h)\right) \left(\hat{d}_{ic, t + \min(0, j)} (g, h) - \bar{\hat{d}}_{ic} (g, h')\right),
\end{align*}
where $\bar{\hat{d}}_{ic} =  \sum_{t = 1}^T \hat{d}_{ict}/T$ and $\hat{d}_{ict}$ is defined as in Section~\ref{sec:test stat} with the double index $ct$ replacing $t$. We use the following modified algorithm for bandwidth selection: 
\begin{enumerate}[label=(\Alph*)]
    \item For $g \in \mathbb{G}$ and $h \in \mathbb{G} \setminus \{g\}$, take each state $i$ such that $\hat{g}_i = g$ and compute $\hat{d}_{ict} (g, h)$ for all $c = 1, \dotsc, n_i$ and $t = 1, \dotsc, T$. 
    Fit an $AR(1)$-model on $(\hat{d}_{ict} (g, h))_{t = 1}^T$ and let 
    $\hat{\rho}_{icgh}$ denote the estimated autoregressive coefficients and $\hat{\sigma}^2_{icgh}$ the estimated variance of the innovation.
    \item 
    Estimate the bandwidth sequence by 
    \begin{align*}
        \hat{\kappa}_N = 1.3221
        \left(
            T \times \frac{\sum_{i = 1}^N \sum_{c = 1}^{n_i} \sum_{g \in \mathbb{G}} \sum_{h \in \mathbb{G} \setminus \{g\}} \hat{\rho}_{icgh}^2  \hat{\sigma}_{icgh}^4  / (1 - \hat{\rho}_{icgh}^2 )^8 }{\sum_{i = 1}^N \sum_{c = 1}^{n_i} \sum_{g \in \mathbb{G}} \sum_{h \in \mathbb{G} \setminus \{g\}} \hat{\sigma}_{icgh}^4  / (1 - \hat{\rho}_{icgh}^2 )^4}
            \right)^{1/5}
    \end{align*}
\end{enumerate}
This algorithm yields $\hat{\kappa}_N = 21.37$.

\begin{table}
\footnotesize 
\centering

\begin{tabular}{lrrrlrrl}
\toprule
\multicolumn{2}{c}{ } & \multicolumn{3}{c}{baseline} & \multicolumn{3}{c}{SNS} \\
\cmidrule(l{3pt}r{3pt}){3-5} \cmidrule(l{3pt}r{3pt}){6-8}
State & $\hat{g}_i$ & p-val $\hat{g}_i$ & card & CS & p-val $\hat{g}_i$ & card & CS\\
\midrule
Arkansas & 2 & 0.027 & 1 & 2 & 0.041 & 1 & 2\\
Colorado & 4 & 0.893 & 2 & 3, 4 & 1.000 & 2 & 3, 4\\
Connecticut & 3 & 0.631 & 2 & 2, 3 & 0.715 & 2 & 2, 3\\
Florida & 4 & 0.049 & 1 & 4 & 0.074 & 2 & 3, 4\\
Idaho & 3 & 0.664 & 2 & 3, 4 & 0.852 & 2 & 3, 4\\
\addlinespace
Indiana & 3 & 0.024 & 1 & 3 & 0.038 & 1 & 3\\
Kansas & 4 & 0.010 & 1 & 4 & 0.015 & 1 & 4\\
Kentucky & 3 & 0.093 & 2 & 3, 4 & 0.151 & 3 & 2, 3, 4\\
Maine & 2 & 0.015 & 1 & 2 & 0.023 & 1 & 2\\
New Hampshire & 3 & 0.080 & 2 & 3, 4 & 0.130 & 2 & 3, 4\\
\addlinespace
New Mexico & 3 & 0.094 & 3 & 2, 3, 4 & 0.121 & 3 & 2, 3, 4\\
Oklahoma & 3 & 0.168 & 2 & 2, 3 & 0.207 & 2 & 2, 3\\
Pennsylvania & 3 & 0.021 & 1 & 3 & 0.030 & 1 & 3\\
West Virginia & 3 & 0.014 & 1 & 3 & 0.041 & 1 & 3\\
Wisconsin & 3 & 0.392 & 2 & 2, 3 & 0.455 & 2 & 2, 3\\
\bottomrule
\end{tabular}
\caption{\label{tab:cs_short}Marginal confidence set at level $1 - \alpha = 0.95$ for the states with group membership that are estimated with $p$-value between $1\%$ and $90\%$. ``p-val $\hat{g}_i$'' is the $p$-value for the significance of the estimated group membership. ``CS cardinality'' is the cardinality of the marginal confidence set for the state.  ``CS'' is the marginal confidence set. ``Baseline'' refers to the procedure with critical values defined in Section~\ref{sec:crit val}. ``SNS'' refers to the procedure with critical values defined in Section~\ref{sec:sns crit val}.}
\end{table}

\label{QE:epsilon application}We compute the joint-confidence set at level $1 - \alpha = 0.95$. For the regularization parameter, we set $\epsilon_N = 0.01$. Our results are robust to different choices of $\epsilon_N$.\footnote{For $\epsilon_N = 0, 0.05$, we obtain a confidence set that differs only in the group assignments for Kentucky. For $\epsilon_N = 0.01, 0.05$, group $g=2$ is contained in the marginal confidence set for Kentucky ($p$-values $0.0662$ and $0.0610$). For $\epsilon = 0$, it is not ($p$-value $0.0372$).}

The full joint confidence set is reported in Table~\ref*{tab:supp:cs_1step} in the Supplemental Material. Marginal confidence sets for a subset of states are given in Table~\ref{tab:cs_short}. The cardinality of the state-wise marginal confidence sets is four for three states, three for seven states, two for twelve states, and one for twenty-nine states.

Based on our joint confidence set, we compute $p$-values for the significance of the estimated group memberships. We say that the estimated group membership $\hat{g}_i$ for state $i$ is significant at level $\alpha$ if the marginal confidence set for state $i$ with confidence level $1-\alpha$ contains only $\hat{g}_i$. Up to a joint failure probability of at most $\alpha$, significantly estimated group memberships reveal the true group membership and cannot be attributed to estimation error. 
The $p$-value for significance of $\hat{g}_i$ is the smallest value of $\alpha$ such that $\hat{g}_i$ is significant at level $\alpha$. The $p$-values for a subset of states are reported in Table~\ref{tab:cs_short}, and results for all states are reported in Table~\ref*{tab:supp:cs_1step} in the Supplemental Material. 

Table~\ref{tab:cs_short} demonstrates that, for our sample, the SNS critical values from Section~\ref{sec:sns crit val} yield slightly larger confidence sets than our baseline procedure. 
For example, the cardinality of the marginal confidence set for Florida increases from 1 to 2 ($p$-value increases from 0.049 to 0.074) when we use the SNS critical values.

Displaying a visual representation of the estimated clusters as we do in Figure~\ref{fig:ghat_plot} is standard practice even in applications where the clustering structure is a nuisance parameter and not of interest in its own right. Visual inspection of the clusters is meant to confirm their economic plausibility and serves as an informal test of model specification. Based on our confidence set, such an informal analysis can be complemented by a graphical representation of clustering uncertainty as illustrated in Figure~\ref{fig:uncertainty_plot}.
\begin{figure}
\input{\Routputpath/figures/figure2.tex}
\caption{\label{fig:uncertainty_plot}Visual representation of clustering uncertainty. The upper panel shows $p$-values for estimated group memberships. The lower panel shows the cardinality of the marginal confidence sets with $1 - \alpha = 0.95$.}
\end{figure}
The upper panel in Figure~\ref{fig:uncertainty_plot} represents clustering uncertainty by shading US states according to the $p$-values of their estimated group memberships. The lower panel shades states by the cardinality of their marginal confidence sets ($1 - \alpha = 0.95$). The figure suggests a low degree of clustering uncertainty. It would suggest a large degree of clustering uncertainty if the maps were shaded mostly in a dark hue. Large clustering uncertainty indicates that the clustering algorithm may be overfitting on the sample, rather than picking up structural heterogeneity. 

Our two-step procedure with parameter $\beta = \alpha / 5 = 0.01$ eliminates only one unit in the first step and does not improve the final confidence set or $p$-values. In Supplemental Material~\ref*{app:application 2step}, we apply the two-step procedure under the implausible assumption of no serial correlation which eliminates more units and lowers the final $p$-values.

\section{\label{sec:simulations}Simulations} 

Our simulation designs are based on our empirical application. The data-generating process is given by 
\begin{align*}
	\widetilde{\texttt{lemp}}_{it} = \theta_{g^0_i, 1} \widetilde{\texttt{lmw}}_{it} + \theta_{g^0_i, 2} \widetilde{\texttt{lpop}}_{it} + \theta_{g^0_i, 3}\widetilde{\texttt{lemp}}_{it}^{\text{TOT}} + \sigma_i v_{it}
\end{align*}
for $i = 1, \dotsc, N$ and $t = 1, \dotsc, T$. For $g = 1, \dotsc, 4$, the coefficients $\theta_{g, 1}$, $\theta_{g, 2}$ and $\theta_{g, 3}$ are set equal to the estimated coefficients in Table~\ref{tab:app:slope_coef}. 

\label{QE:dgp serial correlation}
To generate $x_{it} = (\widetilde{\texttt{lmw}}_{it}, \widetilde{\texttt{lpop}}_{it}, \widetilde{\texttt{lemp}}_{it}^{\text{TOT}})$ that exhibit similar patterns of serial correlation as the time series in our empirical application, we estimate a $\text{VAR}(1)$ model of $x_{it}$ for each county $c$ in our data sample and save the estimated coefficient $\hat{\Gamma}_c$ and the empirical residuals $\hat{e}_{ct}$. We then average over all county-wise coefficients to obtain a single $\text{VAR}(1)$ coefficient matrix $\bar{\Gamma}$. To generate a time series $(x_{it})_{t = 1}^T$, we sample the initial value $x_{i0}$ from the covariates observed in our data and then iteratively generate $x_{i,t+1} = \bar{\Gamma} x_{it} + e_t$, where $e_t$ is a randomly drawn empirical residual $\hat{e}_{ct}$. 

Each unit $i$ is assigned to one of the four groups with equal probability and assigned a heteroscedasticity parameter $\sigma_i$. For $\sigma = 0.1, 0.2$, we draw $\sigma_i = \sigma \chi^2(4)/4$, where $\chi^2 (df)$ is a random draw from a $\chi^2$-distribution with $df$ degrees of freedom. 

For $\rho = 0, 0.25, 0.5$, we set $v_{it} = \sqrt{1 - \rho^2} \tilde{v}_{it}$ where $\tilde{v}_{it}$ follows an $AR(1)$ process with autoregressive parameter $\rho$ and standard normal innovations and $\tilde{v}_{i0}$ is drawn from the stationary distribution of $(\tilde{v}_{it})_{t = 1}^T$. To introduce relevant serial correlation, we require both $v_{it}$ and $x_{it}$ to be serially correlated (cf. Section~\ref{sec:no serial correlation}). Therefore, setting $\rho = 0$ turns off the serial correlation (but not all temporal dependence) in the moment inequalities  even though $x_{it}$ is still serially correlated. 

\begin{table}
    \scriptsize
    \centering 
    
\begin{tabular}{rrrrrrrrrrrrrr}
\toprule
\multicolumn{5}{c}{ } & \multicolumn{4}{c}{coverage} & \multicolumn{3}{c}{cardinality} & \multicolumn{2}{c}{bandwidth} \\
\cmidrule(l{3pt}r{3pt}){6-9} \cmidrule(l{3pt}r{3pt}){10-12} \cmidrule(l{3pt}r{3pt}){13-14}
\multicolumn{5}{c}{ } & \multicolumn{3}{c}{kmeans} & \multicolumn{1}{c}{oracle} & \multicolumn{2}{c}{kmeans} & \multicolumn{1}{c}{oracle} & \multicolumn{1}{c}{kmeans} & \multicolumn{1}{c}{oracle} \\
\cmidrule(l{3pt}r{3pt}){6-8} \cmidrule(l{3pt}r{3pt}){9-9} \cmidrule(l{3pt}r{3pt}){10-11} \cmidrule(l{3pt}r{3pt}){12-12} \cmidrule(l{3pt}r{3pt}){13-13} \cmidrule(l{3pt}r{3pt}){14-14}
$\rho$ & $\sigma$ & $N$ & $T$ & $\hat{\theta} - \theta$ & hac & h & naive & hac & hac & h & hac & hac & hac\\
\midrule
 &  &  & 60 & 0.18 & 0.99 & 0.99 & 0.21 & 0.99 & 1.60 & 1.27 & 1.57 & 2.12 & 2.04\\

 &  & \multirow[t]{-2}{*}{\raggedleft\arraybackslash 50} & 120 & 0.13 & 1.00 & 1.00 & 0.83 & 1.00 & 1.07 & 1.02 & 1.07 & 2.27 & 2.27\\

 &  &  & 60 & 0.13 & 0.99 & 0.98 & 0.04 & 1.00 & 1.76 & 1.30 & 1.75 & 2.20 & 2.19\\

 &  & \multirow[t]{-2}{*}{\raggedleft\arraybackslash 100} & 120 & 0.09 & 1.00 & 1.00 & 0.80 & 1.00 & 1.12 & 1.02 & 1.11 & 2.32 & 2.36\\

 &  &  & 60 & 0.09 & 0.99 & 0.98 & 0.00 & 0.98 & 1.95 & 1.34 & 1.95 & 2.30 & 2.30\\

 & \multirow[t]{-6}{*}{\raggedleft\arraybackslash 0.1} & \multirow[t]{-2}{*}{\raggedleft\arraybackslash 200} & 120 & 0.07 & 1.00 & 0.99 & 0.60 & 1.00 & 1.17 & 1.03 & 1.18 & 2.39 & 2.38\\
\cmidrule{2-14}
 &  &  & 60 & 0.40 & 0.92 & 0.92 & 0.01 & 0.98 & 1.97 & 1.62 & 1.95 & 2.11 & 2.05\\

 &  & \multirow[t]{-2}{*}{\raggedleft\arraybackslash 50} & 120 & 0.26 & 1.00 & 0.99 & 0.65 & 0.99 & 1.16 & 1.08 & 1.14 & 2.23 & 2.23\\

 &  &  & 60 & 0.30 & 0.95 & 0.93 & 0.00 & 0.99 & 2.13 & 1.67 & 2.14 & 2.18 & 2.20\\

 &  & \multirow[t]{-2}{*}{\raggedleft\arraybackslash 100} & 120 & 0.18 & 1.00 & 0.99 & 0.36 & 0.98 & 1.21 & 1.08 & 1.20 & 2.35 & 2.34\\

 &  &  & 60 & 0.22 & 0.96 & 0.96 & 0.00 & 0.97 & 2.30 & 1.74 & 2.31 & 2.27 & 2.31\\

\multirow[t]{-12}{*}{\raggedleft\arraybackslash 0.00} & \multirow[t]{-6}{*}{\raggedleft\arraybackslash 0.2} & \multirow[t]{-2}{*}{\raggedleft\arraybackslash 200} & 120 & 0.13 & 0.98 & 1.00 & 0.14 & 1.00 & 1.28 & 1.09 & 1.27 & 2.36 & 2.39\\
\cmidrule{1-14}
 &  &  & 60 & 0.22 & 0.94 & 0.85 & 0.10 & 0.97 & 2.40 & 1.26 & 2.38 & 3.73 & 3.71\\

 &  & \multirow[t]{-2}{*}{\raggedleft\arraybackslash 50} & 120 & 0.15 & 1.00 & 0.99 & 0.82 & 1.00 & 1.73 & 1.02 & 1.75 & 4.14 & 4.15\\

 &  &  & 60 & 0.15 & 0.96 & 0.87 & 0.01 & 0.97 & 2.68 & 1.30 & 2.71 & 3.85 & 3.87\\

 &  & \multirow[t]{-2}{*}{\raggedleft\arraybackslash 100} & 120 & 0.11 & 1.00 & 0.97 & 0.65 & 1.00 & 2.23 & 1.02 & 2.29 & 4.21 & 4.29\\

 &  &  & 60 & 0.11 & 0.94 & 0.84 & 0.00 & 0.96 & 3.03 & 1.33 & 3.05 & 3.99 & 4.02\\

 & \multirow[t]{-6}{*}{\raggedleft\arraybackslash 0.1} & \multirow[t]{-2}{*}{\raggedleft\arraybackslash 200} & 120 & 0.08 & 1.00 & 0.96 & 0.44 & 1.00 & 2.89 & 1.03 & 2.88 & 4.37 & 4.37\\
\cmidrule{2-14}
 &  &  & 60 & 0.51 & 0.84 & 0.64 & 0.00 & 0.94 & 2.63 & 1.60 & 2.69 & 3.59 & 3.67\\

 &  & \multirow[t]{-2}{*}{\raggedleft\arraybackslash 50} & 120 & 0.31 & 0.99 & 0.94 & 0.49 & 1.00 & 1.87 & 1.08 & 1.85 & 4.06 & 4.11\\

 &  &  & 60 & 0.36 & 0.88 & 0.64 & 0.00 & 0.95 & 2.96 & 1.67 & 2.97 & 3.78 & 3.82\\

 &  & \multirow[t]{-2}{*}{\raggedleft\arraybackslash 100} & 120 & 0.22 & 0.99 & 0.92 & 0.21 & 0.99 & 2.25 & 1.08 & 2.38 & 4.13 & 4.29\\

 &  &  & 60 & 0.27 & 0.89 & 0.64 & 0.00 & 0.91 & 3.24 & 1.71 & 3.24 & 4.00 & 3.98\\

\multirow[t]{-12}{*}{\raggedleft\arraybackslash 0.25} & \multirow[t]{-6}{*}{\raggedleft\arraybackslash 0.2} & \multirow[t]{-2}{*}{\raggedleft\arraybackslash 200} & 120 & 0.15 & 0.99 & 0.93 & 0.04 & 0.99 & 2.95 & 1.09 & 3.02 & 4.33 & 4.42\\
\cmidrule{1-14}
 &  &  & 60 & 0.26 & 0.87 & 0.47 & 0.03 & 0.95 & 3.35 & 1.25 & 3.39 & 6.46 & 6.55\\

 &  & \multirow[t]{-2}{*}{\raggedleft\arraybackslash 50} & 120 & 0.17 & 0.99 & 0.90 & 0.68 & 0.99 & 3.71 & 1.02 & 3.71 & 7.21 & 7.21\\

 &  &  & 60 & 0.18 & 0.86 & 0.32 & 0.00 & 0.90 & 3.52 & 1.28 & 3.55 & 6.45 & 6.50\\

 &  & \multirow[t]{-2}{*}{\raggedleft\arraybackslash 100} & 120 & 0.13 & 0.99 & 0.85 & 0.50 & 1.00 & 3.80 & 1.02 & 3.81 & 7.42 & 7.46\\

 &  &  & 60 & 0.13 & 0.82 & 0.19 & 0.00 & 0.84 & 3.67 & 1.31 & 3.67 & 6.93 & 6.90\\

 & \multirow[t]{-6}{*}{\raggedleft\arraybackslash 0.1} & \multirow[t]{-2}{*}{\raggedleft\arraybackslash 200} & 120 & 0.09 & 1.00 & 0.77 & 0.26 & 1.00 & 3.84 & 1.03 & 3.84 & 7.68 & 7.70\\
\cmidrule{2-14}
 &  &  & 60 & 0.64 & 0.69 & 0.21 & 0.00 & 0.88 & 3.42 & 1.55 & 3.49 & 6.14 & 6.34\\

 &  & \multirow[t]{-2}{*}{\raggedleft\arraybackslash 50} & 120 & 0.36 & 0.98 & 0.74 & 0.30 & 0.98 & 3.72 & 1.07 & 3.74 & 7.16 & 7.19\\

 &  &  & 60 & 0.43 & 0.69 & 0.11 & 0.00 & 0.84 & 3.62 & 1.61 & 3.66 & 6.50 & 6.65\\

 &  & \multirow[t]{-2}{*}{\raggedleft\arraybackslash 100} & 120 & 0.26 & 0.98 & 0.62 & 0.09 & 0.99 & 3.84 & 1.08 & 3.84 & 7.35 & 7.48\\

 &  &  & 60 & 0.32 & 0.70 & 0.02 & 0.00 & 0.76 & 3.74 & 1.67 & 3.75 & 6.80 & 6.88\\

\multirow[t]{-12}{*}{\raggedleft\arraybackslash 0.50} & \multirow[t]{-6}{*}{\raggedleft\arraybackslash 0.2} & \multirow[t]{-2}{*}{\raggedleft\arraybackslash 200} & 120 & 0.18 & 0.97 & 0.48 & 0.00 & 0.98 & 3.87 & 1.09 & 3.87 & 7.62 & 7.66\\
\bottomrule
\end{tabular}
    \caption{\label{tab:sim_one_step} Simulation results. Nominal level $1 - \alpha = 0.95$. ``$\hat{\theta} - \theta$'' is defined in the text. ``hac'' uses the data-driven bandwidth, ``h'' sets the bandwidth equal to zero, and ``na\"ive'' is the na\"ive confidence set defined in the text. ``kmeans'' uses the \emph{kmeans} algorithm to compute group-specific coefficients, and ``oracle'' uses the true coefficients. ``coverage'' is the empirical coverage probability of the joint confidence set. ``cardinality'' is the Monte Carlo average of the average (over all units) cardinality of the marginal unit-wise confidence sets. ``bandwidth'' is the Monte Carlo average of chosen bandwidth.}
\end{table}
We simulate panels of size $N = 50, 100, 200$ and $T = 60, 120$. The lower values of these ranges are in the ballpark of the number of states (51 states) and time periods (66 quarters) observed in our application. 
We set the regularization parameter for covariance estimation to $\epsilon_N = 0.01$. In Supplemental Material~\ref*{appendix:epsilon_N}, we demonstrate that our results are robust to alternative specifications of $\epsilon_N$. 
We simulate both the oracle confidence set, for which we take the true values of the group-specific coefficients as given, and the confidence set with group-specific coefficients estimated by the \emph{kmeans} estimator. We steer the \emph{kmeans} algorithm towards recovering the population labels of the groups by using the true group memberships as initial values.
The simulation results are based on 500 replications and reported in Table~\ref{tab:sim_one_step}.\footnote{Our Monte Carlo simulations were carried out on computing resources of the Swedish National Infrastructure for Computing (SNIC) funded by Swedish Research Council grant 2018-05973.}

We first discuss the empirical coverage probability of our confidence set with \emph{kmeans} estimates and a data-driven bandwidth (coverage\textrightarrow kmeans\textrightarrow hac in Table~\ref{tab:sim_one_step}). The joint confidence set has coverage probability of at least $1 - \alpha = 95\%$ in most cases.
This is in line with our theoretical result on the asymptotic validity of the confidence set. 
When $T=60$, our confidence set under-covers for some designs. 
In these designs, we compute a large confidence set. This result indicates that our confidence set can detect that statistical uncertainty is large even when it is under-covering. 
In the designs with moderate serial correlation ($\rho = 0.25$), under-coverage can be alleviated by increasing the number of cross-sectional observations, keeping the number of time series observations constant. 
In the designs with large serial correlation ($\rho = 0.5$), improving coverage requires increasing the number of time-series observations. 
This result may be explained by the fact that our algorithm for bandwidth selection mechanically chooses a small bandwidth when the time series is short. With a small bandwidth, our procedure cannot control for temporal dependence over many time periods.  

\label{QE:discuss conservative CS}
In some of our designs, our confidence set is conservative, with an empirical coverage probability of up to close to one. Since group membership is a discrete parameter, this is not necessarily a sign that the confidence set is underpowered. 
To see this, consider the na\"ive confidence set that contains only the vector of estimated group memberships. This is the smallest possible confidence set. 
As $T$ increases, true group memberships are revealed with increasing probability, and even the na\"ive confidence set can become conservative eventually (see Section~\ref{sec:2step}). 
The coverage probability of the na\"ive confidence set (coverage\textrightarrow kmeans\textrightarrow na\"ive in Table~\ref{tab:sim_one_step}) gives the probability that data-driven clustering recovers the true group structure. In our designs, this probability varies between 0\% and 83\%. 

We now turn to the power of our confidence set with \emph{kmeans} estimates and a data-driven bandwidth. The simulated average cardinality of a unit-wise marginal confidence set is reported in Table~\ref{tab:sim_one_step} under cardinality\textrightarrow kmeans\textrightarrow hac. 
Increasing $\sigma$ or $\rho$ makes the data more noisy as time-series shocks become larger or more persistent. As is expected, the power of our test decreases, and the size of our confidence set increases as the data become more noisy. 

In the less noisy designs, the confidence set is highly informative. For example, if $\rho = 0$ and $\sigma = 0.1$, the confidence set rules out 2-3 out of 4 possible group assignments for most units. Our confidence set is only slightly larger than the na\"ive set (i.e., the set containing only the vector of estimated group memberships) in the designs where the na\"ive set performs well. The na\"ive set still under-covers, and some enlargement improves the coverage.

In the noisiest designs, our confidence set becomes quite uninformative and assigns the trivial marginal confidence set (all four possible groups) to most units.  

\label{QE:simulations parameter uncertainty explain purpose}
Our asymptotic result is valid under a high-level assumption about the rate at which the estimator $\hat{\theta} = \{\hat{\theta}_g'\}'_{g \in \mathbb{G}}$ converges to the true group-specific coefficients $\theta = \{\theta_g'\}'_{g \in \mathbb{G}}$. This rate may be affected by misclassification. In this section, we study the effect of model estimation in finite samples, focusing on misclassification driven by the noise variance $\sigma^2$. 
Additional results for settings where model estimation is potentially affected by weak group separation are reported in Supplemental Material~\ref*{appendix:weak separation}. 

The column ``$\hat{\theta} - \theta$'' in Table~\ref{tab:sim_one_step} quantifies estimation error in $\hat{\theta}$ and reports the simulated value of 
\begin{align*}
	\frac{\lVert \hat{\theta} - \theta \rVert}{\lVert \theta \rVert} = 
    \left.\sqrt{\sum_{g = 1}^4 \E \lVert \hat{\theta}_g - \theta_g \rVert_2^2}\middle/\sqrt{\sum_{g = 1}^4\lVert \theta_g \rVert_2^2} \right., 
\end{align*} where $\lVert \cdot \rVert_2$ is the $L_2$-norm. 
\label{QE:simulations parameter uncertainty}
The simulated estimation error is not negligible, varying between 7\% and 64\% of the magnitude of the true coefficient vector.
In all our designs, increasing $T$ keeping $N$ fixed or increasing $N$ keeping $T$ fixed decreases estimation error. This shows that, even though states are misclassified, the \emph{kmeans} estimator can exploit both time-series and cross-sectional variation. This is necessary for fulfilling the rate condition imposed in our asymptotic result. 

\label{QE:simulations direct test of parameter uncertainty}
As a more direct test of the effect of parameter estimation, we compare the confidence set using the true coefficients (oracle\textrightarrow hac in Table~\ref{tab:sim_one_step}) to the confidence set using \emph{kmeans} estimates (kmeans\textrightarrow hac in Table~\ref{tab:sim_one_step}). In many designs, both confidence sets are valid with a coverage probability of at least 95\%. In most designs where the oracle confidence set has coverage of at least 95\%, the confidence set with \emph{kmeans} estimates is valid or undercovers only slightly (1-3\%). 
For all designs, increasing either $N$ or $T$ decreases the difference in coverage probability between the oracle confidence set and the confidence set using the \emph{kmeans} estimates. This simulation evidence aligns with our theoretical argument that the asymptotic effect of parameter estimation is not of first order.

Next, we compare our benchmark procedure based on the heteroskedasticity-and-autocorrelation-robust (HAC) variance estimator (``hac'' in results table) and the alternative procedure discussed in Section~\ref{sec:no serial correlation} with variance estimation that is robust to heteroscedasticity but not auto-correlation (``h'' in results table). 

In the designs with no serial correlation ($\rho = 0$), the confidence set using the heteroscedasticity-robust variance estimator is valid. It covers the true group structure at least with probability $1 - \alpha = 95\%$ provided that the sample size is large enough. In these designs, choosing the heteroscedasticity-robust estimator over the HAC estimator increases the power of the confidence set. If $T = 120$, the power gain is modest. If $T = 60$, the power gain can be quite large. For example, in the design with $\rho = 0$, $\sigma = 0.1$, $N = 200$, and $T = 60$, the average cardinality of a unit-wise confidence set is $1.95$ for the HAC estimator and $1.34$ for the heteroscedasticity-robust estimator. This power gain points to the inherent difficulty of estimating long-run variances for a high-dimensional number of relatively short time series. 

For the settings with serial correlation ($\rho = 0.25, 0.5$), the confidence set using the variance estimator that is only robust to heteroscedasticity is not valid. In the designs with a lot of serial correlation ($\rho = 0.5$), it under-covers severely. For example, for $\rho = 0.5$, $\sigma = 0.2$, $N = 200$, $T = 120$, its coverage probability is only 48\%, even though the sample size is fairly large. Our benchmark confidence set based on the HAC estimator set has correct coverage in this design. 

An interpretation of our confidence set is that it distinguishes low-noise units for which the clustering algorithm reveals the true group memberships from noisy units for which group membership is uncertain. Our model and simulation designs parameterize the unit-specific noise-level by the latent parameter $\sigma_i$. Our confidence set picks up the latent heteroscedasticity and assigns, on average, a small marginal confidence set for a unit $i$ with small $\sigma_i$ and a large marginal confidence set for a unit $i$ with large $\sigma_i$. This property is demonstrated in Table~\ref{tab:sim_one_step_by_sigma}, where we report the average cardinality of the marginal confidence set for units with different magnitudes of $\sigma_i$. 
For example, in the first reported design ($\rho = 0$, $\sigma = 0.1$, $N = 50$, $T = 60$), the average cardinality for the units in the lowest quintile of the distribution of $\sigma_i$ is $1.31$. The average cardinality for the units in the highest quintile of the distribution of $\sigma_i$ is $2.07$. 

\begin{table}
    \scriptsize
    \centering 
    
\begin{tabular}{rrrrrrrrrr}
\toprule
\multicolumn{4}{c}{ } & \multicolumn{6}{c}{average cardinality} \\
\cmidrule(l{3pt}r{3pt}){5-10}
$\rho$ & $\sigma$ & $N$ & $T$ & all & 0-20perc & 20-40perc & 40-60perc & 60-80perc & 80-100perc\\
\midrule
 &  &  & 60 & 1.60 & 1.31 & 1.40 & 1.52 & 1.70 & 2.07\\

 &  & \multirow[t]{-2}{*}{\raggedleft\arraybackslash 50} & 120 & 1.07 & 1.03 & 1.04 & 1.05 & 1.07 & 1.15\\

 &  &  & 60 & 1.76 & 1.45 & 1.56 & 1.67 & 1.87 & 2.24\\

 &  & \multirow[t]{-2}{*}{\raggedleft\arraybackslash 100} & 120 & 1.12 & 1.08 & 1.08 & 1.10 & 1.13 & 1.21\\

 &  &  & 60 & 1.95 & 1.63 & 1.75 & 1.88 & 2.08 & 2.43\\

 & \multirow[t]{-6}{*}{\raggedleft\arraybackslash 0.1} & \multirow[t]{-2}{*}{\raggedleft\arraybackslash 200} & 120 & 1.17 & 1.12 & 1.12 & 1.16 & 1.18 & 1.28\\
\cmidrule{2-10}
 &  &  & 60 & 1.97 & 1.41 & 1.62 & 1.91 & 2.21 & 2.67\\

 &  & \multirow[t]{-2}{*}{\raggedleft\arraybackslash 50} & 120 & 1.16 & 1.05 & 1.08 & 1.13 & 1.18 & 1.37\\

 &  &  & 60 & 2.13 & 1.55 & 1.82 & 2.07 & 2.38 & 2.84\\

 &  & \multirow[t]{-2}{*}{\raggedleft\arraybackslash 100} & 120 & 1.21 & 1.07 & 1.11 & 1.17 & 1.25 & 1.45\\

 &  &  & 60 & 2.30 & 1.70 & 2.00 & 2.25 & 2.53 & 3.00\\

\multirow[t]{-12}{*}{\raggedleft\arraybackslash 0.00} & \multirow[t]{-6}{*}{\raggedleft\arraybackslash 0.2} & \multirow[t]{-2}{*}{\raggedleft\arraybackslash 200} & 120 & 1.28 & 1.14 & 1.18 & 1.23 & 1.32 & 1.55\\
\cmidrule{1-10}
 &  &  & 60 & 2.40 & 2.08 & 2.26 & 2.34 & 2.50 & 2.80\\

 &  & \multirow[t]{-2}{*}{\raggedleft\arraybackslash 50} & 120 & 1.73 & 1.64 & 1.68 & 1.69 & 1.75 & 1.86\\

 &  &  & 60 & 2.68 & 2.40 & 2.53 & 2.64 & 2.79 & 3.05\\

 &  & \multirow[t]{-2}{*}{\raggedleft\arraybackslash 100} & 120 & 2.23 & 2.15 & 2.18 & 2.19 & 2.26 & 2.35\\

 &  &  & 60 & 3.03 & 2.81 & 2.90 & 3.00 & 3.13 & 3.32\\

 & \multirow[t]{-6}{*}{\raggedleft\arraybackslash 0.1} & \multirow[t]{-2}{*}{\raggedleft\arraybackslash 200} & 120 & 2.89 & 2.84 & 2.85 & 2.87 & 2.90 & 2.97\\
\cmidrule{2-10}
 &  &  & 60 & 2.63 & 2.11 & 2.36 & 2.59 & 2.85 & 3.24\\

 &  & \multirow[t]{-2}{*}{\raggedleft\arraybackslash 50} & 120 & 1.87 & 1.69 & 1.77 & 1.84 & 1.92 & 2.16\\

 &  &  & 60 & 2.96 & 2.52 & 2.74 & 2.93 & 3.15 & 3.44\\

 &  & \multirow[t]{-2}{*}{\raggedleft\arraybackslash 100} & 120 & 2.25 & 2.09 & 2.14 & 2.21 & 2.28 & 2.52\\

 &  &  & 60 & 3.24 & 2.89 & 3.08 & 3.23 & 3.40 & 3.61\\

\multirow[t]{-12}{*}{\raggedleft\arraybackslash 0.25} & \multirow[t]{-6}{*}{\raggedleft\arraybackslash 0.2} & \multirow[t]{-2}{*}{\raggedleft\arraybackslash 200} & 120 & 2.95 & 2.82 & 2.87 & 2.93 & 2.99 & 3.12\\
\cmidrule{1-10}
 &  &  & 60 & 3.35 & 3.18 & 3.30 & 3.34 & 3.40 & 3.53\\

 &  & \multirow[t]{-2}{*}{\raggedleft\arraybackslash 50} & 120 & 3.71 & 3.68 & 3.70 & 3.71 & 3.72 & 3.74\\

 &  &  & 60 & 3.52 & 3.40 & 3.45 & 3.51 & 3.58 & 3.67\\

 &  & \multirow[t]{-2}{*}{\raggedleft\arraybackslash 100} & 120 & 3.80 & 3.76 & 3.79 & 3.80 & 3.81 & 3.86\\

 &  &  & 60 & 3.67 & 3.57 & 3.62 & 3.65 & 3.71 & 3.78\\

 & \multirow[t]{-6}{*}{\raggedleft\arraybackslash 0.1} & \multirow[t]{-2}{*}{\raggedleft\arraybackslash 200} & 120 & 3.84 & 3.81 & 3.82 & 3.82 & 3.85 & 3.88\\
\cmidrule{2-10}
 &  &  & 60 & 3.42 & 3.16 & 3.30 & 3.42 & 3.54 & 3.67\\

 &  & \multirow[t]{-2}{*}{\raggedleft\arraybackslash 50} & 120 & 3.72 & 3.67 & 3.69 & 3.73 & 3.74 & 3.79\\

 &  &  & 60 & 3.62 & 3.43 & 3.55 & 3.63 & 3.70 & 3.80\\

 &  & \multirow[t]{-2}{*}{\raggedleft\arraybackslash 100} & 120 & 3.84 & 3.79 & 3.81 & 3.84 & 3.86 & 3.90\\

 &  &  & 60 & 3.74 & 3.60 & 3.68 & 3.75 & 3.80 & 3.87\\

\multirow[t]{-12}{*}{\raggedleft\arraybackslash 0.50} & \multirow[t]{-6}{*}{\raggedleft\arraybackslash 0.2} & \multirow[t]{-2}{*}{\raggedleft\arraybackslash 200} & 120 & 3.87 & 3.81 & 3.84 & 3.86 & 3.89 & 3.93\\
\bottomrule
\end{tabular}
    \caption{\label{tab:sim_one_step_by_sigma} Simulated average cardinality of the marginal confidence sets by unit $
    \sigma_i$. Nominal level $1 - \alpha = 0.95$. 0-20perc refers to the units with a $\sigma_i$-value that lies between the 0 and 20 percentile of the distribution of $\sigma \chi^2 (4) / 4$. 20-40perc, 40-60perc, 60-80perc and 80-100perc are defined similarly.}
\end{table}

\section{\label{sec:conclusion}Conclusion}

We have constructed a confidence set for group membership for grouped panel models with time-invariant group-specific regression curves. Our confidence set can be easily tabulated and visualized as demonstrated in our empirical application. Empirical researchers can use our confidence set to quantify the statistical uncertainty about the true group structure in a grouped panel model.  

Our method extends naturally to other models with a latent group structure. We construct our confidence set based on a test for the best fit in a least-squares sense. This idea can be adapted to a wide variety of settings with possibly different notions of what constitutes a best fit.
For example, it may be possible to compute a confidence set for group membership in non-linear likelihood-based models \parencites{liu2020identification, wang2021identifying} by testing group membership based on the fit measured by the log-likelihood function. This and other extensions of our method require new and non-trivial theoretical work and are interesting avenues for future research. 

\label{QE:time varying}
Our approach can be generalized to models with time-varying coefficients $\theta_g = \theta_{g, t}$, including models with time-varying intercepts as in \textcite{bonhomme2015grouped}. We studied this extension in a previous version of this paper \parencite{dzemskiokui2018confidence}. Our present focus on time-invariant coefficients is motivated by the literature \parencite[see, e.g.,][]{su2016identifying,wang2016homogeneity,vogt2015classification} and theoretical considerations. 
The conditions for establishing the validity of our method for models with time-varying coefficients are less transparent and substantially more restrictive. Intuitively,  time-varying coefficients are identified purely from cross-sectional variation and are, therefore, estimated at a slower rate than time-invariant coefficients. Therefore, stronger rate conditions requiring at least $T \log N / N \to 0$ are needed to control estimation error in the time-varying coefficients. 		

\label{QE:unit identities matter}
Our confidence set is tailored to address research questions where unit identities are relevant. For example, in our application, we can identify (at a pre-specified confidence level) a set of US states that exhibit a positive effect of the minimum wage on employment. Identifying such states is relevant for conducting further research or implementing targeted policies. 
Our method cannot be applied directly to assess the effect of misclassification on statistics that average over units without regard to unit labels. One example of such a statistic is the estimator of the group-specific slope coefficients. 
Developing a theory tailored to averages is an interesting research agenda that is complementary to our work. 

\printbibliography

\end{document}


\maketitle

\appendix

\renewcommand{\theequation}{\thesection.\arabic{equation}}
\numberwithin{table}{section}
\numberwithin{figure}{section}

\tableofcontents

\section{\label{appendix:proofs}Proofs of main results}

\subsection{Notation}  

We introduce additional notation. We consider statistics that replace some (not all) estimated components in the statistics defined in the main text with population quantities. 
Let 
\begin{align*}
	D_{i}(h) = \frac{\sum_{t = 1}^T d_{it}(g_i^0, h) /\sqrt{T}}{ \sqrt{ \Xi_i (h,h)} }
\end{align*}
and
\begin{align*}
    \widetilde{H}_{ij} (h,h')= \frac{1}{T}\sum_{t=|j|+1}^T  \left( d_{i, t+ \min(0,j)} (g_i^0, h) - \bar{d}_{i} (g_i^0, h) \right)\left( d_{i,t- \max(0,j)} (g_i^0, h') - \bar{d}_{i} (g_i^0, h') \right),
\end{align*}
where $\bar{d}_{i} (g_i^0, h) = \sum_{t = 1}^T d_{it} (g^0, h) / T$. Let
\begin{align*}
\widetilde{\Xi}_{i} (h,h')= \sum_{j=-T+1}^{T-1} K \left( \frac{j}{\kappa_N} \right) \widetilde{H}_{ij} (h,h').
\end{align*}
Let 
\begin{align*}
	\widetilde{D}_{i}(h) = \frac{\sum_{t = 1}^T d_{it}(g_i^0, h) /\sqrt{T}}{ \sqrt{ \widetilde{\Xi}_i (h, h)} },
\end{align*}
and let $\widetilde{\Omega}_i(g_i^0)$ denote the $(G-1) \times (G-1)$ matrix with entries 
\begin{align*}
	\left(\widetilde{\Omega}_i(g_i^0)\right)_{j, j'} =
	\frac{ \widetilde{\Xi}_i (h,h') }{\sqrt{ \widetilde{\Xi}_i (h,h)   \widetilde{\Xi}_i (h',h')}}.
\end{align*}
We write $d_{it}(h) = d_{it} (g_i^0, h)$.

\subsection{Proofs}

\begin{proof}[Proof of Theorem~\ref*{thm:Bonferroni:bounds}]
    The result follows directly from Lemma~\ref{lem:bonferroni}.
\end{proof}

\begin{proof}[Proof of Theorem~\ref*{thm:MAX-dep-lv}]
Write $k_{\alpha, N} (\Omega)$ for the $1 - \alpha/N$-quantile of a $N(0, \Omega)$ random variable. Abbreviate $\Omega_i = \Omega_i(g_i^0)$ and $\widehat{\Omega}_i = \widehat{\Omega}_i(g_i^0)$.

Note first that by Assumption~\ref*{assumption:basic}.\ref*{assumption:max:G_consistent}, for any sequence $P_N$ such that $P_N \in \mathcal{P}_N$, 
\begin{align*}
	P_N \left(\widehat{G} \neq G \right) = o (1). 
\end{align*}
Therefore, proving the theorem on the event $\{\widehat{G} = G\}$ suffices.
Let $\zeta_N$ denote a diverging sequence, $\zeta_N \to \infty$.
\label{QE:referee2 typos3}
For $\bar{\alpha}_N = \alpha\left(1 + 2 C_{\text{\ref*{lem:comparison_bound}}} \sqrt{\epsilon_N \log (N/\alpha)}\right)$ and $C_{\text{\ref*{lem:comparison_bound}}}$ the constant from Lemma~\ref{lem:comparison_bound}, we first establish the following chain of inequalities 
\begin{align}
\label{eq:ineq_t_normal_crit_val}
k_{\bar{\alpha}_N, N} \left(\Omega_i \right) \leq k_{\alpha, N} \left(\rho(\widehat{\Omega}_i, \epsilon_N) \right) \leq c_{\alpha, N} \left(\rho(\widehat{\Omega}_i, \epsilon_N) \right). 
\end{align}
To prove the second inequality, let $t_{T-1} (\cdot)$ denote the cumulative distribution function of a $t$-distributed random variable with $(T - 1)$-degrees of freedom and let $X$ denote a $(G-1)$ random vector distributed according to centered multivariate $t$-distribution with $T - 1$ degrees of freedom and scale matrix $\rho (\widehat{\Omega}_i, \epsilon_N)$.
The marginal distribution of the first component of $X$, denoted by $X_1$, is $X_1 \sim t_{T-1}$. Let $d_N = t_{T-1}^{-1}(1- \alpha/N)$ and note that $d_N \to \infty$. Moreover, 
\begin{align*}
	\alpha/N = P \left(X_1 > d_N \right) \leq P \left(\max_{h \in 1, \dotsc, G-1} X_h > d_N \right).
\end{align*}
Therefore, $c_{\alpha, N} \big( \rho (\widehat{\Omega}_i, \epsilon_N) \big) \geq \sqrt{T/(T-1)} d_N$ and for $N_0$ and $T_0$ independent of $\rho (\widehat{\Omega}_i(g_i^0), \epsilon_N)$ and $t^*$ the constant defined in Lemma~\ref{lem:bound_MVT_by_normal} we can take 
\begin{align*}
	c_{\alpha, N} \big( \rho (\widehat{\Omega}_i, \epsilon_N) \big) > t^*,
\end{align*}
for all $N \geq N_0$, $T \geq T_0$. Therefore, the assumptions of Lemma~\ref{lem:bound_MVT_by_normal} are satisfied for $t =  c_{\alpha, N} \big( \rho (\widehat{\Omega}_i, \epsilon_N) \big)$ and $N$, $T$ large enough and Lemma~\ref{lem:bound_MVT_by_normal} implies 
\begin{align*}
    & \Phi_{\max, \rho (\widehat{\Omega}_i, \epsilon_N)} \left( k_{\alpha, N} \big( \rho (\widehat{\Omega}_i, \epsilon_N) \big) \right) 
\\
    =& 
	1 - \alpha/ N 
\\
    =& t_{\max, \rho (\widehat{\Omega}_i, \epsilon_N), T - 1} \left( \sqrt{(T-1)/T} c_{\alpha, N} \big( \rho (\widehat{\Omega}_i, \epsilon_N) \big) \right)
\\
	\leq & \Phi_{\max, \rho (\widehat{\Omega}_i, \epsilon_N)} \left( c_{\alpha, N} \big( \rho (\widehat{\Omega}_i, \epsilon_N) \big) \right)
\end{align*}
and therefore $k_{\alpha, N} \big( \rho (\widehat{\Omega}_i, \epsilon_N) \big) \leq c_{\alpha, N} \big( \rho (\widehat{\Omega}_i, \epsilon_N) \big)$. 
We now establish the first inequality in \eqref{eq:ineq_t_normal_crit_val}. Note that
\begin{align*}
    \alpha_N \equiv \alpha \left(1 + C_{\ref*{lem:comparison_bound}} c_{\alpha, N} \left(\rho(\widehat{\Omega}_i, \epsilon_N) \right) \right) \leq \alpha \left(1 + 2 C_{\ref*{lem:comparison_bound}} \sqrt{\log (N/\alpha)} \right)
\end{align*}
by Lemma~\ref{lem:bounds:large:quantiles} and that $k_{\alpha, N}$ is decreasing in $\alpha$.
Therefore, it suffices to establish 
\begin{align}
    \label{eq:critical value inequality}
	k_{\alpha_N, N} \left(\Omega_i \right) \leq k_{\alpha, N} \left(\rho(\widehat{\Omega}_i, \epsilon_N) \right).
\end{align}
It can be established by applying Lemma~\ref{lem:comparison_bound}. 
Lemma~\ref{lem:lv-cor-consistency} yields
\begin{align*}
	\norm{\widehat{\Omega}_{i} - \Omega_i}_{\max} \leq  \zeta_N \left(\frac{r_{\theta, N}}{ \iota_N \min_{1 \leq i \leq N}\sigma_i} +  T^{-c_1} (\log N)^{c_2} + T^{-\rho} \right) \equiv \delta_N,
\end{align*} 
on an event whose probability approaches one.
On this set, Lemma~\ref{lem:bounds:large:quantiles} states that we can take $k_{\alpha, N} (\rho(\widehat{\Omega}_i, \epsilon_N)) > \sqrt{\log (N/\alpha)} > \sqrt{2}$ for $N$ large enough. 
We now show that $\delta_N / \epsilon_N \to 0$. Since $\zeta_N$ can be taken to diverge at an arbitrarily slow rate, it suffices to show that, for $\pi > 0$, 
\begin{align*}
    r_{\theta, N} / \big (\iota_N \min_{1 \leq i \leq N}\sigma_i \big) \leq & \epsilon_N T^{-\pi},
\\
    T^{-c_1} (\log N)^{c_2}  \leq & \epsilon_N T^{-\pi},
\\
    T^{-\rho} \leq & \epsilon_N T^{-\pi}.
\end{align*}
\label{QE:referee2 typo4}
Since $r_{\theta, N} \sqrt{T \log N} = o (\iota_N\min_{1 \leq i \leq N} \sigma_i)$ and $\epsilon_N \geq (\log N)^{-k_2}$, the first condition is met provided 
\begin{align*}
    T^{-1 + 2\pi} (\log N)^{-1 + 2 k_2} = O(1).
\end{align*}
Under the assumed rate condition $N \leq o(1) T^{\delta}_2$, this holds for any $0 < \pi < \frac{1}{2}$.  
The second and third conditions can be checked similarly.
Now that we have established $\delta_N/\epsilon_N \to 0$ we can take $4 \delta_N \leq \epsilon_N$, satisfying one of the condition of Lemma~\ref{lem:comparison_bound}. The condition $\epsilon_N \leq 4 c_\omega / 3$ is assumed. 
This argument verifies all conditions of Lemma~\ref{lem:comparison_bound} and yields inequality \eqref{eq:critical value inequality} and therefore \eqref{eq:ineq_t_normal_crit_val}.

Our assumptions, 
\begin{align*}
    r_{\theta, N} \sqrt{T \log N} / \big (\iota_N \min_{1 \leq i \leq N}\sigma_i \big) \to 0
\end{align*}
and $N \leq o(1) T^{\delta_2}$ guarantee that $(\sqrt{T} \vee \log N) r_{\theta, N} / (\iota_N \min_{1 \leq i \leq N} \sigma_i)$ and $T^{-c_1} (\log N)^{c_2}$ vanish. Therefore, 
\begin{align*}
	b_N^{LV*} =  \frac{ (\sqrt{T}  \vee \log N) r_{\theta, N}}{ \iota_N \min_{1 \leq i \leq N}\sigma_i} +  T^{-c_1} (\log N)^{c_2} + T^{-\rho}
\end{align*}
vanishes and Lemma~\ref{lem:dhat-d-lv} can be applied and yields
\begin{align}
    \label{eq:Dhat convergence}
	P\left(\max_{1 \leq i \leq N} \max_{h \in \mathbb{G} \setminus \{g_i^0\}} \abs{\widehat{D}_i(h) - D_i(h)} > \zeta_N b_N^{LV*} \right) = o(1).
\end{align}
%
We now prepare to apply the high-dimensional CLT in Lemma~\ref{lem:clt-dependent}.
    Let $\delta_i(h) = \theta_{g^0_i} - \theta_h$ and 
	\begin{align*}
	X_{it} (h) = \frac{d_{it} (h)}{\sqrt{ T^{-1}\sum_{t=1}^T \sum_{s=1}^T \E_p [ d_{it} (h) d_{is} (h)]}} 
	= - \frac{x_{it}' \left(\delta_i(h)/\norm{\delta_i(h)} \right) v_{it}}{ \xi_i(h)/(\sigma_i \norm{\delta_i (h)})}, 	
	\end{align*}
	where
	\begin{align*}
		\xi_i (h) = \sqrt{\frac{1}{T} \sum_{t = 1}^T \sum_{s=1}^T \E [d_{it} (h) d_{is} (h)]} =  \sqrt{\delta_i(h)' \left(\frac{1}{T} \sum_{t = 1}^T \sum_{s = 1}^T \E_P \left[ x_{it} x_{is}' u_{it} u_{is} \right] \right) \delta_i (h)}.  
	\end{align*}
	Define the vector 
	\begin{align*}
		X_t = \big((X_{1t}(h))_{h \in \mathbb{G} \setminus \{g^0_1\}}, \dotsc, (X_{Nt}(h))_{h \in \mathbb{G} \setminus \{g^0_N\}}\big)'.
	\end{align*}
	This vector has length $J = N (G- 1)$. 
	Let $\Xi$ denote the long-run variance of the time series $X_t$ defined as the $J \times J$ matrix with entry $(j, k)$ given by 
	\begin{align*}
		\cov \left(
			\frac{1}{\sqrt{T}} \sum_{t = 1}^T X_{t, j},
			\frac{1}{\sqrt{T}} \sum_{t = 1}^T X_{t, k}
		\right).
	\end{align*}
	Let $G$ denote a centered normal vector with variance matrix $\Xi$. 
    Clearly, $X_i$ is a normal random vector with covariance matrix $\Omega_i$. 
    \label{QE:new CLT argument}
    Taking complements, we have
	\begin{align}
        \notag
		& \sup_{(r_1, \dotsc, r_N) \in \reals_{++}^N} 
		\begin{aligned}[t]
		\bigg\lvert
		& P \bigg(\max_{1 \leq i \leq N} \bigg( \max_{h \in \mathbb{G} \setminus \{g^0_i\}} D_i(h) - r_i  \bigg) > 0 \bigg)
		\\
        \notag
		& - P \bigg(\max_{1 \leq i \leq N} \left( \max_{1 \leq h \leq G - 1} X_{i, h} - r_i \right) >0 \bigg)
		\bigg\rvert
		\end{aligned}
		\\
        \notag
		= & 
		\sup_{(r_1, \dotsc, r_N) \in \reals_{++}^N} 
        \begin{aligned}[t]
		\big\lvert
		& P \left(D_i (h) \leq r_i \: \text{for all $h \in \mathbb{G} \setminus \{g_i^0\}$ and $i=1,\dotsc, N$}\right)
		\\
        \notag
        & - P \left(X_{i, h} \leq r_i \: \text{for all $h =1, \dotsc, G-1$ and $i=1,\dotsc, N$}\right) \big\rvert
		\end{aligned}
        \\
        \notag
        \leq & \sup_{a \in \reals^{(G-1)N}} \left\lvert P \left(\frac{1}{\sqrt{T}} \sum_{t = 1}^T Z_{t} \leq a \right)
        - P \left(X \leq a \right)
        \right\rvert 
		\\
        \label{eq:large CLT}
	 \leq & C \left( \frac{( \log N)^{(1 + 2 d_1 )/(3 d_1 )}}{T^{1/9}} + \frac{ ( \log N )^{7/6}}{T^{1/9}} \right) = o (1).
	 \end{align} 
    Here, the last inequality holds by Lemma \ref{lem:clt-dependent} and the asymptotic order follows from $N \leq o(1) T^{\delta_2}$. 
    Now, we have 
    \begin{align}
        \notag
        & P\left(\exists i \in 1, \dotsc, N \text{ such that } \hat{T}_i (g_i^0) > c_{\alpha, N} (\widehat{\Omega}_i) \right) 
    \\ \notag	
        \leq & P\left(\exists i \in 1, \dotsc, N \text{ such that } \hat{T}_i (g_i^0) > k_{\bar{\alpha}_N, N} (\Omega_i) \right) 
    \\ \notag
        \leq & P\left(\max_{1 \leq i \leq N} \max_{h \in \mathbb{G} \setminus \{g_i^0\}} \left( D_i( h) - k_{\bar{\alpha}_N, N}(\Omega_i) +\zeta_N b_N^{LV*}\right) > 0 \right) 
        + o (1).
    \\ \notag 
        \leq & P \bigg(\max_{1 \leq i \leq N} \left( \max_{1 \leq h \leq G - 1} X_{i, h} - k_{\bar{\alpha}_N, N} (\Omega_i) + \zeta_N b_N^{LV*} \right) >0 \bigg)
    \\ \notag
        & + \sup_{(r_1, \dotsc, r_N) \in \reals_{++}^N} 
            \begin{aligned}[t]
                \bigg\lvert
                & P \bigg(\max_{1 \leq i \leq N} \bigg( \max_{h \in \mathbb{G} \setminus \{g^0_i\}} D_i(h) - r_i  \bigg) > 0 \bigg)
                \\
                & - P \bigg(\max_{1 \leq i \leq N} \left( \max_{1 \leq h \leq G - 1} X_{i, h} - r_i \right) >0 \bigg)
                \bigg\rvert + o (1)
            \end{aligned} 
    \\ \label{eq:coverage probability as function of Gaussian}
        \leq & P \bigg(\max_{1 \leq i \leq N} \left( \max_{1 \leq h \leq G - 1} X_{i, h} - k_{\bar{\alpha}_N, N} (\Omega_i) + \zeta_N b_N^{LV*} \right) >0 \bigg) + o (1),
    \end{align}
    where the first inequality follows by \eqref{eq:ineq_t_normal_crit_val}, the second inequality follows by \eqref{eq:Dhat convergence}, the third inequality holds because the sup bounds deviations for all choices of $r_i$ and therefore in particular $r_i = k_{\bar{\alpha}_N, N} (\Omega_i) - \zeta_N b_N^{LV*}$, and the fourth inequality follows by \eqref{eq:large CLT}.

Next, applying an anti-concentration argument eliminates the $\zeta_N b_N^{LV*}$ term on the right-hand side of the previous display. 
\label{QE:new Nasarov argument}
To this end, let $a > 0$ and write $k_{N,i} = k_{\bar{\alpha}_N, N} (\Omega_i)$. Then 
\begin{align*}
    &  P \left(\max_{1 \leq i \leq N} 
	\left(\max_{1 \leq h \leq G - 1} X_{i, h} - k_{N,i} \right) + a > 0 
	\right) -  P \left(\max_{1 \leq i \leq N} 
	\left(\max_{1 \leq h \leq G - 1} X_{i, h} - k_{N,i} \right) > 0 
	\right)
\\
    = &  P \left(\max_{1 \leq i \leq N} 
    \left(\max_{1 \leq h \leq G - 1} X_{i, h} - k_{N,i} \right) \leq 0 
    \right) -  P \left(\max_{1 \leq i \leq N} 
    \left(\max_{1 \leq h \leq G - 1} X_{i, h} - k_{N,i} \right) \leq - a
    \right)
\\
    =& P (X \leq x + a) - P (X \leq x), 
\end{align*}
where 
\begin{align*}
    x = \big(
        \underbrace{k_{N, 1}-a, \dotsc, k_{N,1} - a}_{\text{$G-1$ times}}, 
        \underbrace{k_{N, 2}-a, \dotsc, k_{N,2} - a}_{\text{$G-1$ times}}, \dotsc, 
        \underbrace{k_{N, N}-a, \dotsc, k_{N,N} - a}_{\text{$G-1$ times}}
    \big)
\end{align*}
The Nasarov-type inequality from Lemma A.1 in \textcite{chernozhukov2016central} applies with $b = 1$ and $p = (G-1)N$ and yields 
\begin{align*}
    P (X \leq x + a) - P (X \leq x) \leq C_{\text{Nasarov}} a \sqrt{\log (N(G-1))}
    \leq O(1) a \log N.
\end{align*}
Therefore, 
\begin{align*}
    & P \left(\max_{1 \leq i \leq N} 
	\left(\max_{1 \leq h \leq G - 1} X_{i, h} - k_{N,i} \right) + a > 0 
	\right) 
\\
    \leq & 
    P \left(\max_{1 \leq i \leq N} 
	\left(\max_{1 \leq h \leq G - 1} X_{i, h} - k_{N,i} \right) > 0 
	\right) + O(1) a \sqrt{\log N}
\end{align*}
Now, combining this inequality with \eqref{eq:coverage probability as function of Gaussian} by putting $a = \zeta_N b_N^{LV*}$ yields
\begin{align*}
	& P\left(\max_{1 \leq i \leq N} \max_{h \in \mathbb{G} \setminus \{g_i^0\}} D_i(h) - k_{\bar{\alpha}_N, N}(\Omega_i) + \zeta_N b_N^{LV*} > 0 \right) 
\\
	\leq &  P \left(\max_{1 \leq i \leq N} 
	\left(\max_{1 \leq h \leq G - 1} X_{i, h} - k_{\bar{\alpha}_N, N} (\Omega_i) \right) + \zeta_N b_N^{LV*} > 0 
	\right) + o(1)
\\
    \leq & 
    P \left(\max_{1 \leq i \leq N} 
	\left(\max_{1 \leq h \leq G - 1} X_{i, h} - k_{\bar{\alpha}_N, N} (\Omega_i)\right) > 0 
	\right) + O(1) \zeta_N b_N^{LV*} \sqrt{\log N} + o(1) 
\\
	\leq & \sum_{1 \leq i \leq N} P \left(
	\max_{1 \leq h \leq G - 1} X_{i, h} - k_{\bar{\alpha}_N, N} (\Omega_i)  > 0 
	\right) + o (1)
\\
	= & \bar{\alpha}_N + o (1)
\\
	= & \alpha + 2 C_{\ref*{lem:comparison_bound}} \sqrt{\epsilon_N \log (N/\alpha)} + o (1) \to \alpha
	.
\end{align*}
\end{proof}

\begin{proof}[Proof of Theorem~\ref*{thm:SNS critical values}]
    Since the marginal distributions of a multi-variate $t$-distribution with $\nu$ degrees of freedom are Student-$t$ with $\nu$ degrees of freedom, it holds that 
    \begin{align*}
        \left(t_{\max, \rho(\widehat{\Omega}_i(g_i^0), \epsilon_N), T - 1}\right)^{-1} \left(1 - \frac{\alpha}{N} \right)
        \leq  t_{T- 1}^{-1}\left(1 - \frac{\alpha}{(G- 1) N}\right)
    \end{align*}
    by the union bound. Thus, replacing our critical value with the SNS critical values yields a more conservative test.
    Now, inspection of the proof of Theorem~\ref*{thm:MAX-dep-lv} shows that Assumption~\ref*{assumption:basic}.\ref*{assum:Omega negative correlations} is only used to argue that $\widehat{\Omega}_i (g_i^0)$ in the definition of the critical value can be replaced by the population quantity $\Omega_i (g_i^0)$. Since we are replacing the critical value that depends on $\widehat{\Omega}_i (g_i^0)$ by a critical value that is independent of $\widehat{\Omega}_i (g_i^0)$, this step is not needed. Similarly, the SNS critical value is independent of the regularization sequence, and the assumptions on $\epsilon_N$ are therefore unnecessary. 
\end{proof}

\begin{proof}[Proof of Theorem~\ref*{thm:no serial correlation}]
    The proof is similar to the proof of Theorem~\ref*{thm:MAX-dep-lv}, replacing the application of Lemma~\ref{lem:lv-cor-consistency} and Lemma~\ref{lem:dhat-d-lv} by and application of Lemma~\ref{lem:no serial correlation}.
\end{proof}

\begin{proof}[Proof of Theorem~\ref*{thm:two step}]
\label{QE:referee2 typos5}
Let
\begin{align*}
	 d_{it}^U (g, h) = (y_{it} - w_{it}'\theta^w - x_{it}' \theta_{g})^2 -  (y_{it} - w_{it}'\theta^w - x_{it}' \theta_{h})^2
\end{align*}
and $d_i^U(h) = d_{it}^U(g_i^0, h)$, $\hat{d}_i^U(h) = \hat{d}_i^U(g_i^0, h)$, $\bar{d}_i^U(h) = N^{-1} \sum_{t = 1}^{T} d_{it}^U(h)$, and $\bar{\hat{d}}_i^U(h) = N^{-1} \sum_{t = 1}^{T} \hat{d}_{it}^U(h)$.
We note that the hypothesis selection part of the procedure does not affect the theoretical analysis. This is because, here, we focus on size and thus need to consider only the behavior of the test statistics under $\{ g^0_i \}_{i=1}^N$.

In the following $o(1)$ is understood such that $a=o(1)$ if $\limsup_{N,T \to \infty} |a| =0$.

Let $J= \{ (i,h) \mid i \in \{1,\dots, N\} , h \in \mathbb{G}\backslash \{g^0_i\}   \} $ and 
\begin{align*}
 J_1 = \left\{(i,h) \mid i \in \{1,\dots, N\} , h \in \mathbb{G}\backslash \{g^0_i\} ,
\frac{\sqrt{T} \E_P (\bar d_{i}^U (h)) }{s_{i,T}^U (h)}
> - c_{\beta, N}^{\mathrm{SNS}}  \right\},
\end{align*}
where $	(s_{i}^U ( h))^2 = \sum_{t=1}^T \var (  d_{it}^U (h ) ) /T = \var (  d_{it}^U (h ) ) $ (the equality follows by stationarity).
Roughly speaking, $J_1$ is the set of pairs of units and groups that are difficult to distinguish from the true group membership.

\paragraph{Step 1:}
We first prove that $ \inf_{P \in \mathbb{P}_N} P \left(  \max_{ (i, h) \in J_1^c} 
 \bar{\hat{d}}_{i}^U ( h) \le 0 \right) > 1 - \beta - N^{-1} - C T^{-c} - a_{\theta, N}$.

Note that $\bar{\hat{d}}_{i}^U ( h) >0 $ for some $(i,h) \in J_1^c$ implies that 
\begin{align*}
\max_{(i,h) \in J}\frac{\sqrt{T} (\bar{\hat{d}}_{i}^U ( h)  - \E_P (\bar d_{i}^U (h)) ) }{  s_{i,T}^U ( h)}
>  c_{\beta, N}^{\mathrm{SNS}} .
\end{align*}

We observe that 
\begin{align*}
&\sup_{P \in \mathbb{P}_N} P \left( 
\max_{(i,h) \in J}\frac{\sqrt{T} (\bar{\hat{d}}_{i}^U ( h)  - \E_P (\bar d_{i}^U (h)) ) }{s_i^U(h) }
>  c_{\beta, N}^{\mathrm{SNS}}
\right)  \\
\le &  
\sup_{P \in \mathbb{P}_N}  P \left( 
\max_{(i,h) \in J}\frac{\sqrt{T} (\bar d_{i}^U ( h)  - \E_P (\bar d_{i}^U (h)) ) }{s_i^U(h) }
>  c_{\beta, N}^{\mathrm{SNS}} - e_{N,1}^U
\right)  
\\
& + \sup_{P \in \mathbb{P}_N} P \left( 
\max_{(i,h) \in J}\left| \frac{\sqrt{T} (\bar{\hat{d}}_{i}^U ( h)  - \bar d_{i}^U (h))  }{
 s_i^U(h) }\right|
> e_{N,1}^U
\right),
\end{align*}
where 
\begin{align*}
	e_{N,1}^U = C \frac{r_{\theta, N}}{\iota_N + \min_{1 \leq i \leq N} \sigma_i}.
\end{align*}
The second term on the right-hand side converges to zero by \eqref{eq:duhat_diff_bound} in Lemma \ref{lem:conv:Du}.
Let $\beta_N$ solve $  c_{\beta_N, N}^{\mathrm{SNS}} =   c_{\beta, N}^{\mathrm{SNS}} - e_{N,1}^U$. To see that $\beta_N$ is well-defined, note that since $c_{\beta ,N}^{\mathrm{SNS}} \to \infty$ and $e_{N,1}^U \to 0$ the right-hand side of the equation is diverging and therefore positive for large $N$. Moreover, $c_{p,N}^{\mathrm{SNS}} \downarrow 0$ as $p \uparrow N/2$. We thus establish the existence of $\beta_N$. Uniqueness follows from the strict monotonicity of the distribution function of the $t$-distribution. 
Thus, we have 
\begin{align*}
&	\sup_{P \in \mathbb{P}_N} P \left( 
\max_{(i,h) \in J}\frac{\sqrt{T} (\bar{\hat{d}}_{i}^U ( h) ) - \E_P (\bar d_{i}^U (h)) ) }{s_i^U(h) }
>  c_{\beta, N}^{\mathrm{SNS}}
\right)  \\
\le & \sup_{P \in \mathbb{P}_N} P \left( 
\max_{(i,h) \in J}\frac{\sqrt{T} (\bar d_{i}^U ( h) ) - \E_P (\bar d_{i}^U (h)) ) }{s_i^U(h) }
>  c_{\beta_N,  N}^{\mathrm{SNS}} 
\right)  + o(1) \\
\le & 1- (G-1)N \Phi ( c_{\beta_N,  N}^{\mathrm{SNS}}) + o(1) \\
\le & \beta_N + o(1)\\
= & \beta + o(1),
\end{align*}
where the second inequality follows by Lemma \ref{lem:clt-dependent} and the Bonferroni inequality, the third inequality holds because $c_{\beta_N,  N}^{\mathrm{SNS}}$ becomes sufficiently large as $N\to \infty$ and the tail of the $t$-distribution is heavier than that of the standard normal distribution (Lemma \ref{lem:bound_MVT_by_normal} under the unidimensional case), and the last inequality follows by the fact that $|\beta_N - \beta| \leq C e_{N,1}^U \sqrt{\log ((G-1)N/\beta)}\to 0$. We now show that  $|\beta_N - \beta| \leq C e_{N,1}^U \sqrt{\log ((G-1)N/\beta)}$. 
Let $F_T$ denote the distribution function of a $t$-distributed random variable with $T-1$ degrees of freedom, and let $f_T$ denote its density function. 
Let $c (\beta) = t^{-1}_{T-1} (1 - \beta / ( (G-1) N ) )$ and $e_{N,1}^{U*} = \sqrt{(T-1)/T} e_{N,1}^U$. 
By the mean-value theorem
\begin{align*}
\frac{\beta_N}{(G-1)N } - \frac{\beta}{(G-1)N } = &   F_T ( c( \beta )) -  F_T (c(\beta_N) )
\\
=& F_T ( c( \beta )) -F_T (c(\beta) - e_{N,1}^U )
=  f_T(c^*) e_{N,1}^{U*}, 
\end{align*}
where $c^*$ is a value between $c\left(\beta_N\right)$ and $c\left(\beta \right)$. Noting that $c \left(\beta_N\right) < c \left(\beta \right)$ and that $f_T$ is decreasing on the positive axis, rearranging this equality yields 
\begin{align*}
	\abs{\beta_N - \beta} \leq & f_T\left(c \left(\beta_N\right) \right) (G-1)N e_{N,1}^{U*}
	\\
	\leq &
	2 c \left(\beta_N\right) \left(1 - F_{T}(c\left(\beta_N \right) \right) (G-1) N e_{N,1}^{U*}
	\\
	\leq & 4 e_{N,1}^{U*} \beta_N \sqrt{\log \left( (G - 1) N / \beta_N\right)} 
	\\
	\leq & 4 e_{N,1}^U \beta \sqrt{\log\left( (G - 1) N / \beta \right)} + 4 e_{N,1}^U \abs{\beta_N - \beta} \sqrt{\log\left( (G - 1) N / \beta \right)}
	\\ 
	\leq & 4 e_{N,1}^U \sqrt{\log\left( (G - 1) N / \beta \right)} + o \left(\abs{\beta_N - \beta} \right)
	, 
\end{align*}
\label{QE:referee2 typo6} where the second inequality follows from Lemma~\ref{lem:tdist:bound-millsratio}, the third inequality follows from Lemma~\ref{lem:tdist:tail:upperbound} (with $\epsilon_N = 1$), the fourth inequality follows from $e_{N,1}^U = \sqrt{T/(T-1)} e_{N,1}^{U*} \geq e_{N,1}^{U*}$ and $\beta_N \geq \beta$, the fifth inequality follows from $e_{N,1}^U \sqrt{\log N} \to 0$. This recursion implies 
\begin{align*}
	\abs{\beta_N - \beta} 
	\leq 5 e_{N,1}^U \sqrt{\log\left( (G - 1) N / \beta \right)}
\end{align*}
for $N$ large enough.

An implication of Step 1 is as follows. Let
\begin{align*}
 \mathbb{N} = \left\{ i \in \{1,\dots, N\} \mid  \max_{h \in \mathbb{G}\backslash \{g^0_i\}} 
\frac{\sqrt{T} \E_P ( \bar d_{i}^U (h)) }{ s_i^U(h)}
> -  c_{\beta, N}^{\mathrm{SNS}} 
 \right\}.
\end{align*}
Then 
\begin{align*}
	\inf_{P \in \mathbb{P}_N} P \left(  \max_{ i \in \mathbb{N}^c} 
\max_{h \in \mathbb{G}\backslash \{g^0_i\}} \bar{\hat{d}}_{i}^U ( h) \le 0\right) 
\geq 1 - \beta + o(1).
\end{align*}


\paragraph{Step 2:}
Next, we prove that $\inf_{P \in \mathbb{P}_N} P ( \bigtimes_{i=1}^N \hat M_i (g^0_i)  \supseteq J_1 ) \ge 1 - \beta + o(1)$.
Here, we drop the $g$ argument for simplicity of notation when arguments are $g^0_i$ and $h$.
 
We note that 
\begin{align*}
	& \sup_{P \in \mathbb{P}_N} P\left( \bigtimes_{i=1}^N \hat M_i (g^0_i)  \nsupseteq J_1 \right) 
	\\
	= &\sup_{P \in \mathbb{P}_N}  P \left(\exists (i,h); \hat D_i^U ( h) \le - 2 c_{\beta_N, N}^{\mathrm{SNS}} \text{ and }  \frac{\sqrt{T} \E_P (\bar d_{i}^U ( h)) }{ s_i^U(h)} > - c_{\beta, N}^{\mathrm{SNS}}  \right)  
	\\
	\leq & \sup_{P \in \mathbb{P}_N}  P \left(\exists (i,h);  D_i^U ( h) \le - 2 c_{\beta, N}^{\mathrm{SNS}} +  e_{N,2}^U\text{ and }  \frac{\sqrt{T} \E_P (\bar d_{i}^U ( h)) }{ s_i^U(h)} > - c_{\beta, N}^{\mathrm{SNS}}  \right) 
	\\
	& + \sup_{P \in \mathbb{P}_N} P\left( \max_{1 \le i \le N} \max_{ h \in \mathbb{G} \backslash \{ g_i^0\} } \left| \hat D_i^U ( h) -  D_i^U ( h) \right| > e_{N,2}^U \right),
\end{align*}
where
\begin{align*}
	e_{N,2}^U = C  \frac{r_{\theta, N} ( \sqrt{T} + \sqrt{ \log N})}{\iota_N \wedge \min_{1 \leq i \leq N} \sigma_i}.
\end{align*}
By \eqref{eq:DhatU_conv} in Lemma \ref{lem:conv:Du}, it holds that 
\begin{align*}
	\sup_{P \in \mathbb{P}_N} P\left( \max_{1 \le i \le N} \max_{ h \in \mathbb{G} \backslash \{ g_i^0\} } \left| \hat D_i^U ( h) - \tilde D_i^U ( h) \right| > e_{N,2}^U \right) = o(1).
\end{align*}
	

We observe	
\begin{align*}
	  & \sup_{P \in \mathbb{P}_N} P \left(\exists (i,h);  D_i^U ( h) \le - 2 c_{\beta, N}^{\mathrm{SNS}} + e_{N,2}^U \text{ and }  \frac{\sqrt{T} \E_P (\bar d_{i}^U ( h)) }{ s_i^U(h)} > - c_{\beta, N}^{\mathrm{SNS}}  \right) 
	\\
\le &\sup_{P \in \mathbb{P}_N}	P \left( \max_{1 \le i \le N} \max_{h \in \mathbb{G}\backslash \{g_i^0\}}
 \frac{\sqrt{T} (\E_P ( \bar d_{i}^U ( h) - \bar d_{i}^U ( h))}{s_{i,T}^U ( h)} > \frac{2 \tilde s_i^U (h) - s_i^U(h)}{s_i^U (h)} c_{\beta, N}^{\mathrm{SNS}} -  \frac{\tilde s_i^U (h) }{s_i^U (h)} e_{N,2}^U\right) .
\end{align*}
\label{QE: referee2 typo7} Note that $ \tilde s_i^U (h) s_i^U(h) > 1- r/2 $ is equivalent to 
\begin{align*}
	\frac{2 \tilde s_i^U (h) - s_i^U(h)}{s_i^U (h)} > 1-r.
\end{align*}
Thus, we have 
\begin{align*}
	&\sup_{P \in \mathbb{P}_N}	P \left( \max_{1 \le i \le N} \max_{h \in \mathbb{G}\backslash \{g_i^0\}}
 \frac{\sqrt{T} (\E_P ( \bar d_{i}^U ( h) - \bar d_{i}^U ( h))}{s_{i,T}^U ( h)} > \frac{2 \tilde s_i^U (h) - s_i^U(h)}{s_i^U (h)} c_{\beta, N}^{\mathrm{SNS}} -  \frac{\tilde s_i^U (h) }{s_i^U (h)} e_{N,2}^U\right)  \\
 \leq & \sup_{P \in \mathbb{P}_N}	P \left( \max_{1 \le i \le N} \max_{h \in \mathbb{G}\backslash \{g_i^0\}}
 \frac{\sqrt{T} (\E_P ( \bar d_{i}^U ( h) - \bar d_{i}^U ( h))}{s_{i,T}^U ( h)} > (1-r) c_{\beta, N}^{\mathrm{SNS}} - A e_{N,2}^U\right) \\
 &+  \sup_{P \in \mathbb{P}_N}	P \left( \left| \frac{\tilde s_i^U (h) }{s_i^U (h)}  -1 \right|   > r/2 \right)  \\
 &+  \sup_{P \in \mathbb{P}_N}	P \left(  \max_{1 \le i \le N} \max_{h \in \mathbb{G}\backslash \{g_i^0\}} \left| \frac{\tilde s_i^U (h) }{s_i^U (h)} \right|   > A \right),
\end{align*}
where $A >1 $ is a fixed number.
We now note that 
\label{QE:comment R2 typos8}
 \begin{align*}
 	& \tilde s_i^U (h)^2 - s_i^U (h)^2
\\
 	=& \frac{1}{T} \sum_{t=1}^T \left\{(d_{it}^U (h))^2 - \E_P \left[(d_{it}^U (h))^2)\right] 
 	- (\bar d_{i}^U (h) - \E_P\left[d_{it}^U (h)\right])(\bar d_{i}^U (h) + \E_P\left[d_{it}^U (h)\right]\right\}.
 \end{align*}
 By Lemma \ref{lem:fuk-nagaev-dependent} and \eqref{eq:du_rate}, it holds that 
 \begin{align}
 	\sup_{P \in \mathbb{P}} \left( 
 		\max_{1 \leq i \leq N} \left|
 	 \frac{	\tilde s_i^U (h)^2 - s_i^U (h)^2}{\sigma_i^2 \norm{\delta_i(h) }^2} \right|\geq C T^{-1/2}\log N \right) =o(1).\label{eq:stildeu-su} 
 \end{align}
Because $s_i^U(h) > s_i(h)$ and $s_i(h)/ (\sigma_i \norm{\delta_i(h)}))$ is bounded from above and from below by Assumption \ref*{assumption:basic}.\ref*{assumption:max:min_eigenvalue}, it holds that 
\begin{align}
	\sup_{P \in \mathbb{P}} \left( 
		\max_{1 \leq i \leq N} \left|
	 \frac{\tilde s_i^U (h)}{s_i^U (h)} -1  \right|\geq C T^{-1/4}\sqrt{\log N} \right) =o(1).
\end{align}
We take $r = T^{-1/4} \sqrt{\log N}$. We then have 
\begin{align*}
	\sup_{P \in \mathbb{P}_N}	P \left( \left| \frac{\tilde s_i^U (h) }{s_i^U (h)}  -1 \right|   > r/2 \right) +  \sup_{P \in \mathbb{P}_N}	P \left(  \max_{1 \le i \le N} \max_{h \in \mathbb{G}\backslash \{g_i^0\}} \left| \frac{\tilde s_i^U (h) }{s_i^U (h)} \right|   > A \right) =o(1).
\end{align*}
Let $\beta_N' $ be such that $c_{\beta_N', N}^{\mathrm{SNS}}  =  (1-r) c_{\beta, N}^{\mathrm{SNS}} - A e_{N,2}^U$. 
We then examine
\begin{align*}
	& \sup_{P \in \mathbb{P}_N}	P \left( \max_{1 \le i \le N} \max_{h \in \mathbb{G}\backslash \{g_i^0\}}
	\frac{\sqrt{T} (\E_P ( \bar d_{i}^U ( h) - \bar d_{i}^U ( h))}{s_{i,T}^U ( h)} > (1-r) c_{\beta, N}^{\mathrm{SNS}} - A e_{N,2}^U\right) \\
	=& \sup_{P \in \mathbb{P}_N}	P \left( \max_{1 \le i \le N} \max_{h \in \mathbb{G}\backslash \{g_i^0\}}
	\frac{\sqrt{T} (\E_P ( \bar d_{i}^U ( h) - \bar d_{i}^U ( h))}{s_{i,T}^U ( h)} >  c_{\beta_N', N}^{\mathrm{SNS}} \right) \\
\leq & 1- (G-1)N \Phi (c_{\beta_N', N}^{\mathrm{SNS}})  + o(1)\\
\leq & \beta_N' + o(1) \\
= & \beta + o(1),
\end{align*}
where the first inequality follows by Lemma \ref{lem:clt-dependent} and the Bonferroni inequality, the second inequality holds because $c_{\beta_N',  N}^{\mathrm{SNS}}$ becomes sufficiently large as $N\to \infty$ and the tail of the $t$-distribution is heavier than that of the standard normal distribution (Lemma \ref{lem:bound_MVT_by_normal} under the unidimensional case), and the last inequality follows by the fact that $|\beta_N' - \beta| \leq C e_{N,2}^U \sqrt{\log ((G-1)N/\beta)}\to 0 $ shown in Step 1.

Summing up, we have 
\begin{align*}
	\sup_{P \in \mathbb{P}_N} P \left( \bigtimes_{i=1}^N \hat M_i (g^0_i)  \nsupseteq J_1 \right) \le \beta  + o(1).
\end{align*}

An implication of Step 2 is as follows. Let 
\begin{align*}
\hat{\mathbb{N}} = \left\{ i \in \{ 1, \dots, N \} \mid M_i (g^0_i) \neq \varnothing \right\}.
\end{align*}
Then 
\begin{align*}
	\inf_{P \in \mathbb{P}_N} P \left( \hat{ \mathbb{N}} \supseteq \mathbb{N} \right) \ge 1 - \beta + o(1). 
\end{align*}

\paragraph{Step 3:}
\label{QE:referee2 typos10} 
First, consider the case in which $J_1 = \varnothing$. In this case, the argument in Step 1 yields that
\begin{align*}
	\inf_{P \in \mathbb{P}_N} P (\hat g_i = g^0_i , \forall i )
=
\inf_{P \in \mathbb{P}_N} P \left( \max_{1 \le i \le N} \max_{h \in \mathbb{G}\backslash \{g_i^0\}}
 \hat D_i^U ( h) \le 0 \right) \geq 1 - \beta + o(1).
\end{align*}
\label{QE:referee2 typo 9} The equality in the above display follows because $\hat g_i$ minimizes the squared loss in the two-step procedure (see (\ref*{eq:ghat:squared_loss}) in the main text).
Because $\{\hat{g}_i\}_{i=1}^N$ is always included in the confidence set, the limiting probability of the confidence set not including $\{ g^0_i\}_{i=1}^N$ is less than $\beta < \alpha$.

Next, consider the case in which $|J_1| \ge 1$. 
Observe that 
\begin{align*}
& \sup_{P \in \mathbb{P}_N} P \left( \{ g^0_i\}_{i=1}^N \notin \hat{C}_{\mathrm{sel}, \alpha, \beta} \right) \\
=& \sup_{P \in \mathbb{P}_N} P \left(\bigcup_{i=1}^N \left( \left\{ 
\hat T_i (g^0_i) >
\hat{c}_{\alpha -2 \beta, \hat N, i} (g_i^0)  
\right\}  \cap \left\{ \max_{h \in \mathbb{G} \backslash \{g^0_i\}} \hat D_i^U ( h)  > 0 \right\}
\right) \right) \\
\le &  \sup_{P \in \mathbb{P}_N} P \left( \bigcup_{i\in \mathbb{N}} \left\{ 
\hat T_i (g^0_i) > 
\hat{c}_{\alpha -2 \beta, \hat N, i} (g_i^0)  
\right\}  \cup
\bigcup_{i\in \mathbb{N}^c} \left\{ \max_{h \in \mathbb{G} \backslash \{g^0_i\}} \hat D_i^U ( h)   > 0 \right\}
 \right) \\
\le &
\sup_{P \in \mathbb{P}_N} P  \left( \bigcup_{i\in \mathbb{N}} \left\{ 
\hat T_i (g^0_i) > 
\hat{c}_{\alpha -2 \beta, \hat N, i} (g_i^0)  
\right\}\right) 
+ \sup_{P \in \mathbb{P}_N} P \left( \bigcup_{i\in \mathbb{N}^c} \left\{ \max_{h \in \mathbb{G} \backslash \{g^0_i\}} \hat D_i^U ( h)   > 0 \right\}
 \right).
\end{align*}
By Step 1, we have 
\begin{align*} 
	\sup_{P \in \mathbb{P}_N} P \left( \bigcup_{i\in \mathbb{N}^c} \left\{ \max_{h \in \mathbb{G} \backslash \{g^0_i\}} \hat D_i^U ( h)  > 0 \right\}
 \right) \le \beta + o(1).
\end{align*}
By Step 2, we have 
\begin{align*}
&\sup_{P \in \mathbb{P}_N}  P  \left( \bigcup_{i\in \mathbb{N}} \left\{ 
\hat T_i (g^0_i) > 
\hat{c}_{\alpha -2 \beta, \hat N, i} (g_i^0)  
\right\}\right) \\
\le &  
\sup_{P \in \mathbb{P}_N} P  \left(\{ \hat{\mathbb{N}} \supseteq \mathbb{N}\} \cap \bigcup_{i\in \mathbb{N}} \left\{ 
\hat T_i (g^0_i) > 
\hat{c}_{\alpha -2 \beta, \hat N, i} (g_i^0)  
\right\}\right) 
+ \sup_{P \in \mathbb{P}_N} P (\{ \hat{\mathbb{N}} \nsupseteq \mathbb{N}\}) \\
\le & 
\sup_{P \in \mathbb{P}_N} P  \left( \bigcup_{i\in \mathbb{N}} \left\{ 
\hat T_i (g^0_i) > 
\hat{c}_{\alpha -2 \beta, |\mathbb{N}|} (g_i^0)
\right\}\right) + \beta+ o(1).
\end{align*}
Thus, we have 
\begin{align*}
	\sup_{P \in \mathbb{P}_N} P \left( \{ g^0_i\}_{i=1}^N \notin \hat{C}_{\mathrm{sel}, \alpha, \beta} \right) 
\le  
\sup_{P \in \mathbb{P}_N} P  \left( \bigcup_{i\in \mathbb{N}} \left\{ 
\hat T_i (g^0_i) > 
\hat{c}_{\alpha -2 \beta, |\mathbb{N}|, i} (g_i^0)
\right\}\right) + 2 \beta+ o(1).
\end{align*}
Theorem \ref*{thm:MAX-dep-lv} implies
\begin{align*}
 \limsup_{N,T \to \infty} \sup_{P\in \mathbb{P}_N} P \left( \{ g^0_i\}_{i=1}^N \notin \hat{C}_{\mathrm{sel}, \alpha, \beta} \right) \le \alpha .
\end{align*}
\end{proof}

\subsection{Supporting lemmas}

\begin{lemma}
    \label{lem:bonferroni}
    Let $(\phi_i)_{i=1}^n$ denote a collection of independent, non-randomized tests and suppose that 
    \begin{align*}
        \alpha_i = n \Prob\left( \phi_i > 0 \right)
    \end{align*}
    with $\alpha_{\max} := \max_{i= 1, \dotsc, n} \alpha_i < 1$. 
    Then 
    \begin{align*}
        \alpha_{\min} - \frac{\alpha_{\min}^2}{2} 
        \leq 
        \Prob \left( \max_{i = 1, \dotsc, n} \phi_i > 0 \right)
        \leq
        \alpha_{\max} - \frac{\alpha_{\max}^2}{2} \left(1 - \frac{\alpha_{\max}}{3} + \frac{1}{n} \left(1 - \frac{\alpha_{\max}}{n}\right)^{-2} \right),
    \end{align*}
    where $ \alpha_{\min} := \min_{i= 1, \dotsc, n} \alpha_i $.
    \end{lemma}
    \begin{proof}
    For fixed $0 < x < 1$, let $\bar{x}$ denote a generic intermediate value between zero and $x$. By a Taylor expansion around $x = 0$,  
    \begin{align}
        \label{eq:bonferroni:exp:bound}
        \exp(-x) 	= 1 - x + \frac{1}{2} x^2 - \frac{1}{6} \exp(-\bar{x}) x^3
                    \geq 1 - x + x^2 \left(\frac{1}{2} - \frac{x}{6} \right).	
    \end{align}
    Moreover, 
    \begin{align}
        \label{eq:bonferroni:log:bound}
        \log \left(1 - x\right) = 0 - x - \frac{x^2}{2\left(1 - \bar{x}\right)^2} 
        \geq 
        - x - \frac{x^2}{2\left(1 - x\right)^2}.
    \end{align}
    Now, for $0 < \alpha < 1$, 
    \begin{align*}
        \left(1 - \frac{\alpha}{n}\right)^n =& \exp \left(n \log \left(1 - \frac{\alpha}{n} \right)\right)
    \\
        \geq & 
        \exp(-\alpha) \exp \left(- \frac{\alpha^2}{2 n} \left(1 - \frac{\alpha}{n}\right)^{-2} \right)
    \\
        \geq & \left(1 - \alpha + \alpha^2 \left(\frac{1}{2} - \frac{\alpha}{6}\right) \right)
        \left(1 - \frac{\alpha^2}{2 n} \left(1 - \frac{\alpha}{n}\right)^{-2} \right)
    \\
        \geq & 1 - \alpha + \frac{\alpha^2}{2} \left(1 - \frac{\alpha}{3}\right)  - \frac{\alpha^2}{2 n} \left(1 - \frac{\alpha}{n}\right)^{-2},
    \end{align*}
    where the first inequality uses \eqref{eq:bonferroni:log:bound}, the second inequality uses \eqref{eq:bonferroni:exp:bound} and the last inequality uses
    \begin{align*}
        1 - \alpha + \alpha^2 \left(\frac{1}{2} - \frac{\alpha}{6}\right) \leq 1.
    \end{align*}
    We conclude that 
    \begin{align*}
        & \Prob \left( \max_{i = 1, \dotsc, n} \phi_i > 0 \right)
        = 1 - \Prob \left( \max_{i = 1, \dotsc, n} \phi_i = 0\right)
    \\
        \leq& 1 - \left(1 - \frac{\alpha_{\max}}{n} \right)^n
        \leq \alpha_{\max} - \frac{\alpha_{\max}^2}{2} \left(1 - \frac{\alpha_{\max}}{3}\right) + \frac{\alpha_{\max}^2}{2 n} \left(1 - \frac{\alpha_{\max}}{n}\right)^{-2}.
    \end{align*}
    Next, note that 
    \begin{align*}
        \left(1 - \frac{\alpha}{n} \right)^n \leq \exp(-\alpha) 
        \leq 1 - \alpha + \frac{\alpha^2}{2}
    \end{align*}
    and therefore 
    \begin{align*}
        \Prob \left( \max_{i = 1, \dotsc, n} \phi_i > 0 \right)
         = 1 - \Prob \left( \max_{i = 1, \dotsc, n} \phi_i = 0\right)
         \geq 1 - \left(1 - \frac{\alpha_{\min}}{n} \right)^n
        \geq
        \alpha_{\min} - \frac{\alpha_{\min}^2}{2}
        .
    \end{align*}
\end{proof}

\begin{lemma}
    \label{lem:bound_MVT_by_normal}
    Let $V$ denote a correlation matrix, which is possibly singular, and let $\Phi_{\max,V}$ denote the distribution function of the maximum element of a multivariate normal random vector with covariance matrix $V$.
    There is $t^* \in \mathbb{R}$ independent of $T$ and $V$ such that for all $t > t^*$
    \begin{align*}
        t_{\max, V, T-1} \left( \sqrt{\frac{T - 1}{T}} t \right) \leq \Phi_{\max, V}(t).
    \end{align*}
    \end{lemma}
    
    \begin{proof}
    Let $\mathbf{x}$ be a vector of random variables such that $\mathbf{x} \sim N(0,V)$. 
    By the definitions of $\Phi_{\max, V}$ and $t_{\max ,V, T-1}$, we have 
    \begin{align*}
        \Phi_{\max,V} (t) = P \left( \mathbf{x} \leq t \right) 
    \end{align*}
    and 
    \begin{align*}
        t_{\max, V, T-1} \left( \sqrt{\frac{T-1}{T}} t \right)
        = P \left( \frac{1}{\sqrt{z/(T-1)}} \mathbf{x} \leq \sqrt{\frac{T-1}{T}} t \right) ,
    \end{align*}
    where an inequality such as $\mathbf{x} \leq t$ is understood in an element-wise way, and $z$ is a $\chi^2$ random variable with degree of freedom $T-1$ independent of $\mathbf{x}$.

    Let $r$ be the rank of $V$. We have the following eigen decomposition of $V$:
    \begin{align*}
        V = U \Sigma U',
    \end{align*}
    where $\Sigma$ is a diagonal matrix with non-negative elements and $U$ is a unitary matrix. We arrange the elements of $U$ and $\Sigma$ such that the first $r$ diagonal elements of $\Sigma$ are non-zero and its other diagonal elements are zero. Let $\Sigma_r$ be the $r \times r$ upper-left block of $\Sigma$.
    Let 
    \begin{align*}
        \mathbf{x}^* = U' \mathbf{x}.
    \end{align*}
    By construction, $\mathbf{x}^* \sim N (0, \Sigma)$. Because $\Sigma$ is diagonal and only the first $r$ diagonal elements are non-zero, the first $r$ elements of $\mathbf{x}^*$ can be non-zero, and its other elements are zero. Let $\mathbf{x}_r$ be the vector of the first $r$ elements of $\mathbf{x}^*$. Note that by the definition of $\Sigma_r$, $\mathbf{x}^* \sim N (0, \Sigma_r)$. This observation implies that
    \begin{align*}
        \mathbf{x} = U \mathbf{x}^* = U_r \mathbf{x}_r,
    \end{align*}
    where $U_r$ is the matrix that consists of the first $r$ columns of $U$.

    We can then write
    \begin{align*}
        \Phi_{\max,V} (t) = \int_{\mathbf{x} \le t } \phi_{\Sigma_r} ( \mathbf{x}_r ) \, d \mathbf{x}_r
    \end{align*}
    and 
    \begin{align*}
        t_{\max, V, T-1} \left( \sqrt{\frac{T-1}{T}} t \right)  = \int_{\mathbf{x} \leq \sqrt{\frac{T-1}{T}} t } f^t_{\Sigma_r,T-1} ( \mathbf{x}_r ) \, d \mathbf{x}_r =\int_{\mathbf{x} \leq  t } f_{\Sigma_r,T-1}^{t,*} ( \mathbf{x}_r ) \, d \mathbf{x}_r,
    \end{align*}
    where 
    \begin{align*}
        \phi_{\Sigma_r} ( \mathbf{x}_r ) = (2\pi )^{-r/2}  (\det (\Sigma_r ))^{-1/2} \exp \left( - \frac{1}{2} \mathbf{x}_r'\Sigma_r^{-1} \mathbf{x}_r \right),
    \end{align*}
    and
    \begin{align*}
        f_{\Sigma_r,T-1}^t ( \mathbf{x}_r ) = & (\pi (T-1))^{-r/2} (\det (\Sigma_r ))^{-1/2} \Gamma\left(\frac{T+r-1}{2} \right) \left(\Gamma\left(\frac{T-1}{2} \right)\right)^{-1}  \\
        & \times \left( 1+ \frac{1}{T-1} \mathbf{x}_r'\Sigma_r^{-1} \mathbf{x}_r \right)^{-(T+r-1)/2},
    \end{align*}
    is the density of the multivariate $t$ distribution with scale matrix $V$ and $T-1$ degrees of freedom, and 
    \begin{align*}
        f_{\Sigma_r,T-1}^{t, *} ( \mathbf{x}_r ) = & (\pi T)^{-r/2} (\det (\Sigma_r ))^{-1/2} \Gamma\left(\frac{T+r-1}{2} \right) \left(\Gamma\left(\frac{T-1}{2} \right)\right)^{-1}  \\
        & \times \left( 1+ \frac{1}{T} \mathbf{x}_r'\Sigma_r^{-1} \mathbf{x}_r \right)^{-(T+r-1)/2}.
    \end{align*}
    We now identify a region in which $f_{\Sigma_r,T-1}^{t, *} (\mathbf{x}_r) > \phi_{\Sigma_r} ( \mathbf{x}_r )$. 
    We have 
    \begin{align*}
        &\log f_{\Sigma_r,T-1} (\mathbf{x}_r) -  \log \phi_{\Sigma_r} ( \mathbf{x}_r ) 
        =  A_T - \frac{T+r-1}{2} \log \left( 1+ \frac{1}{T} \mathbf{x}_r'\Sigma_r^{-1} \mathbf{x}_r \right)
        + \frac{1}{2} \mathbf{x}_r'\Sigma_r^{-1} \mathbf{x}_r,
    \end{align*}
    where 
    \begin{align*}
        A_T = - \frac{r}{2}\log ( T) + \log \Gamma\left(\frac{T+r-1}{2} \right) - \log \Gamma\left(\frac{T-1}{2} \right) + \frac{r}{2} \log (2 ).
    \end{align*}
    By the property of the logarithm function and the linear function, there is a unique value, denoted by $x^*_T$, such that $f_{\Sigma_r,T-1}^{t,*} (\mathbf{x}_r) \le \phi_{\Sigma_r} ( \mathbf{x}_r )$ implies $\mathbf{x}_r'\Sigma_r^{-1} \mathbf{x}_r \le x^*_T$. To see this, we consider the two functions $\log (1+y)$ and $ay + b$, where $a= T/(T+r-1)$ and $b= 2A_T/(T+r-1)$. We want to find a value of $y$, say $y'$, such that if $y\geq y'$ then $\log (1+y) \leq ay +b$. Because $\log(1+y)$ is increasing and concave and $a>0$ there are two possibilities: 1) $ay+b \geq \log(1+y)$ for any $y$ and $ay+b > \log(1+y)$ almost always; 2) the curves $\log (1+y)$ and $ ay +b$ intersect with each other at two points, say $y_1$ and $y_2$ such that $\log(1+y) < ay + b$ for $y < y_1$, $\log (1+y) \geq ay+b$ for $y_1 \leq y \leq y_2$, and $\log (1+y) < ay+b$ for $y>y_2$. The first case does not apply to our situation because if this was the case, then $f_{\Sigma_r,T-1}^{t,*} (\mathbf{x}_r) > \phi_{\Sigma_r} ( \mathbf{x}_r )$ almost always, contradicting the fact that both curves integrate to one. Thus, the second case applies. The values of $y_1$ and $y_2$ can be obtained by solving $\log (1+y) = ay +b$. It holds $y_2 >0$ because the slope of $\log (1+y)$ at $y_2$ must be smaller than $a$ and $0<a<1$.

    Choose $t$ large enough such that $\mathbf{x}_r'\Sigma_r^{-1} \mathbf{x}_r \leq x^*_T$ implies $\mathbf{x} \leq t$. 
    This choice of $t$ depends on $T$ only through $x_T^*$. In particular, if $x_T^* = O(1)$ then $t$ can be chosen independently of $T$.   
    To prove this set $t = \sqrt{x_T^* \dim(\mathbf{x})}$. 
    Since $V$ is a correlation matrix, its largest eigenvalue is bounded by $r$. This implies that and $\mathbf{x}_r'\Sigma_r^{-1} \mathbf{x}_r \geq \lVert \mathbf{x}_r \rVert^2 / r$. Because $\mathbf{x}^*$ is a vector whose first $r$ elements are those of $\mathbf{x}_r$ and other elements are zero, $\lVert \mathbf{x}_r \rVert^2 = \lVert \mathbf{x}^* \rVert^2$. By the definition of $\mathbf{x}^*$, it holds that $\lVert \mathbf{x}^* \rVert^2 = \lVert U' \mathbf{x} \rVert^2 = \lVert  \mathbf{x} \rVert^2$, where the last equality uses the fact that $U$ is a unitary matrix. Observe that if $\mathbf{x} \nleq t$ so that an element of $\mathbf{x}$ exceeds $t$, then $\lVert \mathbf{x} \rVert^2 > t^2 \geq x_T^* r$. This implies that $ \mathbf{x}_r'\Sigma_r^{-1} \mathbf{x}_r \geq \lVert  \mathbf{x} \rVert^2 /r > x_T^* r /r = x_T^*$.

    We have 
    \begin{align*}
        \Phi_{\max , \Sigma_r} (t) - 	t_{\max, \Sigma_r, T-1} \left( \sqrt{\frac{T-1}{T}} t \right) 
        =& \int_{\mathbf{x} \le t } \left( \phi_{\Sigma_r} ( \mathbf{x}_r ) - f_{\Sigma_r,T-1}^{t,*} ( \mathbf{x}_r ) \right)  d \mathbf{x}_r \\
        =& \int_{\mathbf{x}_r'\Sigma_r^{-1} \mathbf{x}_r \le x^*_T } \left( \phi_{\Sigma_r} ( \mathbf{x}_r ) - f_{\Sigma_r,T-1}^{t,*} ( \mathbf{x}_r ) \right)  d \mathbf{x}_r \\
            & + \int_{\mathbf{x} \le t ,  \mathbf{x}_r'\Sigma_r^{-1} \mathbf{x}_r > x^*_T } \left( \phi_{\Sigma_r} ( \mathbf{x}_r ) - f_{V,T-1}^{t,*} ( \mathbf{x}_r ) \right)  d \mathbf{x}_r,
    \end{align*}
    where the first integral on the right--hand side of the equation is taken over $\mathbf{x}_r'\Sigma_r^{-1} \mathbf{x}_r \le a $ because $\{\mathbf{x}_r: \mathbf{x}_r'\Sigma_r^{-1} \mathbf{x}_r \le x^*_T, \mathbf{x} \le t \} = \{ \mathbf{x}_r: \mathbf{x}_r'\Sigma_r^{-1} \mathbf{x}_r \le x^*_T\} $ by our choice of $t$.
    Because both $\phi_{\Sigma_r} ( \mathbf{x}_r )$ and $ f_{\Sigma_r,T-1}^{t,*} ( \mathbf{x}_r )$ are densities and integrate to one, we have 
    \begin{align*}
        \int_{\mathbf{x}_r'\Sigma_r^{-1} \mathbf{x}_r \le x^*_T } \left( \phi_{\Sigma_r} ( \mathbf{x}_r ) - f_{\Sigma_r,T-1}^{t,*} ( \mathbf{x}_r ) \right)  d \mathbf{x}_r
        = - \int_{\mathbf{x}_r'\Sigma_r^{-1} \mathbf{x}_r > x^*_T } \left( \phi_{\Sigma_r} ( \mathbf{x}_r ) - f_{\Sigma_r,T-1}^{t,*} ( \mathbf{x}_r ) \right)  d \mathbf{x}_r,
    \end{align*}
    Thus, for $t$ large enough such that $\mathbf{x}'V^{-1} \mathbf{x} \le x^*_T$ implies $\mathbf{x} \le t$, we have
    \begin{align*}
    &	\Phi_{\max , \Sigma_r} (t) - 	t_{\max, \Sigma_r, T-1} \left( \sqrt{\frac{T-1}{T}} t \right) \\
        =& - \int_{\mathbf{x}_r'\Sigma_r^{-1} \mathbf{x}_r > x^*_T } \left( \phi_{\Sigma_r} ( \mathbf{x}_r ) - f_{\Sigma_r,T-1}^{t,*} ( \mathbf{x}_r ) \right)  d \mathbf{x}_r \\
            & + \int_{\mathbf{x} \le t ,  \mathbf{x}_r'\Sigma_r^{-1} \mathbf{x}_r > x^*_T } \left( \phi_{\Sigma_r} ( \mathbf{x}_r ) - f_{V,T-1}^{t,*} ( \mathbf{x}_r ) \right)  d \mathbf{x}_r \\
            =& \int_{\mathbf{x} \nleq t ,  \mathbf{x}_r'\Sigma_r^{-1} \mathbf{x}_r > x^*_T } \left( \phi_{\Sigma_r} ( \mathbf{x}_r ) - f_{V,T-1}^{t,*} ( \mathbf{x}_r ) \right)  d \mathbf{x}_r \geq 0,
    \end{align*}
    where the last inequality follows because $ \mathbf{x}_r'\Sigma_r^{-1} \mathbf{x}_r > x^*_T$ implies $\phi_{\Sigma_r} ( \mathbf{x}_r ) > f_{V,T-1}^{t,*} ( \mathbf{x}_r )$.
    
    Next, we evaluate the order of $x^*_T$. Note that $x^*_T$ solves 
    \begin{align*}
        \frac{1}{2} x^*_T + A_T  =  \frac{T+r-1}{2} \log \left( 1+ \frac{1}{T} x^*_T \right).
    \end{align*}
    We first show that $A_T=O(1)$ where the order is taken with respect to $T$. We separately examine the cases of odd and even $G$.  
    Suppose that $r$ is even (we may assume $r \geq 2$). Then, the recurrent relation of the Gamma function implies that 
    \begin{align*}
        A_T =& - \frac{r}{2}\log (T) + \sum_{j=0}^{r/2 -1} \log \left( \frac{T-1}{2} + j\right)  + \frac{r}{2} \log (2 ) \\
        =& - \frac{r}{2}\log (T) + \sum_{j=0}^{r/2 -1} \log \left( T-1 + 2j \right) - \frac{r}{2}  \log (2) + \frac{r}{2} \log (2 ) \\
        =&  \sum_{j=0}^{r/2 -1} \log \left( \frac{T-1 + 2j}{T} \right) = O(1)
    \end{align*}
    as $T \to \infty$.
    Next, we consider cases in which $r$ is odd. 
    For $r = 1$, 
    $A_T = O(1)$ follows from 
    \begin{align}
        \label{eq:lemma:max:finite:gamma:conv}
        \sqrt{\frac{T}{2}} \frac{\Gamma\left(\frac{T-1}{2}\right)}{\Gamma \left(\frac{T}{2}\right)} \to 1.
    \end{align}
    For $r \geq 3$, by the recurrent relation of the Gamma function, we have
    \begin{align*}
        A_T =& - \frac{r}{2}\log (T) + \sum_{j=0}^{r/2 -1} \log \left( \frac{T}{2} + j\right)  + \log\Gamma\left(\frac{T}{2}\right) - \log \Gamma \left(\frac{T-1}{2}\right) + \frac{r}{2} \log (2 ) \\
        = &  \sum_{j=0}^{(r-1)/2 -1} \log \left( \frac{T+ 2j}{T} \right) + \frac{1}{2} \log \left(\frac{2}{T} \right) + \log\Gamma\left(\frac{T}{2}\right) - \log \Gamma \left(\frac{T-1}{2}\right) .
    \end{align*}
    By \eqref{eq:lemma:max:finite:gamma:conv} 
    \begin{align*}
        \log\Gamma\left(\frac{T}{2}\right) - \log \left( \Gamma \left(\frac{T-1}{2}\right) \left(\frac{T}{2}\right)^{1/2} \right) =O(1).
    \end{align*}
    We have established that $A_T=O(1)$ for all $r \geq 1$. 
    To prove the lemma, it now suffices to prove $x^*_T =O(1)$. Suppose the opposite is true. Then, there is a subsequence $T_1, \dots, T_k, \dots$ such that $x^*_{T_k}$ monotonically diverges to infinity. By the definition of $x^*_T$ we have
    \begin{align*}
        x^*_T + A_T = (T+r-1) \log \left( 1 + \frac{1}{T}x^*_T  \right) .
    \end{align*} 
    For sufficiently large $y$, $y/2 > \log (1+y)$. Therefore, for sufficiently large $k$, we have 
    \begin{align*}
        x^*_{T_k} + A_T < \frac{T+r-1}{2T}  x^*_{T_k}
    \end{align*}
    Rearranging terms yields
    \begin{align*}
        \frac{T-r+1}{2T} 	x^*_{T_k} + A_T < 0,
    \end{align*}
    contradicting that $A_T= O(1)$ and $x^*_{T_k}$ diverging to infinity can both be true. This proves $x_T^* = O(1)$.
    
\end{proof}

\begin{lemma}[Comparison bound for critical values with regularization\label{lem:comparison_bound}]
    Let $\Omega$ and $\widehat{\Omega}$ denote $p \times p$ correlation matrices and let $\epsilon$, $\delta$ and $c_{\omega}$ denote positive constants such that
    $4 \delta \leq \epsilon \leq 4 c_\omega / 3$.  
    Suppose that $\Omega_{ij} > -1 + c_\omega$ for all $i,j = 1, \dotsc, p$ and  
    \begin{align*}
        \norm{\widehat{\Omega} - \Omega}_{\max} \leq \delta.
    \end{align*}
    Let $X \sim N(0, \Omega)$ and $\hat{X}^{\epsilon} \sim N(0, \rho(\widehat{\Omega}, \epsilon))$. 
    Then, there is a universal constant $C$ such that for all $a > \sqrt{2}$ 
    \begin{align*}
        \frac{P\left(\max_{j = 1, \dotsc, p} X_j > a \right)}
        {P\left(\max_{j = 1, \dotsc, p} \hat{X}^{\epsilon}_j > a \right)} - 1 
        \leq C a \sqrt{\epsilon}.
    \end{align*}
    In particular, suppose that $\hat{c}_{\alpha, N}$ is the $1 - \alpha/N$ quantile of $\max_{j = 1, \dotsc, p} \hat{X}^{\epsilon}_j$ and let $c_{\alpha_N, N}$ denote the $1 - \alpha_N/N$ quantile of $\max_{j = 1, \dotsc, p} X_j$, where
    \begin{align*}
        \alpha_N = \alpha (1 + \hat{c}_{\alpha, N}  C \sqrt{\epsilon}). 
    \end{align*}
    If $\hat{c}_{\alpha, N} > \sqrt{2}$ then 
    $\hat{c}_{\alpha, N} \geq c_{\alpha_N, N}$.
    \end{lemma}
    
    \begin{proof} 
    Write $\hat{\epsilon} = \epsilon^* (\widehat{\Omega}, \epsilon)$.
    By the Cholesky decomposition, there is a lower-triangular matrix $L$ (possibly with some diagonal elements equal to zero) such that $\widehat{\Omega} = L L'$. $\widehat{\Omega}$ can be interpreted as the covariance matrix of the random vector $L W$, where $W$ is a random vector in $\reals^p$ with expectation zero and covariance matrix $I_p$. $V = \widehat{\Omega} + \hat{\epsilon} I_p$ can be interpreted as the covariance matrix of the random vector that is generated by adding independent, component-specific noise $E_i$ to the $i$th components of $LW$, where $E_i$ has mean zero and variance $\hat{\epsilon}$. Then, $\rho (\widehat{\Omega}, \hat{\epsilon})$ transforms $V$ into a correlation matrix. 
    Since $\rho (\widehat{\Omega}, \epsilon)$ and $\widehat{\Omega}$ are both correlation matrices, $\rho (\widehat{\Omega}, \epsilon)_{ii} = \widehat{\Omega}_{ii} = 1$, for all $i = 1, \dotsc, p$. Let $\ell_i$ denote the $i$th row of $L$. For $i \neq j$, $\widehat{\Omega}_{ij}$ is equal to the covariance between $\ell_i' W$ and $\ell_j' W$, i.e., 
    \begin{align*}
        \widehat{\Omega}_{ij} = \cov (\ell_i' W, \ell_j' W) = \ell_i' \cov (W) \ell_j = \ell_i' \ell_j.
    \end{align*}
    $V_{ij}$ is equal to the covariance between $\ell_i' W + E_i$ and $\ell_j'W + E_j$, i.e., 
    \begin{align*}
        V_{ij} = \cov (\ell_i' W + E_i, \ell_j' W + E_j) = \ell_i' \cov (W) \ell_j = \ell_i' \ell_j = \widehat{\Omega}_{ij}.
    \end{align*}
    For $i = 1, \dotsc, p$, 
    \begin{align*}
        V_{ii} = \cov (\ell_i' W + E_i, \ell_i' W + E_i) = \ell_i' \cov (W) \ell_i + \hat{\epsilon} = \widehat{\Omega}_{ii} + \hat{\epsilon} = 1 + \hat{\epsilon}.    
    \end{align*}
    Therefore, for $i \neq j$, 
    \begin{align*}
        \rho (\widehat{\Omega}, \epsilon)_{ij} = \frac{V_{ij}}{\sqrt{V_{ii} V_{jj}}} 
        = \frac{\widehat{\Omega}_{ij}}{1 + \hat{\epsilon}}.
    \end{align*}
    We now derive a bound on
    \begin{align*}
        \Delta_{ij} = \left( 
            \arcsin \left(\rho(\widehat{\Omega}, \epsilon)_{ij} \right) - \arcsin \left( \Omega_{ij} \right)
        \right)^+. 
    \end{align*}
    Since $\arcsin (\cdot)$ is strictly increasing on $(0,1)$,
    a necessary condition for $\Delta_{ij} \neq 0$ is $\rho(\widehat{\Omega}, \epsilon)_{ij} > \Omega_{ij}$. This condition bounds $\rho (\widehat{\Omega}, \epsilon)_{ij}$ and $\Omega_{ij}$ away from -1 and 1. In particular, 
    \begin{align*}
        \widehat{\Omega} / (1 + \hat{\epsilon}) 
        = \rho (\widehat{\Omega}, \epsilon)_{ij} > \Omega_{ij}
    \end{align*}
    implies 
    \begin{align*}
        \widehat{\Omega}_{ij} > (1 + \hat{\epsilon}) \Omega_{ij}.
    \end{align*}
    Since we also have $\widehat{\Omega}_{ij} \leq \Omega_{ij} + \delta$, an $\widehat{\Omega}_{ij}$ fulfilling both conditions exists only if 
    \begin{align*}
        (1 + \hat{\epsilon}) \Omega_{ij} < \Omega_{ij} + \delta
    \end{align*}
    or equivalently if $\hat{\epsilon} \Omega_{ij} < \delta$. Suppose that $\Omega_{ij} > 1 - \epsilon/2$. Then we have 
    \begin{align*}
        \hat{\epsilon} 
        \geq & \epsilon - (1 - \widehat{\Omega}_{ij} )
    \\
        \geq & \epsilon + \Omega_{ij} - \delta - 1 
        > \epsilon + (1 - \epsilon/2) - \delta - 1 = \epsilon/2 - \delta. 
    \end{align*}
    Therefore, $\hat{\epsilon} \Omega_{ij} < \delta$ is only possible if 
    \begin{align*}
        (\epsilon/2 - \delta) (1 - \epsilon/2) < \delta. 
    \end{align*} 
    This inequality contradicts $\epsilon \geq 4 \delta$ and hence we can take $\Omega_{ij} \leq 1 - \epsilon/2$. 
    Moreover, since $\epsilon \leq 2 c_\omega$, $\Omega_{ij} \geq -1 - c_\omega \geq -1 - \epsilon/2$.
    
    For an upper bound on $\rho (\widehat{\Omega}, \epsilon)_{ij}$, we have 
    \begin{align*}
        \rho (\widehat{\Omega}, \epsilon)_{ij} = \frac{\widehat{\Omega}_{ij}}{1 + \hat{\epsilon}} 
        \leq & 
        \frac{\widehat{\Omega}_{ij}}{1 + \epsilon - (1 - \widehat{\Omega}_{ij})} 
        \\
        \leq &
        \frac{\widehat{\Omega}_{ij}}{\widehat{\Omega}_{ij} + \epsilon} \leq \frac{1}{1 + \epsilon} \leq 1 - \epsilon/2.
     \end{align*}
    For a lower bound on $\rho (\widehat{\Omega}, \epsilon)_{ij}$, suppose that $\widehat{\Omega}_{ij} < 0$, in which case, 
    \begin{align*}
        \rho (\widehat{\Omega}, \epsilon)_{ij} 
        = \frac{\widehat{\Omega}_{ij}}{1 + \hat{\epsilon}} 
        \geq \widehat{\Omega}_{ij} = \Omega_{ij} + \widehat{\Omega}_{ij} - \Omega_{ij} 
        \geq -1 + c_\omega - \delta 
        \geq -1 + \epsilon/2, 
    \end{align*}
    provided that $c_\omega - \epsilon/2 - \delta \geq 0$. This condition is satisfied if $\epsilon \leq 4 c_\omega / 3$.
    Therefore, we have the bounds 
    \begin{align*}
        -1 + \epsilon/2 \leq \Omega_{ij} \leq 1 - \epsilon/2
    \end{align*}
    and 
    \begin{align*}
        -1 + \epsilon/2 \leq \rho(\widehat{\Omega}_{ij}, \epsilon) \leq 1 - \epsilon/2.
    \end{align*}
    We also have 
    \begin{align*}
        \rho (\widehat{\Omega}, \epsilon)_{ij} - \Omega_{ij} 
        = \widehat{\Omega}_{ij} / (1+ \epsilon) - \Omega_{ij}
        \leq \delta + \epsilon \leq 5 \epsilon. 
    \end{align*}
    By the intermediate value theorem there is an intermediate value $\rho^*$ between $-1 + \epsilon/2$ and $1 - \epsilon / 2$ such that, on $\rho(\widehat{\Omega}, \epsilon)_{ij} > \Omega_{ij}$,
    \begin{align*}
        \Delta_{ij} = \frac{\rho(\widehat{\Omega}, \epsilon)_{ij} - \Omega_{ij}}{\sqrt{1 - (\rho^*)^2}} 
        \leq 
        \frac{5 \epsilon}{\sqrt{1 - (1 - \epsilon/2)^2}}
        \leq 
        \frac{5 \epsilon}{\sqrt{\epsilon(1 - \epsilon/4)}}
        \leq \frac{10 \sqrt{\epsilon}}{\sqrt{3}}.
    \end{align*}
    Let $\Phi$ denote the cumulative distribution function of a standard normal random variable, and let $\phi$ denote its probability density function. Gordon's lower bound (see, e.g., \cite{duembgen2010bounding}) states that 
    \begin{align*}
        1 - \Phi(a) > \frac{\phi(a)}{a(1 - 1/a^2)}
    \end{align*}
    for $a > 0$ and thus $1 - \Phi(a) > \frac{1}{2} \phi(a)/a$ for $a > \sqrt{2}$. Therefore, 
    \begin{align*}
        P\left(\max_{j = 1, \dotsc, p} \hat{X}_j > a\right) 
        \geq P\left(\hat{X}^{\epsilon}_1 > a \right)
        = & 1 - \Phi(a) 
        >   \frac{\phi(a)}{2 a} 
        = \frac{\exp\left(- \frac{a^2}{2} \right)}{a \sqrt{8 \pi}}
    \end{align*}
    By Theorem~2.1 in \textcite{li2002normal}, 
    \begin{align*}
        &  P \left(\max_{j = 1, \dotsc, p} X_j > a \right) - P \left(\max_{j = 1, \dotsc, p} \hat{X}^{\epsilon}_j > a \right)
    \\
        = & P \left(\max_{j = 1, \dotsc, p} \hat{X}^{\epsilon}_j \leq a \right) - P \left(\max_{j = 1, \dotsc, p} X_j \leq a \right)
    \\
        \leq & 
        \frac{1}{2 \pi} \exp\left(- \frac{a^2}{2}\right) \sum_{1 \leq i < j \leq p} 
        \Delta_{ij}
        \leq \frac{5}{\pi \sqrt{3}} \exp\left(- \frac{a^2}{2} \right) \sqrt{\epsilon}
    \end{align*}
    We may assume $P \left(\max_{j = 1, \dotsc, p} X_j > a \right) > P \left(\max_{j = 1, \dotsc, p} \hat{X}^{\epsilon}_j > a \right)$ since the statement of the theorem holds trivially otherwise.
    Then, combining the bounds derived above yields
    \begin{align*}
        &  \frac{P \left(\max_{j = 1, \dotsc, p} X_j > a \right) - P \left(\max_{j = 1, \dotsc, p} \hat{X}^{\epsilon}_j > a \right)}{P \left(\max_{j = 1, \dotsc, p} \hat{X}^{\epsilon}_j > a \right)}
        < \frac{10 a \sqrt{\epsilon}}{\sqrt{3}}
        .
    \end{align*}
    This proves the first claim of the lemma.
    To prove the second claim of the lemma, note that the first claim of the lemma implies 
    \begin{align*}
        P \left(\max_{j=1, \dotsc, p} X_j > \hat{c}_{\alpha, N} \right)
        = &
        \frac{P \left(\max_{j=1, \dotsc, p} X_j > \hat{c}_{\alpha, N} \right)}{P \left(\max_{j=1, \dotsc, p} \hat{X}_j > \hat{c}_{\alpha, N} \right)} \alpha/N 
    \\
        \leq & \alpha/N (1 + \hat{c}_{\alpha, N} C \sqrt{\epsilon}) \leq \alpha_N / N.
    \end{align*}
\end{proof}

\begin{lemma}[Consistency of $\widehat{\Omega}$]
	\label{lem:lv-cor-consistency}
	Let $\mathbb{P}_N$ be the set of probability measures which satisfy Assumptions~\ref*{assumption:basic} with identical choices of $a$, $b$, $d_1$, and $d_2$. 
	Assume $T^{\delta_1} \leq N$ for some universal constant $\delta_1 >0$. Let Assumption~\ref*{assum:kernel} hold and $\kappa_N \asymp T^{\rho}$ where $0 <\rho < (\vartheta -1)/(3\vartheta -2)$.
	Let 
	\begin{align*}
		b_N^{LV} =  \frac{r_{\theta, N}}{ \iota_N \min_{1 \leq i \leq N}\sigma_i} +  T^{-c_1} (\log N)^{c_2} + T^{-\rho},
	\end{align*}
	where $c_1>0$ and $c_2>0$ are two constants defined in Lemma \ref{lem:lv-consistency} with $c_1$ depending only on $(\rho, \vartheta)$ and $c_2$ depending only on $(d_2, \vartheta)$.
	Assume that $b_N^{LV} \to 0$.
    For any sequence $\zeta_N$ such that $\zeta_N \to \infty $ as $N,T \to \infty$, 
	\begin{align*}
	\sup_{P\in \mathbb{P}_N} P \left( \max_{1 \leq i \leq N} \max_{(h,h^*) \in \mathbb{G}^2} \left| ( \widehat{\Omega}_i(g_i^0) )_{h, h^*} - ( \Omega_i(g_i^0) )_{h, h^*} \right| > \zeta_N b_N^{LV} \right) = o(1).
\end{align*}	
\end{lemma}
	
\begin{proof}
	Throughout the proof, let $C$, $C'$, and $C''$ denote generic constants that do not depend on $P \in \mathbb{P}$.
	Let $\xi_i (h) = \sqrt{\Xi_i (h,h)}$ and $\widehat{\xi}_i (h) = \sqrt{\widehat{\Xi}_i (h,h)}$.

	We observe the following decomposition:
	\begin{align*}
		& ( \widehat{\Omega}_i )_{h, h^*} - ( \Omega_i )_{h, h^*} 
    \\
        = & \frac{\widehat{\Xi}_{i} (h, h^*)}{\hat{\xi}_{i} (h) \hat{\xi}_{i} (h^*)}
		- \frac{\Xi_{i} (h, h^*)}{\xi_{i} (h ) \xi_{i}( h^* )}
	\\
		= & \left[ \left(\frac{\xi_{i} ( h )}{\hat{\xi}_{i}( h) }\right)\left(\frac{\xi_{i} (h^*)}{\hat{\xi}_{i} ( h^* )}\right) - 1 \right] \left( 
			\frac{\widehat{\Xi}_{i} (h, h^*)}{\xi_{i} ( h ) \xi_{i} ( h^* )}
			- \frac{\Xi_{i} ( h, h^* )}{\xi_{i} ( h ) \xi_{i} ( h^* )}
		\right)
	\\
		& + \left[ \left(\frac{\xi_{i} ( h )}{\hat{\xi}_{i} ( h )}\right)\left(\frac{\xi_{i} ( h^* )}{\hat{\xi}_{i} ( h^* )}\right) - 1 \right] \frac{\Xi_{i} ( h, h^* )}{\xi_{i} ( h ) \xi_{i} (h^* )}
		+ \left(
			\frac{\widehat{\Xi}_{i} ( h, h^* )}{\xi_{i} ( h ) \xi_{i} ( h^*) }- 
			\frac{\Xi_{i} (h, h^* )}{\xi_{i} (h ) \xi_{i} ( h^* )}
		\right).
	\end{align*}
	Noting that 
	\begin{align*}
		\left| \Xi_{i} (h, h) \right| = \left| \sigma_i^2 \delta_i(h)' \left( \frac{1}{T}\sum_{t=1}^{T} \sum_{s=1}^T \E_P [v_{it} v_{is} x_{it} x_{is}' ] \right)  \delta_i(h) \right|,
	\end{align*}
	Assumption~\ref*{assumption:basic}.\ref*{assumption:max:min_eigenvalue} implies that $\xi_i (h) / (\sigma_i) \norm{\delta (h)}$ (and $\xi_i (h^*) / (\sigma_i) \norm{\delta (h^*)}$) is bounded away from zero. 
	The inequality $|\sqrt{a} - 1 | \leq | a- 1 | $ for $a>0$ implies that 
	\begin{align*}
		\left| \left(\frac{\widehat{\xi}_{i} ( h )}{\xi_{i} ( h )}\right) - 1   \right| 
		\leq 
		\left| \left(\frac{\widehat{\Xi}_{i} ( h ,h )}{\Xi_{i} ( h,h )}\right) - 1   \right| .
	\end{align*}
	Thus, we have 
	\begin{align*}
		& \left| \left(\frac{\xi_{i} ( h )}{\hat{\xi}_{i} ( h )}\right)\left(\frac{\xi_{i} ( h^* )}{\hat{\xi}_{i} ( h^* )}\right) - 1 \right|\\
		=&  \left|		\left(\frac{\xi_{i} ( h )}{\hat{\xi}_{i} ( h )} - 1 + 1 \right)\left(\frac{\xi_{i} ( h^* )}{\hat{\xi}_{i} ( h^* )}\right) - 1  \right| \\
		=& \left| \left(\frac{\xi_{i} ( h )}{\hat{\xi}_{i} ( h )} - 1 \right)\left(\frac{\xi_{i} ( h^* )}{\hat{\xi}_{i} ( h^* )}\right) + \left(\frac{\xi_{i} ( h^* )}{\hat{\xi}_{i} ( h^* )} - 1\right) \right| \\
		=& \left| \left(\frac{\xi_{i} ( h )}{\hat{\xi}_{i} ( h )} - 1 \right)\left(\frac{\xi_{i} ( h^* )}{\hat{\xi}_{i} ( h^* )}\right) + \left(\frac{\xi_{i} ( h^* )}{\hat{\xi}_{i} ( h^* )} - 1\right) \right| \\
		\leq & \left| \frac{\Xi_{i} ( h,h )}{\widehat{\Xi}_{i} ( h,h )} - 1 \right| \left(\frac{\Xi_{i} ( h^*, h^* )}{\widehat{\Xi}_{i} ( h^*,h^* )}\right)^{1/2} + \left| \frac{\Xi_{i} ( h^* , h^* )}{\widehat{\Xi}_{i} ( h^* ,h^*)} - 1 \right|.
	\end{align*}
	Now, it holds that 
	\begin{align*}
		\left| \frac{\Xi_{i} ( h,h )}{\widehat{\Xi}_{i} ( h,h )} - 1 \right|
		\leq \left| \frac{\widehat{\Xi}_{i} ( h,h ) -\Xi_{i} ( h,h ) }{\Xi_{i} ( h,h )} + 1 \right|^{-1} \left| \frac{\widehat{\Xi}_{i} ( h,h ) -\Xi_{i} ( h,h ) }{\Xi_{i} ( h,h )}\right|.
	\end{align*}
	Provided that $b_N^{LV} \to 0$, we have 
	\begin{align*}
		\left| \left(\frac{\xi_{i} ( h )}{\hat{\xi}_{i} ( h )}\right)\left(\frac{\xi_{i} ( h^* )}{\hat{\xi}_{i} ( h^* )}\right) - 1 \right| < C \zeta_N b_N^{LV}
	\end{align*}
	with probability approaching one by Lemma \ref{lem:lv-consistency}. Therefore, by Lemma \ref{lem:lv-consistency}, the desired result holds.
\end{proof}

\begin{lemma}[Consistency of long-run variance estimator]
	\label{lem:lv-consistency}
	Let $\mathbb{P}_N$ be the set of probability measures which satisfy Assumption~\ref*{assumption:basic} with identical choices of $a$, $b$, $d_1$ and $d_2$ 
	Assume $T^{\delta_1} \leq N$ for some universal constant $\delta_1 >0$. Let Assumption~\ref*{assum:kernel} hold and $\kappa_N \asymp T^{\rho}$ where $0 <\rho < (\vartheta -1)/(3\vartheta -2)$. 
    Let 
    \begin{align*}
    	b_N^{LV} =  \frac{r_{\theta, N}}{ \iota_N \min_{1 \leq i \leq N}\sigma_i} +  T^{-c_1} (\log N)^{c_2} + T^{-\rho},
    \end{align*}
    and assume that $b_N^{LV} \to 0$. 
    Then, there exist two constants $c_1>0$ depending only on $(\rho, \vartheta)$ and $c_2>0$ depending only on $(d_2, \vartheta)$, such that, for any sequence $\zeta_N$ such that $\zeta_N \to \infty$ as $N, T \to \infty$, 
    \begin{align*}
        \sup_{P\in \mathbb{P}_N} P \left( \max_{1 \leq i \leq N} \max_{(h,h') \in \mathbb{G}^2} \left| \frac{ \widehat{\Xi}_i (h,h') - \Xi_i (h,h') }{ \sigma_i^2 \norm{\delta_i(h) } \norm{\delta_i(h')}} \right| > \zeta_N b_N^{LV} \right) = o(1).
    \end{align*}	
\end{lemma}

\begin{proof}
To conserve notation, we introduce the short-hands
\begin{align*}
    K_N^{(j)} =& K \left(\frac{j}{\kappa_N}\right)
\end{align*}
and 
\begin{align*}
    u_{it}^{(+j)} = u_{i, t + \max(0,j)},
\quad 
    v_{it}^{(+j)} = v_{i, t + \max(0,j)},
\quad  
    w_{it}^{(+j)} = w_{i, t + \max(0,j)},
\quad
    x_{it}^{(+j)} = x_{i, t + \max(0,j)}, 
\\
    u_{it}^{(-j)} = u_{i, t - \max(0,j)},
\quad 
    v_{it}^{(-j)} = v_{i, t - \max(0,j)},
\quad  
    w_{it}^{(-j)} = w_{i, t - \max(0,j)},
\quad
    x_{it}^{(-j)} = x_{i, t - \max(0,j)} 
\end{align*}
and 
\begin{align*}
    \overline{x_i x_i} = \frac{1}{T}\sum_{u = 1}^T x_{iu} x_{iu}',
    \quad \overline{x_i w_i} = \frac{1}{T}\sum_{u = 1}^T x_{iu} w_{iu}',
\\
    \quad \overline{u_i x_i} = \frac{1}{T}\sum_{u = 1}^T u_{iu} x_{iu},
    \quad \overline{v_i x_i} = \frac{1}{T}\sum_{u = 1}^T v_{iu} x_{iu}.
\end{align*}

By the triangular inequality, we have
\begin{align*}
	\left| \frac{ \widehat{\Xi}_i (h,h') - \Xi_i (h,h') }{ \sigma_i^2 \norm{\delta_i(h) } \norm{\delta_i(h')}} \right| 
\leq \left| \frac{ \widehat{\Xi}_i (h,h') - \widetilde{\Xi}_i (h,h') }{ \sigma_i^2 \norm{\delta_i(h) } \norm{\delta_i(h')}} \right| 
+\left| \frac{ \widetilde{\Xi}_i (h,h') - \Xi_i (h,h') }{ \sigma_i^2 \norm{\delta_i(h) } \norm{\delta_i(h')}} \right| .
\end{align*}
We examine each of the two terms on the right-hand side. We first examine the second term and then the first term.
We note that 
\begin{align*}
	\widetilde{\Xi}_i (h,h') 
	 = &  \sigma_i^2 \delta_i(h)' \sum_{j=-T+1}^{T-1}K_N^{(j)} \frac{1}{T} \sum_{t=|j|+1}^T \left( v_{it}^{(+j)} x_{it}^{(+j)} - \overline{v_i x_i} \right) \left( v_{it}^{(-j)} x_{it}^{(-j)} - \overline{v_i x_i} \right)'   \delta_i(h')
\end{align*}
and 
\begin{align*}
	\Xi_i (h,h')= \sigma_i^2 \delta_i(h)' \left( \frac{1}{T}\sum_{t=1}^{T} \sum_{s=1}^T \E_P [v_{it} v_{is} x_{it} x_{is}' ] \right)  \delta_i(h'). 
\end{align*}
Lemma \ref{lem:lv-xi-consistency} gives the bound of 
\begin{align*}
	\sup_{i,h,h'} \bigg| & \sum_{j=-T+1}^{T-1}K_N^{(j)} \frac{1}{T} \sum_{t=|j|+1}^T \left( v_{it}^{(+j)} x_{it}^{(+j)} - \overline{v_i x_i} \right) \left( v_{it}^{(-j)} x_{it}^{(-j)} - \overline{v_i x_i} \right)' 
\\
    & \quad  -  \left( \frac{1}{T}\sum_{t=1}^{T} \sum_{s=1}^T \E_P [v_{it} v_{is} x_{it} x_{is}' ] \right) \bigg|_{\infty}.
\end{align*}
We thus have
\begin{align*}
	\sup_{P\in \mathbb{P}_N} P \left( \max_{1 \leq i \leq N} \max_{(h,h') \in \mathbb{G}^2} \left| \frac{ \widetilde{\Xi}_i (h,h') - \Xi_i (h,h') }{ \sigma_i^2 \norm{\delta_i(h) } \norm{\delta_i(h')}} \right| > \zeta_{1,N} ( T^{-c_1} (\log N)^{c_2} + T^{-\rho} )\right) = o(1),
\end{align*}
where $\zeta_{1,N} \to \infty$.
Next, we derive the bound of
\begin{align*}
	\sup_{i,h,h'} \frac{1}{\sigma_i^2\norm{\delta_i (h)}\norm{\delta_i (h')} } \left| \widehat{\Xi}_i (h,h') - \widetilde{\Xi}_i (h,h') \right|. 
\end{align*}
First, note that
\begin{align*}
	d_{it} (h) = - \sigma_i v_{it} x_{it}' (\theta_{g_i^0} - \theta_h) = - \sigma_i v_{it} x_{it}' \delta_i(h) ,
\end{align*}
and
\begin{align*}
	\bar d_i (h) = - \sigma_i \overline{v_i x_i}' (\theta_{g_i^0} - \theta_h) = - \sigma_i \overline{v_i x_i}'\delta_i(h) .
\end{align*}
Let 
	\begin{align*}
		\hat{u}_{it} = y_{it} - x_{it}' \hat{\theta}_{g_i^0} - w_{it}'\hat{\theta}^w ,
	\end{align*}
	so that 
	\begin{align*}
		\hat{u}_{it} - u_{it} = - x_{it}' \left(\hat{\theta}_{g_i^0} - \theta_{g_i^0}\right) - w_{it}' \left(\hat{\theta}^w - \theta^w\right).
	\end{align*}
	With this notation 
	\begin{align*}
		\hat{d}_{it}(g_i^0, h) = \hat{d}_{it} (h) =& - \hat{u}_{it} x_{it}' \hat{\delta}_i (h).
	\end{align*}
Consider the decomposition
\begin{align*}
	& \left( \hat{d}_{it} (g^0, h) - \bar{\hat{d}}_{i} (g^0, h) \right)\left( \hat{d}_{is} (g^0, h') - \bar{\hat{d}}_{i} (g^0, h') \right) \\
	& -  \left( d_{it} (g^0, h) - \bar{d}_{i} (g^0, h) \right)\left( d_{is} (g^0, h') - \bar{d}_{i} (g^0, h') \right) \\
	=&  (\hat \delta_i(h) - \delta_i(h))' \left( \hat u_{it}x_{it} - \frac{1}{T}\sum_{u=1}^T \hat u_{iu}x_{iu}   \right) \left( \hat u_{is}x_{is} - \frac{1}{T}\sum_{u=1}^T \hat u_{iu}x_{iu}   \right)' ( \hat \delta_i(h')- \delta_i(h')  ) \\
	& -  (\hat \delta_i(h) - \delta_i(h))' \left(  u_{it}x_{it} - \frac{1}{T}\sum_{u=1}^T u_{iu}x_{iu}   \right) \left( u_{is}x_{is} - \frac{1}{T}\sum_{u=1}^T u_{iu}x_{iu}  \right)' ( \hat \delta_i(h')- \delta_i(h')  ) \\
	& +  ( \hat \delta_i(h)- \delta_i(h)  ) ' \left( \hat u_{it}x_{it} - \frac{1}{T}\sum_{u=1}^T \hat u_{iu}x_{iu}   \right) \left( \hat u_{is}x_{is} - \frac{1}{T}\sum_{u=1}^T \hat u_{iu}x_{iu}   \right)' \delta_i(h')   \\
	& -  (\hat \delta_i(h) - \delta_i(h))' \left(  u_{it}x_{it} - \frac{1}{T}\sum_{u=1}^T u_{iu}x_{iu}   \right) \left( u_{is}x_{is} - \frac{1}{T}\sum_{u=1}^T u_{iu}x_{iu}  \right)' \delta_i(h')  \\
	& +    \delta_i(h)' \left( \hat u_{it}x_{it} - \frac{1}{T}\sum_{u=1}^T \hat u_{iu}x_{iu}   \right) \left( \hat u_{is}x_{is} - \frac{1}{T}\sum_{u=1}^T \hat u_{iu}x_{iu}   \right)' ( \hat \delta_i(h')- \delta_i(h')  )  \\
	& -  \delta_i(h)' \left(  u_{it}x_{it} - \frac{1}{T}\sum_{u=1}^T u_{iu}x_{iu}   \right) \left( u_{is}x_{is} - \frac{1}{T}\sum_{u=1}^T u_{iu}x_{iu}  \right)' ( \hat \delta_i(h')- \delta_i(h')  ) \\
	& +  \delta_i(h)'  \left( \hat u_{it}x_{it} - \frac{1}{T}\sum_{u=1}^T \hat u_{iu}x_{iu}   \right) \left( \hat u_{is}x_{is} - \frac{1}{T}\sum_{u=1}^T \hat u_{iu}x_{iu}   \right)' \delta_i (h') \\ 
	& - \delta_i (h)' \left(  u_{it}x_{it} - \frac{1}{T}\sum_{u=1}^T u_{iu}x_{iu}   \right) \left( u_{is}x_{is} - \frac{1}{T}\sum_{u=1}^T u_{iu}x_{iu}  \right)' \delta_i(h') \\
	& +(\hat \delta_i(h) - \delta_i(h))' \left(  u_{it}x_{it} - \frac{1}{T}\sum_{u=1}^T u_{iu}x_{iu}   \right) \left( u_{is}x_{is} - \frac{1}{T}\sum_{u=1}^T u_{iu}x_{iu}  \right)' ( \hat \delta_i(h')- \delta_i(h')  )\\
	& +(\hat \delta_i(h) - \delta_i(h))' \left(  u_{it}x_{it} - \frac{1}{T}\sum_{u=1}^T u_{iu}x_{iu}   \right) \left( u_{is}x_{is} - \frac{1}{T}\sum_{u=1}^T u_{iu}x_{iu}  \right)'  \delta_i(h') \\
	& + \delta_i(h)' \left(  u_{it}x_{it} - \frac{1}{T}\sum_{u=1}^T u_{iu}x_{iu}   \right) \left( u_{is}x_{is} - \frac{1}{T}\sum_{u=1}^T u_{iu}x_{iu}  \right)' ( \hat \delta_i(h')- \delta_i(h') ).
\end{align*}

Next, we consider the following decomposition:
\begin{align*}
	\hat u_{it}x_{it} - \frac{1}{T}\sum_{u=1}^T \hat u_{iu}x_{iu}  
	= & u_{it} x_{it} -  \frac{1}{T}\sum_{u=1}^T u_{iu}x_{iu} 
    - \left( x_{it} x_{it}' - \overline{x_i x_i} \right) \left(\hat{\theta}_{g_i^0} - \theta_{g_i^0}\right) 
\\   
    & - \left( x_{it} w_{it}'  -  \overline{x_i w_i}\right) \left(\hat{\theta}^w - \theta^w\right).
\end{align*}
We thus have 
\begin{align*}
	& \left( \hat u_{it}x_{it} - \frac{1}{T}\sum_{u=1}^T \hat u_{iu}x_{iu}   \right) \left( \hat u_{is}x_{is} - \frac{1}{T}\sum_{u=1}^T \hat u_{iu}x_{iu}   \right)' \\
	= & \left(  u_{it}x_{it} - \overline{u_ix_i}   \right) \left( u_{is}x_{is} - \overline{u_ix_i}  \right)'  \\
	& - \left(  u_{it}x_{it} - \overline{u_ix_i}   \right) \left(  \left( x_{is} x_{is}' - \overline{x_i x_i} \right) \left(\hat{\theta}_{g_i^0} - \theta_{g_i^0}\right) + \left( x_{is} w_{is}' -  \overline{x_i w_i}\right) \left(\hat{\theta}^w - \theta^w\right) \right)' \\
	& - \left(  \left( x_{it} x_{it}' - \overline{x_i x_i} \right) \left(\hat{\theta}_{g_i^0} - \theta_{g_i^0}\right) + \left( x_{it} w_{it}' -  \overline{x_i w_i}\right) \left(\hat{\theta}^w - \theta^w\right) \right) \left( u_{is}x_{is} - \overline{u_ix_i}  \right)'
\\
    &+  \left( \left( x_{it} x_{it}' - \overline{x_i x_i} \right) \left(\hat{\theta}_{g_i^0} - \theta_{g_i^0}\right) + \left( x_{it} w_{it}' -  \overline{x_i w_i}\right) \left(\hat{\theta}^w - \theta^w\right) \right) \\
	& \qquad \times  \left(  \left( x_{is} x_{is}' - \overline{x_i x_i} \right) \left(\hat{\theta}_{g_i^0} - \theta_{g_i^0}\right) + \left( x_{is} w_{is}' -  \overline{x_i w_i}\right) \left(\hat{\theta}^w - \theta^w\right) \right)' \\
	= & \left(  u_{it}x_{it} - \overline{u_ix_i}   \right) \left( u_{is}x_{is} - \overline{u_ix_i}  \right)'  \\
	& - \left(  u_{it}x_{it} - \overline{u_ix_i}   \right) \left(  \left( x_{is} x_{is}' - \overline{x_i x_i} \right) \left(\hat{\theta}_{g_i^0} - \theta_{g_i^0}\right) \right)' \\
	& - \left(  u_{it}x_{it} - \overline{u_ix_i}   \right) \left(   \left( x_{is} w_{is}' -  \overline{x_i w_i}\right) \left(\hat{\theta}^w - \theta^w\right) \right)' \\
	& - \left(  \left( x_{it} x_{it}' - \overline{x_i x_i} \right) \left(\hat{\theta}_{g_i^0} - \theta_{g_i^0}\right) \right) \left( u_{is}x_{is} - \overline{u_ix_i}  \right)' \\
	& - \left(  \left( x_{it} w_{it}' -  \overline{x_i w_i}\right) \left(\hat{\theta}^w - \theta^w\right) \right) \left( u_{is}x_{is} - \overline{u_ix_i}  \right)' \\
	& +  \left( \left( x_{it} x_{it}' - \overline{x_i x_i} \right) \left(\hat{\theta}_{g_i^0} - \theta_{g_i^0}\right) \right) \left(  \left( x_{is} x_{is}' - \overline{x_i x_i} \right) \left(\hat{\theta}_{g_i^0} - \theta_{g_i^0}\right) \right)'\\
	& +  \left( \left( x_{it} x_{it}' - \overline{x_i x_i} \right) \left(\hat{\theta}_{g_i^0} - \theta_{g_i^0}\right)\right)  \left(  \left( x_{is} w_{is}' -  \overline{x_i w_i}\right) \left(\hat{\theta}^w - \theta^w\right) \right)'\\
	& +  \left( \left( x_{it} w_{it}' -  \overline{x_i w_i}\right) \left(\hat{\theta}^w - \theta^w\right) \right)  \left(  \left( x_{is} x_{is}' - \overline{x_i x_i} \right) \left(\hat{\theta}_{g_i^0} - \theta_{g_i^0}\right) \right)'\\
	& +  \left( \left( x_{it} w_{it}' -  \overline{x_i w_i}\right) \left(\hat{\theta}^w - \theta^w\right) \right)  \left(   \left( x_{is} w_{is}' -  \overline{x_i w_i}\right) \left(\hat{\theta}^w - \theta^w\right) \right)'.
\end{align*}
Combining these two decomposition results, we have 
\begin{align*}
	& \left( \hat{d}_{it} (g^0, h) - \bar{\hat{d}}_{i} (g^0, h) \right)\left( \hat{d}_{is} (g^0, h') - \bar{\hat{d}}_{i} (g^0, h') \right) \\
	& -  \left( d_{it} (g^0, h) - \bar{d}_{i} (g^0, h) \right)\left( d_{is} (g^0, h') - \bar{d}_{i} (g^0, h') \right) \\
	=& (\hat \delta_i(h) - \delta_i(h))'\bigg(  - \left(  u_{it}x_{it} - \overline{u_ix_i}   \right) \left(  \left( x_{is} x_{is}' - \overline{x_i x_i} \right) \left(\hat{\theta}_{g_i^0} - \theta_{g_i^0}\right) \right)' \\
	& - \left(  u_{it}x_{it} - \overline{u_i x_i}   \right) \left(   \left( x_{is} w_{is}' -  \overline{x_i w_i}\right) \left(\hat{\theta}^w - \theta^w\right) \right)' \\
	& - \left(  \left( x_{it} x_{it}' - \overline{x_i x_i} \right) \left(\hat{\theta}_{g_i^0} - \theta_{g_i^0}\right) \right) \left( u_{is}x_{is} - \overline{u_i x_i}  \right)' \\
	& - \left(  \left( x_{it} w_{it}' -  \overline{x_i w_i}\right) \left(\hat{\theta}^w - \theta^w\right) \right) \left( u_{is}x_{is} - \overline{u_i x_i}  \right)' \\
	& +  \left( \left( x_{it} x_{it}' - \overline{x_i x_i} \right) \left(\hat{\theta}_{g_i^0} - \theta_{g_i^0}\right) \right) \left(  \left( x_{is} x_{is}' - \overline{x_i x_i} \right) \left(\hat{\theta}_{g_i^0} - \theta_{g_i^0}\right) \right)'\\
	& +  \left( \left( x_{it} x_{it}' - \overline{x_i x_i} \right) \left(\hat{\theta}_{g_i^0} - \theta_{g_i^0}\right)\right)  \left(  \left( x_{is} w_{is}' -  \overline{x_i w_i}\right) \left(\hat{\theta}^w - \theta^w\right) \right)'\\
	& +  \left( \left( x_{it} w_{it}' -  \overline{x_i w_i}\right) \left(\hat{\theta}^w - \theta^w\right) \right)  \left(  \left( x_{is} x_{is}' - \overline{x_i x_i} \right) \left(\hat{\theta}_{g_i^0} - \theta_{g_i^0}\right) \right)'\\
	& +  \left( \left( x_{it} w_{it}' -  \overline{x_i w_i}\right) \left(\hat{\theta}^w - \theta^w\right) \right)  \left(   \left( x_{is} w_{is}' -  \overline{x_i w_i}\right) \left(\hat{\theta}^w - \theta^w\right) \right)' \bigg)  ( \hat \delta_i(h')- \delta_i(h')  ) \\
	& + (\hat \delta_i(h) - \delta_i(h))'\bigg( - \left(  u_{it}x_{it} - \overline{u_i x_i}   \right) \left(  \left( x_{is} x_{is}' - \overline{x_i x_i} \right) \left(\hat{\theta}_{g_i^0} - \theta_{g_i^0}\right) \right)' \\
	& - \left(  u_{it}x_{it} - \overline{u_i x_i}   \right) \left(   \left( x_{is} w_{is}' -  \overline{x_i w_i}\right) \left(\hat{\theta}^w - \theta^w\right) \right)' \\
	& - \left(  \left( x_{it} x_{it}' - \overline{x_i x_i} \right) \left(\hat{\theta}_{g_i^0} - \theta_{g_i^0}\right) \right) \left( u_{is}x_{is} - \overline{u_i x_i}  \right)' \\
	& - \left(  \left( x_{it} w_{it}' -  \overline{x_i w_i}\right) \left(\hat{\theta}^w - \theta^w\right) \right) \left( u_{is}x_{is} - \overline{u_i x_i}  \right)' \\
	& +  \left( \left( x_{it} x_{it}' - \overline{x_i x_i} \right) \left(\hat{\theta}_{g_i^0} - \theta_{g_i^0}\right) \right) \left(  \left( x_{is} x_{is}' - \overline{x_i x_i} \right) \left(\hat{\theta}_{g_i^0} - \theta_{g_i^0}\right) \right)'\\
	& +  \left( \left( x_{it} x_{it}' - \overline{x_i x_i} \right) \left(\hat{\theta}_{g_i^0} - \theta_{g_i^0}\right)\right)  \left(  \left( x_{is} w_{is}' -  \overline{x_i w_i}\right) \left(\hat{\theta}^w - \theta^w\right) \right)'\\
	& +  \left( \left( x_{it} w_{it}' -  \overline{x_i w_i}\right) \left(\hat{\theta}^w - \theta^w\right) \right)  \left(  \left( x_{is} x_{is}' - \overline{x_i x_i} \right) \left(\hat{\theta}_{g_i^0} - \theta_{g_i^0}\right) \right)'\\
	& +  \left( \left( x_{it} w_{it}' -  \overline{x_i w_i}\right) \left(\hat{\theta}^w - \theta^w\right) \right)  \left(   \left( x_{is} w_{is}' -  \overline{x_i w_i}\right) \left(\hat{\theta}^w - \theta^w\right) \right)' \bigg)   \delta_i(h') \\
	& +  \delta_i(h)'\bigg(  - \left(  u_{it}x_{it} - \overline{u_i x_i}   \right) \left(  \left( x_{is} x_{is}' - \overline{x_i x_i} \right) \left(\hat{\theta}_{g_i^0} - \theta_{g_i^0}\right) \right)' \\
	& - \left(  u_{it}x_{it} - \overline{u_i x_i}   \right) \left(   \left( x_{is} w_{is}' -  \overline{x_i w_i}\right) \left(\hat{\theta}^w - \theta^w\right) \right)' \\
	& - \left(  \left( x_{it} x_{it}' - \overline{x_i x_i} \right) \left(\hat{\theta}_{g_i^0} - \theta_{g_i^0}\right) \right) \left( u_{is}x_{is} - \overline{u_i x_i}  \right)' \\
	& - \left(  \left( x_{it} w_{it}' -  \overline{x_i w_i}\right) \left(\hat{\theta}^w - \theta^w\right) \right) \left( u_{is}x_{is} - \overline{u_i x_i}  \right)' \\
	& +  \left( \left( x_{it} x_{it}' - \overline{x_i x_i} \right) \left(\hat{\theta}_{g_i^0} - \theta_{g_i^0}\right) \right) \left(  \left( x_{is} x_{is}' - \overline{x_i x_i} \right) \left(\hat{\theta}_{g_i^0} - \theta_{g_i^0}\right) \right)'\\
	& +  \left( \left( x_{it} x_{it}' - \overline{x_i x_i} \right) \left(\hat{\theta}_{g_i^0} - \theta_{g_i^0}\right)\right)  \left(  \left( x_{is} w_{is}' -  \overline{x_i w_i}\right) \left(\hat{\theta}^w - \theta^w\right) \right)'\\
	& +  \left( \left( x_{it} w_{it}' -  \overline{x_i w_i}\right) \left(\hat{\theta}^w - \theta^w\right) \right)  \left(  \left( x_{is} x_{is}' - \overline{x_i x_i} \right) \left(\hat{\theta}_{g_i^0} - \theta_{g_i^0}\right) \right)'\\
	& +  \left( \left( x_{it} w_{it}' -  \overline{x_i w_i}\right) \left(\hat{\theta}^w - \theta^w\right) \right)  \left(   \left( x_{is} w_{is}' -  \overline{x_i w_i}\right) \left(\hat{\theta}^w - \theta^w\right) \right)' \bigg)  ( \hat \delta_i(h')- \delta_i(h')  ) \\
	& +  \delta_i(h)'\bigg( - \left(  u_{it}x_{it} - \overline{u_i x_i}   \right) \left(  \left( x_{is} x_{is}' - \overline{x_i x_i} \right) \left(\hat{\theta}_{g_i^0} - \theta_{g_i^0}\right) \right)' \\
	& - \left(  u_{it}x_{it} - \overline{u_i x_i}   \right) \left(   \left( x_{is} w_{is}' -  \overline{x_i w_i}\right) \left(\hat{\theta}^w - \theta^w\right) \right)' \\
	& - \left(  \left( x_{it} x_{it}' - \overline{x_i x_i} \right) \left(\hat{\theta}_{g_i^0} - \theta_{g_i^0}\right) \right) \left( u_{is}x_{is} - \overline{u_i x_i}  \right)' \\
	& - \left(  \left( x_{it} w_{it}' -  \overline{x_i w_i}\right) \left(\hat{\theta}^w - \theta^w\right) \right) \left( u_{is}x_{is} - \overline{u_i x_i}  \right)' \\
	& +  \left( \left( x_{it} x_{it}' - \overline{x_i x_i} \right) \left(\hat{\theta}_{g_i^0} - \theta_{g_i^0}\right) \right) \left(  \left( x_{is} x_{is}' - \overline{x_i x_i} \right) \left(\hat{\theta}_{g_i^0} - \theta_{g_i^0}\right) \right)'\\
	& +  \left( \left( x_{it} x_{it}' - \overline{x_i x_i} \right) \left(\hat{\theta}_{g_i^0} - \theta_{g_i^0}\right)\right)  \left(  \left( x_{is} w_{is}' -  \overline{x_i w_i}\right) \left(\hat{\theta}^w - \theta^w\right) \right)'\\
	& +  \left( \left( x_{it} w_{it}' -  \overline{x_i w_i}\right) \left(\hat{\theta}^w - \theta^w\right) \right)  \left(  \left( x_{is} x_{is}' - \overline{x_i x_i} \right) \left(\hat{\theta}_{g_i^0} - \theta_{g_i^0}\right) \right)'\\
	& +  \left( \left( x_{it} w_{it}' -  \overline{x_i w_i}\right) \left(\hat{\theta}^w - \theta^w\right) \right)  \left(   \left( x_{is} w_{is}' -  \overline{x_i w_i}\right) \left(\hat{\theta}^w - \theta^w\right) \right)' \bigg)  \delta_i(h')   \\
	& +(\hat \delta_i(h) - \delta_i(h))' \left(  u_{it}x_{it} - \overline{u_i x_i}   \right) \left( u_{is}x_{is} - \overline{u_i x_i}  \right)' ( \hat \delta_i(h')- \delta_i(h')  )\\
	& +(\hat \delta_i(h) - \delta_i(h))' \left(  u_{it}x_{it} - \overline{u_i x_i}   \right) \left( u_{is}x_{is} - \overline{u_i x_i}  \right)'  \delta_i(h') \\
	& + \delta_i(h)' \left(  u_{it}x_{it} - \overline{u_i x_i}   \right) \left( u_{is}x_{is} - \overline{u_i x_i}  \right)' ( \hat \delta_i(h')- \delta_i(h') ).
\end{align*}

It thus holds that
\begin{align*}
&	\widehat{\Xi}_i (h,h') - \widetilde{\Xi}_i (h,h') \\
& = - (\hat \delta_i(h) - \delta_i(h))' \sum_{j=-T+1}^{T-1}K_N^{(j)} \frac{1}{T} \sum_{t=|j|+1}^T\left(  u_{it}^{(+j)}x_{it}^{(+j)} - \overline{u_i x_i}   \right) \\
	&  \quad \times   \left(  \left( x_{it}^{(-j)} {x_{it}^{(-j)}}' - \overline{x_i x_i} \right) \left(\hat{\theta}_{g_i^0} - \theta_{g_i^0}\right) \right)' (\hat \delta_i(h) - \delta_i(h))\\
	& - (\hat \delta_i(h) - \delta_i(h))' \sum_{j=-T+1}^{T-1}K_N^{(j)} \frac{1}{T} \sum_{t=|j|+1}^T\left(  u_{it}^{(+j)}x_{it}^{(+j)} - \overline{u_i x_i}   \right) \\
	&  \quad \times \left(   \left( x_{it}^{(-j)} {w_{it}^{(-j)}}' -  \overline{x_i w_i}\right) \left(\hat{\theta}^w - \theta^w\right) \right)' (\hat \delta_i(h) - \delta_i(h))\\
	& - (\hat \delta_i(h) - \delta_i(h))' \sum_{j=-T+1}^{T-1}K_N^{(j)} \frac{1}{T} \sum_{t=|j|+1}^T\left(  \left( x_{it}^{(+j)} {x_{it}^{(+j)}}' - \overline{x_i x_i} \right) \left(\hat{\theta}_{g_i^0} - \theta_{g_i^0}\right) \right) \\
	&  \quad \times \left( u_{i,t-\max (0,j)}x_{it}^{(-j)} - \overline{u_i x_i}  \right)' (\hat \delta_i(h) - \delta_i(h)) \\
	& - (\hat \delta_i(h) - \delta_i(h))' \sum_{j=-T+1}^{T-1}K_N^{(j)} \frac{1}{T} \sum_{t=|j|+1}^T\left(  \left( x_{it}^{(+j)} {w_{it}^{(+j)}}' -  \overline{x_i w_i}\right) \left(\hat{\theta}^w - \theta^w\right) \right) \\
	&  \quad \times  \left( u_{it}^{(-j)}x_{it}^{(-j)} - \overline{u_i x_i}  \right)' (\hat \delta_i(h) - \delta_i(h)) \\
	& + (\hat \delta_i(h) - \delta_i(h))' \sum_{j=-T+1}^{T-1}K_N^{(j)} \frac{1}{T} \sum_{t=|j|+1}^T \left( \left( x_{it}^{(+j)} {x_{it}^{(+j)}}' - \overline{x_i x_i} \right) \left(\hat{\theta}_{g_i^0} - \theta_{g_i^0}\right) \right) \\
	&  \quad \times  \left(  \left( x_{it}^{(-j)} {x_{it}^{(-j)}}' - \overline{x_i x_i} \right) \left(\hat{\theta}_{g_i^0} - \theta_{g_i^0}\right) \right)' (\hat \delta_i(h) - \delta_i(h))\\
	& + (\hat \delta_i(h) - \delta_i(h))' \sum_{j=-T+1}^{T-1}K_N^{(j)} \frac{1}{T} \sum_{t=|j|+1}^T \left( \left( x_{it}^{(+j)} {x_{it}^{(+j)}}' - \overline{x_i x_i} \right) \left(\hat{\theta}_{g_i^0} - \theta_{g_i^0}\right)\right) \\
	&  \quad \times  \left(  \left( x_{it}^{(-j)} {w_{it}^{(-j)}}' -  \overline{x_i w_i}\right) \left(\hat{\theta}^w - \theta^w\right) \right)' (\hat \delta_i(h) - \delta_i(h))\\
	& + (\hat \delta_i(h) - \delta_i(h))' \sum_{j=-T+1}^{T-1}K_N^{(j)} \frac{1}{T} \sum_{t=|j|+1}^T \left( \left( x_{it}^{(+j)} {w_{it}^{(+j)}}' -  \overline{x_i w_i}\right) \left(\hat{\theta}^w - \theta^w\right) \right) \\
	&  \quad \times  \left(  \left( x_{it}^{(-j)} {x_{it}^{(-j)}}' - \overline{x_i x_i} \right) \left(\hat{\theta}_{g_i^0} - \theta_{g_i^0}\right) \right)' (\hat \delta_i(h) - \delta_i(h))\\
	& + (\hat \delta_i(h) - \delta_i(h))' \sum_{j=-T+1}^{T-1}K_N^{(j)} \frac{1}{T} \sum_{t=|j|+1}^T \left( \left( x_{it}^{(+j)} {w_{it}^{(+j)}}' -  \overline{x_i w_i}\right) \left(\hat{\theta}^w - \theta^w\right) \right) \\
	&  \quad \times  \left(   \left( x_{it}^{(-j)} {w_{it}^{(-j)}}' -  \overline{x_i w_i}\right) \left(\hat{\theta}^w - \theta^w\right) \right)' \bigg)  ( \hat \delta_i(h')- \delta_i(h')  ) \\
	& - (\hat \delta_i(h) - \delta_i(h))' \sum_{j=-T+1}^{T-1}K_N^{(j)} \frac{1}{T} \sum_{t=|j|+1}^T\left(  u_{it}^{(+j)}x_{it}^{(+j)} - \overline{u_i x_i}   \right) \\
	&  \quad \times  \left(  \left( x_{it}^{(-j)} {x_{it}^{(-j)}}' - \overline{x_i x_i} \right) \left(\hat{\theta}_{g_i^0} - \theta_{g_i^0}\right) \right)' \delta_i(h) \\
	& - (\hat \delta_i(h) - \delta_i(h))' \sum_{j=-T+1}^{T-1}K_N^{(j)} \frac{1}{T} \sum_{t=|j|+1}^T\left(  u_{it}^{(+j)}x_{it}^{(+j)} - \overline{u_i x_i}   \right)  \\
	&  \quad \times \left(   \left( x_{it}^{(-j)} {w_{it}^{(-j)}}' -  \overline{x_i w_i}\right) \left(\hat{\theta}^w - \theta^w\right) \right)'  \delta_i(h) \\
	& - (\hat \delta_i(h) - \delta_i(h))' \sum_{j=-T+1}^{T-1}K_N^{(j)} \frac{1}{T} \sum_{t=|j|+1}^T\left(  \left( x_{it}^{(+j)} {x_{it}^{(+j)}}' - \overline{x_i x_i} \right) \left(\hat{\theta}_{g_i^0} - \theta_{g_i^0}\right) \right) \\
	&  \quad \times  \left( u_{it}^{(-j)}x_{it}^{(-j)} - \overline{u_i x_i}  \right)'  \delta_i(h) \\
	& - (\hat \delta_i(h) - \delta_i(h))' \sum_{j=-T+1}^{T-1}K_N^{(j)} \frac{1}{T} \sum_{t=|j|+1}^T\left(  \left( x_{it}^{(+j)} {w_{it}^{(+j)}}' -  \overline{x_i w_i}\right) \left(\hat{\theta}^w - \theta^w\right) \right) \\
	&  \quad \times  \left( u_{it}^{(-j)}x_{it}^{(-j)} - \overline{u_i x_i}  \right)'  \delta_i(h) \\
	& + (\hat \delta_i(h) - \delta_i(h))' \sum_{j=-T+1}^{T-1}K_N^{(j)} \frac{1}{T} \sum_{t=|j|+1}^T \left( \left( x_{it}^{(+j)} {x_{it}^{(+j)}}' - \overline{x_i x_i} \right) \left(\hat{\theta}_{g_i^0} - \theta_{g_i^0}\right) \right) \\
	&  \quad \times  \left(  \left( x_{it}^{(-j)} {x_{it}^{(-j)}}' - \overline{x_i x_i} \right) \left(\hat{\theta}_{g_i^0} - \theta_{g_i^0}\right) \right)'  \delta_i(h) \\
	& + (\hat \delta_i(h) - \delta_i(h))' \sum_{j=-T+1}^{T-1}K_N^{(j)} \frac{1}{T} \sum_{t=|j|+1}^T \left( \left( x_{it}^{(+j)} {x_{it}^{(+j)}}' - \overline{x_i x_i} \right) \left(\hat{\theta}_{g_i^0} - \theta_{g_i^0}\right)\right) \\
	&  \quad \times  \left(  \left( x_{it}^{(-j)} {w_{it}^{(-j)}}' -  \overline{x_i w_i}\right) \left(\hat{\theta}^w - \theta^w\right) \right)'  \delta_i(h) \\
	& + (\hat \delta_i(h) - \delta_i(h))' \sum_{j=-T+1}^{T-1}K_N^{(j)} \frac{1}{T} \sum_{t=|j|+1}^T \left( \left( x_{it}^{(+j)} {w_{it}^{(+j)}}' -  \overline{x_i w_i}\right) \left(\hat{\theta}^w - \theta^w\right) \right) \\
	&  \quad \times  \left(  \left( x_{it}^{(-j)} {x_{it}^{(-j)}}' - \overline{x_i x_i} \right) \left(\hat{\theta}_{g_i^0} - \theta_{g_i^0}\right) \right)'  \delta_i(h) \\
	& +  (\hat \delta_i(h) - \delta_i(h))' \sum_{j=-T+1}^{T-1}K_N^{(j)} \frac{1}{T} \sum_{t=|j|+1}^T\left( \left( x_{it}^{(+j)} {w_{it}^{(+j)}}' -  \overline{x_i w_i}\right) \left(\hat{\theta}^w - \theta^w\right) \right) \\
	&  \quad \times  \left(   \left( x_{it}^{(-j)} {w_{it}^{(-j)}}' -  \overline{x_i w_i}\right) \left(\hat{\theta}^w - \theta^w\right) \right)' \bigg)   \delta_i(h') \\
	& -  \delta_i(h)' \sum_{j=-T+1}^{T-1}K_N^{(j)} \frac{1}{T} \sum_{t=|j|+1}^T \left(  u_{it}^{(+j)}x_{it}^{(+j)} - \overline{u_i x_i}   \right) \\
	&  \quad \times  \left(  \left( x_{it}^{(-j)} {x_{it}^{(-j)}}' - \overline{x_i x_i} \right) \left(\hat{\theta}_{g_i^0} - \theta_{g_i^0}\right) \right)'(\hat \delta_i(h) - \delta_i(h)) \\
	& -  \delta_i(h)' \sum_{j=-T+1}^{T-1}K_N^{(j)} \frac{1}{T} \sum_{t=|j|+1}^T \left(  u_{it}^{(+j)}x_{it}^{(+j)} - \overline{u_i x_i}   \right) \\
	&  \quad \times  \left(   \left( x_{it}^{(-j)} {w_{it}^{(-j)}}' -  \overline{x_i w_i}\right) \left(\hat{\theta}^w - \theta^w\right) \right)' (\hat \delta_i(h) - \delta_i(h)) \\
	& - \delta_i(h)' \sum_{j=-T+1}^{T-1}K_N^{(j)} \frac{1}{T} \sum_{t=|j|+1}^T \left(  \left( x_{it}^{(+j)} {x_{it}^{(+j)}}' - \overline{x_i x_i} \right) \left(\hat{\theta}_{g_i^0} - \theta_{g_i^0}\right) \right) \\
	&  \quad \times  \left( u_{it}^{(-j)}x_{it}^{(-j)} - \overline{u_i x_i}  \right)' (\hat \delta_i(h) - \delta_i(h))\\
	& -  \delta_i(h)' \sum_{j=-T+1}^{T-1}K_N^{(j)} \frac{1}{T} \sum_{t=|j|+1}^T \left(  \left( x_{it}^{(+j)} {w_{it}^{(+j)}}' -  \overline{x_i w_i}\right) \left(\hat{\theta}^w - \theta^w\right) \right) \\
	&  \quad \times  \left( u_{it}^{(-j)}x_{it}^{(-j)} - \overline{u_i x_i}  \right)' (\hat \delta_i(h) - \delta_i(h))\\
	& +  \delta_i(h)' \sum_{j=-T+1}^{T-1}K_N^{(j)} \frac{1}{T} \sum_{t=|j|+1}^T \left( \left( x_{it}^{(+j)} {x_{it}^{(+j)}}' - \overline{x_i x_i} \right) \left(\hat{\theta}_{g_i^0} - \theta_{g_i^0}\right) \right) \\
	&  \quad \times  \left(  \left( x_{it}^{(-j)} {x_{it}^{(-j)}}' - \overline{x_i x_i} \right) \left(\hat{\theta}_{g_i^0} - \theta_{g_i^0}\right) \right)' (\hat \delta_i(h) - \delta_i(h)) \\
	& +  \delta_i(h)' \sum_{j=-T+1}^{T-1}K_N^{(j)} \frac{1}{T} \sum_{t=|j|+1}^T \left( \left( x_{it}^{(+j)} {x_{it}^{(+j)}}' - \overline{x_i x_i} \right) \left(\hat{\theta}_{g_i^0} - \theta_{g_i^0}\right)\right) \\
	&  \quad \times  \left(  \left( x_{it}^{(-j)} {w_{it}^{(-j)}}' -  \overline{x_i w_i}\right) \left(\hat{\theta}^w - \theta^w\right) \right)' (\hat \delta_i(h) - \delta_i(h)) \\
	& + \delta_i(h)' \sum_{j=-T+1}^{T-1}K_N^{(j)} \frac{1}{T} \sum_{t=|j|+1}^T  \left( \left( x_{it}^{(+j)} {w_{it}^{(+j)}}' -  \overline{x_i w_i}\right) \left(\hat{\theta}^w - \theta^w\right) \right) \\
	&  \quad \times  \left(  \left( x_{it}^{(-j)} {x_{it}^{(-j)}}' - \overline{x_i x_i} \right) \left(\hat{\theta}_{g_i^0} - \theta_{g_i^0}\right) \right)' (\hat \delta_i(h) - \delta_i(h)) \\
	& +  \delta_i(h)' \sum_{j=-T+1}^{T-1}K_N^{(j)} \frac{1}{T} \sum_{t=|j|+1}^T \left( \left( x_{it}^{(+j)} {w_{it}^{(+j)}}' -  \overline{x_i w_i}\right) \left(\hat{\theta}^w - \theta^w\right) \right) \\
	&  \quad \times  \left(   \left( x_{it}^{(-j)} {w_{it}^{(-j)}}' -  \overline{x_i w_i}\right) \left(\hat{\theta}^w - \theta^w\right) \right)' \bigg)  ( \hat \delta_i(h')- \delta_i(h')  ) \\
	& +  \delta_i(h)' \sum_{j=-T+1}^{T-1}K_N^{(j)} \frac{1}{T} \sum_{t=|j|+1}^T \bigg( - \left(  u_{it}^{(+j)}x_{it}^{(+j)} - \overline{u_i x_i}   \right) \\
	&  \quad \times  \left(  \left( x_{it}^{(-j)} {x_{it}^{(-j)}}' - \overline{x_i x_i} \right) \left(\hat{\theta}_{g_i^0} - \theta_{g_i^0}\right) \right)'  \delta_i(h)\\
	& - \delta_i(h)' \sum_{j=-T+1}^{T-1}K_N^{(j)} \frac{1}{T} \sum_{t=|j|+1}^T \left(  u_{it}^{(+j)}x_{it}^{(+j)} - \overline{u_i x_i}   \right) \\
	&  \quad \times  \left(   \left( x_{it}^{(-j)} {w_{it}^{(-j)}}' -  \overline{x_i w_i}\right) \left(\hat{\theta}^w - \theta^w\right) \right)'  \delta_i(h) \\
	& - \delta_i(h)' \sum_{j=-T+1}^{T-1}K_N^{(j)} \frac{1}{T} \sum_{t=|j|+1}^T \left(  \left( x_{it}^{(+j)} {x_{it}^{(+j)}}' - \overline{x_i x_i} \right) \left(\hat{\theta}_{g_i^0} - \theta_{g_i^0}\right) \right) \\
	&  \quad \times  \left( u_{it}^{(-j)}x_{it}^{(-j)} - \overline{u_i x_i}  \right)'  \delta_i(h) \\
	& - \delta_i(h)' \sum_{j=-T+1}^{T-1}K_N^{(j)} \frac{1}{T} \sum_{t=|j|+1}^T \left(  \left( x_{it}^{(+j)} {w_{it}^{(+j)}}' -  \overline{x_i w_i}\right) \left(\hat{\theta}^w - \theta^w\right) \right) \\
	&  \quad \times  \left( u_{it}^{(-j)}x_{it}^{(-j)} - \overline{u_i x_i}  \right)' \delta_i(h) \\
	& +  \delta_i(h)' \sum_{j=-T+1}^{T-1}K_N^{(j)} \frac{1}{T} \sum_{t=|j|+1}^T \left( \left( x_{it}^{(+j)} {x_{it}^{(+j)}}' - \overline{x_i x_i} \right) \left(\hat{\theta}_{g_i^0} - \theta_{g_i^0}\right) \right) \\
	&  \quad \times  \left(  \left( x_{it}^{(-j)} {x_{it}^{(-j)}}' - \overline{x_i x_i} \right) \left(\hat{\theta}_{g_i^0} - \theta_{g_i^0}\right) \right)'  \delta_i(h) \\
	& +  \delta_i(h)' \sum_{j=-T+1}^{T-1}K_N^{(j)} \frac{1}{T} \sum_{t=|j|+1}^T \left( \left( x_{it}^{(+j)} {x_{it}^{(+j)}}' - \overline{x_i x_i} \right) \left(\hat{\theta}_{g_i^0} - \theta_{g_i^0}\right)\right) \\
	&  \quad \times  \left(  \left( x_{it}^{(-j)} {w_{it}^{(-j)}}' -  \overline{x_i w_i}\right) \left(\hat{\theta}^w - \theta^w\right) \right)'  \delta_i(h) \\
	& +  \delta_i(h)' \sum_{j=-T+1}^{T-1}K_N^{(j)} \frac{1}{T} \sum_{t=|j|+1}^T \left( \left( x_{it}^{(+j)} {w_{it}^{(+j)}}' -  \overline{x_i w_i}\right) \left(\hat{\theta}^w - \theta^w\right) \right) \\
	&  \quad \times  \left(  \left( x_{it}^{(-j)} {x_{it}^{(-j)}}' - \overline{x_i x_i} \right) \left(\hat{\theta}_{g_i^0} - \theta_{g_i^0}\right) \right)'  \delta_i(h) \\
	& +  \delta_i(h)' \sum_{j=-T+1}^{T-1}K_N^{(j)} \frac{1}{T} \sum_{t=|j|+1}^T \left( \left( x_{it}^{(+j)} {w_{it}^{(+j)}}' -  \overline{x_i w_i}\right) \left(\hat{\theta}^w - \theta^w\right) \right) \\
	&  \quad \times  \left(   \left( x_{it}^{(-j)} {w_{it}^{(-j)}}' -  \overline{x_i w_i}\right) \left(\hat{\theta}^w - \theta^w\right) \right)' \bigg)  \delta_i(h')   \\
	& + (\hat \delta_i(h) - \delta_i(h))'  \sum_{j=-T+1}^{T-1}K_N^{(j)} \frac{1}{T} \sum_{t=|j|+1}^T\left(  u_{it}^{(+j)}x_{it}^{(+j)} - \overline{u_i x_i}   \right) \\
	&  \quad \times  \left( u_{it}^{(-j)}x_{it}^{(-j)} - \overline{u_i x_i}  \right)' ( \hat \delta_i(h')- \delta_i(h')  )\\
	& + (\hat \delta_i(h) - \delta_i(h))'  \sum_{j=-T+1}^{T-1}K_N^{(j)} \frac{1}{T} \sum_{t=|j|+1}^T \left(  u_{it}^{(+j)}x_{it}^{(+j)} - \overline{u_i x_i}   \right) \\
	&  \quad \times  \left( u_{it}^{(-j)}x_{it}^{(-j)} - \overline{u_i x_i}  \right)'  \delta_i(h') \\
	& +  \delta_i(h)'  \sum_{j=-T+1}^{T-1}K_N^{(j)} \frac{1}{T} \sum_{t=|j|+1}^T \left(  u_{it}^{(+j)}x_{it}^{(+j)} - \overline{u_i x_i}   \right) \\
	&  \quad \times  \left( u_{it}^{(-j)}x_{it}^{(-j)} - \overline{u_i x_i}  \right)' ( \hat \delta_i(h')- \delta_i(h') ).
\end{align*}
We examine each term. For vector $a$, let $a_p$ denote the $p$-th element of $a$. Let $d_x$ be the dimension of $x_{it}$. 
With probability at least $a_{\theta, N}$, it holds
\begin{align*}
	&\frac{1}{\sigma_i^2\norm{\delta_i (h)}\norm{\delta_i (h')} } \bigg|  (\hat \delta_i(h) - \delta_i(h))' \sum_{j=-T+1}^{T-1}K_N^{(j)} \frac{1}{T} \sum_{t=|j|+1}^T\left(  u_{it}^{(+j)}x_{it}^{(+j)} - \overline{u_i x_i}   \right) \\
		&  \quad \times   \left(  \left( x_{it}^{(-j)} {x_{it}^{(-j)}}' - \overline{x_i x_i} \right) \left(\hat{\theta}_{g_i^0} - \theta_{g_i^0}\right) \right)' (\hat \delta_i(h) - \delta_i(h)) \bigg| \\
	\leq  & C r_{\theta, N}^2 \iota_N^{-2} \frac{1}{\sigma_i } \bigg|  \sum_{j=-T+1}^{T-1}K_N^{(j)} \frac{1}{T} \sum_{t=|j|+1}^T\left(  v_{it}^{(j)} x_{it}^{(+j)} - \overline{v_i x_i}   \right) \\
	&  \quad \times   \left(  \left( x_{it}^{(-j)} {x_{it}^{(-j)}}' - \overline{x_i x_i} \right) \left(\hat{\theta}_{g_i^0} - \theta_{g_i^0}\right) \right)' \bigg|_{\infty} \\
	\leq  & C r_{\theta, N}^3 \iota_N^{-2} (\min_{1 \leq i \leq N} \sigma_i)^{-1} \sum_{p=1}^{d_x}\bigg|  \sum_{j=-T+1}^{T-1}K_N^{(j)} \frac{1}{T} \sum_{t=|j|+1}^T\left(  v_{it}^{(j)} x_{it}^{(+j)} - \overline{v_i x_i}   \right) \\
	&  \quad \times  \left( x_{i,t-\max (0,j),p} {x_{it}^{(-j)}}' - \frac{1}{T}\sum_{u=1}^T  x_{iu,p} x_{iu}' \right) \bigg|_{\infty} \\
	\leq  & C r_{\theta, N}^3 \iota_N^{-2} \sum_{p=1}^{d_x} (\min_{1 \leq i \leq N} \sigma_i)^{-1}\bigg|  \sum_{j=-T+1}^{T-1}K_N^{(j)} \frac{1}{T} \sum_{t=|j|+1}^T\left(  v_{it}^{(j)} x_{it}^{(+j)} - \overline{v_i x_i}   \right) \\
	&  \quad \times  \left( x_{i,t-\max (0,j),p} {x_{it}^{(-j)}}' - \frac{1}{T}\sum_{u=1}^T  x_{iu,p} x_{iu}' \right) \\
	& \quad - \frac{1}{T} \sum_{t=1}^T \sum_{s=1}^T\E_P [v_{it}x_{it} (x_{is,p}x_{is} - \E_p (x_{is,p} x_{is}))] \bigg|_{\infty} \\
	& +  C r_{\theta, N}^3 \iota_N^{-2} (\min_{1 \leq i \leq N} \sigma_i)^{-1} \sum_{p=1}^{d_x}\left|  \frac{1}{T} \sum_{t=1}^T \sum_{s=1}^T\E_P [v_{it}x_{it} (x_{is,p}x_{is} - \E_p (x_{is,p} x_{is}))]  \right|_{\infty}
\end{align*}
Applying Lemma \ref{lem:lv-xi-consistency} to
\begin{align*}
	& \bigg|  \sum_{j=-T+1}^{T-1}K_N^{(j)} \frac{1}{T} \sum_{t=|j|+1}^T\left(  v_{it}^{(j)} x_{it}^{(+j)} - \overline{v_i x_i}   \right) \left( x_{i,t-\max (0,j),p} {x_{it}^{(-j)}}' - \frac{1}{T}\sum_{u=1}^T  x_{iu,p} x_{iu}' \right) \\
	& \qquad -  \frac{1}{T} \sum_{t=1}^T \sum_{s=1}^T\E_P [v_{it}x_{it} (x_{is,p}x_{is} - \E_p (x_{is,p} x_{is}))] \bigg|_{\infty},
\end{align*}
and Lemma \ref{lem:lv-exists} to 
\begin{align*}
	\frac{1}{T} \sum_{t=1}^T \sum_{s=1}^T\E_P [v_{it}x_{it} (x_{is,p}x_{is} - \E_p (x_{is,p} x_{is}))], 
\end{align*}
we obtain 
\begin{align*}
	& \sup_{i,h,h'} \frac{1}{\sigma_i^2\norm{\delta_i (h)}\norm{\delta_i (h')} } \bigg|  (\hat \delta_i(h) - \delta_i(h))' \sum_{j=-T+1}^{T-1}
    K_N^{(j)} \frac{1}{T}     \Bigg\{ \sum_{t=|j|+1}^T\left(  u_{it}^{(+j)}x_{it}^{(+j)} - \overline{u_i x_i}   \right) \\
	&  \qquad \times   \left(  \left( x_{it}^{(-j)} {x_{it}^{(-j)}}' - \overline{x_i x_i} \right) \left(\hat{\theta}_{g_i^0} - \theta_{g_i^0}\right) \right)' (\hat \delta_i(h) - \delta_i(h)) \bigg| \Bigg\}\\
	\precsim_P &  r_{\theta, N}^3 \iota_N^{-2} T^{-c_1} (\log (N(G-1)))^{c_2}(\min_{1 \leq i \leq N}\sigma_i)^{-1}+  r_{\theta, N}^3  \iota_N^{-2} T^{-\rho}(\min_{1 \leq i \leq N}\sigma_i)^{-1} \\
	&+ r_{\theta, N}^3 \iota_N^{-2}(\min_{1 \leq i \leq N}\sigma_i)^{-1}  \\
	\precsim_P & r_{\theta, N}^3 \iota_N^{-2}(\min_{1 \leq i \leq N}\sigma_i)^{-1},
\end{align*}
where $\precsim_P$ signifies that for $A_N$ and $B_N$, $A_N \precsim_P B_N$ if $ \sup_{P\in \mathbb{P}_N}\Pr (A_N > B_N \zeta_N) = o(1)$ for any $\zeta_N \to \infty$, and the last $\precsim_P $ follows by $b_N^{LV} \to 0$. 
Following the same argument, we have 
\begin{align*}
	& \sup_{i,h,h'} \frac{1}{\sigma_i^2\norm{\delta_i (h)}\norm{\delta_i (h')} } \bigg| (\hat \delta_i(h) - \delta_i(h))' \sum_{j=-T+1}^{T-1}K_N^{(j)} \frac{1}{T} \Bigg\{ \sum_{t=|j|+1}^T\left(  u_{it}^{(+j)}x_{it}^{(+j)} - \overline{u_i x_i}   \right) \\
	&  \quad \times \left(   \left( x_{it}^{(-j)} {w_{it}^{(-j)}}' -  \overline{x_i w_i}\right) \left(\hat{\theta}^w - \theta^w\right) \right)' (\hat \delta_i(h) - \delta_i(h)) \Bigg\} \bigg| \\
	\precsim_P &  r_{\theta, N}^3 \iota_N^{-2}(\min_{1 \leq i \leq N}\sigma_i)^{-1},
\end{align*}
\begin{align*}
	& \sup_{i,h,h'} \frac{1}{\sigma_i^2\norm{\delta_i (h)}\norm{\delta_i (h')} } \bigg| (\hat \delta_i(h) - \delta_i(h))' \sum_{j=-T+1}^{T-1}K_N^{(j)} \\
	& \times \frac{1}{T} \sum_{t=|j|+1}^T \Bigg \{ \left(  \left( x_{it}^{(+j)} {x_{it}^{(+j)}}' - \overline{x_i x_i} \right) 
    \left(\hat{\theta}_{g_i^0} - \theta_{g_i^0}\right) \right) 
    \\
    & \times \left( u_{it}^{(-j)}x_{it}^{(-j)} - \overline{u_i x_i}  \right)' (\hat \delta_i(h) - \delta_i(h)) \Bigg\} \bigg| \\
	\precsim_P &  r_{\theta, N}^3 \iota_N^{-2}(\min_{1 \leq i \leq N}\sigma_i)^{-1},
\end{align*}
\begin{align*}
	& \sup_{i,h,h'} \frac{1}{\sigma_i^2\norm{\delta_i (h)}\norm{\delta_i (h')} } \bigg| (\hat \delta_i(h) - \delta_i(h))' \sum_{j=-T+1}^{T-1}K_N^{(j)} \\
	& \quad \times  \frac{1}{T} \sum_{t=|j|+1}^T\left(  \left( x_{it}^{(+j)} {w_{it}^{(+j)}}' -  \overline{x_i w_i}\right) \left(\hat{\theta}^w - \theta^w\right) \right) \\
	&  \quad \times  \left( u_{it}^{(-j)}x_{it}^{(-j)} - \overline{u_i x_i}  \right)' (\hat \delta_i(h) - \delta_i(h))  \bigg| \\
	\precsim_P &  r_{\theta, N}^3 \iota_N^{-2}(\min_{1 \leq i \leq N}\sigma_i)^{-1},
\end{align*}
\begin{align*}
	& \sup_{i,h,h'} \frac{1}{\sigma_i^2\norm{\delta_i (h)}\norm{\delta_i (h')} } \bigg| (\hat \delta_i(h) - \delta_i(h))' \sum_{j=-T+1}^{T-1}K_N^{(j)} \\
	& \quad \times \frac{1}{T} \sum_{t=|j|+1}^T \left( \left( x_{it}^{(+j)} {x_{it}^{(+j)}}' - \overline{x_i x_i} \right) \left(\hat{\theta}_{g_i^0} - \theta_{g_i^0}\right) \right) \\
	&  \quad \times  \left(  \left( x_{it}^{(-j)} {x_{it}^{(-j)}}' - \overline{x_i x_i} \right) \left(\hat{\theta}_{g_i^0} - \theta_{g_i^0}\right) \right)' (\hat \delta_i(h) - \delta_i(h)) \bigg| \\
	\precsim_P &  r_{\theta, N}^4 \iota_N^{-2}(\min_{1 \leq i \leq N}\sigma_i)^{-2},
\end{align*}
\begin{align*}
	& \sup_{i,h,h'} \frac{1}{\sigma_i^2\norm{\delta_i (h)}\norm{\delta_i (h')} } \bigg| (\hat \delta_i(h) - \delta_i(h))' \sum_{j=-T+1}^{T-1}K_N^{(j)} \\
	& \quad \times \frac{1}{T} \sum_{t=|j|+1}^T \left( \left( x_{it}^{(+j)} {x_{it}^{(+j)}}' - \overline{x_i x_i} \right) \left(\hat{\theta}_{g_i^0} - \theta_{g_i^0}\right)\right) \\
	&  \quad \times  \left(  \left( x_{it}^{(-j)} {w_{it}^{(-j)}}' -  \overline{x_i w_i}\right) \left(\hat{\theta}^w - \theta^w\right) \right)' (\hat \delta_i(h) - \delta_i(h)) \bigg| \\
	\precsim_P &  r_{\theta, N}^4 \iota_N^{-2}(\min_{1 \leq i \leq N}\sigma_i)^{-2},
\end{align*}
\begin{align*}
	& \sup_{i,h,h'} \frac{1}{\sigma_i^2\norm{\delta_i (h)}\norm{\delta_i (h')} } \bigg| (\hat \delta_i(h) - \delta_i(h))' \sum_{j=-T+1}^{T-1}K_N^{(j)} \\
	& \quad \times  \frac{1}{T} \sum_{t=|j|+1}^T \left( \left( x_{it}^{(+j)} {w_{it}^{(+j)}}' -  \overline{x_i w_i}\right) \left(\hat{\theta}^w - \theta^w\right) \right) \\
	&  \quad \times  \left(  \left( x_{it}^{(-j)} {x_{it}^{(-j)}}' - \overline{x_i x_i} \right) \left(\hat{\theta}_{g_i^0} - \theta_{g_i^0}\right) \right)' (\hat \delta_i(h) - \delta_i(h))\bigg| \\
	\precsim_P &  r_{\theta, N}^4 \iota_N^{-2}(\min_{1 \leq i \leq N}\sigma_i)^{-2},
\end{align*}
\begin{align*}
	& \sup_{i,h,h'} \frac{1}{\sigma_i^2\norm{\delta_i (h)}\norm{\delta_i (h')} } \bigg| (\hat \delta_i(h) - \delta_i(h))' \sum_{j=-T+1}^{T-1}K_N^{(j)} \\
	& \quad \times  \frac{1}{T} \sum_{t=|j|+1}^T \left( \left( x_{it}^{(+j)} {w_{it}^{(+j)}}' -  \overline{x_i w_i}\right) \left(\hat{\theta}^w - \theta^w\right) \right) \\
	&  \quad \times  \left(   \left( x_{it}^{(-j)} {w_{it}^{(-j)}}' -  \overline{x_i w_i}\right) \left(\hat{\theta}^w - \theta^w\right) \right)' \bigg)  ( \hat \delta_i(h')- \delta_i(h')  ) \bigg| \\
	\precsim_P &  r_{\theta, N}^4 \iota_N^{-2}(\min_{1 \leq i \leq N}\sigma_i)^{-2},
\end{align*}
\begin{align*}
	& \sup_{i,h,h'} \frac{1}{\sigma_i^2\norm{\delta_i (h)}\norm{\delta_i (h')} } \bigg| (\hat \delta_i(h) - \delta_i(h))' \sum_{j=-T+1}^{T-1}K_N^{(j)} \\
	& \quad \times  \frac{1}{T} \sum_{t=|j|+1}^T\left(  u_{it}^{(+j)}x_{it}^{(+j)} - \overline{u_i x_i}   \right) \\
	&  \quad \times  \left(  \left( x_{it}^{(-j)} {x_{it}^{(-j)}}' - \overline{x_i x_i} \right) \left(\hat{\theta}_{g_i^0} - \theta_{g_i^0}\right) \right)' \delta_i(h) \bigg| \\
	\precsim_P &  r_{\theta, N}^2 \iota_N^{-1}(\min_{1 \leq i \leq N}\sigma_i)^{-1},
\end{align*}
\begin{align*}
	& \sup_{i,h,h'} \frac{1}{\sigma_i^2\norm{\delta_i (h)}\norm{\delta_i (h')} } \bigg| (\hat \delta_i(h) - \delta_i(h))' \sum_{j=-T+1}^{T-1}K_N^{(j)} \\
	& \quad \times \frac{1}{T} \sum_{t=|j|+1}^T\left(  u_{it}^{(+j)}x_{it}^{(+j)} - \overline{u_i x_i}   \right)  \\
	&  \quad \times \left(   \left( x_{it}^{(-j)} {w_{it}^{(-j)}}' -  \overline{x_i w_i}\right) \left(\hat{\theta}^w - \theta^w\right) \right)'  \delta_i(h) \bigg| \\
	\precsim_P &  r_{\theta, N}^2 \iota_N^{-1}(\min_{1 \leq i \leq N}\sigma_i)^{-1},
\end{align*}
\begin{align*}
	& \sup_{i,h,h'} \frac{1}{\sigma_i^2\norm{\delta_i (h)}\norm{\delta_i (h')} } \bigg| (\hat \delta_i(h) - \delta_i(h))' \sum_{j=-T+1}^{T-1}K_N^{(j)} \\
	& \quad \times \frac{1}{T} \sum_{t=|j|+1}^T\left(  \left( x_{it}^{(+j)} {x_{it}^{(+j)}}' - \overline{x_i x_i} \right) \left(\hat{\theta}_{g_i^0} - \theta_{g_i^0}\right) \right) \\
	&  \quad \times  \left( u_{it}^{(-j)}x_{it}^{(-j)} - \overline{u_i x_i}  \right)'  \delta_i(h) \bigg| \\
	\precsim_P &  r_{\theta, N}^2 \iota_N^{-1}(\min_{1 \leq i \leq N}\sigma_i)^{-1},
\end{align*}
\begin{align*}
	& \sup_{i,h,h'} \frac{1}{\sigma_i^2\norm{\delta_i (h)}\norm{\delta_i (h')} } \bigg| (\hat \delta_i(h) - \delta_i(h))' \sum_{j=-T+1}^{T-1}K_N^{(j)} \\
	& \quad \times \frac{1}{T} \sum_{t=|j|+1}^T\left(  \left( x_{it}^{(+j)} {w_{it}^{(+j)}}' -  \overline{x_i w_i}\right) \left(\hat{\theta}^w - \theta^w\right) \right) \\
	&  \quad \times  \left( u_{it}^{(-j)}x_{it}^{(-j)} - \overline{u_i x_i}  \right)'  \delta_i(h)\bigg| \\
	\precsim_P &  r_{\theta, N}^2 \iota_N^{-1}(\min_{1 \leq i \leq N}\sigma_i)^{-1},
\end{align*}
\begin{align*}
	& \sup_{i,h,h'} \frac{1}{\sigma_i^2\norm{\delta_i (h)}\norm{\delta_i (h')} } \bigg| (\hat \delta_i(h) - \delta_i(h))' \sum_{j=-T+1}^{T-1}K_N^{(j)} \\
	& \quad \times  \frac{1}{T} \sum_{t=|j|+1}^T \left( \left( x_{it}^{(+j)} {x_{it}^{(+j)}}' - \overline{x_i x_i} \right) \left(\hat{\theta}_{g_i^0} - \theta_{g_i^0}\right) \right) \\
	&  \quad \times  \left(  \left( x_{it}^{(-j)} {x_{it}^{(-j)}}' - \overline{x_i x_i} \right) \left(\hat{\theta}_{g_i^0} - \theta_{g_i^0}\right) \right)'  \delta_i(h)  \bigg| \\
	\precsim_P &  r_{\theta, N}^3 \iota_N^{-1}(\min_{1 \leq i \leq N}\sigma_i)^{-2},
\end{align*}
\begin{align*}
	& \sup_{i,h,h'} \frac{1}{\sigma_i^2\norm{\delta_i (h)}\norm{\delta_i (h')} } \bigg| (\hat \delta_i(h) - \delta_i(h))' \sum_{j=-T+1}^{T-1}K_N^{(j)} \\
	& \quad \times \frac{1}{T} \sum_{t=|j|+1}^T \left( \left( x_{it}^{(+j)} {x_{it}^{(+j)}}' - \overline{x_i x_i} \right) \left(\hat{\theta}_{g_i^0} - \theta_{g_i^0}\right)\right) \\
	&  \quad \times  \left(  \left( x_{it}^{(-j)} {w_{it}^{(-j)}}' -  \overline{x_i w_i}\right) \left(\hat{\theta}^w - \theta^w\right) \right)'  \delta_i(h) \bigg| \\
	\precsim_P &  r_{\theta, N}^3 \iota_N^{-1}(\min_{1 \leq i \leq N}\sigma_i)^{-2},
\end{align*}
\begin{align*}
	& \sup_{i,h,h'} \frac{1}{\sigma_i^2\norm{\delta_i (h)}\norm{\delta_i (h')} } \bigg| (\hat \delta_i(h) - \delta_i(h))' \sum_{j=-T+1}^{T-1}K_N^{(j)} \\
	& \quad \times  \frac{1}{T} \sum_{t=|j|+1}^T \left( \left( x_{it}^{(+j)} {w_{it}^{(+j)}}' -  \overline{x_i w_i}\right) \left(\hat{\theta}^w - \theta^w\right) \right) \\
	&  \quad \times  \left(  \left( x_{it}^{(-j)} {x_{it}^{(-j)}}' - \overline{x_i x_i} \right) \left(\hat{\theta}_{g_i^0} - \theta_{g_i^0}\right) \right)'  \delta_i(h) \bigg| \\
	\precsim_P & r_{\theta, N}^3 \iota_N^{-1}(\min_{1 \leq i \leq N}\sigma_i)^{-2}),
\end{align*}
\begin{align*}
	& \sup_{i,h,h'} \frac{1}{\sigma_i^2\norm{\delta_i (h)}\norm{\delta_i (h')} } \bigg| (\hat \delta_i(h) - \delta_i(h))' \sum_{j=-T+1}^{T-1}K_N^{(j)} \\
	& \quad \times \frac{1}{T} \sum_{t=|j|+1}^T\left( \left( x_{it}^{(+j)} {w_{it}^{(+j)}}' -  \overline{x_i w_i}\right) \left(\hat{\theta}^w - \theta^w\right) \right) \\
	&  \quad \times  \left(   \left( x_{it}^{(-j)} {w_{it}^{(-j)}}' -  \overline{x_i w_i}\right) \left(\hat{\theta}^w - \theta^w\right) \right)' \bigg)   \delta_i(h') \bigg| \\
	\precsim_P &  r_{\theta, N}^3 \iota_N^{-1}(\min_{1 \leq i \leq N}\sigma_i)^{-2},
\end{align*}
\begin{align*}
	& \sup_{i,h,h'} \frac{1}{\sigma_i^2\norm{\delta_i (h)}\norm{\delta_i (h')} } \bigg| \delta_i(h)' \sum_{j=-T+1}^{T-1}K_N^{(j)} \\
	& \quad \times \frac{1}{T} \sum_{t=|j|+1}^T \left(  u_{it}^{(+j)}x_{it}^{(+j)} - \frac{1}{T}\sum_{u=1}^T u_{iu}x_{iu}   \right) \\
	&  \quad \times  \left(  \left( x_{it}^{(-j)} {x_{it}^{(-j)}}' - \overline{x_i x_i} \right) \left(\hat{\theta}_{g_i^0} - \theta_{g_i^0}\right) \right)'(\hat \delta_i(h) - \delta_i(h)) \bigg| \\
	\precsim_P & r_{\theta, N}^2 \iota_N^{-1}(\min_{1 \leq i \leq N}\sigma_i)^{-1}),
\end{align*}
\begin{align*}
	& \sup_{i,h,h'} \frac{1}{\sigma_i^2\norm{\delta_i (h)}\norm{\delta_i (h')} } \bigg| \delta_i(h)' \sum_{j=-T+1}^{T-1}K_N^{(j)} \\
	& \quad \times \frac{1}{T} \sum_{t=|j|+1}^T \left(  u_{it}^{(+j)}x_{it}^{(+j)} - \frac{1}{T}\sum_{u=1}^T u_{iu}x_{iu}   \right) \\
	&  \quad \times  \left(   \left( x_{it}^{(-j)} {w_{it}^{(-j)}}' -  \overline{x_i w_i}\right) \left(\hat{\theta}^w - \theta^w\right) \right)' (\hat \delta_i(h) - \delta_i(h)) \bigg| \\
	\precsim_P &  r_{\theta, N}^2 \iota_N^{-1}(\min_{1 \leq i \leq N}\sigma_i)^{-1},
\end{align*}
\begin{align*}
	& \sup_{i,h,h'} \frac{1}{\sigma_i^2\norm{\delta_i (h)}\norm{\delta_i (h')} } \bigg| \delta_i(h)' \sum_{j=-T+1}^{T-1}K_N^{(j)} \\
	& \quad \times  \frac{1}{T} \sum_{t=|j|+1}^T \left(  \left( x_{it}^{(+j)} {x_{it}^{(+j)}}' - \overline{x_i x_i} \right) \left(\hat{\theta}_{g_i^0} - \theta_{g_i^0}\right) \right) \\
	&  \quad \times  \left( u_{it}^{(-j)}x_{it}^{(-j)} - \frac{1}{T}\sum_{u=1}^T u_{iu}x_{iu}  \right)' (\hat \delta_i(h) - \delta_i(h)) \bigg| \\
	\precsim_P &  r_{\theta, N}^2 \iota_N^{-1}(\min_{1 \leq i \leq N}\sigma_i)^{-1},
\end{align*}
\begin{align*}
	& \sup_{i,h,h'} \frac{1}{\sigma_i^2\norm{\delta_i (h)}\norm{\delta_i (h')} } \bigg| \delta_i(h)' \sum_{j=-T+1}^{T-1}K_N^{(j)} \\
	& \quad \times  \frac{1}{T} \sum_{t=|j|+1}^T \left(  \left( x_{it}^{(+j)} {w_{it}^{(+j)}}' -  \overline{x_i w_i}\right) \left(\hat{\theta}^w - \theta^w\right) \right) \\
	&  \quad \times  \left( u_{it}^{(-j)}x_{it}^{(-j)} - \frac{1}{T}\sum_{u=1}^T u_{iu}x_{iu}  \right)' (\hat \delta_i(h) - \delta_i(h)) \bigg| \\
	\precsim_P &  r_{\theta, N}^2 \iota_N^{-1}(\min_{1 \leq i \leq N}\sigma_i)^{-1},
\end{align*}
\begin{align*}
	& \sup_{i,h,h'} \frac{1}{\sigma_i^2\norm{\delta_i (h)}\norm{\delta_i (h')} } \bigg| \delta_i(h)' \sum_{j=-T+1}^{T-1}K_N^{(j)} \\
	& \quad \times \frac{1}{T} \sum_{t=|j|+1}^T \left( \left( x_{it}^{(+j)} {x_{it}^{(+j)}}' - \overline{x_i x_i} \right) \left(\hat{\theta}_{g_i^0} - \theta_{g_i^0}\right) \right) \\
	&  \quad \times  \left(  \left( x_{it}^{(-j)} {x_{it}^{(-j)}}' - \overline{x_i x_i} \right) \left(\hat{\theta}_{g_i^0} - \theta_{g_i^0}\right) \right)' (\hat \delta_i(h) - \delta_i(h))  \bigg| \\
	\precsim_P &  r_{\theta, N}^3 \iota_N^{-1}(\min_{1 \leq i \leq N}\sigma_i)^{-2},
\end{align*}
\begin{align*}
	& \sup_{i,h,h'} \frac{1}{\sigma_i^2\norm{\delta_i (h)}\norm{\delta_i (h')} } \bigg| \delta_i(h)' \sum_{j=-T+1}^{T-1}K_N^{(j)} \\
	& \quad \times \frac{1}{T} \sum_{t=|j|+1}^T \left( \left( x_{it}^{(+j)} {x_{it}^{(+j)}}' - \overline{x_i x_i} \right) \left(\hat{\theta}_{g_i^0} - \theta_{g_i^0}\right)\right) \\
	&  \quad \times  \left(  \left( x_{it}^{(-j)} {w_{it}^{(-j)}}' -  \overline{x_i w_i}\right) \left(\hat{\theta}^w - \theta^w\right) \right)' (\hat \delta_i(h) - \delta_i(h)) \bigg| \\
	\precsim_P &  r_{\theta, N}^3 \iota_N^{-1}(\min_{1 \leq i \leq N}\sigma_i)^{-2},
\end{align*}
\begin{align*}
	& \sup_{i,h,h'} \frac{1}{\sigma_i^2\norm{\delta_i (h)}\norm{\delta_i (h')} } \bigg| \delta_i(h)' \sum_{j=-T+1}^{T-1}K_N^{(j)} \\
	& \quad \times \frac{1}{T} \sum_{t=|j|+1}^T  \left( \left( x_{it}^{(+j)} {w_{it}^{(+j)}}' -  \overline{x_i w_i}\right) \left(\hat{\theta}^w - \theta^w\right) \right) \\
	&  \quad \times  \left(  \left( x_{it}^{(-j)} {x_{it}^{(-j)}}' - \overline{x_i x_i} \right) \left(\hat{\theta}_{g_i^0} - \theta_{g_i^0}\right) \right)' (\hat \delta_i(h) - \delta_i(h)) \bigg| \\
	\precsim_P &  r_{\theta, N}^3 \iota_N^{-1}(\min_{1 \leq i \leq N}\sigma_i)^{-2},
\end{align*}
\begin{align*}
	& \sup_{i,h,h'} \frac{1}{\sigma_i^2\norm{\delta_i (h)}\norm{\delta_i (h')} } \bigg| \delta_i(h)' \sum_{j=-T+1}^{T-1}K_N^{(j)} \\
	& \quad \times\frac{1}{T} \sum_{t=|j|+1}^T \left( \left( x_{it}^{(+j)} {w_{it}^{(+j)}}' -  \overline{x_i w_i}\right) \left(\hat{\theta}^w - \theta^w\right) \right) \\
	&  \quad \times  \left(   \left( x_{it}^{(-j)} {w_{it}^{(-j)}}' -  \overline{x_i w_i}\right) \left(\hat{\theta}^w - \theta^w\right) \right)' \bigg)  ( \hat \delta_i(h')- \delta_i(h')  ) \bigg| \\
	\precsim_P &  r_{\theta, N}^3 \iota_N^{-1}(\min_{1 \leq i \leq N}\sigma_i)^{-2},
\end{align*}
\begin{align*}
	& \sup_{i,h,h'} \frac{1}{\sigma_i^2\norm{\delta_i (h)}\norm{\delta_i (h')} } \bigg| \delta_i(h)' \sum_{j=-T+1}^{T-1}K_N^{(j)} \\
	& \quad \times\frac{1}{T} \sum_{t=|j|+1}^T \bigg( - \left(  u_{it}^{(+j)}x_{it}^{(+j)} - \frac{1}{T}\sum_{u=1}^T u_{iu}x_{iu}   \right) \\
	&  \quad \times  \left(  \left( x_{it}^{(-j)} {x_{it}^{(-j)}}' - \overline{x_i x_i} \right) \left(\hat{\theta}_{g_i^0} - \theta_{g_i^0}\right) \right)'  \delta_i(h) \bigg| \\
	\precsim_P & r_{\theta, N} (\min_{1 \leq i \leq N}\sigma_i)^{-1} ,
\end{align*}
\begin{align*}
	& \sup_{i,h,h'} \frac{1}{\sigma_i^2\norm{\delta_i (h)}\norm{\delta_i (h')} } \bigg| \delta_i(h)' \sum_{j=-T+1}^{T-1}K_N^{(j)} \\
	& \quad \times\frac{1}{T} \sum_{t=|j|+1}^T \left(  u_{it}^{(+j)}x_{it}^{(+j)} - \frac{1}{T}\sum_{u=1}^T u_{iu}x_{iu}   \right) \\
	&  \quad \times  \left(   \left( x_{it}^{(-j)} {w_{it}^{(-j)}}' -  \overline{x_i w_i}\right) \left(\hat{\theta}^w - \theta^w\right) \right)'  \delta_i(h)  \bigg| \\
	\precsim_P & r_{\theta, N} (\min_{1 \leq i \leq N}\sigma_i)^{-1},
\end{align*}
\begin{align*}
	& \sup_{i,h,h'} \frac{1}{\sigma_i^2\norm{\delta_i (h)}\norm{\delta_i (h')} } \bigg| \delta_i(h)' \sum_{j=-T+1}^{T-1}K_N^{(j)} \\
	& \quad \times\frac{1}{T} \sum_{t=|j|+1}^T \left(  \left( x_{it}^{(+j)} {x_{it}^{(+j)}}' - \overline{x_i x_i} \right) \left(\hat{\theta}_{g_i^0} - \theta_{g_i^0}\right) \right) \\
	&  \quad \times  \left( u_{it}^{(-j)}x_{it}^{(-j)} - \frac{1}{T}\sum_{u=1}^T u_{iu}x_{iu}  \right)'  \delta_i(h)  \bigg| \\
	\precsim_P & r_{\theta, N}(\min_{1 \leq i \leq N}\sigma_i)^{-1}),
\end{align*}
\begin{align*}
	& \sup_{i,h,h'} \frac{1}{\sigma_i^2\norm{\delta_i (h)}\norm{\delta_i (h')} } \bigg| \delta_i(h)' \sum_{j=-T+1}^{T-1}K_N^{(j)} \\
	& \quad \times\frac{1}{T} \sum_{t=|j|+1}^T \left(  \left( x_{it}^{(+j)} {w_{it}^{(+j)}}' -  \overline{x_i w_i}\right) \left(\hat{\theta}^w - \theta^w\right) \right) \\
	&  \quad \times  \left( u_{it}^{(-j)}x_{it}^{(-j)} - \frac{1}{T}\sum_{u=1}^T u_{iu}x_{iu}  \right)' \delta_i(h) \bigg| \\
	\precsim_P & r_{\theta, N} (\min_{1 \leq i \leq N}\sigma_i)^{-1},
\end{align*}
\begin{align*}
	& \sup_{i,h,h'} \frac{1}{\sigma_i^2\norm{\delta_i (h)}\norm{\delta_i (h')} } \bigg| \delta_i(h)' \sum_{j=-T+1}^{T-1}K_N^{(j)} \\
	& \quad \times \frac{1}{T} \sum_{t=|j|+1}^T \left( \left( x_{it}^{(+j)} {x_{it}^{(+j)}}' - \overline{x_i x_i} \right) \left(\hat{\theta}_{g_i^0} - \theta_{g_i^0}\right) \right) \\
	&  \quad \times  \left(  \left( x_{it}^{(-j)} {x_{it}^{(-j)}}' - \overline{x_i x_i} \right) \left(\hat{\theta}_{g_i^0} - \theta_{g_i^0}\right) \right)'  \delta_i(h) \bigg| \\
	\precsim_P & r_{\theta, N}^2(\min_{1 \leq i \leq N}\sigma_i)^{-2},
\end{align*}
\begin{align*}
	& \sup_{i,h,h'} \frac{1}{\sigma_i^2\norm{\delta_i (h)}\norm{\delta_i (h')} } \bigg| \delta_i(h)' \sum_{j=-T+1}^{T-1}K_N^{(j)} \\
	& \quad \times\frac{1}{T} \sum_{t=|j|+1}^T \left( \left( x_{it}^{(+j)} {x_{it}^{(+j)}}' - \overline{x_i x_i} \right) \left(\hat{\theta}_{g_i^0} - \theta_{g_i^0}\right)\right) \\
	&  \quad \times  \left(  \left( x_{it}^{(-j)} {w_{it}^{(-j)}}' -  \overline{x_i w_i}\right) \left(\hat{\theta}^w - \theta^w\right) \right)'  \delta_i(h) \bigg| \\
	\precsim_P & r_{\theta, N}^2 (\min_{1 \leq i \leq N}\sigma_i)^{-2}),
\end{align*}
\begin{align*}
	& \sup_{i,h,h'} \frac{1}{\sigma_i^2\norm{\delta_i (h)}\norm{\delta_i (h')} } \bigg| \delta_i(h)' \sum_{j=-T+1}^{T-1}K_N^{(j)} \\
	& \quad \times \frac{1}{T} \sum_{t=|j|+1}^T \left( \left( x_{it}^{(+j)} {w_{it}^{(+j)}}' -  \overline{x_i w_i}\right) \left(\hat{\theta}^w - \theta^w\right) \right) \\
	&  \quad \times  \left(  \left( x_{it}^{(-j)} {x_{it}^{(-j)}}' - \overline{x_i x_i} \right) \left(\hat{\theta}_{g_i^0} - \theta_{g_i^0}\right) \right)'  \delta_i(h) \bigg| \\
	\precsim_P & r_{\theta, N}^2 (\min_{1 \leq i \leq N}\sigma_i)^{-2}),
\end{align*}
\begin{align*}
	& \sup_{i,h,h'} \frac{1}{\sigma_i^2\norm{\delta_i (h)}\norm{\delta_i (h')} } \bigg| \delta_i(h)' \sum_{j=-T+1}^{T-1}K_N^{(j)} \\
	& \quad \times\frac{1}{T} \sum_{t=|j|+1}^T \left( \left( x_{it}^{(+j)} {w_{it}^{(+j)}}' -  \overline{x_i w_i}\right) \left(\hat{\theta}^w - \theta^w\right) \right) \\
	&  \quad \times  \left(   \left( x_{it}^{(-j)} {w_{it}^{(-j)}}' -  \overline{x_i w_i}\right) \left(\hat{\theta}^w - \theta^w\right) \right)' \bigg)  \delta_i(h')  \bigg| \\
	\precsim_P & r_{\theta, N}^2(\min_{1 \leq i \leq N}\sigma_i)^{-2} ,
\end{align*}
\begin{align*}
	& \sup_{i,h,h'} \frac{1}{\sigma_i^2\norm{\delta_i (h)}\norm{\delta_i (h')} } \bigg|(\hat \delta_i(h) - \delta_i(h))' \sum_{j=-T+1}^{T-1}K_N^{(j)} \\
	& \quad \times\frac{1}{T} \sum_{t=|j|+1}^T\left(  u_{it}^{(+j)}x_{it}^{(+j)} - \overline{u_i x_i}   \right) \\
	&  \quad \times  \left( u_{it}^{(-j)}x_{it}^{(-j)} - \overline{u_i x_i}  \right)' ( \hat \delta_i(h')- \delta_i(h')  ) \bigg| \\
	\precsim_P & r_{\theta, N}^2 \iota_N^{-2} (\min_{1 \leq i \leq N}\sigma_i)^{-1},
\end{align*}
\begin{align*}
	& \sup_{i,h,h'} \frac{1}{\sigma_i^2\norm{\delta_i (h)}\norm{\delta_i (h')} } \bigg|  (\hat \delta_i(h) - \delta_i(h))' \sum_{j=-T+1}^{T-1}K_N^{(j)} \\
	& \quad \times \frac{1}{T} \sum_{t=|j|+1}^T \left(  u_{it}^{(+j)}x_{it}^{(+j)} - \overline{u_i x_i}   \right) \\
	&  \quad \times  \left( u_{it}^{(-j)}x_{it}^{(-j)} - \overline{u_i x_i}  \right)'  \delta_i(h') \bigg| \\
	\precsim_P & r_{\theta, N} \iota_N^{-1}(\min_{1 \leq i \leq N}\sigma_i)^{-1},
\end{align*}
and 
\begin{align*}
	& \sup_{i,h,h'} \frac{1}{\sigma_i^2\norm{\delta_i (h)}\norm{\delta_i (h')} } \bigg| \delta_i(h)' \sum_{j=-T+1}^{T-1}K_N^{(j)} \\
	& \quad \times\frac{1}{T} \sum_{t=|j|+1}^T \left(  u_{it}^{(+j)}x_{it}^{(+j)} - \overline{u_i x_i}   \right) \\
	&  \quad \times  \left( u_{it}^{(-j)}x_{it}^{(-j)} - \overline{u_i x_i}  \right)' ( \hat \delta_i(h')- \delta_i(h') )  \bigg| \\
	\precsim_P & r_{\theta, N} \iota_N^{-1}(\min_{1 \leq i \leq N}\sigma_i)^{-1}).
\end{align*}

To sum up, we have 
\begin{align*}
	&\sup_{i,h,h'}  \frac{1}{\sigma_i^2\norm{\delta_i (h)}\norm{\delta_i (h')} } \left| \widehat{\Xi}_i (h,h') - \widetilde{\Xi}_i (h,h') \right| \\
	\precsim_P &  r_{\theta, N}^3 \iota_N^{-2}(\min_{1 \leq i \leq N}\sigma_i)^{-1} + r_{\theta, N}^4 \iota_N^{-2}(\min_{1 \leq i \leq N}\sigma_i)^{-2} + r_{\theta, N}^2 \iota_N^{-1}(\min_{1 \leq i \leq N}\sigma_i)^{-1} \\
	& + r_{\theta, N}^3 \iota_N^{-1}(\min_{1 \leq i \leq N}\sigma_i)^{-2}  + r_{\theta, N}(\min_{1 \leq i \leq N}\sigma_i)^{-1}) + r_{\theta, N}^2(\min_{1 \leq i \leq N}\sigma_i)^{-2} \\ 
	& + r_{\theta, N}^2 \iota_N^{-2}(\min_{1 \leq i \leq N}\sigma_i)^{-1} +  r_{\theta, N} \iota_N^{-1}(\min_{1 \leq i \leq N}\sigma_i)^{-1} \\
	\precsim_P & r_{\theta, N} \iota_N^{-1}(\min_{1 \leq i \leq N}\sigma_i)^{-1},
\end{align*}
where the last $\precsim_P $  follows by $b_N^{LV} \to 0$. 
We conclude that we have 
\begin{align*}
	& \sup_{i,h,h'}  \frac{1}{\sigma_i^2\norm{\delta_i (h)}\norm{\delta_i (h')} } \left| \widehat{\Xi}_i (h,h') - \Xi_i (h,h') \right| \\
	\precsim_P  & r_{\theta, N} \iota_N^{-1}(\min_{1 \leq i \leq N}\sigma_i)^{-1} +  T^{-c_1} (\log N)^{c_2} + T^{-\rho}.
\end{align*}
\end{proof}

\begin{lemma}
	\label{lem:lv-xi-consistency}
Let $\xi_{it}$ be the random vector consisting of distinct elements of $v_{it} x_{it}$, $x_{it}x_{it}' - \E_P (x_{it}x_{it}')$ and $w_{it}x_{it}' - \E_P (w_{it}x_{it}')$.
Let $\mathbb{P}_N$ be the set of probability measures which satisfy Assumption~\ref*{assumption:basic}.\ref*{assum:tail} and Assumption~\ref*{assumption:basic}.\ref*{assum:mixing} with identical choices of $a$, $b$, $d_1$, and $d_2$. Let $\zeta_N $ be a sequence satisfying $\zeta_N \to \infty $ as $N \to \infty$. 
Assume $T^{\delta_1} \leq N$ for some universal constant $\delta_1 >0$. Let Assumption~\ref*{assum:kernel} hold and $\kappa_N \asymp T^{\rho}$ where $0 <\rho < (\vartheta -1)/(3\vartheta -2)$. Then, there exist two constants $c_1>0$ depending only on $(\rho, \vartheta)$ and $c_2>0$ depending only on $(d_2, \vartheta)$, such that 
\begin{align*}
	\sup_{P\in \mathbb{P}_N} P \bigg( \sup_{1 \leq i \leq N} \bigg| &  \sum_{j=-T+1}^{T-1} K\left( \frac{j}{\kappa_N} \right)  \\
	& \times \frac{1}{T} \sum_{t=|j|+1}^T \left(\xi_{i,t+ \min (0, j)} - \frac{1}{T}\sum_{u=1}^T \xi_{iu} \right)  \left( \xi_{i,t- \max (0, j)} - \frac{1}{T}\sum_{u=1}^T \xi_{iu} \right)' \\ 
	& \quad  -  \left( \frac{1}{T}\sum_{t=1}^{T} \sum_{s=1}^T \E_P [\xi_{it} \xi_{is} ' ] \right) \bigg|_{\infty} \\
	>  & \zeta_N \left( T^{-c_1} (\log N)^{c_2} + T^{-\rho}\right)\bigg) = o(1).
\end{align*}
\end{lemma}

\begin{proof}
	We apply Theorem 11(i) of \textcite{chang2022central}. In their theorem, $B_n$ is the bound for the Orlicz norm. In our case, this bound depends only on $K$, $a$ and $d_1$ under Assumption~\ref*{assumption:basic}.\ref*{assum:tail} by \textcite[][Lemma 8.1]{kosorok2008introduction} and Lemma~\ref*{lem:exp-tail}. Their $\gamma_1$ is 1 in our case by Assumption~\ref*{assumption:basic}.\ref*{assum:tail}. By Lemma~\ref*{lem:mixing-g}, $\gamma_2$ in \textcite{chang2022central} depends only on $d_2$.  The conditions for the kernel and the bandwidth are assumed. Thus, Theorem~11(i) in \textcite{chang2022central} can be applied. Note that their results are stated in terms of stochastic order, but an inspection of their proof reveals that the constant terms hidden in the stochastic order depend only on the constants in the assumptions. 
\end{proof}

\begin{lemma}[Tail bounds for functions]
	\label{lem:exp-tail}
	Suppose that two random variances $X_1$ and $X_2$ satisfy $P(|X_a| >x ) \leq C_a \exp( - b_a x^{d_a}) $ for $a=0,1$, then $P(|X_1 X_2| >x ) \leq C \exp( - b x^{d}) $ for some positive constants $C$, $b$, and $d_2$, and $P(|X_1 + X_2| >x ) \leq C' \exp( - b' x^{d'}) $ for some constants $C'$, $b'$, and $d'$.	
\end{lemma}
\begin{proof}
    The first statement follows because
    \begin{align*}
        P (|X_1 X_2| >x) 
        \leq & P(|X_1 | > \sqrt{x} ) + P ( |X_2|> \sqrt{x}) \\
        \leq & C_1 \exp (-b_1 x^{d_{1}/2} ) + C_2 \exp ( - b_2 x^{d_{2}/2} ) \\
        \leq & 2 \max (C_1, C_2 ) \exp ( - \min ( b_1, b_2) z^{\min (d_{1}, d_{2}) /2} ).
    \end{align*}
    For the second statement, we have 	
    \begin{align*}
        P (|X_1 +  X_2| >x) 
        \leq & P(|X_1 | > x / 2 ) + P ( |X_2|> x / 2) \\
        \leq & C_1 \exp (-b_1/2^{d_{1}} x^{d_{1}} ) + C_2 \exp ( - b_2/ 2^{d_{2}} x^{d_{2}} ) \\
        \leq & 2 \max (C_1, C_2 ) \exp \left( - \min ( b_1/2^{d_{1}}, b_2/ 2^{d_{2}}) x^{\min (d_{1}, d_{2}) } \right).
    \end{align*}
\end{proof}
	
\begin{lemma}[Tail bounds for norms]
    \label{lem:tail bounds norms}
    Let $X_1$ and $X_2$ denote two random vectors such that there are constants $K$, $b$ and $d$ such that for any component $Y$ of $X_1$ and $X_2$, $P (\lvert Y \rvert > x) \leq C \exp \left(-b x^d\right)$. Then, there are constants $C'$, $b'$ and $d'$ such that 
    \begin{align*}
        P \left( \lVert X_1 \rVert > x \right) \leq & C' \exp \left(- b' x^{d'} \right),
    \\
        P \left( \lVert X_1 \rVert^2 > x \right) \leq & C' \exp \left(- b' x^{d'} \right),
    \\
        P \left( \lVert X_1 \rVert \lVert X_2 \rVert > x \right) \leq & C' \exp \left(- b' x^{d'} \right).
    \end{align*}
\end{lemma}

\begin{proof}
The second statement follows from the first statement. The third statement follows from the second statement of this lemma and the first statement of Lemma~\ref{lem:exp-tail}. It remains to prove the first statement. Let $X_1 = (Y_1, \dotsc, Y_p)$ and note that 
$P (\lvert Y_j^2 \rvert > x) \leq C \exp \left(-b x^{d/2}\right)$ for $j = 1, \dotsc, p$. Now, the first statement follows from writing  
\begin{align*}
    \lVert X_1 \rVert^2 = Y_1^2 + \dotsm Y_p^2 
\end{align*}
and applying the second statement of Lemma~\ref{lem:exp-tail} repeatedly.
\end{proof}

\begin{lemma}[Functions of mixing sequences]
	\label{lem:mixing-g}
	Suppose that $(x_{it}, w_{it}, v_{it})$ is a strong mixing sequence over $t$ with mixing coefficients $\sup_i a_i [t] \leq C \exp (-at^{d}) $ for constants $C$, $a$ and $d$, then so is $ g( x_{it}, w_{it}, v_{it})$ where $g$ is a measurable function. 
\end{lemma}
	
\begin{proof}
    The proof follows the argument in the proof of Theorem 14.1 in \textcite{davidson1994stochastic}.
\end{proof}

\begin{lemma}[Large quantiles of the normal distribution]
	\label{lem:bounds:large:quantiles}
	Let $X$ denote a standard normal vector with $p \times p$ correlation matrix $\Omega$ and let $0 \leq d < 2$. Let $c_{\alpha, N}$ denote the $1 - \alpha/N$ quantile of $X$. 
	Then there is a constant $N_0$ that depends only on $\underline{\alpha}$ and $d$ such that for $\underline{\alpha} \leq \alpha \leq 1$ and $N \geq N_0$
	\begin{align*}
		\sqrt{d \log (N/\alpha)} \leq c_{\alpha, N} \leq \sqrt{2 \log p} + \sqrt{2 \log (N/\alpha)}.
	\end{align*}
\end{lemma}
\begin{proof}
    The upper bound is given in Lemma D.4 in \textcite{chernozhukov2013testing}. To prove the lower bound put $a_N = \sqrt{d \log (N/\alpha)}$. Let $\Phi$ denote the cumulative distribution function of a standard normal random variable, and let $\phi$ denote its probability density function. Gordon's lower bound (see, e.g., \textcite{duembgen2010bounding}) states that 
    \begin{align*}
        1 - \Phi(x) > \frac{\phi(x)}{x(1 - 1/x^2)}
    \end{align*}
    for $x > 0$ and thus $1 - \Phi(x) > \frac{1}{2} \phi(x)/x$ for $x > \sqrt{2}$. Therefore, 
    \begin{align*}
        P\left(\max_{j = 1, \dotsc, p} X_j > a_N\right) 
        \geq & P\left(X_1 > a_N \right)
    \\
        = & 1 - \Phi(a_N) 
    \\
        > &  \frac{\phi(a_N)}{2 a_N} 
        = \frac{\exp\left(- \frac{a_N^2}{2} \right)}{a_N \sqrt{8 \pi}}
        = \frac{\left(\alpha/N\right)^{d/2}}{a_N \sqrt{8 \pi}}
        =  \alpha / N 
        \left( 
            \frac{\left(N/\alpha\right)^{1 - d/2}}{a_N \sqrt{8 \pi}}
        \right)
    \\
        \geq & \alpha / N 
        \left( 
            \frac{N^{1 - d/2}}{ \sqrt{8 d \pi \log (N/\underline{\alpha})}}
        \right) \geq \alpha / N,
    \end{align*}
    where the last inequality holds for $N \geq N_0$ and $N_0$ is chosen such that $N \geq N_0$ implies that $N^{1-d/2}/\sqrt{8 d \pi \log (N/\underline{\alpha})} \geq 1$. Such an $N_0$ can be found since  $N^{1-d/2}/\sqrt{8 d \pi \log (N/\underline{\alpha})} \to \infty$.
    The inequality $P(\max_{j=1, \dotsc, p} X_j > a_N) > \alpha/N$ implies $c_{\alpha, N} \geq a_N$.
\end{proof}
\begin{lemma}[Long-run variance is finite]
	\label{lem:lv-exists}
    Let $\xi_{it}$ denote any element of the vectors
    $v_{it} x_{it}$, 
    $\vecop (x_{it}x_{it}') - \E_P \vecop (x_{it}x_{it}')$, 
    $\vecop (w_{it}x_{it}') - \E_P \vecop (w_{it}x_{it}')$, 
    and $\vecop (v_{it}^2 x_{it}x_{it}') - \E_P \vecop (v_{it}^2 x_{it}x_{it}')$,
    or any of the random variables $\lVert x_{it} \rVert^2$, $\lVert x_{it} \rVert \lVert w_{it} \rVert$, $\lvert v_{it} \rvert \lVert x_{it} \rVert$,  $\lVert x_{it} \rVert^4$, $\lVert x_{it}^2 \rVert \lVert w_{it}^2 \rVert$ and $\lvert v_{it}^2 \rvert \lVert x_{it}^2 \rVert$. 
    Let $\mathbb{P}_N$ be a set of probability measures which satisfy Assumption~\ref*{assumption:basic}.\ref*{assum:tail}-\ref*{assumption:basic}.\ref*{assum:stationarity} with identical choices of $a$, $b$, $d_1$, and $d_2$. 
    Let \begin{align*}
        s_{i, T}^2 (P)   = \max_{1 \leq t \leq T} \left(\E_P (\xi_{it}^2) + 2 \sum_{s>t} \left| \E( \xi_{it} \xi_{is})\right| \right).
    \end{align*} 
    Then, there exists a constant $C_{\xi} < \infty$ such that 
    \begin{align*}
        \limsup_{N,T \to \infty} \sup_{P \in \mathbb{P}_N} \max_{1 \leq i \leq N} s_{i, T}^2 (P) < C_\xi.       
    \end{align*}
    In particular,  
    \begin{align*}
    \limsup_{N,T \to \infty} \sup_{P\in \mathbb{P}_N} \max_{1 \leq i \leq N}
    \Bigg\lvert \var_P \left( \frac{1}{\sqrt{T}} \sum_{t=1}^T \xi_{it} \right) \Bigg\rvert
    < C_\xi.
    \end{align*}
\end{lemma}

\begin{proof}
	Lemma \ref{lem:exp-tail} and Lemma~\ref{lem:tail bounds norms} imply $P (|\xi_{it}| > z ) < \exp (- (z/a_{\xi})^{d_{1, \xi}})$ for some $a_{\xi}$ and $d_{1, \xi}$, which in turns implies that $\E_P (\xi_{it}^m) < M_{\xi_p}$ for some universal constant $M_{\xi} < \infty $ for any integer $m$ by Lemma \ref{lem:moment}. Moreover, Lemma \ref{lem:mixing-g} implies that $\xi_{it}$ is an $\alpha$-mixing sequence with mixing coefficient $\sup_i \alpha_{i, \xi} [k] \leq \exp (1- b_{\xi} k^{d_{2, \xi}})$. Thus, by the argument in  \textcite[][Section C.1]{GalvaoKato14}, which is an application of \textcite{davidov68}, it holds that, for any $s > t$ and any integer $m$, 
    \begin{align*}
        \lvert \E_P (\xi_{it} \xi_{is}) \rvert \leq  12 (\E_P ( |\xi_{it} |^m ))^{2/m}  \left(\alpha_{i, \xi}[s - t]\right)^{1-2/m}.
    \end{align*}
    In particular, 
    \begin{align*}
        2 \sum_{s > T}\lvert \E_P (\xi_{it} \xi_{is}) \rvert \leq  24 (\E_P ( |\xi_{it} |^m ))^{2/m}  \sum_{s > T} \left(\alpha_{i, \xi}[s - t]\right)^{1-2/m}.
    \end{align*}
    The right-hand side is bounded by a constant $C_\xi'$ that depends only on $a$, $b$, $d_1$, and $d_2$. This follows from the existence of moments and the mixing property of $\xi_{it}$. Note that the stationarity assumption is used to apply the result of \textcite{davidov68}.
\end{proof}    

\begin{lemma}[Exponential tail bound implies existence of moments]
	\label{lem:moment}
	Suppose that a random variable $X$ satisfies that $P (|X| > x) < C \exp (-(x/a)^{d})$ for some $C, a >0 $ and $d>1$. Then, for any integer $p$, $\E | X|^p  <M$ for $M$ depending only on $C,a, d$ and $p$. 
\end{lemma}
\begin{proof}
		By the argument given in \textcite[][page 129]{kosorok2008introduction}, which is based on the series expansion of the exponential function, we have $( \E(|X|^p ))^{1/p} \leq p! || X||_{\psi_1}$ where $|| \cdot ||_{\psi_a}$ is the Orlicz norm with $\psi_a (x) = \exp (x^a)-1$ as defined in the proof of Lemma \ref{lem:clt-dependent}. By \textcite[][Lemma 8.1]{kosorok2008introduction}, $\lVert X \rVert_{\psi_1}$ is bounded by a constant which depends on $C$, $a$ and $d$. 
	\end{proof}

\begin{lemma}[Large CLT for mixing sequences]
	\label{lem:clt-dependent}
		Suppose that $\{\{ X_{jt}\}_{j=1}^J\}_{t=1}^T$ is an $\alpha$-mixing sequence (as a sequence indexed by $t$) with mixing coefficients $\alpha (k)$. Suppose that $T^{\delta_1} \leq J$ for some $\delta_1 >0$. Let $S_J = T^{-1/2} \sum_{t=1}^T(X_{1t}, \dots, X_{Jt})'$. Let $G \sim N(0,\Xi)$, where $\Xi$ is the long-run covariance matrix of $(X_{1t}, \dots, X_{Jt})$. Assume the following three conditions:
		\begin{enumerate}
			\item There exist some universal constants $C_1>0$, $a>0$ and $d_1>0$ such that $P( |X_{jt} | >x ) < C_1 \exp (- (1/a)^{d_1} x^{d_1})$ for all $t \in \{1,\dots, T\}$ and $j \in \{ 1, \dots, J\}$.
			\item There exist some universal constants $C_2>1$, $b>0 $ and $d_2 >0 $ such that $\alpha (k) \leq C_2 \exp (-b k^{d_2})$ for any $k\geq 1$. 
			\item There exists a universal constant $C_3 >0 $ such that $V_{T,j} \geq C_3$ for any $j \in \{1 , \dots, J\} $, where $V_{T,j} = \var ( \sum_{t=1}^T X_{jt} / \sqrt{T}) $.  
		\end{enumerate}
		For $a\in \mathbb{R}^N$, define $A (a) = \{x \in \mathbb{R}^N : x_j \leq a_j \; \text{for $j =1, \dotsc, J$.}\}$. Let $\mathcal{A}= \bigcup_{a\in  \mathbb{R}^J} A(a)$. Then, it holds that
		\begin{align*}
		\sup_{P\in \mathbb{P}} \sup_{A \in \mathcal{A}} \left| P(S_J \in A ) - P (G \in A) \right| 
		 \lesssim    \frac{ ( \log J)^{(1 + 2 d_2 )/(3 d_2 )}}{T^{1/9}} + \frac{( \log J )^{7/6}}{T^{1/9}}
		\end{align*}
		provided that $(\log J)^{3-d_2} = o ( T^{d_2/3} )$ and $\mathbb{P}$ is a collection of probabilities measures under which the above three conditions are satisfied with identical choices of $C_1$, $C_2$. $C_3$, $a$, $b$, $d_1$ and $d_2$.
	\end{lemma}

	\begin{proof}
		The lemma follows by Theorem 1 of \textcite{chang2022central}, noting the remark at the beginning of Section 2.1 of \textcite{chang2022central}. Theorem 1 of \textcite{chang2022central} has three conditions, and the second and third conditions are given in the statement of the lemma. The first condition is ``There exist a sequence of constants $B_J \geq 1$ and a universal constant $d_1 \geq 1$ such that $|| X_{jt} ||_{\psi_{d_1}} \leq B_J$ for all $t \in \{1,\dots, T\}$ and $j \in \{ 1, \dots, J\}$'', where $|| \xi ||_{\psi_{\alpha}} = \inf [ \lambda >0 : E( \psi_{\alpha} ( | \xi | / \lambda )) \leq 1 ] $ for $\psi_{\alpha} (x) = \exp (x^{\alpha} )-1$ (the Orlicz norm with $\psi_{\alpha}$). By Lemma 8.1 of \textcite{kosorok2008introduction}, $P( |X_{jt} | >x ) < C_1 \exp (- (1/a)^{d_1} x^{d_1})$ implies this condition by taking  $B_J = ((1+ C_1 / (1/a)^{d_1}))^{1/d_1}$, which is constant if $C_1$, $a$ and $d_1$ are constant.
	\end{proof}

    \begin{lemma}
        \label{lem:dhat-d-lv}
        Let $\mathbb{P}_N$ be the set of probability measures which satisfy Assumption~\ref*{assumption:basic} with identical choices of $a$, $b$, $d_1$ and $d_2$. 
        Assume that there are finite constants $0 < \delta_1 < \delta_2$ such that $T^{\delta_1} \leq N \leq o(1)T^{\delta_2}$. Let Assumption~\ref*{assum:kernel} hold with $\kappa_N \asymp T^{\rho}$ where $0 <\rho < (\vartheta -1)/(3\vartheta -2)$.
        Let 
        \begin{align*}
            b_N^{LV*} =  \frac{ (\sqrt{T}  \vee \log N) r_{\theta, N}}{ \iota_N \min_{1 \leq i \leq N}\sigma_i} +  T^{-c_1} (\log N)^{c_2} + T^{-\rho},
        \end{align*}
        where $c_1>0$ and $c_2>0$ are two constants defined in Lemma \ref{lem:lv-consistency} with $c_1$ depending only on $(\rho, \vartheta)$ and $c_2$ depending only on $(d_2, \vartheta)$.
        Assume that $b_N^{LV*} \to 0$. 
        Then for any sequence $\zeta_N$ such that $\zeta_N \to \infty$ as $N,T \to \infty$, 
        \begin{align*}
        \sup_{P\in \mathbb{P}_N} P \left( \max_{1 \leq i \leq N} \max_{(h,h^*) \in \mathbb{G}^2} \left| \widehat{D}_{i}(g_i^0, h) - D_{i}(g_i^0, h) \right| > \zeta_N b_N^{LV*} \right) = o(1).
    \end{align*}	
    \end{lemma}

        \begin{proof}
        Throughout the proof, let $C$ denote a generic constant that does not depend on $P \in \mathbb{P}$ and whose value may change between different equations.
        Let $\delta_i(h) = \theta_{g_i^0} - \theta_h$ and $\hat{\delta}_i(h) = \hat{\theta}_{g_i^0} - \hat{\theta}_h$. Let $\widehat{\xi}_i(h) = \sqrt{\widehat{\Xi}_i (h,h)}$ and $\xi(h) = \sqrt{\Xi_i(h,h)}$. Let 
        \begin{align*}
            b_N^{LV} =  \frac{r_{\theta, N}}{ \iota_N \min_{1 \leq i \leq N}\sigma_i} +  T^{-c_1} (\log N)^{c_2} + T^{-\rho}. 
        \end{align*}
        By the inequality $\lvert 1 - \sqrt{a} \rvert \leq \lvert 1 - a \rvert$, 
        \begin{align*}
            \left\lvert 1 - \frac{\hat{\xi}_i (h)}{\xi_i (h)} \right\rvert 
            \leq & \frac{\left\lvert \hat{\xi}_i (h, h') - \xi_i (h, h')\right\rvert}{\xi_i (h, h')} 
        \\ 
            \leq & \frac{\sigma_i^2 \lVert \delta_i(h) \rVert \lVert \delta_i (h') \rVert}{\xi_i (h, h')} \frac{\left\lvert \hat{\xi}_i (h, h') - \xi_i (h, h')\right\rvert}{\sigma_i^2 \lVert \delta_i(h) \rVert \lVert \delta_i (h') \rVert}.
        \end{align*}
        By Assumption~\ref*{assumption:basic}.\ref*{assumption:max:min_eigenvalue}, $\xi_i (h, h') / (\sigma_i^2 \lVert \delta_i(h) \rVert \lVert \delta_i (h') \rVert)$ is bounded away from zero. Moreover, with probability approaching one, 
        \begin{align*}
            \frac{\left\lvert \hat{\xi}_i (h, h') - \xi_i (h, h')\right\rvert}{\sigma_i^2 \lVert \delta_i(h) \rVert \lVert \delta_i (h') \rVert} \leq \zeta_N b_N^{LV}
        \end{align*}
        uniformly over $i = 1, \dotsc, N$ and $h, h' \in \mathbb{G} \setminus \{g_i^0\}$. 
        Therefore, we can take 
        \begin{align}
            \label{eq:conv little xi hat}
            \lvert 1 - \hat{\xi}_i (h) \big/ \xi_i (h) \rvert \leq \zeta_N b_N^{LR}. 
        \end{align}
        Next we consider $\norm{T^{-1/2} \sum_{t = 1}^T v_{it} x_{it}}$. Consider any component $x_{it, p}$ of $x_{it}$. Set $\xi_{it} = v_{it} x_{it, p}$ in Lemma~\ref{lem:lv-exists} and conclude that 
        \begin{align*}
            s_T^2 = \max_{1 \leq i \leq N} \max_{1 \leq t \leq T} \left(
                \E (\xi_{it}^2) + 2 \sum_{s > t} \left\lvert \E (\xi_{it} \xi_{is}) \right\rvert 
            \right)
        \end{align*}
        is bounded and fulfills the condition in Lemma~\ref{lem:fuk-nagaev-dependent}. By Lemma~\ref{lem:exp-tail} and Lemma~\ref{lem:mixing-g}, $\xi_{it}$ satisfies the tail and mixing conditions for $X_{it}$ in Lemma~\ref{lem:fuk-nagaev-dependent}. Now, applying Lemma~\ref{lem:fuk-nagaev-dependent}
        \begin{align}
            \label{eq:vit xit bound}
            \max_{1 \leq i \leq N}   \left\lVert
                T^{-1/2} \sum_{t = 1}^T v_{it} x_{it}
            \right\rVert 
            = O_p \left(\log N\right).
        \end{align}
        Similarly, it can be argued that 
        \begin{align}
            \label{eq:conv norm to expected norm}
            \begin{aligned}
                \max_{1 \leq i \leq N}    \left|  T^{-1/2} \sum_{t = 1}^T \left( \lVert x_{it} \rVert^2 - \E \lVert x_{it} \rVert \right) \right|
            =& O_p \left(\log N\right),            
        \\
        \max_{1 \leq i \leq N}   \left|  T^{-1/2} \sum_{t = 1}^T \left( \lVert x_{it} \rVert \lVert w_{it} \rVert - \E \lVert x_{it} \lVert w_{it} \rVert \rVert \right) \right|
                =& O_p \left(\log N\right), 
        \\
        \max_{1 \leq i \leq N}    \left|   T^{-1/2} \sum_{t = 1}^T \left( \lvert v_{it} \rvert \lVert x_{it} \rVert - \E \lvert v_{it} \rvert \lVert x_{it} \rVert \right) \right|
            =& O_p \left(\log N\right).   
            \end{aligned}         
        \end{align}
        We write
        \begin{align*}
            \widehat{D}_i(h) - D_i(h) = & \left( \frac{\xi_i(h)}{\hat{\xi}_i(h)} - 1 \right) D_i(h)
            + 
            \frac{\frac{1}{\sqrt{T}} \sum_{t = 1}^T \left(\hat{d}_{it} (h) - d_{it} (h) \right)/(\sigma_i \norm{\delta_i(h)})}{\hat{\xi}_{i}(h)/(\sigma_i \norm{\delta_i(h)})} 
        \\
            \equiv & J_1 + J_2.
        \end{align*}
        We bound $J_1$ by writing 
        \begin{align*}
            \abs{ \left(\frac{\xi_i(h)}{\hat{\xi}_i(h)} - 1 \right) D_i (h)} 
            = & 
            \abs{
                \frac{\xi_i(h)}{\hat{\xi}_i(h)} \left(1 - \frac{\hat{\xi}_i (h)}{\xi_i(h)}\right)
                \frac{\frac{1}{T} \sum_{t = 1}^T v_i x_i'\delta_i(h) / \norm{\delta_i(h)}}{\xi_i(h) / (\sigma_i \norm{\delta_i(h)})}
            }
            \\ 
            \leq & 
                \frac{\xi_i(h)}{\hat{\xi}_i(h)} \abs{1 - \frac{\hat{\xi}_i (h)}{\xi_i(h)}}
                \left\lVert\frac{1}{\sqrt{T}} \sum_{t = 1}^T v_{it} x_i\right\rVert
                \left(
                \frac{\xi_i(h)}{\sigma_i \norm{\delta_i (h)}}
                \right)^{-1}.
        \end{align*}
        The right-hand side is bounded by \eqref{eq:conv little xi hat}, \eqref{eq:vit xit bound}, noting that $\xi_i(h)/ ( \sigma_i \norm{\delta_i(h)})$ is bounded away from zero by Assumption~\ref*{assumption:basic}.\ref*{assumption:max:min_eigenvalue}, and observing that 
        \begin{align*}
            \frac{\hat{\xi}_i(h)}{\xi_i(h)}
            \geq 1 - \abs{\frac{\hat{\xi}_i(h)}{\xi_i(h)} - 1}
        \end{align*}
        in conjunction with \eqref{eq:conv little xi hat} implies a lower bound on $\hat{\xi}_i(h)/\xi_i(h)$. Hence, $J_1$ is bounded by $C\zeta_N b_N^{LV} \log N$ with probability approaching one. 
        
        To bound $J_2$, we derive a lower bound on its denominator from 
        \begin{align*}
            \frac{\hat{\xi}_i (h)}{\sigma_i \norm{\delta_i(h)}} = 
            \frac{\xi_i (h)}{\sigma_i \norm{\delta_i(h)}} \left \{ 
            \left( \frac{\hat{\xi}_i(h)}{\xi_i(h)} - 1\right) + 1
            \right\}
        \end{align*}
        in conjunction with \eqref{eq:conv little xi hat} and noting that $\xi_i(h)/ ( \sigma_i \norm{\delta_i(h)})$ is bounded away from zero by Assumption~\ref*{assumption:basic}.\ref*{assumption:max:min_eigenvalue}.
        For the numerator in $J_2$, we observe the following decomposition:
        \begin{align*}
            \frac{\hat{d}_{it}(h) - d_{it}(h)}{\sigma_i \norm{\delta_i(h)}} 
            = & \frac{1}{2} \frac{x_{it}' \left(\hat{\theta}_{g_i^0} - \theta_{g_i^0}\right) + w_{it}' \left(\hat{\theta}^w - \theta^w \right)}{\sigma_i} x_{it}' \left( \frac{\hat{\delta}_i(h)}{\norm{\delta_i(h)}} \right) 
        \\
            & - \frac{1}{2} v_{it} x_{it}' \frac{\left(\hat{\delta}_i(h) - \delta_i (h)\right)}{\norm{\delta_i(h)}}.
        \end{align*}
        In the following arguments, we use that 
        \begin{align*}
            \norm{ \hat{\delta}_i(h)}/\norm{\delta_i(h)} 
            \leq 1 + \frac{\norm{\hat{\delta}_i(h) - \delta_i(h)}}{\norm{\delta_i(h)}}
        \end{align*}
        is bounded by the fact that $r_{\theta, N} = o (1 \wedge \iota_N)$. 
        With probability at least $1-a_{\theta, N}$, we bound 
        \begin{align*}
            \left\lvert \frac{\hat{d}_{it}(h) - d_{it}(h)}{\sigma_i \norm{\delta_i(h)}} \right\rvert 
            \leq & C \left( \frac{\norm{x_{it}}^2 + \norm{x_{it}} \norm{w_{it}}}{\min_{1 \leq i \leq N} \sigma_i} + \frac{\abs{v_{it}}\norm{x_{it}}}{\iota_N} \right) r_{\theta, N}
        \\
            \leq &  r_{\theta, N} C \left(
                \frac{\E \norm{x_{it}}^2 + \E \norm{x_{it}} \norm{w_{it}}}{\min_{1 \leq i \leq N} \sigma_i} 
                + 
                \frac{\E \abs{v_{it}}\norm{x_{it}}}{\iota_N}
            \right)
        \\  \leq & r_{\theta, N} C \bigg(
            \begin{aligned}[t]
                & \frac{\norm{x_{it}}^2 +  \norm{x_{it}} \norm{w_{it}} - (\E \norm{x_{it}} \norm{w_{it}} + \E \norm{x_{it}} \norm{w_{it}})}{\min_{1 \leq i \leq N} \sigma_i} 
        \\    
                & + \frac{\abs{v_{it}} \norm{x_{it}} - \E \abs{v_{it}}\norm{x_{it}}}{\iota_N} \bigg).
            \end{aligned}
        \end{align*}	
        Noting that $\E \norm{x_{it}}^2$, $\E \norm{x_{it}} \norm{w_{it}}$ and $\E \abs{v_{it}} \norm{x_{it}}$ are bounded uniformly over $i$ and $t$ by Assumption~\ref*{assumption:basic}.\ref*{assum:tail}, \eqref{eq:conv norm to expected norm} implies 
        \begin{align*}
        \max_{1 \leq i \leq T}
        \left| \frac{1}{\sqrt{T}} \sum_{t=1}^T \frac{\hat{d}_{it}(h) - d_{it}(h)}{\sigma_i \norm{\delta_i(h)}} \right|
        \leq & r_{\theta, N} C (\min_{1 \leq i \leq N} \sigma_i \wedge \iota_N)^{-1} \left(\sqrt{T} + \log N \right) 
        \\
        \leq & r_{\theta, N} C (\min_{1 \leq i \leq N} \sigma_i \wedge \iota_N)^{-1} \sqrt{T},
        \end{align*}
        where the last inequality follows since $N \leq o (1) T^{\delta_2}$.
        The bounds on $J_1$ and $J_2$ yield the desired result.
    \end{proof}

    \begin{lemma}[Fuk-Nagaev-type inequality for mixing sequences]
        \label{lem:fuk-nagaev-dependent}
             Suppose that $X_{it}$ is a strongly mixing process with zero mean for each $i=1, \dots, N$  with tail probabilities $\sup_{i = 1, \dotsc, N} P (|X_{it}| > x) \leq \exp (1-(x/a)^{d_1})$ and with strong mixing coefficients $\sup_{i = 1, \dotsc, N} a_i [t] \leq \exp ( -bt^{d_2} ) $, where $a$, $b$, $d_1$, and $d_2$ are positive constants. Let $\mathbb{P}_N$ denote a sequence of sets of probability measures that satisfy the above conditions with given values of $a$, $b$, $d_1$, and $d_2$. 
            Let \begin{align*}
            s_T^2   = \max_{1 \leq i \leq N} \max_{1 \leq t \leq T} \left(\E (X_{it}^2) + 2 \sum_{s>t} \left| \E( X_{it} X_{is})\right| \right) .
            \end{align*}
        Assume that $s_T^2 < C_s \log^{a_s} N$ for constants $C_s$ and $0 \leq a_s \leq 1$ which do not depend on $N, T$ nor $P$.
             Then, it holds that for any constant $C>0$, as $N,T\to \infty$ with $NT^{-\delta_2} \to 0$ for some $\delta_2>0$,
             \begin{align*}
        \sup_{P \in \mathbb{P}_N} P \left( \max_{1 \leq i \leq N}  \left| \frac{1}{T}\sum_{t=1}^T X_{it} \right| \geq  C T^{-1/2} \log N  \right) \to 0 .
        \end{align*}
        
        \end{lemma}
        
        \begin{proof}
            By the Bonferroni inequality and inequality (1.7) in \textcite{merlevede2011bernstein} which is an application of \textcite[][Theorem 6.2]{rio2017asymptotic} (the original French version was published in 2000), we have 
        \begin{align*}
            &\sup_{P \in \mathbb{P}} P  \left( \max_{1 \leq i \leq N} \left| \frac{1}{T}\sum_{t=1}^T X_{it} \right| \geq  x  \right) 
            \leq  \sup_{P \in \mathbb{P}}  \sum_{i=1}^N P \left( \left| \frac{1}{T}\sum_{t=1}^T X_{it} \right| \geq x \right)
             \\
            \leq & 4 N \left( 1+ \frac{T(x/4)^2}{r s_T^2 } \right)^{-r/2}  + 4 C N (x/4)^{-1} \exp \left(  - a \frac{(Tx/4)^{d}}{b^{d} r^{d}} \right),
        \end{align*}
        where $r\geq 1$, $d = (d_1^{-1} + d_2^{-1})^{-1}$ and $C'$ is a positive constant. 
        Thus, for $x= C T^{-1/2} \log N$, it holds that 
        \begin{align*}
        &	\sup_{P \in \mathbb{P}} P  \left( \max_{i = 1, \dotsc, N} \left| \frac{1}{T}\sum_{t=1}^T X_{it} \right| \geq  C T^{-1/2} \log N  \right) \\
            \leq & 4 N \left( 1+ \frac{ C^2 \log^2 N }{16 r s_T^2 } \right)^{-r/2} 
              + 16  (C'/C) N  T^{1/2} \log^{-1} N \exp \left(  - a \frac{(C T^{1/2} \log N )^{d}}{4^d b^{d} r^{d}} \right) \\
              =& 4 N \exp \left( -\frac{r}{2} \log \left( 1 +  \frac{ C^2 \log^2 N }{16 r s_T^2 }\right)  \right) 
        \\
            &  + 16  (C'/C) \frac{N  T^{1/2}}{\log N} \exp \left(  - a \frac{c^d}{4^db^d} \left(\frac{T^{1/2}\log N}{r} \right)^d  \right) 
        \end{align*}
        We take $r= T^{1/2 -c}$ for $0 < c < 1/2 $. The second term on the last line in the above display converges to zero because $T^{1/2} \log N /r = T^{c} \log N$ and $N T^{-\delta_2} \to 0$. 
        We now argue that the first term vanishes as well. 
        For $a$ close to zero, a second-order Taylor expansion of the natural logarithm function yields 
        \begin{align*}
            \log (1 + a) = a - \frac{1}{(1 + a^*)^2} a^2, 
        \end{align*}
        where $a^*$ is an intermediate value between zero and $a$. For $a$ close to zero, $1/(1 + a^*)$ is bounded and therefore 
        \begin{align*}
            \log (1 + a) \geq a + O (a^2).
        \end{align*}
        In particular, $\log (1 - a) \geq a + O (a^2)$. We set $a = C^2 \log^2 N / (16 T^{1/2 - c} s_T^2)$. Under the assumption of the lemma, $a \to 0$. 
        Now, the term in the exponential function can be bounded by 
        \begin{align*}
            -\frac{T^{1/2-c}}{2} \log \left( 1 +  \frac{ C^2 \log^2 N }{16 T^{1/2-c}  s_T^2 }\right) 
            \leq & -\frac{T^{1/2-c}}{2} \left\{ 
                \frac{ C^2 \log^2 N }{16 T^{1/2-c}  s_T^2 }
                + O \left(\left[\frac{ C^2 \log^2 N }{16 T^{1/2-c}  s_T^2 }\right]^2\right)
            \right\}
        \\
            \leq & -\frac{T^{1/2-c} C^2 \log^2 N }{32 T^{1/2-c}  s_T^2 } (1 + o (1)) 
        \\
            \leq & - \frac{ C^2 \log^2 N }{64 s_T^2 }.
        \end{align*}
        Thus, under $s_T^2 \leq  K \log^{a_s} N$ with $0 <a_s < 1$,
        \begin{align*}
            4 N \exp \left( -\frac{r}{2} \log \left( 1 +  \frac{ C^2 \log^2 N }{16 r s_T^2 }\right)  \right) \leq 4 N \exp \left(  - \frac{ C^2 \log^2 N }{64 s_T^2 }\right) \to 0.
        \end{align*}
    \end{proof}

    \begin{lemma}   
        \label{lem:no serial correlation}
        Let $\mathbb{P}_N$ be the set of probability measures which satisfy Assumption~\ref*{assumption:basic}.\ref*{assum:exogeneity}--\ref*{assumption:basic}.\ref*{assum:mixing} and Assumption~\ref*{assumption:basic}.\ref{assum:Omega negative correlations} with identical choices of $a$, $b$, $d_1$ and $d_2$. 
        Assume that there are finite constants $0 < \delta_1 < \delta_2$ such that $T^{\delta_1} \leq N \leq o(1)T^{\delta_2}$ and that 
        \begin{align*}
            \frac{r_{\theta, N}}{\iota_N \wedge \min_{1 \leq i \leq N} \sigma_i} \to 0.
       \end{align*}
       Suppose that Assumption~\ref*{assumption:no serial correlation} holds and $\kappa_N = 0$.
        Then, there is a constant $C$ such that  
        \begin{align*}
        \sup_{P\in \mathbb{P}_N} P \left( \max_{1 \leq i \leq N} \max_{(h,h^*) \in \mathbb{G}^2} \left| \widehat{D}_{i}(g_i^0, h) - D_{i}(g_i^0, h) \right| > C \frac{r_{\theta, N}\sqrt{T} }{\iota_N \wedge \min_{1 \leq i \leq N} \sigma_i} \right) = o(1)
        \end{align*}
        and for all $c > 0$
        \begin{align*}
            \sup_{P\in \mathbb{P}_N} P \left( \max_{1 \leq i \leq N} \max_{(h,h^*) \in \mathbb{G}^2} \left\lVert \widehat{\Omega}_{i}(g_i^0) - \Omega_{i}(g_i^0) \right\rVert > c \right) = o(1).
        \end{align*} 
    \end{lemma}
    
    \begin{proof}
    Following the arguments in Lemma~\ref{lem:dhat-d-lv}, we bound, with probability at least $1-a_{\theta, N}$,
    \begin{align*}
        \left\lvert \frac{\hat{d}_{it}(h) - d_{it}(h)}{\sigma_i \norm{\delta_i(h)}} \right\rvert 
        \leq & C \left( \frac{\norm{x_{it}}^2 + \norm{x_{it}} \norm{w_{it}}}{\min_{1 \leq i \leq N} \sigma_i} + \frac{\abs{v_{it}}\norm{x_{it}}}{\iota_N} \right) r_{\theta, N}
    \end{align*}
    and hence 
    \begin{align}
        \label{eq:no sc dhat minus d}
        \begin{aligned}
        \frac{1}{T} \sum_{t = 1}^T \frac{\hat{d}_{it}(h) - d_{it}(h)}{\sigma_i \norm{\delta_i (h)}} = O_p \left( \frac{r_{\theta, N}}{\iota_N \wedge \min_{1 \leq i \leq N} \sigma_i}\right)
        \\
        \frac{1}{T} \sum_{t = 1}^T \frac{(\hat{d}_{it}(h) - d_{it}(h))(\hat{d}_{it}(h') - d_{it}(h'))}{\sigma_i^2 \norm{\delta_i (h)}\norm{\delta_i (h')}} = O_p \left( \frac{r_{\theta, N}^2}{\iota_N^2 \wedge \min_{1 \leq i \leq N} \sigma_i^2}\right)
        \end{aligned}
    \end{align}
    uniformly in unit $i$ and probability measure $P \in \mathbb{P}_N$. 
    Following similar arguments as in the proof of Lemma~\ref{lem:dhat-d-lv}, Lemma~\ref{lem:fuk-nagaev-dependent} yields 
    \begin{align}
        \label{eq:no sc dh dhprime convergence}
        \begin{aligned}
        & \bigg\lvert \frac{1}{T} \sum_{t = 1}^T 
         \frac{d_{it}(h) d_{it}(h')}{\sigma_i^2 \norm{\delta_i(h)}\norm{\delta_i(h')}} 
         - \frac{\frac{1}{T} \sum_{t = 1}^T \E_P \left[d_{it} (h) d_{it}(h')\right]
        }{\sigma_i^2 \norm{\delta_i(h)}\norm{\delta_i(h')}} \bigg\rvert
    \\
        \leq & C 
        \frac{1}{T} \sum_{t = 1}^T \left( \frac{\norm{x_{it}}^2 + \norm{x_{it}} \norm{w_{it}}}{\min_{1 \leq i \leq N} \sigma_i} + \frac{\abs{v_{it}}\norm{x_{it}}}{\iota_N} \right)^2 r_{\theta, N}^2
    \\
        \leq & \frac{C \, r_{\theta, N}^2}{\iota_N^2 \wedge \min_{1 \leq i \leq N} \sigma_i^2} + o_p(1),
        \end{aligned}
    \end{align}
    where the $o_p(1)$ term is uniform over units $i$ and probability measures $P \in \mathbb{P}_N$. Noting that 
    \begin{align*}
        \bigg\lvert \frac{1}{\sqrt{T}}\sum_{t = 1}^T \frac{d_{it}(h)}{\sigma_i \lVert \delta_i (h) \rVert} \bigg\rvert 
        = \bigg\lVert \frac{1}{\sqrt{T}}\sum_{t = 1}^T v_{it} x_{it} \bigg\rVert, 
    \end{align*}
    Lemma~\ref{lem:fuk-nagaev-dependent} implies 
    \begin{align}
        \label{eq:no ser D bound}
        \bigg\lvert \frac{1}{\sqrt{T}}\sum_{t = 1}^T \frac{d_{it}(h)}{\sigma_i \lVert \delta_i (h) \rVert} \bigg\rvert = O_p \left(\log N\right)
    \end{align}
    uniformly in $i$ and $P$. Since Assumption~\ref*{assumption:basic}.\ref*{assum:tail} bounds the expectation $\E [v_{it}^2 x_{it} x_{it}']$ by a finite constant, we have 
    \begin{align}
        \label{eq:no sc dh square convergence}
        \begin{aligned}
        \frac{1}{T} \sum_{t = 1}^T \frac{d_{it}(h) d_{it}(h')}{\sigma_i^2 \norm{\delta_i(h)} \norm{\delta_i(h')}} 
        = & \frac{\delta_i(h)' \frac{1}{T} \sum_{t = 1}^T v_{it}^2 x_{it} x_{it}' \delta_i(h')}{\norm{\delta_i(h)} \norm{\delta_i(h')}} .
        \\
        \leq & \left\lVert \frac{1}{T} \sum_{t = 1}^T  \left( v_{it}^2 x_{it} x_{it}' - \E [v_{it}^2 x_{it} x_{it}'] \right)\right\rVert 
        \\
        & + \left\lVert \frac{1}{T} \sum_{t = 1}^T  \E [v_{it}^2 x_{it} x_{it}']\right\rVert 
        \\
        = & O_p \left(\log N / \sqrt{T}\right) = o_p(1)
        \end{aligned}
    \end{align}
    uniformly in $i$ and $P$.

    Now, combining the decomposition
    \begin{align*}
        & \frac{
        \frac{1}{T} \sum_{t = 1}^T \left(\hat{d}_{it} (h) - \bar{\hat{d}}_{it} (h) \right)\left(\hat{d}_{it} (h') - \bar{\hat{d}}_{it} (h') \right)
        }{
        \sigma_i^2 \norm{\delta_i(h)}\norm{\delta_i(h')}
        } - \frac{\frac{1}{T} \sum_{t= 1}^T \E_P \left[d_{it} (h) d_{it}(h')\right]
        }{\sigma_i^2 \norm{\delta_i(h)}\norm{\delta_i(h')}}
    \\
         =& \frac{1}{T} \sum_{t = 1}^T \left(\frac{\hat{d}_{it}(h) - d_{it}(h)}{\sigma_i \norm{\delta_i(h)}}\right)\left(\frac{\hat{d}_{it}(h') - d_{it}(h')}{\sigma_i \norm{\delta_i(h')}}\right)
    \\
        & + \frac{1}{T} \sum_{t = 1}^T \frac{d_{it}(h)}{\sigma_i \norm{\delta_i(h)}} \left(\frac{\hat{d}_{it} (h') - d_{it}(h')}{\sigma_i \norm{\delta_i(h')}}\right)
         + \frac{1}{T} \sum_{t = 1}^T \frac{d_{it}(h')}{\sigma_i \norm{\delta_i(h')}} \left(\frac{\hat{d}_{it} (h) - d_{it}(h)}{\sigma_i \norm{\delta_i(h)}}\right). 
    \\
        & - \left(\frac{1}{T} \sum_{t = 1}^T \frac{\hat{d}_{it} (h)}{\sigma_i \norm{\delta_i(h)}} \right)
        \left(\frac{1}{T} \sum_{t = 1}^T \frac{\hat{d}_{it} (h')}{\sigma_i \norm{\delta_i(h')}} \right) 
    \\
        & + \frac{1}{T} \sum_{t = 1}^T 
         \frac{d_{it}(h) d_{it}(h')}{\sigma_i^2 \norm{\delta_i(h)}\norm{\delta_i(h')}} 
         - \frac{\frac{1}{T} \sum_{t = 1}^T \E_P \left[d_{it} (h) d_{it}(h')\right]
        }{\sigma_i^2 \norm{\delta_i(h)}\norm{\delta_i(h')}}
    \end{align*}
    with \eqref{eq:no sc dhat minus d}, \eqref{eq:no sc dh dhprime convergence}, \eqref{eq:no ser D bound} and \eqref{eq:no sc dh square convergence} yields a constant $C$ such that 
    \begin{align*}
    \begin{aligned}
        \sup_{P \in \mathbb{P}} P \Bigg( 
        \max_{1 \leq i \leq N} \Bigg\lvert &
        \frac{
        \frac{1}{T} \sum_{t = 1}^T \left(\hat{d}_{it} (h) - \bar{\hat{d}}_{it} (h) \right)\left(\hat{d}_{it} (h') - \bar{\hat{d}}_{it} (h') \right)
        }{
        \sigma_i^2 \norm{\delta_i(h)}\norm{\delta_i(h')}
        }
        - \frac{1}{T} \sum_{t = 1}^T 
         \frac{d_{it}(h) d_{it}(h')}{\sigma_i^2 \norm{\delta_i(h)}\norm{\delta_i(h')}}
        \Bigg\rvert 
    \\
        & \geq C 
            \frac{r_{\theta, N}}{\iota_N \wedge \min_{1 \leq i \leq N} \sigma_i} 
        \Bigg) =o(1).
    \end{aligned}
    \end{align*}
    This implies the first statement of the lemma. 
    The proof of the second statement of the lemma is similar to the proof of Lemma~\ref{lem:dhat-d-lv}, but replacing all references to Lemma~\ref{lem:lv-consistency} by a reference to the result in the previous display.
    \end{proof}

    \begin{lemma}
        \label{lem:tdist:tail:upperbound}
        Let $\nu(N) \geq 1$ denote a sequence that converge to infinity and let $c_{N}(\alpha)$ denote the $(1 - \alpha/N)$-quantile of the $t$-distribution with $\nu(N)$ degrees of freedom. Suppose that $(\log N)/\nu(N) \to 0$.
        For each $\epsilon > 0$ and $0 < \underline{\alpha} < 1$ there is a threshold $N_0$ such that for $N \geq N_0$
        \begin{align*}
            \sup_{\underline{\alpha} \leq \alpha < 1} c_{N} (\alpha) \leq \sqrt{2 (1 + \epsilon) \log (N/\alpha)}.
        \end{align*}
        \end{lemma}
\begin{proof}
        For notational convenience, write $\nu = \nu(N)$. We prove the bound for $\alpha = \underline{\alpha}$ and write $c_{N} = c_N(\alpha)$. The uniformity then follows from the monotonicity of the distribution function. 
        Clearly, $c_{N} \to \infty$ so we can take $c_N \geq 1$, provided that $N$ is large enough. 
        The density function of the $t$-distribution with $\nu$ degrees of freedom is given by 
        $f_\nu(x) = c(\nu) \left(1 + x^2/\nu\right)^{-\frac{\nu + 1}{2}}$, where
        \begin{align*}
            c(\nu) = \frac{\Gamma\left(\frac{\nu + 1}{2}\right)}{\sqrt{\nu \pi} \Gamma\left(\frac{\nu}{2}\right)} \to \frac{1}{\sqrt{2 \pi}}
        \end{align*}
        as $\nu \to \infty$. It follows that there is a universal constant $C$ such that $c(\nu) \leq C$.
        We first show that $c_N^2/\nu = O(1)$. The proof is by contradiction. Suppose that $\limsup_{N \to \infty} c_N^2/\nu = \infty$.
        Applying Theorem~1 in \textcite{soms1976asymptotic} with $n = 1$ yields 
        \begin{align}
            \label{eq:tdist:soms-upperbound}
            1 - F_{\nu}(c_N) \leq f_{\nu} (c_N) \frac{1}{c_N} \left(1 + \frac{c_N^2}{\nu}\right).
        \end{align}
        This implies that 
        \begin{align*}
            \frac{\alpha}{N} \leq c(\nu) \left(1 + \frac{c_N^2}{\nu}\right)^{-\frac{\nu + 1}{2}}\left(1 + \frac{c_N^2}{\nu}\right)
            \leq C \left(1 + \frac{c_N^2}{\nu}\right)^{-\frac{\nu - 1}{2}}.
        \end{align*}
        Taking logs and re-arranging gives 
        \begin{align*}
            \frac{\log(N/\alpha)}{\nu} \geq \frac{1}{2} \frac{\nu - 1}{\nu} \left(\log \left(1 + \frac{c_N^2}{\nu} \right) - C \right).
        \end{align*}
        The left-hand side of the inequality vanishes under the assumptions of the lemma, whereas a sub-sequence of the right-hand side diverges to infinity. This establishes that the inequality is impossible and therefore $c_N^2/\nu = O(1)$.
        This implies that there exists a constant $b$ such that 
        \begin{align*}
            1 < b \leq \left(  1 + \frac{c_N^2}{\nu} \right)^{\frac{\nu}{c_N^2}} \leq e, 
        \end{align*}
        so that we can take 
        \begin{align*}
            \left(\left(  1 + \frac{c_N^2}{\nu} \right)^{\frac{\nu}{c_N^2}} \right)^{-1}
            \leq e^{- \frac{\nu}{\nu + 1}(1 + \epsilon^*/2)^{-1}}
        \end{align*}
        for a positive $\epsilon^*$.
        Then
        \begin{align*}
            f_\nu(c_N) \leq C \left[ \left(  1 + \frac{c_N^2}{\nu} \right)^{\frac{\nu}{c_N^2}}\right]^{-\frac{c_N^2}{2} \left[\frac{\nu + 1}{\nu}\right]}
            \leq 
            C \exp \left( -\frac{c_N^2}{2} (1 + \epsilon^*/2)^{-1} \right).
        \end{align*}
        Take $N$ large enough that 
        \begin{align*}
        \frac{1}{1 + \epsilon^*/2} - \frac{4 \log c_N}{c_N^2} > \frac{1}{1 + \epsilon^*}.
        \end{align*}
        Then, the right-hand side of \eqref{eq:tdist:soms-upperbound} can be bounded by 
        \begin{align*}
            C \exp \left( - \frac{c_N^2}{2} (1 - \epsilon^*/2)^{-1}\right)
            \left(1 + \frac{c_N^2}{\nu}\right)
            \leq & 2 C \exp \left( - \frac{c_N^2}{2} \left((1 + \epsilon^*/2)^{-1} - \frac{4 \log c_N}{c_N^2} \right)\right)
            \\
            \leq & 2 C \exp \left( - \frac{c_N^2}{2} \left(1 + \epsilon^*\right)^{-1} \right)
            .
        \end{align*}
        Plugging in $1 - F_{\nu}(c_N) = \alpha/N$ and taking logs gives 
        \begin{align*}
            c_N^2 \leq & (1 + \epsilon^*) \log \left(N/\alpha\right) + \log (2C)
            \\
            \leq & 2 (1 + \epsilon^*) \log \left(N/\alpha\right) 
            \left(1 + \frac{1}{2(1 + \epsilon^*)} \frac{\log (2C)}{\log(N/\alpha)}\right).
        \end{align*}
        Hence, there is a constant $C$ such that $c_N^2 \leq C \log (N/\alpha)$. Using this inequality, we can now verify that 
        $c_N^2/\nu \to 0$ so that 
        \begin{align*}
            \left(  1 + \frac{c_N^2}{\nu} \right)^{\frac{\nu}{c_N^2}} \to e,
        \end{align*}
        allowing us to take $\epsilon^* = \epsilon/2$ for sufficiently large $N$. Taking $N$ large enough that 
        \begin{align*}
        (1 + \epsilon/2) \left(1 + \frac{1}{2(1 + \epsilon/2)} \frac{\log (2C)}{\log(N)}\right)
        \leq 1 + \epsilon
        \end{align*}
        yields 
        $c_N^2 \leq 2 (1 + \epsilon) \log \left(N/\alpha\right)$.  \end{proof}
        \begin{lemma}
        \label{lem:tdist:bound-millsratio}
        For $\nu \geq 1$, let $F_\nu$ and $f_\nu$ denote the distribution and density function of a $t$-distributed random variable with $\nu$ degrees of freedom.
        For $x^2 > 2$
        \begin{align*}
            f_{\nu}(x) < 2 x \left(1 - F_{\nu}(x) \right).
        \end{align*}
        \end{lemma}
        \begin{proof}
        Applying Theorem~1 in \textcite{soms1976asymptotic} with $n = 2$ yields the inequality
        \begin{align*}
            1 - F_\nu(x) \geq (1 + x^2/\nu)\left(1 - \frac{\nu}{(\nu + 2) x^2}\right) f_\nu(x) /x.
        \end{align*}
        Now, $x^2 > 2$ implies
        \begin{align*}
            1 - F_\nu(x)
            > \left(1 - \frac{1}{2}\right) f_\nu(x)/x.
        \end{align*}
        \end{proof}

    \begin{lemma}
        \label{lem:conv:Du}

        Let $\mathbb{P}_N$ denote a family of probability measures satisfying Assumptions~\ref*{assumption:basic} with parameters satisfying  $\norm{\theta_h} < M$ for some finite $M$ for any $h$.
        Assume $r_{\theta, N} = o (1 \wedge \iota_N)$ and $ r_{\theta, N}(\sqrt{T} + \sqrt{\log N})(\iota_N + \min_{1 \leq i \leq N} \sigma_i) =o(1)$.
        There are constants $C$ and $C'$ such that
        \begin{align}
        \label{eq:duhat_diff_bound}
         \sup_{P \in \mathbb{P}_N} P \left(
                \max_{1 \leq i \leq N} \bigg\lvert \frac{1}{T}\sum_{t= 1}^T 
                    \frac{ \hat{d}_{it}^U(g_i^0,h) - d_{it}^U (g_i^0,h)}{s_i^U(h)} 
                    \bigg\rvert 
                    \geq C \frac{r_{\theta, N}}{\iota_N + \min_{1 \leq i \leq N} \sigma_i} \right) =& o(1).
        \\
            \label{eq:DhatU_conv}
                \sup_{P \in \mathbb{P}_N} P 
                \bigg(
                    \max_{1 \leq i \leq N} \left\lvert 
                        \widehat{D}_{i}^U(g_i^0, h) - \tilde{D}_{i}^U(g_i^0, h)
                    \right\rvert 
                    \geq C' \frac{r_{\theta, N}(\sqrt{T} + \sqrt{\log N})}{\iota_N + \min_{1 \leq i \leq N} \sigma_i} 
                \bigg)
            = & o(1).
        \end{align}
        \end{lemma}

        \begin{proof}
        
        Throughout the proof, let $C$ denote a generic constant that does not depend on $P \in \mathbb{P}$.
        Let $\delta_i(h) = \theta_{g_i^0} - \theta_h$ and $\hat{\delta}_i(h) = \hat{\theta}_{g_i^0} - \hat{\theta}_h$. Note that \begin{align*}
            \norm{ \hat{\delta}_i(h)}/\norm{\delta_i(h)} 
            \leq 1 + \frac{\norm{\hat{\delta}_i(h) - \delta_i(h)}}{\norm{\delta_i(h)}}
        \end{align*}
        is bounded by the fact that $r_{\theta, N} = o (1 \wedge \iota_N)$.

        We observe
        \begin{align*}
            \hat d_{it}^U (h) - d_{it}^U (h) =& \frac{1}{2} \left((y_{it} - w_{it}'\hat \theta^w - x_{it}'\hat\theta_{g_i^0})^2 - (y_{it} - w_{it}'\hat \theta^w - x_{it}'\hat\theta_{h})^2  \right) \\
            & - \frac{1}{2}\left((y_{it} - w_{it}' \theta^w - x_{it}'\theta_{g_i^0})^2 - (y_{it} - w_{it}' \theta^w - x_{it}'\theta_{h})^2  \right)\\
             =& -u_{it} x_{it}'(\hat \delta_i (h) - \delta_i(h)) + w_{it}'(\hat \theta^w -\theta^w)x_{it}' \hat \delta_i (h) +  x_{it}' \hat \delta_i (h) x_{it}'(\hat \theta_{g_i^0} - \theta_{g_i^0}) \\
             & - \frac{1}{2} (x_{it}' \hat \delta_i (h) )^2 + \frac{1}{2} (x_{it}'  \delta_i (h) )^2.
        \end{align*} 
        Thus, we have, with probability at least $1- a_{\theta, N}$, 
        \begin{align}
            \left\lvert \frac{\hat{d}_{it}^U(h) - d_{it}^U(h)}{\sigma_i \norm{\delta_i(h)}} \right\rvert 
            \leq & C \left( \frac{\norm{x_{it}}^2 + \norm{x_{it}} \norm{w_{it}}}{\min_{1 \leq i \leq N} \sigma_i} + \frac{\abs{v_{it}}\norm{x_{it}}}{\iota_N} \right) r_{\theta, N}, \label{eq:dhatdu}\\
            \left(
                \frac{\hat{d}_{it}^U(h) - d_{it}^U(h)}{\sigma_i \norm{\delta_i(h)}}
            \right)^2 
            \leq & C \left( 
                \frac{\norm{x_{it}}^4 + \norm{x_{it}}^2 \norm{w_{it}}^2}{\min_{1 \leq i \leq N} \sigma_i^2}
                + \frac{\abs{v_{it}}^2 \norm{x_{it}}^2}{\iota_{N}^2}
            \right)r_{\theta, N}^2 \label{eq:dhatdu2}
        \end{align}
        and 
        \begin{align}
            \label{eq:ddhatdu}
        \begin{aligned}
            \left\lvert \frac{d_{it}^U(h)}{\sigma_i \norm{\delta_i(h)}}
            \left(
                \frac{\hat{d}_{it}^U(h) - d_{it}^U(h)}{\sigma_i \norm{\delta_i(h)}}
            \right) 
            \right\rvert 
            \leq 
            C \bigg(
                \frac{\norm{x_{it}}^4 + \abs{v_i} \norm{x_{it}}^3 + \abs{v_i} \norm{x_{it}}^2 \norm{w_{it}} + \norm{x_{it}}^3 \norm{w_{it}}}{\min_{1 \leq i \leq N} \sigma_i} \\ 
                + \frac{\abs{v_i}^2 \norm{x_{it}}^2 + \abs{v_i} \norm{x_{it}}^3}{\iota_{N}} 
            \bigg) r_{\theta, N},
            \end{aligned}
        \end{align}
        where \eqref{eq:ddhatdu} also rely on the compactness assumption in the theorem.

        Combining \eqref{eq:dhatdu2}, Lemmas \ref{lem:moment} and \ref{lem:fuk-nagaev-dependent}, the fact that $s_i^U(h) = \sigma_i  \norm{\delta_i(h)} \E \norm{x_{it}} + \E (\norm{x_{it}}^2) \norm{\delta_i(h)}^2 \geq  C\sigma_i \norm{\delta_i(h)}$ yields \eqref{eq:duhat_diff_bound}.
        
        Let 
        \begin{align*}
            \hat s_i^U (h)^2 =\frac{1}{T} \sum_{t=1}^T \left(\hat d_{it}^U (h) - \bar{\hat d }_{it}^U (h)\right)^2 , \quad 	\tilde s_i^U (h)^2 =\frac{1}{T} \sum_{t=1}^T \left(d_{it}^U (h) - \bar{ d }_{it}^U (h)\right)^2 
        \end{align*}
        We observe that 
        \begin{align*}
         \frac{	\hat s_i^U (h)^2 - \tilde s_i^U (h)^2}{\sigma_i^2 \norm{\delta_i(h) }^2} 
        = & \frac{1}{T} \sum_{t=1}^T \left( \frac{\hat d_{it}^U (h) - d_{it}^U (h)}{\sigma_i || \delta_i(h)||} \right)^2 
        + 2 \frac{1}{T} \sum_{t=1}^T \frac{d_{it}^U (h)}{\sigma_i \norm{ \delta_i(h)}} \frac{(\hat d_{it}^U (h) - d_{it}^U(h) )}{\sigma_i \norm{ \delta_i(h) }} \\
        & - T \left(\frac{1}{T}\sum_{t=1}^T  \frac{\hat d_{it}^U (h) - d_{it}^U (h)}{\sigma_i || \delta_i(h)||}  \right) \left( \frac{1}{T} \sum_{t=1}^T \frac{\hat d_{it}^U (h) - d_{it}^U (h)}{\sigma_i || \delta_i(h)||}  +\frac{2}{T} \sum_{t=1}^T \frac{d_{it}^U (h)}{\sigma_i || \delta_i(h)||} \right).
        \end{align*}
        Note that 
        \begin{align*}
        \left| \frac{1}{T} \sum_{t=1}^T \frac{d_{it}^U (h)}{\sigma_i || \delta_i(h)||} \right| \leq \left\Vert \frac{1}{T} \sum_{t=1}^T v_{it} x_{it} \right\Vert + \norm{\delta_i (h)} \frac{1}{T}\sum_{t=1}^T \norm{x_{it}}^2
        \end{align*}
        By Lemmas \ref{lem:moment} and \ref{lem:fuk-nagaev-dependent}, and the compactness condition, we have 
        \begin{align}
        \label{eq:du_rate}
             \sup_{P \in \mathbb{P}} P \left(
                \max_{1 \leq i \leq N} \bigg\lvert \frac{1}{T} \sum_{t= 1}^T
                    \frac{ d_{it}^U(h)}{\sigma_i \norm{\delta_i(h)}} 
                    \bigg\rvert 
                    \geq C (1+ T^{-1/2} \sqrt{\log N} )	\right) =o(1),
        \end{align}
        Combining Lemma \ref{lem:fuk-nagaev-dependent},  \eqref{eq:dhatdu}, \eqref{eq:dhatdu2}, \eqref{eq:ddhatdu}, and \eqref{eq:du_rate} yields
        \begin{align}
            \sup_{P \in \mathbb{P}} \left( 
                \max_{1 \leq i \leq N} \left|
             \frac{	\hat s_i^U (h)^2 - \tilde s_i^U (h)^2}{\sigma_i^2 \norm{\delta_i(h) }^2} \right|\geq C\frac{r_{\theta, N}}{\iota_N \wedge \min_{1 \leq i \leq N} \sigma_i} \right) =o(1).\label{eq:shatu-stildeu}
        \end{align}
    Observing that $s_i^U(h) > s_i (h)$ which holds under the $E(u_{it} x_{itk_1} x_{itk_2} x_{itk_3})=0$, $\sigma_i \norm{\delta_i(h)}/ s_i^U (h)$ is bounded away from infinity by Assumption~\ref*{assumption:basic}.\ref*{assumption:max:min_eigenvalue}. This in turn implies that $\sigma_i \norm{\delta_i(h)}/ \tilde s_i^U (h)$ is bounded away from infinity. We thus have 
        \begin{align}
            \label{eq:shatstilde}
            \sup_{P \in \mathbb{P}} \left( 
                \max_{1 \leq i \leq N} \bigg\lvert 
                    \frac{\hat {s}_i^U (h)}{\tilde s_i^U(h)}
                    - 1 
                \bigg\rvert  \geq C 
            \left(
                \frac{r_{\theta, N}}{\iota_N \wedge \min_{1 \leq i \leq N} \sigma_i}  
                    \right)
            \right) =o(1).
        \end{align}

        Lastly, we consider 
        \begin{align*}
            \widehat{D}_i^U(h) - \tilde D_i^U(h) = & \left( \frac{ \tilde s_i^U(h)}{\hat{s}_i^U(h)} - 1 \right) D_i(h)
            + 
            \frac{\frac{1}{\sqrt{T}} \sum_{t = 1}^T \left(\hat{d}_{it}^U (h) - d_{it}^U (h) \right)/(\sigma_i \norm{\delta_i(h)})}{\hat{s}_{i}^U(h)/(\sigma_i \norm{\delta_i(h)})} = J_1^U + J_2^U.
        \end{align*}
        We bound $J_1$ by writing 
        \begin{align*}
            & \abs{ \left(\frac{ \tilde s_i^U(h)}{\hat{s}_i^U(h)} - 1 \right)\tilde  D_i (h)}  \\
            = & 
            \abs{
                \frac{ \tilde s_i^U(h)}{\hat{s}_i^U(h)} \left(1 - \frac{\hat{s}_i^U (h)}{\tilde s_i^U(h)}\right)
                \frac{\frac{1}{T} \sum_{t = 1}^T d_{it}^U(h) / (\sigma_i\norm{\delta_i(h)})}{ \tilde s_i^U(h) / (\sigma_i \norm{\delta_i(h)})}
            }
            \\ 
            \leq & 
                \frac{ \tilde s_i^U(h)}{\hat{s}_i^U(h)} \abs{1 - \frac{\hat{s}_i^U (h)}{\tilde  s_i^U(h)}}
                \left(  \left\lVert\frac{1}{\sqrt{T}} \sum_{t = 1}^T v_{it} x_i\right\rVert + \frac{1}{\sigma_i \sqrt{T}} \sum_{t=1}^T \norm {x_{it}}^2  \norm{\delta_i(h)} \right)  
                \left(
                \frac{\tilde s_i^U(h)}{\sigma_i \norm{\delta_i (h)}}
                \right)^{-1}
        \end{align*}
        and applying Lemma \ref{lem:fuk-nagaev-dependent} to bound $\norm{T^{-1/2} \sum_{t = 1}^T v_{it} x_{it}}$, \eqref{eq:shatstilde} to bound $\abs{1 - \hat{s}_i (h)/ \tilde s_i(h)}$, the discussion above \eqref{eq:shatstilde} for a lower bound on $\tilde s_i^U(h)/(\sigma_i \norm{\delta_i (h)})$ and 
        \begin{align*}
            \frac{\hat{s}_i^U(h)}{\tilde s_i^U(h)}
            \geq 1 - \abs{\frac{\hat{s}_i^U(h)}{\tilde s_i^U(h)} - 1}
        \end{align*}
        in conjunction with \eqref{eq:shatstilde} to derive a lower bound on $\hat{s}_i^U(h)/\tilde s_i^U(h)$.
        To bound $J^U_2$, we derive a lower bound on its denominator from 
        \begin{align*}
            \frac{\hat{s}_i^U (h)}{\sigma_i \norm{\delta_i(h)}} = 
            \frac{\tilde s_i^U (h)}{\sigma_i \norm{\delta_i(h)}} \left \{ 
            \left( \frac{\hat{s}_i^U(h)}{\tilde s_i^U(h)} - 1\right) + 1
            \right\}
        \end{align*}
        in conjunction with \eqref{eq:shatstilde} and the lower bound on $\tilde s_i^U(h)/(\sigma_i \norm{\delta_i (h)})$ discussed above.
        The numerator in $J_2^U$ is bounded by the argument used to prove \eqref{eq:duhat_diff_bound}.
        The bounds on $J_1^U$ and $J_2^U$ yield \eqref{eq:DhatU_conv}
    \end{proof}

\section{\label{appendix:additional application results}Additional results for empirical application}

\subsection{One-step confidence set}
\begin{center}
    \scriptsize 
    
\begin{tabular}{lrrrlrrl}
\toprule
\multicolumn{2}{c}{ } & \multicolumn{3}{c}{baseline} & \multicolumn{3}{c}{SNS} \\
\cmidrule(l{3pt}r{3pt}){3-5} \cmidrule(l{3pt}r{3pt}){6-8}
State & $\hat{g}_i$ & p-val $\hat{g}_i$ & card & CS & p-val $\hat{g}_i$ & card & CS\\
\midrule
Alabama & 1 & 0.000 & 1 & 1 & 0.000 & 1 & 1\\
Alaska & 3 & 1.000 & 3 & 2, 3, 4 & 1.000 & 3 & 2, 3, 4\\
Arizona & 3 & 1.000 & 3 & 2, 3, 4 & 1.000 & 3 & 2, 3, 4\\
Arkansas & 2 & 0.027 & 1 & 2 & 0.041 & 1 & 2\\
California & 3 & 0.000 & 1 & 3 & 0.000 & 1 & 3\\
\addlinespace
Colorado & 4 & 0.893 & 2 & 3, 4 & 1.000 & 2 & 3, 4\\
Connecticut & 3 & 0.631 & 2 & 2, 3 & 0.715 & 2 & 2, 3\\
Delaware & 2 & 1.000 & 2 & 2, 3 & 1.000 & 2 & 2, 3\\
D.C. & 2 & 1.000 & 4 & 1, 2, 3, 4 & 1.000 & 4 & 1, 2, 3, 4\\
Florida & 4 & 0.049 & 1 & 4 & 0.074 & 2 & 3, 4\\
\addlinespace
Georgia & 1 & 0.000 & 1 & 1 & 0.000 & 1 & 1\\
Hawaii & 2 & 1.000 & 4 & 1, 2, 3, 4 & 1.000 & 4 & 1, 2, 3, 4\\
Idaho & 3 & 0.664 & 2 & 3, 4 & 0.852 & 2 & 3, 4\\
Illinois & 3 & 0.003 & 1 & 3 & 0.004 & 1 & 3\\
Indiana & 3 & 0.024 & 1 & 3 & 0.038 & 1 & 3\\
\addlinespace
Iowa & 4 & 0.000 & 1 & 4 & 0.000 & 1 & 4\\
Kansas & 4 & 0.010 & 1 & 4 & 0.015 & 1 & 4\\
Kentucky & 3 & 0.093 & 2 & 3, 4 & 0.151 & 3 & 2, 3, 4\\
Louisiana & 1 & 0.000 & 1 & 1 & 0.000 & 1 & 1\\
Maine & 2 & 0.015 & 1 & 2 & 0.023 & 1 & 2\\
\addlinespace
Maryland & 2 & 0.000 & 1 & 2 & 0.000 & 1 & 2\\
Massachusetts & 2 & 1.000 & 2 & 2, 3 & 1.000 & 2 & 2, 3\\
Michigan & 2 & 0.001 & 1 & 2 & 0.001 & 1 & 2\\
Minnesota & 2 & 0.000 & 1 & 2 & 0.000 & 1 & 2\\
Mississippi & 1 & 0.001 & 1 & 1 & 0.001 & 1 & 1\\
\addlinespace
Missouri & 2 & 0.002 & 1 & 2 & 0.003 & 1 & 2\\
Montana & 3 & 1.000 & 3 & 2, 3, 4 & 1.000 & 3 & 2, 3, 4\\
Nebraska & 4 & 1.000 & 3 & 2, 3, 4 & 1.000 & 3 & 2, 3, 4\\
Nevada & 2 & 1.000 & 4 & 1, 2, 3, 4 & 1.000 & 4 & 1, 2, 3, 4\\
New Hampshire & 3 & 0.080 & 2 & 3, 4 & 0.130 & 2 & 3, 4\\
\addlinespace
New Jersey & 2 & 1.000 & 2 & 2, 3 & 1.000 & 2 & 2, 3\\
New Mexico & 3 & 0.094 & 3 & 2, 3, 4 & 0.121 & 3 & 2, 3, 4\\
New York & 2 & 0.000 & 1 & 2 & 0.000 & 1 & 2\\
North Carolina & 2 & 0.000 & 1 & 2 & 0.000 & 1 & 2\\
North Dakota & 4 & 0.010 & 1 & 4 & 0.015 & 1 & 4\\
\addlinespace
Ohio & 1 & 0.000 & 1 & 1 & 0.000 & 1 & 1\\
Oklahoma & 3 & 0.168 & 2 & 2, 3 & 0.207 & 2 & 2, 3\\
Oregon & 4 & 0.000 & 1 & 4 & 0.000 & 1 & 4\\
Pennsylvania & 3 & 0.021 & 1 & 3 & 0.030 & 1 & 3\\
Rhode Island & 2 & 0.000 & 1 & 2 & 0.001 & 1 & 2\\
\addlinespace
South Carolina & 1 & 0.000 & 1 & 1 & 0.000 & 1 & 1\\
South Dakota & 4 & 1.000 & 3 & 2, 3, 4 & 1.000 & 3 & 2, 3, 4\\
Tennessee & 2 & 0.000 & 1 & 2 & 0.000 & 1 & 2\\
Texas & 1 & 0.000 & 1 & 1 & 0.000 & 1 & 1\\
Utah & 4 & 1.000 & 2 & 3, 4 & 1.000 & 2 & 3, 4\\
\addlinespace
Vermont & 4 & 0.000 & 1 & 4 & 0.000 & 1 & 4\\
Virginia & 2 & 1.000 & 2 & 2, 3 & 1.000 & 2 & 2, 3\\
Washington & 4 & 0.000 & 1 & 4 & 0.000 & 1 & 4\\
West Virginia & 3 & 0.014 & 1 & 3 & 0.041 & 1 & 3\\
Wisconsin & 3 & 0.392 & 2 & 2, 3 & 0.455 & 2 & 2, 3\\
\addlinespace
Wyoming & 3 & 1.000 & 3 & 2, 3, 4 & 1.000 & 3 & 2, 3, 4\\
\bottomrule
\end{tabular}
    \captionof{table}{Marginal confidence set at level $1 - \alpha = 0.95$. ``p-val $\hat{g}_i$'' is the p-value for the significance of the estimated group membership. ``CS cardinality'' is the cardinality of the marginal confidence set for the state.  ``CS'' is the marginal confidence set. ``Baseline'' refers to the procedure with critical values defined in Section~\ref*{sec:crit val}. ``SNS'' refers to the procedure with critical values defined in Section~\ref*{sec:sns crit val}.}
    \label{tab:supp:cs_1step}
\end{center}

\subsection{\label{app:application 2step}Two-step confidence set assuming no serial correlation}

Under the assumption of no serial correlation, we can use the non-HAC variance estimator from Section~\ref*{sec:no serial correlation} in the main paper and eliminate nine units in the first step.  

\begin{center}
    \scriptsize 
    \input{\Routputpath/figures/figureB2.tex}
    \captionof{table}{Two-step procedure assuming no serial correlation. Second-step $p$-values for the significance of the estimated group memberships with and without unit selection ($\alpha = 0.05$, $\beta = 0.01$). The dashed horizontal line indicates the threshold for significance without unit selection. The solid horizontal line indicates the threshold for significance with unit selection.}
    \label{fig:supp:2step}
\end{center}
As illustrated in Figure~\ref{fig:supp:2step}, the elimination at the first stage decreases the second-step $p$-values since the Bonferroni adjustment is over a smaller number of simultaneous tests. On the other hand, turning on unit selection lowers the threshold $p$-value at which we can conclude significance from $\alpha = 0.05$ to $\alpha - 2 \beta = 0.03$. In this example, the two-step procedure reduces $p$-values in the second step but does not produce a smaller confidence set.

\section{\label{appendix:additional simulation results}Additional simulation results}

\subsection{\label{appendix:epsilon_N}Choice of the regularization sequence $\epsilon_N$}

\subsubsection{Benchmark simulation designs from Section~\ref*{sec:simulations}}

For our simulation results in Section~\ref*{sec:simulations} in the main text, we set the regularization sequence $\epsilon_N$ equal to the constant sequence $\epsilon_N = 0.01$. In this appendix, we investigate the robustness of our simulation results to this choice of regularization sequence. 

The simulation design is identical to the specification simulated in Section~\ref*{sec:simulations}. For all simulation results presented in this section, we estimate the group-specific model parameters by the \emph{kmeans} estimator and use the HAC-type estimator of the long-run variance with data-driven bandwidth. Simulation results are based on 500 replications.

\label{QE:online epsilon sim}
We first consider constant sequences $\epsilon_N = 0, 0.01, 0.05$. Here, $\epsilon_N = 0.01$ is the value used in the simulations in the main text, and $\epsilon_N = 0$ turns off the regularization of the variance matrix.
Table~\ref{tab:epsilon_N constant} reports the simulation results.

\begin{center}
  \scriptsize 
  
\begin{tabular}{rrrrrrrrr}
\toprule
\multicolumn{3}{c}{ } & \multicolumn{3}{c}{coverage} & \multicolumn{3}{c}{average cardinality} \\
\cmidrule(l{3pt}r{3pt}){4-6} \cmidrule(l{3pt}r{3pt}){7-9}
$\rho$ & $N$ & $T$ & $\epsilon = 0$ & $\epsilon = 0.01$ & $\epsilon = 0.05$ & $\epsilon = 0$ & $\epsilon = 0.01$ & $\epsilon = 0.05$\\
\midrule
 &  & 60 & 0.92 & 0.92 & 0.93 & 1.94 & 1.97 & 2.00\\

 & \multirow[t]{-2}{*}{\raggedleft\arraybackslash 50} & 120 & 1.00 & 1.00 & 1.00 & 1.14 & 1.16 & 1.18\\

 &  & 60 & 0.95 & 0.95 & 0.95 & 2.10 & 2.13 & 2.19\\

 & \multirow[t]{-2}{*}{\raggedleft\arraybackslash 100} & 120 & 0.99 & 1.00 & 1.00 & 1.19 & 1.21 & 1.26\\

 &  & 60 & 0.97 & 0.96 & 0.95 & 2.28 & 2.30 & 2.37\\

\multirow[t]{-6}{*}{\raggedleft\arraybackslash 0.0} & \multirow[t]{-2}{*}{\raggedleft\arraybackslash 200} & 120 & 1.00 & 0.98 & 0.99 & 1.27 & 1.28 & 1.34\\
\cmidrule{1-9}
 &  & 60 & 0.69 & 0.69 & 0.66 & 3.39 & 3.42 & 3.49\\

 & \multirow[t]{-2}{*}{\raggedleft\arraybackslash 50} & 120 & 0.98 & 0.98 & 0.98 & 3.67 & 3.72 & 3.78\\

 &  & 60 & 0.67 & 0.69 & 0.73 & 3.60 & 3.62 & 3.66\\

 & \multirow[t]{-2}{*}{\raggedleft\arraybackslash 100} & 120 & 0.97 & 0.98 & 0.97 & 3.83 & 3.84 & 3.85\\

 &  & 60 & 0.68 & 0.70 & 0.68 & 3.73 & 3.74 & 3.76\\

\multirow[t]{-6}{*}{\raggedleft\arraybackslash 0.5} & \multirow[t]{-2}{*}{\raggedleft\arraybackslash 200} & 120 & 0.97 & 0.97 & 0.98 & 3.86 & 3.87 & 3.87\\
\bottomrule
\end{tabular}
  \captionof{table}{Simulation results for $\epsilon_N = 0, 0.01, 0.05$. Nominal level $1 - \alpha = 0.95$.  ``coverage'' is the empirical coverage probability of the joint confidence set. ``cardinality'' is the expected average (over all units) cardinality of the marginal unit-wise confidence sets.}
  \label{tab:epsilon_N constant}
\end{center}

The results are not very sensitive to the choice of $\epsilon_N$. Notably, the performance of our method is not affected substantially by turning off regularization completely (i.e., choosing $\epsilon_N = 0$). In the next section, we show that regularization plays a greater, though still limited, role in an alternative design that is tailored to make regularization relevant.

We also simulate vanishing sequences $\epsilon_N = \log^{-2} N, \log^{-3/2} N$. These sequences satisfy the rate condition imposed in Theorem~\ref*{thm:MAX-dep-lv}. Table~\ref{tab:epsilon_N log_N} reports the simulation results.

\begin{center}
  \scriptsize 
  
\begin{tabular}{rrrrrrrrr}
\toprule
\multicolumn{3}{c}{ } & \multicolumn{2}{c}{$\epsilon$} & \multicolumn{2}{c}{coverage} & \multicolumn{2}{c}{average cardinality} \\
\cmidrule(l{3pt}r{3pt}){4-5} \cmidrule(l{3pt}r{3pt}){6-7} \cmidrule(l{3pt}r{3pt}){8-9}
$\rho$ & $N$ & $T$ & $\epsilon = \log^{-2} N$ & $\epsilon = \log^{-3/2} N$ & $\epsilon = \log^{-2} N$ & $\epsilon = \log^{-3/2} N$ & $\epsilon = \log^{-2} N$ & $\epsilon = \log^{-3/2} N$\\
\midrule
 &  & 60 & 0.07 & 0.13 & 0.92 & 0.93 & 2.03 & 2.04\\

 & \multirow[t]{-2}{*}{\raggedleft\arraybackslash 50} & 120 & 0.07 & 0.13 & 1.00 & 0.99 & 1.19 & 1.18\\

 &  & 60 & 0.05 & 0.10 & 0.96 & 0.96 & 2.19 & 2.25\\

 & \multirow[t]{-2}{*}{\raggedleft\arraybackslash 100} & 120 & 0.05 & 0.10 & 1.00 & 1.00 & 1.24 & 1.26\\

 &  & 60 & 0.04 & 0.08 & 0.95 & 0.97 & 2.37 & 2.44\\

\multirow[t]{-6}{*}{\raggedleft\arraybackslash 0.0} & \multirow[t]{-2}{*}{\raggedleft\arraybackslash 200} & 120 & 0.04 & 0.08 & 1.00 & 1.00 & 1.34 & 1.35\\
\cmidrule{1-9}
 &  & 60 & 0.07 & 0.13 & 0.72 & 0.74 & 3.49 & 3.53\\

 & \multirow[t]{-2}{*}{\raggedleft\arraybackslash 50} & 120 & 0.07 & 0.13 & 0.98 & 0.97 & 3.79 & 3.80\\

 &  & 60 & 0.05 & 0.10 & 0.75 & 0.71 & 3.67 & 3.68\\

 & \multirow[t]{-2}{*}{\raggedleft\arraybackslash 100} & 120 & 0.05 & 0.10 & 0.99 & 0.98 & 3.85 & 3.86\\

 &  & 60 & 0.04 & 0.08 & 0.67 & 0.68 & 3.75 & 3.77\\

\multirow[t]{-6}{*}{\raggedleft\arraybackslash 0.5} & \multirow[t]{-2}{*}{\raggedleft\arraybackslash 200} & 120 & 0.04 & 0.08 & 0.98 & 0.98 & 3.87 & 3.87\\
\bottomrule
\end{tabular}
  \captionof{table}{Simulation results for $\epsilon_N = \log^{-2} N, \log^{-3/2} N$. Nominal level $1 - \alpha = 0.95$.  ``coverage'' is the empirical coverage probability of the joint confidence set. ``cardinality'' is the expected average (over all units) cardinality of the marginal unit-wise confidence sets.}
  \label{tab:epsilon_N log_N}
\end{center}

Again, we find that the simulation results are not sensitive to the choice of regularization sequence.

\subsubsection{\label{sec appendix:regularization matters}A design where regularization matters} 

In our benchmark designs, the choice of regularization parameter hardly affects the performance of our procedure, raising the question whether regularization is indeed necessary. It seems possible that regularization may be a purely technical device to facilitate the mathematical proof of the validity of our procedure, but that it may not have any practical relevance. 

We address this concern by presenting an alternative simulation design where regularization affects the finite-sample performance of our procedure.

The design is very stylized and exhibits close-to-perfect correlations among the moment inequalities. For such correlations, our comparison bound relies on regularization to bound the estimation error in the critical values (see proof of Lemma~\ref{lem:comparison_bound}).

Similar to the simulation designs in Section~\ref*{sec:simulations}, the data generating process is given by 
\begin{align*}
	\widetilde{\texttt{lemp}}_{it} = \theta_{g^0_i, 1} \widetilde{\texttt{lmw}}_{it} + \theta_{g^0_i, 2} \widetilde{\texttt{lpop}}_{it} + \theta_{g^0_i, 3}\widetilde{\texttt{lemp}}_{it}^{\text{TOT}} + \sigma_i v_{it}
\end{align*}
for $i = 1, \dotsc, N$ and $t = 1, \dotsc, T$. 
We simplify the generating process of the covariates and obtain $x_{it} = (\widetilde{\texttt{lmw}}_{it}, \widetilde{\texttt{lpop}}_{it}, \widetilde{\texttt{lemp}}_{it}^{\text{TOT}})$ by sampling independently three times from the empirical distribution of $\widetilde{\texttt{lmw}}_{it}$ observed in the data for our application. This guarantees that the components of $x_{it}$ have identical and independent marginal distributions which makes it easier to parameterize the correlation structure of the moment inequalities with the parameter $\kappa$ below.

For group $g = 1$, we set the group-specific coefficient $\theta_1 = (\theta_{1, 1}, \theta_{1, 2}, \theta_{1, 3})$ equal to $(0.5, 0.5, 0.5)$. For the remaining groups, the coefficients are a convex combination of a design with parallel groups and a design with orthogonal groups. For $g=2, 3, 4$, the coefficients with parallel groups are $\theta_g^{\text{parallel}} = c_g \theta_1$, with $c_2 = 0.7$, $c_3 = 0.4$ and $c_4 = 0.1$. The coefficients with orthogonal groups are $\theta_2^{\text{orthogonal}} = (0.5, 0, 0)$, $\theta_3^{\text{orthogonal}} = (0, 0.5, 0)$, $\theta_4^{\text{orthogonal}} = (0, 0, 0.5)$.  
For $g=2, 3, 4$, the group-specific coefficients are given by the convex combination $\theta_g = (1 - \kappa) \theta_g^{\text{parallel}} + \kappa \theta_g^{\text{orthogonal}}$, where $\kappa = 0, 0.05, 0.1, 0.2$.
For $\kappa = 0$, groups are parallel, and all off-diagonal entries of the population correlation matrix $\Omega_i(g)$ are perfect correlations. For $\kappa = 1$, groups are orthogonal, and the matrix $\Omega_i(g)$ is diagonal.

As in the designs in Section~\ref*{sec:simulations}, each unit $i$ is assigned to one of the four groups with equal probability and exhibits a random heteroscedasticity parameter $\sigma_i = 0.1 \times \chi^2(4)/4$, where $\chi^2 (df)$ is a random draw from a $\chi^2$-distribution with $df$ degrees of freedom. 

To establish a conjecture about the role of regularization in this design, we briefly review where regularization enters our theoretical arguments. 
Regularization is part of our strategy to control estimation errors in the critical values. This estimation error comes from the fact that the group-specific critical values are estimated from data. It is, therefore, a greater concern in small panels ($T$ and $N$ small) than in large panels ($T$ or $N$ large).
Consider a positive off-diagonal entry in $\Omega_i(g_i^0)$. 
From the proof of Lemma~\ref{lem:comparison_bound}, it is apparent that estimation error is easily controlled if the entry is bounded away from unity. We apply an argument that relies on our regularization scheme to control estimation error if the entry is close to unity. 
In summary, regularization is expected to be relevant if $T$ and/or $N$ are small and $\kappa$ is small.

\begin{center}
  \scriptsize 
  
\begin{tabular}{rrrrrrrrrrr}
\toprule
\multicolumn{3}{c}{ } & \multicolumn{4}{c}{coverage} & \multicolumn{4}{c}{average cardinality} \\
\cmidrule(l{3pt}r{3pt}){4-7} \cmidrule(l{3pt}r{3pt}){8-11}
\multicolumn{3}{c}{ } & \multicolumn{3}{c}{MVT} & \multicolumn{1}{c}{SNS} & \multicolumn{3}{c}{MVT} & \multicolumn{1}{c}{SNS} \\
\cmidrule(l{3pt}r{3pt}){4-6} \cmidrule(l{3pt}r{3pt}){7-7} \cmidrule(l{3pt}r{3pt}){8-10} \cmidrule(l{3pt}r{3pt}){11-11}
$\kappa$ & $N$ & $T$ & $\epsilon = 0$ & $\epsilon = 0.01$ & $\epsilon = 0.05$ &  & $\epsilon = 0$ & $\epsilon = 0.01$ & $\epsilon = 0.05$ & \\
\midrule
 &  & 60 & 0.82 & 0.80 & 0.84 & 0.86 & 2.19 & 2.18 & 2.22 & 2.30\\

 & \multirow[t]{-2}{*}{\raggedleft\arraybackslash 50} & 120 & 0.91 & 0.94 & 0.94 & 0.95 & 1.66 & 1.65 & 1.69 & 1.74\\

 &  & 60 & 0.92 & 0.91 & 0.93 & 0.95 & 2.32 & 2.32 & 2.36 & 2.47\\

\multirow[t]{-4}{*}{\raggedleft\arraybackslash 0.0} & \multirow[t]{-2}{*}{\raggedleft\arraybackslash 100} & 120 & 0.93 & 0.96 & 0.95 & 0.96 & 1.72 & 1.73 & 1.77 & 1.83\\
\cmidrule{1-11}
 &  & 60 & 0.69 & 0.74 & 0.73 & 0.75 & 2.32 & 2.32 & 2.35 & 2.42\\

 & \multirow[t]{-2}{*}{\raggedleft\arraybackslash 50} & 120 & 0.86 & 0.87 & 0.89 & 0.91 & 1.80 & 1.79 & 1.80 & 1.85\\

 &  & 60 & 0.84 & 0.82 & 0.85 & 0.87 & 2.46 & 2.46 & 2.50 & 2.58\\

\multirow[t]{-4}{*}{\raggedleft\arraybackslash 0.1} & \multirow[t]{-2}{*}{\raggedleft\arraybackslash 100} & 120 & 0.93 & 0.92 & 0.92 & 0.95 & 1.87 & 1.87 & 1.89 & 1.95\\
\cmidrule{1-11}
 &  & 60 & 0.65 & 0.62 & 0.66 & 0.68 & 2.42 & 2.42 & 2.43 & 2.50\\

 & \multirow[t]{-2}{*}{\raggedleft\arraybackslash 50} & 120 & 0.87 & 0.86 & 0.85 & 0.88 & 1.88 & 1.90 & 1.89 & 1.93\\

 &  & 60 & 0.78 & 0.78 & 0.76 & 0.81 & 2.59 & 2.60 & 2.59 & 2.67\\

\multirow[t]{-4}{*}{\raggedleft\arraybackslash 0.2} & \multirow[t]{-2}{*}{\raggedleft\arraybackslash 100} & 120 & 0.91 & 0.90 & 0.91 & 0.94 & 1.99 & 1.98 & 1.99 & 2.03\\
\bottomrule
\end{tabular}
  \captionof{table}{\label{tab:simulations EPS} Simulation results for a stylized design with strongly correlated moment inequalities. Nominal level $1 - \alpha = 0.95$. ``coverage'' is the empirical coverage probability of the joint confidence set. ``cardinality'' is the expected average (over all units) cardinality of the marginal unit-wise confidence sets. MVT = use MVT critical values. SNS = use SNS critical values.}
  \label{tab:simulations regularization}
\end{center}

This conjecture is confirmed by the simulation results in Table~\ref{tab:simulations regularization} For $\kappa = 0, 0.1$, regularization improves size control in the designs with small samples. In particular, we see improvements if $N = 50$. For $\kappa = 0.2$, regularization leads to slightly worse size control. We interpret this as a sign that for $\kappa = 0.2$, the cost of regularization in terms of a biased variance estimator is not outweighed by the benefit of guarding against underestimating close-to-perfect positive correlations. 

Simulation designs that investigate the role of regularization are by necessity designs with substantial sampling error in the the group-specific coefficients. Without sampling error, there is no uncertainty about the critical values and regularization is not needed. The overall noisiness that makes the designs presented here informative about regularization also affects the performance of our procedure directly, leading to a confidence set that is underpowered independently of imprecisely estimated critical values. This can be seen by comparing the performance of the regularized procedure with MVT critical values to the procedure with data-independent SNS critical values. The coverage under SNS critical values provides an upper bound on the coverage that can be achieved by eliminating estimation error in the critical values, i.e., an upper bound on what better regularization can achieve. This has to be considered when interpreting the improvements in size control from regularization. For example, for $\kappa = 0.1$, $N = 50$, and $T = 60$, regularization improves the size by about five percentage points, bringing the size within a percentage point of the size under SNS critical values.

\label{QE:regularization matters summary}
The simulation results offer some evidence that the theoretical considerations that motivate our regularization approach have practical relevance. This suggests that it may not be possible to rigorously justify a version of our procedure that does not use regularization. On the other hand, even in this highly stylized design, gains from regularization are limited. From a practical perspective, correct regularization may not be a key concern.

\subsection{\label{sec:testing ghat}Testing the estimated group membership $\hat{g}_i$}

In the definition of $\widehat{C}_{\alpha, N, i}$, we explicitly add $\hat{g}_i$ to the confidence set. Not doing this changes the marginal confidence set of unit $i$ only if $\hat{g}_i$ is not already included anyway, i.e., if 
\begin{align}
  \label{eq:ghat not rejected}
  \widehat{T}_i \left(\hat{g}_i\right) > \hat{c}_{\alpha, N, i} \left(\hat{g}_i\right). 
\end{align}
We simulate the finite sample probability of this happening in our simulation designs from Section~\ref*{sec:simulations} in the main text. The simulation results are summarized in the following table.

\begin{center}
  \scriptsize
  
\begin{tabular}{rrrrrrr}
\toprule
$\rho$ & $\sigma$ & $N$ & $T$ & coverage & cardinality & $\hat{g}_i$ not rej\\
\midrule
 &  &  & 60 & 0.99 & 1.60 & 0.99\\

 &  & \multirow[t]{-2}{*}{\raggedleft\arraybackslash 50} & 120 & 1.00 & 1.07 & 0.99\\

 &  &  & 60 & 0.99 & 1.76 & 1.00\\

 &  & \multirow[t]{-2}{*}{\raggedleft\arraybackslash 100} & 120 & 1.00 & 1.12 & 1.00\\

 &  &  & 60 & 0.99 & 1.95 & 1.00\\

 & \multirow[t]{-6}{*}{\raggedleft\arraybackslash 0.1} & \multirow[t]{-2}{*}{\raggedleft\arraybackslash 200} & 120 & 1.00 & 1.17 & 1.00\\
\cmidrule{2-7}
 &  &  & 60 & 0.92 & 1.97 & 0.99\\

 &  & \multirow[t]{-2}{*}{\raggedleft\arraybackslash 50} & 120 & 1.00 & 1.16 & 0.99\\

 &  &  & 60 & 0.95 & 2.13 & 0.99\\

 &  & \multirow[t]{-2}{*}{\raggedleft\arraybackslash 100} & 120 & 1.00 & 1.21 & 1.00\\

 &  &  & 60 & 0.96 & 2.30 & 1.00\\

\multirow[t]{-12}{*}{\raggedleft\arraybackslash 0.0} & \multirow[t]{-6}{*}{\raggedleft\arraybackslash 0.2} & \multirow[t]{-2}{*}{\raggedleft\arraybackslash 200} & 120 & 0.98 & 1.28 & 1.00\\
\cmidrule{1-7}
 &  &  & 60 & 0.87 & 3.35 & 0.99\\

 &  & \multirow[t]{-2}{*}{\raggedleft\arraybackslash 50} & 120 & 0.99 & 3.71 & 1.00\\

 &  &  & 60 & 0.86 & 3.52 & 1.00\\

 &  & \multirow[t]{-2}{*}{\raggedleft\arraybackslash 100} & 120 & 0.99 & 3.80 & 1.00\\

 &  &  & 60 & 0.82 & 3.67 & 1.00\\

 & \multirow[t]{-6}{*}{\raggedleft\arraybackslash 0.1} & \multirow[t]{-2}{*}{\raggedleft\arraybackslash 200} & 120 & 1.00 & 3.84 & 1.00\\
\cmidrule{2-7}
 &  &  & 60 & 0.69 & 3.42 & 0.99\\

 &  & \multirow[t]{-2}{*}{\raggedleft\arraybackslash 50} & 120 & 0.98 & 3.72 & 1.00\\

 &  &  & 60 & 0.69 & 3.62 & 0.99\\

 &  & \multirow[t]{-2}{*}{\raggedleft\arraybackslash 100} & 120 & 0.98 & 3.84 & 1.00\\

 &  &  & 60 & 0.70 & 3.74 & 1.00\\

\multirow[t]{-12}{*}{\raggedleft\arraybackslash 0.5} & \multirow[t]{-6}{*}{\raggedleft\arraybackslash 0.2} & \multirow[t]{-2}{*}{\raggedleft\arraybackslash 200} & 120 & 0.97 & 3.87 & 1.00\\
\bottomrule
\end{tabular}
  \captionof{table}{Simulated probability of the event \eqref{eq:ghat not rejected}.}
  \label{tab:simulations include ghat}
\end{center}

In Table~\ref{tab:simulations include ghat}, the column ``$\hat{g}_i$ not rej'' gives the simulated probability of our group membership not rejecting the estimated group membership (i.e., one minus the probability of the event defined in equation \eqref{eq:ghat not rejected}). We find that our test for group membership does not reject the estimated group membership with probability close to, but not equal to, one. 

\subsection{\label{appendix:simulations two-step}Two-step procedure}

In this appendix, we report simulation results regarding the finite-sample performance of our two-step procedure. 

We simulate a design with independent time periods. Like our main design in Section~\ref*{sec:simulations} in the main text, the design studied here builds on the model estimated in Section~\ref*{sec:application} in the main text. A unit $i$ corresponds to a US state and the ``time periods'' are given by observations of different counties in different quarters. The panel model is specified as in equation (\ref*{eq:application FE model}) in the main text, with coefficients equal to the estimate coefficients in Table~\ref*{tab:app:slope_coef} in the main text. 
The joint distribution of the regressors $\widetilde{\texttt{lmw}}_{it}$, $\widetilde{\texttt{lpop}}_{it}$ and $\widetilde{\texttt{lemp}}_{it}^{\text{TOT}}$ is defined from the data used in our empirical application. 
In particular, $\widetilde{\texttt{lmw}}_{it}$, $\widetilde{\texttt{lpop}}_{it}$ and $\widetilde{\texttt{lemp}}_{it}^{\text{TOT}}$ are sampled from the pooled empirical distribution of the respective fixed-effect transformations of $\log(\texttt{mw}_{it})$, $\log(\texttt{pop}_{it})$ and $\log(\texttt{emp}_{it}^{\text{TOT}})$. 
The error component $v_{it}$ is a standard normal noise term.

We set the distribution of heteroscedasticity and group membership, i.e., to the joint distribution of $(\sigma_i, T_i, g_i^0)$ so that the simulation results reveal different aspects of the performances of the two-step procedure. We note that the two-step procedure is sensitive to this distribution. We determine it from the data by mapping each simulated unit $i$ to one of the $N = 51$ units from our empirical application. We set $\sigma_i$ equal to $m_{\sigma}$ times the standard deviation of the empirical residuals for unit $i$, $g_i^0$ equal to the estimated group membership of $i$ and $T_i$ equal to the number of observed ``time periods'' for unit $i$ (i.e. counties times quarters). The parameter $m_{\sigma} = 1/4, 1, 4$ shifts the global level of uncertainty. 

The other parameters for the simulations are set as follows. The nominal level of the simulated joint confidence set is $1- \alpha = 0.95$. We simulate different values of the first-step parameter $\beta = \alpha/ 5, \alpha / 10 = 0.01, 0.005$. The regularization sequence is specified as $\epsilon_N = 0.01$. 
We simulate the confidence set using our benchmark critical values defined in Section~\ref*{sec:crit val} of the main text (labelled MVT = multi-variate $t$-distribution), as well as the SNS critical values defined in Section~\ref*{sec:sns crit val} of the main text (labelled SNS).

\begin{table}[h]
\scriptsize
\centering 

\begin{tabular}{lrrrrrrrrrrrr}
\toprule
\multicolumn{3}{c}{ } & \multicolumn{2}{c}{success} & \multicolumn{2}{c}{failure} & \multicolumn{2}{c}{card with sel} & \multicolumn{2}{c}{card without sel} & \multicolumn{2}{c}{ } \\
\cmidrule(l{3pt}r{3pt}){4-5} \cmidrule(l{3pt}r{3pt}){6-7} \cmidrule(l{3pt}r{3pt}){8-9} \cmidrule(l{3pt}r{3pt}){10-11}
 & $m_{\sigma}$ & $\alpha / \beta$ & insignif & signif & insignif & signif & insignif & signif & insignif & signif & $\hat{N}$ & coverage\\
\midrule
\addlinespace[0.3em]
\multicolumn{13}{l}{\textbf{MVT}}\\
\hspace{1em} & 0.25 & 10 & 0.55 & 0 & 0.00 & 0.00 & 1.48 & 1.00 & 2.01 & 1 & 10.06 & 1.00\\

\hspace{1em} &  & 5 & 0.48 & 0 & 0.00 & 0.00 & 1.55 & 1.00 & 2.01 & 1 & 9.29 & 1.00\\

\hspace{1em} & 1.00 & 10 & 0.13 & 0 & 0.00 & 0.00 & 2.15 & 1.00 & 2.21 & 1 & 33.73 & 0.99\\

\hspace{1em} &  & 5 & 0.00 & 0 & 0.01 & 0.03 & 2.22 & 1.00 & 2.21 & 1 & 32.35 & 1.00\\

\hspace{1em} & 4.00 & 10 & 0.00 & 0 & 0.50 & 0.45 & 2.75 & 1.02 & 2.72 & 1 & 50.95 & 0.96\\

\hspace{1em} &  & 5 & 0.00 & 0 & 0.80 & 0.77 & 2.79 & 1.05 & 2.72 & 1 & 50.90 & 0.97\\

\addlinespace[0.3em]
\multicolumn{13}{l}{\textbf{SNS}}\\
\hspace{1em} & 0.25 & 10 & 0.53 & 0 & 0.00 & 0.00 & 1.51 & 1.00 & 2.00 & 1 & 10.04 & 1.00\\

\hspace{1em} &  & 5 & 0.46 & 0 & 0.00 & 0.00 & 1.58 & 1.00 & 2.00 & 1 & 9.32 & 1.00\\

\hspace{1em} & 1.00 & 10 & 0.13 & 0 & 0.00 & 0.00 & 2.21 & 1.00 & 2.27 & 1 & 33.73 & 1.00\\

\hspace{1em} &  & 5 & 0.00 & 0 & 0.01 & 0.04 & 2.27 & 1.00 & 2.27 & 1 & 32.42 & 1.00\\

\hspace{1em} & 4.00 & 10 & 0.00 & 0 & 0.52 & 0.44 & 2.78 & 1.02 & 2.75 & 1 & 50.95 & 0.98\\

\hspace{1em} &  & 5 & 0.00 & 0 & 0.81 & 0.73 & 2.82 & 1.05 & 2.75 & 1 & 50.92 & 0.98\\
\bottomrule
\end{tabular}
\caption{Simulation results for the two-step procedures (unit selection).}
\label{tab:sim:step2}
\end{table}

The simulation results are based on 1000 replications and reported in Table~\ref{tab:sim:step2}. 
The columns labeled ``insignif'' give averages over units that are insignificant under no unit-selection. Columns labeled ``signif'' give averages over units that are significant under no unit-selection. A unit is labeled as a ``success'' (``failure'') if its marginal confidence set is strictly smaller (strictly larger) under unit-selection than under no unit-selection. 
The columns labelled ``card with sel'' (``card without sel'') give the cardinality of unit-wise marginal confidence sets if unit-selection is turned on (turned off). The column labeled $\hat{N}$ gives the simulated expected number of units that survive unit selection ($N = 51$). ``Coverage'' gives the simulated joint coverage probability of the two-step joint confidence set (nominal level $1 - \alpha = 0.95$).

In all designs, unit selection produces a valid joint confidence set that covers the true group structure at the prescribed nominal level.

Unit selection aims to tighten the marginal confidence sets for units for which estimated group memberships are insignificant under a one-step procedure. Among such units, the expected proportion of units for which a two-step procedure tightens the marginal confidence set varies across the different designs. In the design with low error variances ($m_{\sigma} = 0.25$), this proportion ranges between 46\% and 55\%. This means that the two-step procedure improves the marginal confidence sets for roughly half of the units for which they can be improved. In the design with medium error variances ($m_{\sigma} = 1$), this proportion lies between 0\% and 13\%. In the design with high error variances ($m_{\sigma} = 4$), there are no improvements. This illustrates that the two-step procedures can only be successful if the overall uncertainty is low but unequally distributed across units. If overall uncertainty is high, then the first step cannot deselect units, and hence the second-step confidence sets cannot be tightened.

The two-step procedures can cause the confidence set to become wider if insufficiently many units are eliminated in the first step. This happens in the designs with high error variance ($m_{\sigma} = 4$): hardly any units are eliminated in the first step and the size of the marginal confidence sets increases both for units with significant and units with insignificant group membership estimates under the one-step procedure.

Using MVT instead of SNS critical values increases the power of our two-step procedure. In our designs, both choices of critical values select a similar number of units for the second step. Therefore, the power gain from using MVT critical values is almost entirely due to more efficient testing in the second step. 

\section{\label{appendix:weak separation}Weak group separation}

\subsection{Introduction}

In this appendix, we consider grouped panel models in which groups are only weakly separated. By weak separation, we mean that groups are distinct but very similar to each other. We formalize this notion using an approach inspired by the local alternatives in asymptotic testing theory. In particular, we let the distance between groups shrink to zero at a fixed rate.

We offer new theoretical results on the rate of consistency of the \emph{kmeans} estimator under weak separation. In particular, we give conditions under which the estimated group-specific coefficients converge at the parametric $\sqrt{NT}$-rate if the distance between groups shrinks at a rate slower than $T^{-1/2}$. We then use this result to derive conditions under which our confidence set is valid under weak group separation. 

In addition to the theoretical analysis, we provide simulation studies to investigate the finite sample behavior of the \emph{kmeans} estimator and our joint confidence set under weak separation and to verify our theoretical predictions. 

This appendix is structured as follows. In Section~\ref{sec:weak-sep_no separation} we discuss existing results on \emph{kmeans} estimation in a setting where groups are not separated at all. We then turn to our analysis of the \emph{kmeans} estimation under weak separation. In Section~\ref{sec:weak-sep_theory}, we present asymptotic results. In Section~\ref{sec:weak-sep_simulations}, we present simulation evidence. Proofs are given in Section~\ref{sec:weak-sep_proofs}.

\subsection{No group separation\label{sec:weak-sep_no separation}}
\label{QE:online no group separation}
We first discuss \emph{kmeans} estimation under no group separation. By ``no group separation'' we mean that there are at least two groups with identical coefficients. This corresponds to over-specification of the number of groups. \textcite{bonhomme2015grouped} study this setting in their supplemental appendix. In this setting, the estimators of the group-specific coefficient converge at most at the rate of $T^{-1/2}$. As discussed in the main text of this paper, this rate is too slow to satisfy our conditions for the validity of the joint confidence set.

We consider a simple mean shift model where we observe $y_{it}$, for $i=1,\dots, N$ and $t=1 ,\dots T$. The parameter of interest is the mean of $y_{it}$. We assume that there is a latent group structure with $G$ groups and that the mean of $y_{it}$ may depend on unit $i$'s group membership. Suppose that there is only one distinct group, i.e., all units have the same mean, but we incorrectly set the number of groups to two. Specifically, the estimated model is 
\begin{align*}
	y_{it} = \alpha_{g_i} + v_{it},
\end{align*}
where $g_i=1,2$ and $v_{it}$ is assumed to be $i.i.d.N(0,\sigma^2)$. Let $\hat \alpha_1$ and $\hat \alpha_2$ be the estimators of $\alpha_1$ and $\alpha_2$, respectively, by the \textit{kmeans} method. By relabelling if necessary, we impose $\hat \alpha_1 \geq \hat \alpha_2$. 
The true model is homogeneous such that $\alpha = \alpha_1=\alpha_2$. 

Proposition S.2 of \textcite{bonhomme2015grouped} (supplemental appendix) states that, as $N\to \infty$ with $T$ fixed, it holds that 
\begin{align*}
	\hat \alpha_1 \to \alpha + \sqrt{\frac{2}{\pi T}},  \quad 	\hat \alpha_2 \to \alpha - \sqrt{\frac{2}{\pi T}}.
\end{align*}
We note that the model considered in Proposition S.2 of \textcite{bonhomme2015grouped} includes regressors with common coefficients, but its presence does not affect the probability limits of $\hat \alpha_1$ and $\hat \alpha_2$. 

The above result indicates that, even we take $T\to \infty$ in addition to $N \to \infty$, the convergence rates of $\hat \alpha_1 $ and $\hat \alpha_2$ are at most of order $T^{-1/2}$. In particular, the probability that $P( |\hat\alpha_g - \alpha | > C T^{-1/2})$ for fixed $C$ does not converge to 0.

\subsection{Asymptotic analysis\label{sec:weak-sep_theory}}

We now turn to the setting of weak group separation, where groups are distinct but very similar. The discussion given here is a simplified version of \textcite[][Supplemental Appendix C]{LumsdaineOkuiWang2022}. 

We observe $(y_{it}, x_{it})$ for $i=1,\dots, N$ and $t=1,\dots, T$. Units are divided into $G$ groups, and all members of a group share the same value of the regression coefficient.
The model is
\begin{align*}
 y_{it} = x_{it}'\theta_{g_i^0}^0 + u_{it},
 \end{align*}
where $\theta_{g}^0$, $g=1, \dots, G$, are group-specific coefficients, $g_i^0 \in \{1, \dots, G\}$ is unit $i$'s true group membership, and $u_{it}$ is an error term.

The parameters are estimated by the \textit{kmeans} method \parencite{bonhomme2015grouped}. 
Let $\mathbb{G} =  \{1, \dots, G\}$ be the set of groups. Then, $\mathbb{G}^N$ is the parameter space for the group membership structure. A typical element of $\mathbb{G}$ is $\gamma =(g_1, \dots, g_N)$. The true group membership structure is $\gamma^0 = (g_1^0, \dots, g_N^0) \in \mathbb{G}$. The parameter space for the coefficients is denoted as $\mathcal{B}\subset \mathbb{R}^{Gp}$. The estimator is
\begin{align*}
 (\hat \gamma, \hat \theta) = \arg \min_{\gamma \in \mathbb{G}^N, \theta \in \mathcal{B}} = \frac{1}{NT} \sum_{t=1}^{T} \sum_{i=1}^N (y_{it} - x_{it}'\theta_{g_i})^2 .
\end{align*}

We prove a rate of consistency of the \emph{kmeans} estimator and the following assumptions that are weaker than the standard set of assumptions imposed in the literature \parencite[see, e.g.,][]{bonhomme2015grouped}.
Most importantly, the group separation condition is relaxed, allowing the difference between the slope coefficients associated with two groups to vanish asymptotically. In addition, we relax the conditions on the existence of moments and the mixing properties of the data.
 
 \begin{app_assumption}
  \label{assumption:weak-sep}
  \begin{enumerate}
    \item
    \label{assumption:weak-sep_asym-tail-w}
    Let $z_{it}$ be $x_{it}'x_{it}$, or $\left\Vert u_{it} x_{it} \right \Vert$. Assume the following holds for any choice of $z_{it}$: $z_{it} $ is a strictly stationary and strong mixing sequence over $t$ whose mixing coefficients $a_i [t]$ are  bounded by $a [ t]$ such that $\max_{1\leq i \leq N} a_i [t] \leq a[t]$ and $\sum_{t=0}^{\infty} (t+1)^{r/2 -1} a [t]^{b / r + b} <\infty $ for some $b>0$, and $\max_{1 \leq i \leq N} E(|z_{it}|^{r+b}) < \infty$ for some $b >0$.
   \item
   \label{assumption:weak-sep_beta-compact}
   $\mathcal{B}$ is compact. 
   \item \label{assumption:weak-sep_x-min-eigen-gp} Let $\rho_{N}(\gamma, g,\tilde g)$ be the minimum eigenvalue of 
   \begin{align*}
    \sum_{i=1}^N \sum_{t=1}^T \mathbf{1} \{g_{i}^0 = g \} \mathbf{1} \{g_{i} = \tilde g \} x_{it}x_{it}'/(NT),  
   \end{align*}
   where $\gamma = (g_{1}, \dots, g_{N})$.
   For any $g\in \mathbb{G}$,
   \begin{align*}
    \min_{\gamma \in (\mathbb{G})^N} \max_{\tilde g \in \mathbb{G}} \rho_{N} (\gamma, g, \tilde g) >\hat \rho,
   \end{align*}
   where $\hat \rho \to_p \rho $ as $N,T \to \infty$ and $\rho >0$. 
   \item 
   \label{assumption:weak-sep_x-min-eigen}
   There exists $\hat \rho^*$ such that for any $i$,
   \begin{align*}
   \lambda_{\min} \left( \frac{1}{T} \sum_{t=1}^T x_{it} x_{it}'\right) \geq \hat \rho^* 
   \end{align*}
   and $\hat \rho^* \to_p \rho^*>0$ as $N,T \to \infty$, where $\lambda_{\min}$ gives the minimum eigenvalue of its argument.
   \item 
   \label{assumption:weak-sep_group-separation-w}
   There exists a nonrandom sequence $c_T > T^{-1/2 + e}$ for some $e>0$ such that for any $g \neq \tilde h $ where $ g, h \in \mathbb{G}$, it holds that
   $\lVert \theta_{g}^0 -  \theta_{h}^0 \rVert > c_T$.
  \end{enumerate} 
 \end{app_assumption}

 Assumption \ref{assumption:weak-sep}.\ref{assumption:weak-sep_group-separation-w} is the key assumption that replaces the usual group separation assumption by weak group separation. Similar to the standard assumptions, any pair of groups must have distinct coefficients. In particular, the distance between their coefficients has to be bounded away from zero in any finite sample. We generalize the standard assumptions and allow the distance to vanish asymptotically. In the limit, groups are not separated. We assume that the rate at which group differences vanish is slower than $T^{-1/2}$. 
 
 The mixing and moment conditions in Assumption \ref{assumption:weak-sep}.\ref{assumption:weak-sep_asym-tail-w} are weaker than the standard assumptions imposed in the literature (see, e.g., \cite{bonhomme2015grouped}). However, we impose the additional assumption of strict stationarity. Under this assumption, we can relate the degree of weak group separation to a condition on the relative magnitudes of $N$ and $T$.
 
 All other assumptions are standard in the literature.

 The following theorem derives an asymptotic equivalence between the \emph{kmeans} estimator $\hat{\theta}$ and the oracle estimator $\tilde{\theta}$ under known group membership structure. The oracle estimator is trivially $\sqrt{NT}$-consistent. 
 
Unlike most existing results on the consistency of the \emph{kmeans} estimator in grouped panels, the theorem holds under weak group separation. However, the degree of group separation affects the required condition on the relative magnitudes of $N$ and $T$. The faster group separation converges to the limit of no group separation, the stronger the conditions on $N$ and $T$. In particular, when group separation is weak, $T$ must be large relative to $N$.

 \begin{app_theorem}\label{thm:weak-sep_gamma-beta-w}
  Suppose that Assumptions \ref{assumption:weak-sep}.\ref{assumption:weak-sep_asym-tail-w}-\ref{assumption:weak-sep}.\ref{assumption:weak-sep_group-separation-w} hold. 
  As $N,T \to \infty $ with $NT^{-er} \to 0$, where $e$ and $r$ are defined in Assumptions \ref{assumption:weak-sep}.\ref{assumption:weak-sep_asym-tail-w} and \ref{assumption:weak-sep}.\ref{assumption:weak-sep_group-separation-w}, respectively, it holds that 
  \begin{align*}
    \hat \theta = \tilde \theta + o_p (1/\sqrt{NT}).
  \end{align*}
 \end{app_theorem}
 
 Since the oracle estimator $\tilde \theta$ is $\sqrt{NT}$-consistent, Theorem~\ref{thm:weak-sep_gamma-beta-w} implies that $\hat \theta$ is $\sqrt{NT}$-consistent.
 
 \label{QE:discuss weak separation thm}
 Under the conditions of Theorem~\ref{thm:weak-sep_gamma-beta-w}, group consistency still holds (see Lemma~\ref{lem-group-consistency-w}). This is restrictive since the relevance of our testing problem relies on uncertainty about the group memberships even in the asymptotic limit. We leave a formal analysis that extends our results to settings with asymptotic misclassification to future research.

 Our results indicate that such an extension is feasible. 
 In our previous work in \textcite{dzemskiokui2021convergence}, we have shown that consistent estimation of group memberships is not a necessary condition for $\sqrt{NT}$-estimation of the group-specific coefficients under weak separation. We proved this result for the mean-shift model estimated from i.i.d. data. Theorem~\ref{thm:weak-sep_gamma-beta-w} extends some of our previous analysis to a grouped panel regression model with weakly dependent time series. To simplify the derivations and make our main point (robustness to weak separation) in a clear and transparent manner, we impose a uniform bound on the variance of the error term (see our Assumption~\ref{assumption:weak-sep}.\ref{assumption:weak-sep_asym-tail-w}). This bound implies group consistency. We conjecture that the uniform variance bound can be replaced by a set of more convoluted conditions (see condition (4) in \cite{dzemskiokui2021convergence}) that do not imply group consistency. We leave the details of this argument to future research.

Putting $\iota_N = c_T$, Theorem~\ref*{thm:MAX-dep-lv} in the main text yields
\begin{align}
    \label{eq:weak separation condition r_NT}
    T^{1 - e} r_{N, T} \to 0
\end{align}
as $N,T \to \infty$ as a necessary condition (ignoring a log term) for the validity of our confidence set. 
Here, $r_{N, T}$ is the rate of convergence of $(\hat{\theta}_1, \dotsc, \hat{\theta}_G)$. It indicates that the validity of our confidence set holds even when groups are only weakly separated as long as the cross-sectional sample is sufficiently large so that the group-specific coefficients converge sufficiently fast.

\subsection{Simulations\label{sec:weak-sep_simulations}}

We now report simulation evidence to study the finite sample effect of weak group separation on the rate of convergence of the \emph{kmeans} estimator and the validity of our joint confidence set for group membership. 

The simulation design is a simplified version of the design in Section~\ref*{sec:simulations} in the main text. The data generating process is given by 
\begin{align*}
	\widetilde{\texttt{lemp}}_{it} = \theta_{g^0_i, 1} \widetilde{\texttt{lmw}}_{it} + \theta_{g^0_i, 2} \widetilde{\texttt{lpop}}_{it} + \theta_{g^0_i, 3}\widetilde{\texttt{lemp}}_{it}^{\text{TOT}} + \sigma_i v_{it}
\end{align*}
for $i = 1, \dotsc, N$ and $t = 1, \dotsc, T$,  where $g^0_i$ is the group membership of unit $i$ and takes either the value one or two with equal probability. 
For $g = 1, 2$, the group specific-coefficient is equal to 
\begin{align*}
    \theta_g = (1 - 2 T^{-1/2 + e}) \frac{\bar{\theta}_1 + \bar{\theta}_2}{2} + 2 T^{-1/2 + e} \bar{\theta}_g,
\end{align*} 
where $(\bar{\theta}_1, \bar{\theta}_2)$ are estimated by fitting the model to the data from the empirical application using the \emph{kmeans} algorithm and setting the number of groups to $G = 2$. The data generating process for the covariates and error is the same as in Section~\ref*{sec:simulations} in the main text, setting $\rho = 0$ and $\sigma = 0.2$. We simulate designs with $N = 50, 100, 200$, $T = 30, 60, 120$ and $e = -0.25, 0, 0.25$. 
\begin{table}
    \centering 
    
\begin{tabular}{rrrrrrr}
\toprule
\multicolumn{1}{c}{ } & \multicolumn{3}{c}{$\lVert \theta_1^0 - \theta_2^0 \rVert$} & \multicolumn{3}{c}{$\lVert \theta^0 \rVert$} \\
\cmidrule(l{3pt}r{3pt}){2-4} \cmidrule(l{3pt}r{3pt}){5-7}
T & e=0.25 & e=0.00 & e=-0.25 & e=0.25 & e=0.00 & e=-0.25\\
\midrule
30 & 0.50 & 0.22 & 0.09 & 1.01 & 0.99 & 0.98\\
60 & 0.43 & 0.16 & 0.06 & 1.00 & 0.98 & 0.98\\
120 & 0.36 & 0.11 & 0.03 & 1.00 & 0.98 & 0.98\\
\bottomrule
\end{tabular}
    \caption{Group separation and $e$.}
    \label{tab:weak-sep_sim-local-coef-distance}
\end{table}
The parameter $e$ controls the rate at which the two group-specific coefficients converge to a common value as $T \to \infty$. Under $e = -0.25$, the two groups converge to a common group the fastest, and under $e = 0.25$, they converge the slowest. The case $e = 0.25$ is covered by the theoretical result in Theorem~\ref{thm:weak-sep_gamma-beta-w}. The case $e = 0$ is the infimum of the $e$ considered in Theorem~\ref{thm:weak-sep_gamma-beta-w}. Under $e = -0.25$, group separation vanishes at a rate too fast to be covered by Theorem~\ref{thm:weak-sep_gamma-beta-w}. Table~\ref{tab:weak-sep_sim-local-coef-distance} reports group separation between the two groups for the different choices of $e$. 

To measure the distance between two sets of group-specific slope coefficients $\theta = (\theta_1, \theta_2)$ and $\theta' = (\theta_1', \theta_2')$ we define
\begin{align*}
    \lVert \theta - \theta' \rVert = \sqrt{\sum_{g = 1}^2 \E \lVert \theta_g - \theta_g' \rVert_2^2}
\end{align*}
and 
\begin{align*}
    \lVert \theta \rVert = \sqrt{\sum_{g = 1}^2 \E \lVert \theta_g \rVert_2^2}
\end{align*}
and $\lVert \cdot \rVert_2$ is the $L_2$-norm.
Table~\ref{tab:weak-sep_sim-local-coef-distance} shows that, for $\theta^0 = (\theta_1^0, \theta_2^0)$, $\lVert \theta^0 \rVert$ is almost independent of $e$. 

We simulate the joint confidence for group membership using the variance estimator for the case of no serial correlation (i.e., setting the bandwidth equal to zero). For all simulations, the nominal level for the joint confidence set is set to $1 - \alpha = 0.95$, and the number of replications is $500$. 
The simulation results are summarized in Table~\ref{tab:weak-sep_sim-local}. 

The column ``coverage'' gives the simulated coverage probability of the joint confidence set for group membership. The coverage is always conservative for the designs with slowly vanishing group separation ($e = 0.25$). For the designs with group separation that vanishes at a moderate or fast rate ($e = 0$ and $e = -0.25$, respectively), the confidence set has appropriate or conservative coverage provided $N$ and $T$ are large enough. 

The columns labelled ``$\hat{\theta} - \theta$'' simulate the expected total error of \emph{kmeans} estimation $\lVert \hat{\theta} - \theta^0 \rVert / \lVert \theta^0 \rVert$, where $\hat{\theta} = (\hat{\theta}_1, \hat{\theta}_2)$ and the norm $\lVert \cdot \rVert$ is defined above. This error is relevant for assessing the finite sample validity of the assumptions we impose on coefficient estimation in Theorem~\ref*{thm:MAX-dep-lv} in the main text. When scaled by $T^{1/2}$, the error is approximately constant when increasing $T$ and leaving $N$ constant. This indicates that this is the rate at which time-series variation reveals information about the panel model. A smaller error can be achieved by using both time-series and cross-sectional variation, i.e., by increasing both $T$ and $N$. As discussed in Section~\ref{sec:weak-sep_theory}, a necessary condition for the asymptotic validity of our confidence set is 
\begin{align*}
    T^{1 - e} \lVert \hat{\theta} - \theta^0 \rVert \to 0
\end{align*}
as $N,T \to \infty$. Clearly, this condition cannot be met by increasing $T$ alone while keeping $N$ constant. However, the estimation error scaled by $T^{1-e}$ vanishes if $N$ increases, suggesting that estimation error from \emph{kmeans} estimation is negligible in panels with a large cross-sectional dimension. 

The columns labelled ``$\hat{\theta} - \tilde{\theta}$'' simulate $\lVert \tilde{\theta} - \hat{\theta} \rVert / \lVert \theta^0 \rVert$. They are the expected differences between the \emph{kmeans} estimator $\hat{\theta}$ and the oracle estimator $\tilde{\theta}$. Based on these columns, we assess the finite-sample relevance of our asymptotic result on \emph{kmeans} estimation under weak separation. Only the designs with $e = 0.25$ are covered by Theorem~\ref{thm:weak-sep_gamma-beta-w}. As predicted by the theorem, the difference between the \emph{kmeans} estimator and the oracle estimation vanishes at a faster rate than $T^{-1/2}$. We find the same result for $e = 0, -0.25$, two cases not covered by our asymptotic results. 

The columns labelled ``$\tilde{\theta} - \theta$'' report the simulated value of $\lVert \hat{\theta} - \theta^0 \rVert / \lVert \theta^0 \rVert$, i.e., the expected error of the oracle estimator.

The columns labeled ``$\hat{\theta} - \theta$'' are equivalent to differently scaled versions of the convergence rate of $\hat{\theta}$, i.e., of $r_{N,T}$ defined in Assumption~\ref*{assumption:basic}.\ref*{assum:est}. The column scaled by $T^{1 - e}$ checks the validity of condition \eqref{eq:weak separation condition r_NT} in finite samples. This condition is necessary for the validity of our joint confidence set under weak separation. For $e = 0.25$, Theorem~\ref{thm:weak-sep_gamma-beta-w} predicts that $T^{1-e} r_{N,T} \to \infty$ (i.e., condition \eqref{eq:weak separation condition r_NT} above) if $N$ is sufficiently large compared to $T$. The simulation evidence confirms this prediction. The settings with $e = 0, -0.25$ are not covered by the asymptotic analysis in Theorem~\ref{thm:weak-sep_gamma-beta-w}. Our simulation evidence suggests that the conditions for $T^{1-e} r_{N,T} \to 0$ are possibly even weaker under these settings.

\begin{table}
\scriptsize
\centering 

\begin{tabular}{rrrrrrrrrrrrr}
\toprule
\multicolumn{4}{c}{ } & \multicolumn{3}{c}{error} & \multicolumn{3}{c}{error scaled by $T^{1/2}$} & \multicolumn{3}{c}{error scaled by $T^{1-e}$} \\
\cmidrule(l{3pt}r{3pt}){5-7} \cmidrule(l{3pt}r{3pt}){8-10} \cmidrule(l{3pt}r{3pt}){11-13}
e & N & T & coverage & $\hat{\theta} - \theta$ & $\hat{\theta} - \tilde{\theta}$ & $\tilde{\theta} - \theta$ & $\hat{\theta} - \theta$ & $\hat{\theta} - \tilde{\theta}$ & $\tilde{\theta} - \theta$ & $\hat{\theta} - \theta$ & $\hat{\theta} - \tilde{\theta}$ & $\tilde{\theta} - \theta$\\
\midrule
 &  & 30 & 0.97 & 0.28 & 0.12 & 0.25 & 1.52 & 0.66 & 1.37 & 3.55 & 1.54 & 3.22\\

 &  & 60 & 0.99 & 0.20 & 0.05 & 0.19 & 1.52 & 0.35 & 1.48 & 4.23 & 0.99 & 4.11\\

 & \multirow[t]{-3}{*}{\raggedleft\arraybackslash 50} & 120 & 1.00 & 0.14 & 0.00 & 0.14 & 1.56 & 0.04 & 1.56 & 5.16 & 0.15 & 5.16\\
\cmidrule{2-13}
 &  & 30 & 0.97 & 0.20 & 0.10 & 0.17 & 1.09 & 0.53 & 0.95 & 2.54 & 1.23 & 2.23\\

 &  & 60 & 0.99 & 0.14 & 0.04 & 0.14 & 1.10 & 0.30 & 1.06 & 3.07 & 0.85 & 2.96\\

 & \multirow[t]{-3}{*}{\raggedleft\arraybackslash 100} & 120 & 1.00 & 0.10 & 0.01 & 0.10 & 1.09 & 0.06 & 1.09 & 3.61 & 0.21 & 3.61\\
\cmidrule{2-13}
 &  & 30 & 0.98 & 0.15 & 0.08 & 0.12 & 0.80 & 0.43 & 0.68 & 1.87 & 1.00 & 1.59\\

 &  & 60 & 0.99 & 0.10 & 0.03 & 0.10 & 0.76 & 0.23 & 0.76 & 2.12 & 0.64 & 2.11\\

\multirow[t]{-9}{*}{\raggedleft\arraybackslash 0.25} & \multirow[t]{-3}{*}{\raggedleft\arraybackslash 200} & 120 & 1.00 & 0.07 & 0.00 & 0.07 & 0.79 & 0.05 & 0.78 & 2.61 & 0.16 & 2.58\\
\cmidrule{1-13}
 &  & 30 & 0.84 & 0.39 & 0.24 & 0.26 & 2.16 & 1.33 & 1.41 & 2.16 & 1.33 & 1.41\\

 &  & 60 & 0.94 & 0.23 & 0.12 & 0.19 & 1.81 & 0.90 & 1.51 & 1.81 & 0.90 & 1.51\\

 & \multirow[t]{-3}{*}{\raggedleft\arraybackslash 50} & 120 & 0.99 & 0.15 & 0.03 & 0.14 & 1.61 & 0.33 & 1.58 & 1.61 & 0.33 & 1.58\\
\cmidrule{2-13}
 &  & 30 & 0.90 & 0.28 & 0.19 & 0.18 & 1.55 & 1.04 & 0.98 & 1.55 & 1.04 & 0.98\\

 &  & 60 & 0.96 & 0.17 & 0.09 & 0.14 & 1.35 & 0.71 & 1.09 & 1.35 & 0.71 & 1.09\\

 & \multirow[t]{-3}{*}{\raggedleft\arraybackslash 100} & 120 & 0.98 & 0.10 & 0.03 & 0.10 & 1.15 & 0.29 & 1.11 & 1.15 & 0.29 & 1.11\\
\cmidrule{2-13}
 &  & 30 & 0.91 & 0.23 & 0.17 & 0.13 & 1.26 & 0.94 & 0.70 & 1.26 & 0.94 & 0.70\\

 &  & 60 & 0.97 & 0.12 & 0.07 & 0.10 & 0.96 & 0.55 & 0.77 & 0.96 & 0.55 & 0.77\\

\multirow[t]{-9}{*}{\raggedleft\arraybackslash 0.00} & \multirow[t]{-3}{*}{\raggedleft\arraybackslash 200} & 120 & 0.99 & 0.07 & 0.02 & 0.07 & 0.82 & 0.21 & 0.79 & 0.82 & 0.21 & 0.79\\
\cmidrule{1-13}
 &  & 30 & 0.42 & 0.76 & 0.60 & 0.26 & 4.17 & 3.28 & 1.42 & 1.78 & 1.40 & 0.61\\

 &  & 60 & 0.58 & 0.42 & 0.30 & 0.19 & 3.23 & 2.31 & 1.51 & 1.16 & 0.83 & 0.54\\

 & \multirow[t]{-3}{*}{\raggedleft\arraybackslash 50} & 120 & 0.89 & 0.18 & 0.09 & 0.14 & 1.95 & 0.97 & 1.58 & 0.59 & 0.29 & 0.48\\
\cmidrule{2-13}
 &  & 30 & 0.46 & 0.55 & 0.45 & 0.18 & 3.02 & 2.44 & 0.98 & 1.29 & 1.04 & 0.42\\

 &  & 60 & 0.65 & 0.31 & 0.23 & 0.14 & 2.40 & 1.79 & 1.09 & 0.86 & 0.64 & 0.39\\

 & \multirow[t]{-3}{*}{\raggedleft\arraybackslash 100} & 120 & 0.93 & 0.13 & 0.07 & 0.10 & 1.40 & 0.77 & 1.11 & 0.42 & 0.23 & 0.33\\
\cmidrule{2-13}
 &  & 30 & 0.52 & 0.48 & 0.41 & 0.13 & 2.60 & 2.23 & 0.70 & 1.11 & 0.95 & 0.30\\

 &  & 60 & 0.68 & 0.24 & 0.20 & 0.10 & 1.88 & 1.52 & 0.77 & 0.67 & 0.55 & 0.28\\

\multirow[t]{-9}{*}{\raggedleft\arraybackslash -0.25} & \multirow[t]{-3}{*}{\raggedleft\arraybackslash 200} & 120 & 0.95 & 0.09 & 0.05 & 0.07 & 1.04 & 0.58 & 0.79 & 0.31 & 0.18 & 0.24\\
\bottomrule
\end{tabular}
\caption{Estimation error and confidence set coverage under shrinking group separation.}
\label{tab:weak-sep_sim-local}
\end{table}

\subsection{Proof of Theorem~\ref{thm:weak-sep_gamma-beta-w}\label{sec:weak-sep_proofs}}

Let 
\begin{align*}
Q(\gamma, \theta ) = \frac{1}{NT} \sum_{t=1}^{T} \sum_{i=1}^N (y_{it} - x_{it}'\theta_{g_i})^2 .
\end{align*}
Note that $Q(\gamma, \theta )$ is the objective function for the estimation.
We also define 
\begin{align*}
\tilde Q (\gamma, \theta) 
= 	\frac{1}{NT}	\sum_{t=1}^T \sum_{i=1}^N (x_{it}'(\theta_{g_i^0}^0 - \theta_{g_i}))^2 + \frac{1}{NT}	 \sum_{t=1}^{T} \sum_{i=1}^N  u_{it}^2 .
\end{align*}
The theorem follows from the following sequence of lemmas. 

\begin{lemma}
    \label{lem-qq-diff}
    Suppose that Assumptions~\ref{assumption:weak-sep}.\ref{assumption:weak-sep_asym-tail-w} and \ref{assumption:weak-sep}.\ref{assumption:weak-sep_beta-compact} hold. Then, 
    \begin{align*}
    \sup_{\gamma \in \mathbb{G}^N, \theta \in \mathcal{B}} \left|		\tilde Q (\gamma, \theta) -  Q (\gamma, \theta) \right| = O_p \left( \frac{1}{\sqrt{T}} \right). 
    \end{align*}
\end{lemma}

\begin{proof}
    The proof is almost identical to the proof of Lemma S.3 of \textcite{bonhomme2015grouped}.
    We have 
    \begin{align*}
    \tilde Q (\gamma, \theta) -  Q (\gamma, \theta)
    = & 
    -2 \frac{1}{NT}	 \sum_{t=1}^{T} \sum_{i=1}^N x_{it}'(\theta_{g_i^0}^0 - \theta_{g_i}) u_{it} .
    \end{align*}	
    We rewrite a part of the right-hand side as
    \begin{align*}
    \frac{1}{NT}	 \sum_{t=1}^{T} \sum_{i=1}^N x_{it}'\theta_{g_i^0}^0  u_{it}
    = 		\frac{1}{NT}	\sum_{g\in\mathbb{G}^B}  \sum_{t=1}^{T} \sum_{i=1}^N \mathbf{1} ( g_i(B) = g ) x_{it}'\theta_{g_i^0}^0  u_{it}.
    \end{align*}
    For each $g\in\mathbb{G}^B$, it holds that 
    \begin{align*}
    & E  \left( \frac{1}{NT}	 \sum_{t=1}^{T} \sum_{i=1}^N \mathbf{1} ( g_i(B) = g ) x_{it}'\theta_{g_i^0}^0  u_{it} \right)^2 
    \le   C 		E \left\Vert  \frac{1}{NT}	 \sum_{t=1}^{T} \sum_{ g_i(B) = g}  x_{it}  u_{it} \right\Vert^2 
    = O \left( \frac{1}{NT}\right),
    \end{align*}
    where the inequality is the Cauchy-Schwarz inequality with $C$ satisfying $\left\Vert \theta_{g_i} \right\Vert^2 < C $ for any $\theta \in \mathcal{B}$ (by Assumption~\ref{assumption:weak-sep}.\ref{assumption:weak-sep_beta-compact}), and the equality follows 
    since Theorem 1 of \textcite{Yokoyama1980} implies that under Assumption \ref{assumption:weak-sep}.\ref{assumption:weak-sep_asym-tail-w}, for any $L \subseteq \{1, \dots,N\} $, there exists $M$ which does not depend on $L$ such that 
   \begin{align}
    \label{eq:weak-sep_xu-square-order}
   E \left( \left\Vert  \frac{1}{NT}  \sum_{t=1}^{T} \sum_{ i \in L } x_{it}  u_{it} \right\Vert^2 \right) \leq M \frac{|L|}{N^2T} ,
   \end{align}
   where $|L|$ denotes the cardinality of $L$.
 We then examine the other part of $\tilde Q (\gamma, \theta) -  Q (\gamma, \theta)$. It follows that 
    \begin{align*}
    \left( \frac{1}{NT}	 \sum_{t=1}^{T} \sum_{i=1}^N x_{it}'\theta_{g_i}  u_{it} \right)^2
    \leq & \left( \frac{1}{NT}	\sum_{i=1}^N \theta_{g_i} \sum_{t=1}^{T}  x_{it}  u_{it} \right)^2 \\
    \leq & \left( \frac{1}{N}	\sum_{i=1}^N || \theta_{g_i} ||^2 \right) \left( \frac{1}{NT^2}\sum_{i=1}^N \left\Vert \sum_{t=1}^{T}  x_{it}  u_{it} \right\Vert^2 \right) \\
    =& O_p \left( \frac{1}{T}\right),
    \end{align*}
    where the first inequality follows by the Cauchy-Schwarz inequality and the second inequality follows by Assumption \ref{assumption:weak-sep}.\ref{assumption:weak-sep_beta-compact} and the Markov inequality by \eqref{eq:weak-sep_xu-square-order}.
    We thus have 
    \begin{align*}
    \tilde Q ( \gamma, \theta) -  Q ( \gamma, \theta) 
    = &  O \left( \frac{1}{\sqrt{NT}}\right) +  O \left( \frac{1}{\sqrt{T}}\right) 
    \end{align*}
    uniformly over $\theta$ and $\gamma$, and consequently
    \begin{align*}
    \sup_{\gamma \in \mathbb{G}, \theta \in \mathcal{B}} \left|		\tilde Q ( \gamma, \theta) -  Q (\gamma, \theta) \right| = O_p \left(\frac{1}{\sqrt{T}} \right). 
    \end{align*}
    
\end{proof}

\begin{lemma} \label{lem-rate-conv-w}
    Suppose that Assumptions~\ref{assumption:weak-sep_asym-tail-w}-\ref{assumption:weak-sep_x-min-eigen-gp} hold. 
    Then, 
    \begin{align*}
        \max_{g \in \mathbb{G} } \min_{\tilde g \in \mathbb{G}} \left\Vert \theta_{g}^0 - \hat \theta_{\tilde g} \right\Vert^2  = O_p(1/\sqrt{T}).        
    \end{align*}
\end{lemma}

\begin{proof}
The proof is almost identical to the proofs of Lemmas A.2 and B.3 of \textcite{bonhomme2015grouped}. 
    Lemma \ref{lem-qq-diff} implies
    \begin{align*}
    \tilde{Q} ( \hat \gamma, \hat \theta) =& Q ( \hat \gamma, \hat \theta) + O_p \left(\frac{1}{\sqrt{T}} \right) \\
    \le & Q ( \gamma^0, \theta^0) +O_p \left(\frac{1}{\sqrt{T}} \right) =  \tilde{Q} ( \gamma^0, \theta^0) + O_p \left(\frac{1}{\sqrt{T}} \right).
    \end{align*}
    The fact that $\tilde{Q} ( \gamma, \theta) $ is minimized at $( \gamma^0, \theta^0)$ implies 
    \begin{align*}
    \tilde{Q} ( \hat \gamma, \hat \theta) -\tilde{Q} ( \gamma^0, \theta^0) = O_p \left(\frac{1}{\sqrt{T}} \right).
    \end{align*}
    
    We now establish a lower bound of $\tilde Q (\gamma, \theta) -\tilde{Q} ( \gamma^0, \theta^0)$ such that 
    \begin{align*}
    &\tilde Q (\gamma, \theta) -\tilde{Q} ( \gamma^0, \theta^0) \\
    = & \frac{1}{NT}	 \sum_{t=1}^{T} \sum_{i=1}^N (x_{it}'(\theta_{g_i^0}^0 - \theta_{g_i}))^2 \\
    = & \frac{1}{NT}	 \sum_{t=1}^{T} \sum_{g=1}^{G} \sum_{\tilde g=1}^{G} \sum_{i=1}^N \mathbf{1} \{g_i^0 = g \} \{g_i = \tilde g \} (x_{it}'(\theta_{g}^0 - \theta_{ \tilde g}))^2 \\
    \ge &  \frac{1}{T}	 \sum_{t=1}^{T} \sum_{g=1}^{G} \sum_{\tilde g=1}^{G} \rho_{N}(\gamma, g,\tilde g) \left\Vert \theta_{g}^0 - \theta_{ \tilde g} \right\Vert^2 	\\
    \ge &  	 \hat \rho  G^2 \max_{g\in \mathbb{G}} \min_{\tilde g \in \mathbb{G}^B} \left\Vert \theta_{g}^0 - \theta_{ \tilde g } \right\Vert^2,
    \end{align*}
    where the first inequality follows by the definition of $\rho_{N}(\gamma, g,\tilde g) $ and the second inequality is from the definition of $\hat \rho$. Now, Assumption \ref{assumption:weak-sep}.\ref{assumption:weak-sep_x-min-eigen} implies 
    \begin{align*}
        \max_{g \in \mathbb{G} } \min_{\tilde g \in \mathbb{G}} \left\Vert \theta_{g}^0 - \hat \theta_{\tilde g} \right\Vert^2 = O_p\left( \frac{1}{\sqrt{T}}\right).
    \end{align*}
\end{proof}

\begin{lemma}\label{lem-coef-rate-w}
    Suppose that Assumptions \ref{assumption:weak-sep}.\ref{assumption:weak-sep_asym-tail-w}-\ref{assumption:weak-sep}.\ref{assumption:weak-sep_x-min-eigen-gp}, and \ref{assumption:weak-sep}.\ref{assumption:weak-sep_group-separation-w} are satisfied.
    Then there exist a permutation $\sigma: \mathbb{G} \mapsto \mathbb{G}$ such that $  \left\Vert \theta_{g}^0 - \hat \theta_{\sigma (g)} \right\Vert^2 = O_p (1/\sqrt{T})$ for any $g \in \mathbb{G}$.
\end{lemma}

\begin{proof}

    We construct a permutation with the property stated in the lemma. Indeed, we show that 
    \begin{align*}
        \sigma (g) =  \arg \min_{\tilde g \in \mathbb{G}} \left\Vert \theta_{g}^0 - \hat \theta_{\tilde g} \right\Vert^2
    \end{align*}
    is such a permutation. We first show that $\sigma$ satisfies $  \left\Vert \theta_{g}^0 - \hat \theta_{\sigma (g)} \right\Vert^2 = O_p (1/\sqrt{T})$ for any $g \in \mathbb{G}$, and that it is a permutation.
    
    Lemma \ref{lem-rate-conv-w} states that 
    \begin{align*}
    \max_{g \in \mathbb{G} } \min_{\tilde g \in \mathbb{G}} \left\Vert \theta_{g}^0 - \hat \theta_{\tilde g} \right\Vert^2  = O_p(1/\sqrt{T}).
    \end{align*}
    The map $\sigma$, by construction, satisfies $ \left\Vert \theta_{g}^0 - \hat \theta_{\sigma (g)} \right\Vert^2 = O_p (1/\sqrt{T})$ for any $g \in \mathbb{G}$.

    It remains to establish that $\sigma$ is a permutation.
    For $g \neq \tilde g$, the triangular inequality gives 
    \begin{align*}
    \left\Vert \hat \theta_{\sigma (g)} - \hat \theta_{\sigma (\tilde g )} \right\Vert \geq  \left\Vert \theta_{g}^0 -  \theta_{\tilde g}^0 \right\Vert -  \left\Vert \theta_{g}^0 - \hat \theta_{\sigma (g)} \right\Vert - \left\Vert \theta_{\tilde g}^0 - \hat \theta_{\sigma (\tilde g)} \right\Vert.
    \end{align*}
    In the above we have seen that $ \left\Vert \theta_{g}^0 - \hat \theta_{\sigma (g)} \right\Vert = O_p(1/\sqrt{T})$ and $ \left\Vert \theta_{\tilde g}^0 - \hat \theta_{\sigma(\tilde g)} \right\Vert = O_p(1/\sqrt{T})$. Assumption~\ref{assumption:weak-sep_group-separation-w} states that $\left\Vert \theta_{g}^0 -  \theta_{\tilde g}^0 \right\Vert >c_T$. The condition that $c_T > T^{-1/2 +e}$ implies that $\left\Vert \theta_{g}^0 -  \theta_{\tilde g}^0 \right\Vert -  \left\Vert \theta_{g}^0 - \hat \theta_{\sigma (g)} \right\Vert - \left\Vert \theta_{\tilde g}^0 - \hat \theta_{\sigma (\tilde g)} \right\Vert >0$ with probability approaching one. This means that $\sigma (g) \neq \sigma (\tilde g)$ for $g\neq \tilde g$ with probability approaching one. Thus, $\sigma$ possesses a well-defined inverse and is bijective; in other words, $\sigma$ is a permutation. 
\end{proof}

From Lemmas \ref{lem-rate-conv-w} and \ref{lem-coef-rate-w}, we observe that the Hausdorff distance between $\theta^0$ and $ \hat \theta$ converges to 0 at the rate of $\sqrt{T}$. By using the labelling such that $\sigma(g) =g$, we can write $\left\Vert \theta_{g}^0 - \hat \theta_{g} \right\Vert^2 = O_p (1/\sqrt{T})$ for any $g \in \mathbb{G}$.

We then establish that the group membership structure is correct asymptotically as long as the coefficients are in a neighborhood of the true value. 	
Let $\mathcal{N} = \{ \theta: \left \Vert \theta_{g}^0 -\theta_{g} \right \Vert < \eta = T^{-1/2+f}, \forall g \in \mathbb{G} \}$ for $0< f < e$, where $e$ is defined in Assumption \ref{assumption:weak-sep}.\ref{assumption:weak-sep_group-separation-w}, for any $g \in \mathbb{G}$.

\begin{lemma}\label{lem-group-consistency-w}
    Suppose that Assumptions~\ref{assumption:weak-sep}.\ref{assumption:weak-sep_asym-tail-w}, \ref{assumption:weak-sep}.\ref{assumption:weak-sep_beta-compact}, \ref{assumption:weak-sep}.\ref{assumption:weak-sep_x-min-eigen} and \ref{assumption:weak-sep}.\ref{assumption:weak-sep_group-separation-w} hold. As $N,T\to \infty $ with $NT^{-er} \to 0$, where $e$ and $r$ are defined in Assumptions \ref{assumption:weak-sep}.\ref{assumption:weak-sep_asym-tail-w} and \ref{assumption:weak-sep}.\ref{assumption:weak-sep_group-separation-w}, respectively, it holds that 
    \begin{align*}
    \Pr \left\{ \hat \gamma (\theta)  \neq \gamma^0 \text{ \emph{for some} }  \theta \in \mathcal{N}\right\} \to 0.
    \end{align*}
\end{lemma}
\begin{proof}
    We establish an equivalent statement:
    \begin{align*}
    \max_{1 \leq i \leq N} \sup_{\theta \in \mathcal{N}} \mathbf{1} \{\hat g_i (\theta) \neq g_i^0\} = o_p(1).
    \end{align*}
    Note that 
    \begin{align*}
    \mathbf{1}\left\{\hat g_i (\theta) \neq g_i^0 \right\} 
    = \max_{g \in \mathbb{G} \backslash \{g_i^0 \}} \mathbf{1} \left( \sum_{t=1}^{T} (y_{it} - x_{it}'\theta_{g})^2 < \sum_{t=1}^{T} (y_{it} - x_{it}'\theta_{g_i^0})^2 \right).\label{eq:indicator-fun-ext}
    \end{align*}
    We have 
    \begin{align*}
    & \sum_{t=1}^{T}  \left( (y_{it} - x_{it}'\theta_{g})^2 - (y_{it} - x_{it}'\theta_{g_i^0})^2 \right) \\
    = &  \sum_{t=1}^{T} 2 u_{it} x_{it} (\theta_{g_i^0}^0 - \theta_{g}^0) 
    + \sum_{t=1}^{T} (x_{it}'( \theta_{g_i^0}^0  -\theta_{g}^0))^2 + \Psi,
    \end{align*}
    where
    \begin{align*}
    \Psi=&  \sum_{t=1}^{T} 2 u_{it} x_{it} (\theta_{g_i^0} - \theta_{g} - \theta_{g_i^0}^0 + \theta_{g}^0) \\
    & + \sum_{t=1}^{T}  (\theta_{g_i^0} - \theta_{g} - \theta_{g_i^0}^0 + \theta_{g}^0)'x_{it} x_{it}'(2 \theta_{g_i^0}^0 - \theta_{g_i^0} -\theta_{g})\\
    & + \sum_{t=1}^{T}  (\theta_{g_i^0}^0 - \theta_{g}^0)'x_{it} x_{it}'( \theta_{g_i^0}^0 - \theta_{g_i^0} -\theta_{g}  + \theta_{g}^0 ) 	.
    \end{align*}
    Applying the Cauchy-Schwarz inequality and then Assumption \ref{assumption:weak-sep}.\ref{assumption:weak-sep_beta-compact} and the definition of $\mathcal{N}$ gives 
    \begin{align*}
    |\Psi| \leq \eta C_1 \left\Vert \sum_{t=1}^{T} u_{it} x_{it}\right\Vert + \eta C_2 \left\Vert \sum_{t=1}^{T} x_{it} x_{it}' \right\Vert ,
    \end{align*}
    where $C_1$ and $C_2$ are constants independent of $\eta$ and $T$.
    We thus have the following inequality
    \begin{align*}
    &\mathbf{1} \left( \sum_{t=1}^{T} (y_{it} - x_{it}'\theta_{g})^2 < \sum_{t=1}^{k-1} (y_{it} - x_{it}'\theta_{g_i^0})^2 \right) \\
    \leq &  \mathbf{1} \Bigg(  \sum_{t=1}^{k^0-1} 2 u_{it} x_{it}' (\theta_{g_i^0}^0 - \theta_{g}^0) \\
    & \quad \quad -  \sum_{t=1}^{T}   (x_{it}'( \theta_{g_i^0}^0  -\theta_{g, B}^0))^2   
    + \eta  C_1 \left\Vert \sum_{t=1}^{T} u_{it} x_{it}\right\Vert + \eta C_2 \left\Vert \sum_{t=1}^{T} x_{it} x_{it}' \right\Vert  \Bigg).
    \end{align*}
    Thus, we have 
    \begin{align*}
    &\Pr \left( \sup_{\theta \in \mathcal{N}} \mathbf{1} (\hat g_i  (\theta) \neq g_i^0) \neq 0 \right) \\
    \leq &  \sum_{g \in \mathbb{G} \backslash  \{g_i^0\}} \Bigg( \Pr \left( \frac{1}{T}\sum_{t=1}^{T}   (x_{it}'( \theta_{g_i^0}^0  -\theta_{g}^0))^2   \leq  \frac{c_T''}{2}  \right)  
    + \Pr \left(\left\Vert \frac{1}{T} \sum_{t=1}^{T} u_{it} x_{it}\right\Vert  \geq M \right) \\ 
    & \quad + \Pr \left(  \left\Vert \frac{1}{T} \sum_{t=1}^{T} x_{it} x_{it}' \right\Vert \geq M \right)  \\
    & \quad +  \Pr \Bigg(  \frac{1}{T} \sum_{t=1}^{T} 2 u_{it} x_{it}' (\theta_{g_i^0}^0 - \theta_{g}^0) 
    < - \frac{c_T''}{2} + \eta  C_1 M  + \eta C_2 M  \Bigg) \Bigg),
    \end{align*}
    where we take $c_T'' = c_T \times \rho^*$ for $c$ in Assumption \ref{assumption:weak-sep}.\ref{assumption:weak-sep_group-separation-w} and $\rho^*$ in Assumption \ref{assumption:weak-sep}.\ref{assumption:weak-sep_x-min-eigen} and $M$ is some large constant.
    
    We now bound the second and third terms on the right-hand side of the inequality.
    First, we have 
    \begin{align*}
    \Pr \left(  \left\Vert \frac{1}{T} \sum_{t=1}^{T} x_{it} x_{it}' \right\Vert \geq M \right)  \leq \Pr \left(   \frac{1}{T} \sum_{t=1}^{T} \left\Vert x_{it} x_{it}' \right\Vert \geq M \right) = \Pr \left(   \frac{1}{T} \sum_{t=1}^{T} x_{it}' x_{it} \geq M \right).
    \end{align*}
    Assumption \ref{assumption:weak-sep}.\ref{assumption:weak-sep_asym-tail-w} enables us to the Markov inequality and Theorem 1 of \textcite{Yokoyama1980} to $ x_{it}' x_{it} - E( x_{it}' x_{it})$, and we establish 
    \begin{align*}
    \Pr \left(  \left\Vert \frac{1}{T} \sum_{t=1}^{T} x_{it} x_{it}' \right\Vert \geq M \right)  =  O\left( T^{-r/2}\right),
    \end{align*}
    by taking $M$ large enough such that $\sum_{t=1}^T E( x_{it}' x_{it}) /T < M$.
    A similar argument under Assumption \ref{assumption:weak-sep}.\ref{assumption:weak-sep_asym-tail-w} implies that 
    $		\Pr \left(\left\Vert (T )^{-1} \sum_{t=1}^{T} u_{it} x_{it}\right\Vert  \geq M \right) = O\left( T^{-r/2}\right)$.

    Next, we consider the first term. We now use Assumptions \ref{assumption:weak-sep}.\ref{assumption:weak-sep_x-min-eigen}, \ref{assumption:weak-sep}.\ref{assumption:weak-sep_group-separation-w} and \ref{assumption:weak-sep}.\ref{assumption:weak-sep_asym-tail-w}. The Markov inequality combined with Theorem 1 of \textcite{Yokoyama1980} implies
    \begin{align*}
    \Pr \left( \left| \frac{1}{T}\sum_{t=1}^{T}   (x_{it}'( \theta_{g_i^0}^0  -\theta_{g}^0))^2 - \frac{1}{T}\sum_{t=1}^{T}   E ((x_{it}'( \theta_{g_i^0}^0  -\theta_{g}^0))^2 ) \right| \geq  \frac{c_T''}{2}  \right)  = O( T^{-er}) .
    \end{align*}
    We thus have uniformly over $g$:
    \begin{align*}
    \Pr \left( \frac{1}{T}\sum_{t=1}^{T}   (x_{it}'( \theta_{g_i^0}^0  -\theta_{g}^0))^2   \geq  \frac{c_T''}{2} \right)  = O (T^{-er}).
    \end{align*}

    Lastly, we consider the fourth term. It follows that 
    \begin{align*}
        & \Pr \Bigg(  \frac{1}{T} \sum_{t=1}^{T} 2 u_{it} x_{it}' (\theta_{g_i^0}^0 - \theta_{g}^0) 
        < - \frac{c_T''}{2} + \eta  C_1 M  + \eta C_2 M  \Bigg) \\
        \leq &  \Pr \Bigg(  \frac{1}{T} \sum_{t=1}^{T} 2 u_{it} x_{it}' (\theta_{g_i^0}^0 - \theta_{g}^0) 
        < - \frac{c_T''}{4}   \Bigg)= O (T^{-er})
        \end{align*}
        uniformly over $g$ under Assumption \ref{assumption:weak-sep}.\ref{assumption:weak-sep_asym-tail-w}. 
        The inequality follows by $c_T'' = O(c_T) = O( T^{-1/2+e})$ and $\eta = o (T^{-1/2 + e}) $. The equality holds by the Markov inequality and Theorem 1 of \textcite{Yokoyama1980}.

To sum up, we have 
    \begin{align*}
    & \Pr \left(\max_{1 \leq i \leq N} \sup_{\theta \in \mathcal{N}}  \mathbf{1} (\hat g_i  (\theta) \neq g_i^0) \neq 0 \right)  
    \leq  \sum_{i=1}^N \Pr \left( \sup_{\theta \in \mathcal{N}} \mathbf{1} (\hat g_i (\theta) \neq g_i^0) \neq 0 \right) \\  
    =&  O (N ( T^{-er} + T^{-r/2} )) 
    =  O(NT^{-er} ).
    \end{align*}
    
\end{proof}

\printbibliography